\documentclass[a4paper,12pt]{report}

\usepackage{mathrsfs}
\usepackage{graphicx}
\usepackage{bm}
\usepackage{amssymb}
\usepackage{amsmath}
\usepackage{ascmac}
\usepackage{array}

\usepackage{authblk}

\usepackage{cite}

\usepackage{url}

\usepackage{theorem} 

\theorembodyfont{\rmfamily}

\newtheorem{theorem}{Theorem}[chapter]
\newtheorem{proposition}{Proposition}[chapter]
\newtheorem{lemma}{Lemma}[chapter]
\newtheorem{corollary}{Corollary}[chapter]
\newtheorem{definition}{Definition}[chapter]
\newtheorem{proof}{Proof.}

\begin{document}


\title{Entropy, Divergence, and Majorization in Classical and Quantum Thermodynamics}

\author{Takahiro Sagawa
\footnote{\scriptsize Department of Applied Physics and Quantum-Phase Electronics Center (QPEC), The University of Tokyo, Tokyo 113-8656, Japan}
}


\date{}

\maketitle


\section*{Preface}

\
\

In these decades, it has been revealed that there is rich information-theoretic structure in thermodynamics of out-of-equilibrium systems in both the classical and quantum regimes.
This has led to  the fruitful interplay among statistical physics, quantum information theory, and mathematical theories including matrix analysis and asymptotic probability theory.

The main purpose of this book is to clarify how information theory works behind  thermodynamics and to shed modern light on it.
We focus  on both of purely information-theoretic concepts and their physical implications:
We  present self-contained and rigorous proofs of several fundamental properties of entropies, divergences, and majorization.
We also discuss the modern formulations of thermodynamics, especially from the perspectives of stochastic thermodynamics and resource theory of thermodynamics.
Resource theory is a recently-developed field as a branch of quantum information theory in order to quantify (energetically or information-theoretically) ``useful resources.''  We show that resource theory has an intrinsic connection to various fundamental ideas of mathematics and information theory.

This book is not intended to be a comprehensive review of the field, but would serve as a concise introduction to several important ingredients of  the information-theoretic formulation of thermodynamics.
We hope that the readers would grasp a compact overview on physics and  mathematics of entropy-like quantities from the modern point of view.

\vspace{\baselineskip}
\begin{flushright}\noindent
Tokyo, July 2020\hfill {\it Takahiro Sagawa}\\
\end{flushright}

\newpage

\section*{Acknowledgements}

\

This book is partially based on lectures by the author at Kyoto University in June 2018 and at Osaka University in January 2020.
The author is grateful to the hosts of these lectures: Hideo Suganuma and Hidenori Fukaya.

The author is grateful to Yosuke Mitsuhashi and Ryuji Takagi for careful reading of the draft, and to Hiroshi Nagaoka, Frederik vom Ende, Hiroyasu Tajima, Naoto Shiraishi,  Kosuke Ito, James P. Crutchfield,
and Nicole Yunger Halpern for valuable comments.
The author is also grateful to the students, Kosuke Kumasaki, Taro Sawada, Hiroshi Kubota, Thanaporn Sichanugrist, and Takahiro Uto, for reading the draft.

The author is grateful to the editors of Springer, Chino Hasebe and  Masayuki Nakamura, for their patience for the long delay of my manuscript.

Finally, the author is most grateful to Tohru Eguchi, who passed away in 2019, for inviting me to this book project of Springer Briefs already several years ago.

\
\

This work is supported by JSPS KAKENHI Grant Numbers JP16H02211 and JP19H05796.


\tableofcontents


\chapter{Introduction}
\label{chap:introduction}

Thermodynamics was originally  established  as a phenomenological theory of thermal equilibrium~\cite{Callen}.
Entropy is at the core of  thermodynamics, as it provides a complete characterization of state convertibility through the second law of thermodynamics; 
In the absence of a heat bath,
a state transformation between equilibrium states is possible \textit{if and only if} the entropy does not decrease.
This remarkable feature of entropy has been formalized by a  rigorous axiomatic theory at the purely phenomenological level by Lieb and Yngvason~\cite{Lieb1999}.

Once one goes beyond conventional equilibrium situations, however, the notion of entropy becomes more subtle.
In fact,  a proper definition of entropy of nonequilibrium states has been a long-standing problem in  statistical physics.
In light of progress in theories of modern thermodynamics, it has been revealed that information-theoretic entropies play  significant roles in thermodynamics of out-of-equilibrium and even quantum systems.

There is a long history of researches of such entropy-like quantities in classical and quantum information theories~\cite{Cover_Thomas,Nielsen,Ohya,Hayashi} as well as matrix analysis~\cite{Bhatia,Hiai2010}, where it has been shown that several information-theoretic quantities exhibit universal features that resemble the second law of thermodynamics.
A particularly important concept is \textit{divergence} as well as information-theoretic \textit{entropy}.  Divergence is also referred to as relative entropy and corresponds to free energy of thermodynamics.

In this book, we will discuss the fundamental properties of entropies and divergences both in the classical and quantum regimes, including the Shannon entropy~\cite{Shannon}, the von Neumann entropy~\cite{Neumann}, the Kullback-Leibler (KL) divergence~\cite{Kullback} and its quantum generalization~\cite{Umegaki}, the R\'enyi entropy and divergence~\cite{Renyi},  the $f$-divergence~\cite{Ali,Csiszar}, and general quantum divergence-like quantities called the Petz's quasi-entropies~\cite{Petz1985,Petz1986}.
In addition, we will briefly discuss the classical and quantum Fisher information~\cite{Fisher,Petz1996,Amari2000}, which has a close connection to divergences.

We will see that the monotonicity properties of these divergences under stochastic or quantum dynamics bring us to the information-theoretic foundation of the second law of thermodynamics.
An advantage of such an approach to the second law lies in the fact that 
information-theoretic entropies and divergences can be defined for arbitrary probability distributions and quantum states including out-of-equilibrium ones,
and therefore informational quantities can take the place of the Boltzmann entropy defined only for the equilibrium ensembles.

From the physics side, recent progress of experimental technologies of manipulating small-scale systems has led us to a fruitful playground of modern thermodynamics.
Small-scale heat engines have been experimentally realized with various systems both in the classical and quantum regimes, such as colloidal particles~\cite{Toyabe2010,Berut2012}, biomolecules~\cite{Crivellari2019}, single electrons~\cite{Koski2014}, superconducting qubits~\cite{Masuyama2018}, and NMR~\cite{Camati2016},
 where the connection between the second law and information theory has been investigated in real laboratories.

In contrast to the case of conventional macroscopic systems, thermodynamic quantities  of small-scale systems 
become random  variables, because dynamics of small systems exhibit intrinsic stochasticity induced by thermal fluctuations of heat baths. 
This is a fundamental reason why modern thermodynamics is relevant to information theory based on probability theory.

There are two main complementary streams that deal with theoretical formulations of small-scale thermodynamics in the above-mentioned spirit.
One is  \textit{stochastic thermodynamics}, which has been developed in the field of nonequilibrium statistical mechanics~\cite{Evans1993,Jarzynski1997,Crooks1999,Seifert2005} (see also review articles~\cite{Jarzynski2011,Seifert2012,Esposito2009,Sagawa2012,Funo2018}).
Stochastic thermodynamics has led to modern understanding of thermodynamics of information~\cite{Sagawa2012a,Parrondo2015,Sagawa2019r}, which sheds new light on ``Maxwell's demon''~\cite{Demon}.

The other stream is \textit{resource theory of thermodynamics}, which has been  developed more recently as a branch of quantum information theory~\cite{Brandao2013,Horodecki2013,Aberg2013,Brandao2015,Weilenmann2016,Faist2018} (see also review articles~\cite{Gour,Goold,Lostaglio2019}).
From this perspective, the second law of thermodynamics quantifies how much ``resource'' such as work is required for a desired thermodynamic task.
One of the key ideas 
 is that work is supposed to be a deterministic, not random quantity, in order to implement the idea that work is a purely ``mechanical'' quantity without any entropic contribution.
This is referred to as \textit{single-shot} (or one-shot) thermodynamics~\cite{Horodecki2013,Aberg2013} and is contrastive to the setup of stochastic thermodynamics that allows work fluctuations.
We note that there are attempts to connect resource-theoretic results with experimentally relevant setups (e.g., Refs.~\cite{Alhambra2019,Halpern2020}).

In general, resource theory is an information-theoretic framework to quantify ``useful resources'' (see also a review article~\cite{Chitambar2019}).
For example, resource theory of entanglement is an earliest resource theory in quantum information theory~\cite{Nielsen,Nielsen1999}, and has a similar mathematical structure to resource theory of thermodynamics at infinite temperature.
To formulate a resource theory,  we need to identify \textit{free states} and \textit{free operations}, which can be prepared and performed without any cost.
It is also important to consider a \textit{monotone}, which is a quantity that monotonically changes (does not increase or does not decrease) under free operations~\cite{Takagi2019}.
In particular, there is a concept called a \textit{complete monotone}, which provides a sufficient, not only necessary, condition that an operation is possible.
We remark that the concept of single-shot  is  significant in various resource theories~\cite{Liu2019}.

In the case of thermodynamics, work and nonequilibrium states are regarded as resources, because work can be extracted only from nonequilibrium states and nonequilibrium states can be created only by investing work.
On the other hand, Gibbs states are free states and relaxation processes (called Gibbs-preserving maps or thermal operations) are free operations.

Conventional thermodynamics can be regarded as a prototypical resource theory, where entropy is a complete monotone that provides a necessary and sufficient condition of transitions between equilibrium states~\cite{Lieb1999}.
In modern nonequilibrium thermodynamics, the divergences (including the KL divergence and the R\'enyi divergence) are monotones, implying that these divergences serve as thermodynamic potentials (or free energies).
However, except for equilibrium transitions, such divergences are not complete monotones in the single-shot scenario.
This brings us to a mathematical concept called \textit{majorization}~\cite{Bhatia,Marshall}, which enables a complete  characterization of thermodynamic state transformations.
Majorization plays a central role in resource theory of entanglement, while a generalized concept called \textit{thermo-majorization}~\cite{Horodecki2013} (and \textit{d-majorization}~\cite{Blackwell,Torgersen,Ruch}) is crucial for resource theory of thermodynamics at finite temperature.

It is also interesting to take the asymptotic limit in resource theories, where many copies of the system are  available.
In  the asymptotic limit, structure of state convertibility often becomes  simpler than standard majorization and can be characterized by a single complete monotone.
This is a main focus of the theory of \textit{information spectrum}~\cite{Han1993,Han2003,Nagaoka2007,Bowen_Datta,Bowen_Datta2,Datta_Renner,Datta2009}, where a quantum version of \textit{asymptotic equipartition properties (AEP)}~\cite{Cover_Thomas} plays a significant role.
To take the asymptotic limit, it is useful to utilize quantities called \textit{smooth} entropy and divergence~\cite{Renner2005,Renner2004}.
From the physics point of view, the asymptotic limit represents the thermodynamic limit of many-body systems, and a complete monotone is regarded as a complete macroscopic thermodynamic potential like the one discussed by Lieb and Yngvason~\cite{Lieb1999}.
We will emphasize that the asymptotic theory of information spectrum can be applied  to interacting many-body systems beyond the independent and identically-distributed (i.i.d.) situations.

\

The organization of this book is as follows.

In Chapter \ref{chap:classical_entropy}, we discuss the  properties of classical information-theoretic entropies and divergences: especially,  the Shannon entropy, the KL divergence, and 
 the R\'enyi $\alpha$-entropies and divergences.
We also briefly mention the $f$-divergence and the Fisher information.

In Chapter \ref{chap:classical_majorization}, we review majorization and thermo-majorization (and d-majorization) for classical stochastic systems.
We introduce the Lorenz curve to visualize majorization and show that the R\'enyi $0$- and $\infty$-divergences provide a useful characterization of state convertibility.

In Chapter \ref{chap:classical_thermodynamics}, we apply the foregoing information-theoretic argument (Chapter \ref{chap:classical_entropy} and Chapter \ref{chap:classical_majorization}) to classical thermodynamics and show that the second law immediately follows from general information-theoretic inequalities.
Specifically, the KL divergence is relevant to the setup of stochastic thermodynamics with fluctuating work, while the R\'enyi $0$- and $\infty$-divergences are relevant to resource theory of thermodynamics in the single-shot scenario.

In Chapter \ref{chap:quantum_entropy}, we move to the quantum case.  After a brief overview of quantum states and quantum dynamics, we discuss quantum information-theoretic entropies and divergences.
In particular, we will focus on the von Neumann entropy, the quantum KL divergence, and the quantum R\'enyi $0$- and $\infty$-divergences.

In Chapter \ref{chap:quantum_majorization}, we consider the quantum counterpart of majorization. 
The mathematical structure of ordinary majorization of the quantum case is similar to that of the classical case.
The quantum version of thermo-majorization (and d-majorization) is a more subtle concept, while  the quantum R\'enyi $0$- and $\infty$-divergences can still characterize state convertibility.

In Chapter \ref{chap:approximate_asymptotic}, we consider approximate and asymptotic state conversion.
We consider the smooth R\'enyi $0$- and $\infty$-divergences, and as their asymptotic limit, introduce a concept called information spectrum (or the spectral divergence rate). 
We then discuss a quantum version of the AEP under certain assumptions including ergodicity, and show that the AEP implies the existence of a single complete monotone, which is nothing but the KL divergence rate. 

In Chapter \ref{chap:quantum_thermodynamics}, we apply the foregoing  argument (Chapter \ref{chap:quantum_entropy}, Chapter \ref{chap:quantum_majorization}, and Chapter \ref{chap:approximate_asymptotic}) to quantum thermodynamics.
We generally formulate Gibbs-preserving maps and thermal operations and derive the second law of quantum thermodynamics.
We discuss both of the fluctuating-work formulation and the single-shot scenario.
We also consider the asymptotic limit, where a complete macroscopic thermodynamic potential emerges.

In Appendix~\ref{apx:general_monotonicity}, we provide a proof of the monotonicity properties of general quantum divergences, by invoking  mathematical techniques of matrix analysis.
In particular, we show that operator monotone and operator convex play crucial roles.
We also discuss the quantum Fisher information and prove its monotonicity property.

In Appendix~\ref{appx:hypothesis_testing}, we briefly overview quantum hypothesis testing, which has a fundamental connection to the asymptotic theory discussed in Chapter \ref{chap:approximate_asymptotic}.
We discuss that quantum hypothesis testing provides essentially the same information as the smooth R\'enyi $0$- and $\infty$-divergences. 
We also discuss semidefinite programming as a useful tool and focus on the quantum Stein's lemma as  another representation of the quantum AEP. 

In Appendix~\ref{app_classical}, we discuss the classical AEP by considering classical ergodic processes, and provide a proof of the classical Stein's lemma.
This is regarded as the classical counterpart of Chapter \ref{chap:approximate_asymptotic} and Appendix~\ref{appx:hypothesis_testing}.

\

Finally, we remark on a notation that is used throughout this book.
$\beta \geq 0$ represents the inverse temperature of the environment (i.e., the heat bath), that is, $\beta := (k_{\rm B}T)^{-1}$ with $0< T \leq \infty$ being the corresponding temperature and $k_{\rm B}$ being the Boltzmann constant.
Also, we set the Planck constant to unity.

We will basically  restrict ourselves to finite-dimensional systems for both of the classical and quantum cases, unless stated otherwise.
Only for a few topics such as continuous-variable majorization and the AEP, we will adopt infinite-dimensional setups.



\chapter{Classical entropy and divergence}
\label{chap:classical_entropy}

In this chapter, we consider the basic concepts of classical information theory.
In Section~\ref{sec:state_dynamics},  we formulate classical states and dynamics as probability distributions and stochastic matrices.
We then introduce the Shannon entropy and the Kullback-Leibler (KL) divergence (relative entropy) in Section~\ref{sec:Shannon_KL} and the R\'enyi $\alpha$-entropy and $\alpha$-divergence in Section~\ref{sec:Renyi}.
In Section~\ref{sec:classical_general_divergence}, we consider more general divergence-like quantities including the $f$-divergence, and as special cases, provide proofs of the properties of the KL and the R\'enyi divergences.
In Section~\ref{sec:classical_Fisher}, we briefly discuss the classical Fisher information, which is a lower-order expansion of the KL divergence.

\section{Classical state and dynamics}
\label{sec:state_dynamics}

As a preliminary, we fix our terminologies and notations for classical probability theory. 
The concept of \textit{state} of a classical system is represented by a probability distribution $p := (p_1, p_2, \cdots, p_d)^{\rm T}$ with $\sum_{i=1}^dp_i=1$ and $p_i \geq 0$.
Here, $p$ is regarded as a $d$-dimensional column vector and $\rm T$ represents the transpose.
In the following, we do not distinguish \textit{state} and  \textit{distribution}.
On the other hand, the classical phase-space point corresponds to index $i$, which is also referred to as a ``state'' in the physics convention, while we do not use this terminology in this book.

Let $\mathcal P_d$ be the set of $d$-dimensional probability distributions. 
We denote the uniform distribution as $u := (1/d, 1/d, \cdots, 1/d)^{\rm T}$.
With two independent distributions $p \in \mathcal P_d$ and $q \in \mathcal P_{d'}$, we write their joint distribution as $p \otimes q \in \mathcal P_{dd'}$ whose components are given in the form $p_iq_j$.  This ``tensor product'' notation is consistent with the quantum case, where a classical probability distribution is regarded as the diagonal elements of a density matrix.
If there are $n$ samples that are independent and identically distributed (i.i.d.), we write the joint distribution as $p^{\otimes n} \in \mathcal P_{d^n}$.

We note that the \textit{support} of $p$, denoted as ${\rm supp}[ p]$, is defined as the set of indexes $\{ i : p_i > 0\}$. 
The \textit{rank} of $p$ is the number of elements of ${\rm supp}[ p]$, i.e., ${\rm rank}[p] := | {\rm supp}[ p] |$.
If $p \in \mathcal P_d$ does not have a zero component, then $p$ is called to have full rank.

A time evolution of classical probability distributions is represented by a \textit{stochastic matrix}, which maps the input state at the initial time into the output state at the final time.
For the sake of simplicity, we assume that the input and the output spaces have the same dimensions in Chapters \ref{chap:classical_entropy}, \ref{chap:classical_majorization}, and \ref{chap:classical_thermodynamics} (i.e., for the classical case), which means that any stochastic matrix is a square matrix.
The components of a stochastic matrix $T$ must satisfy $\sum_{i=1}^dT_{ij} = 1$ and $T_{ij} \geq 0$.  
The time evolution of  a probability vector is then given by $p' = Tp$, or equivalently, $p'_i = \sum_{j=1}^d T_{ij}p_j$.
Such a map on probability distributions is called a stochastic map or a Markov map, for which we use the same notation $T$.

If a stochastic matrix $T$ further satisfies $\sum_{j=1}^d T_{ij} = 1$, it is called a \textit{doubly stochastic matrix}.
By definition, a stochastic matrix $T$ is doubly stochastic if and only if the uniform distribution is its fixed point: $u=Tu$.

Meanwhile, we note that for distributions $p,q$, the \textit{trace distance} is defined as
\begin{equation}
D(p,q) := \frac{1}{2} \| p - q \|_1 := \frac{1}{2}\sum_{i=1}^d | p_i - q_i |,
\end{equation}
where $\| \cdot \|_1$ is called the trace norm.
The trace distance does not increase under any stochastic map $T$, which is called the monotonicity (or the data processing inequality):
\begin{equation}
D(p,q)  \geq  D(Tp,  Tq ).
\label{monotonicity_trace_norm}
\end{equation}
The proof of this is easy, but we will postpone it to  Section~\ref{sec:classical_general_divergence} where we give a more general perspective on the monotonicity.

\section{Shannon entropy and the KL divergence}
\label{sec:Shannon_KL}

A most basic concept in information theory is the \textit{Shannon entropy}, which is defined for a classical probability distribution $p \in \mathcal P_d$ as
\begin{equation}
S_1(p) := -\sum_{i=1}^d p_i \ln p_i.
\end{equation}
Here, we added the subscript ``$1$'' because the Shannon entropy is the R\'enyi $1$-entropy as discussed later.
Obviously, $S_1 (p) \geq 0$.

We next consider two distributions $p, q \in \mathcal P_d$.
The \textit{Kullback-Leibler (KL) divergence} (or the \textit{relative entropy}) is defined as
\begin{equation}
S_1 (p \| q) := \sum_{i=1}^d p_i \ln \frac{p_i}{q_i}.
\end{equation}
If the support of $p$ is not included in that of $q$, we define $S_1(p \| q ) := + \infty$.
In order to avoid too much complexity, however, we always assume that the support of  $p$ is included in that of $q$ throughout this book, whenever we consider divergence-like quantities.

The KL divergence is regarded as an asymmetric ``distance'' between two distributions.
It is non-negative:
\begin{equation}
S_1 (p \| q) \geq 0,
\end{equation} 
where the equality $S_1 (p \| q) = 0$ holds if and only if $p=q$.
A simplest way to see this is as follows: Since $\ln (x^{-1}) \geq 1-x$ holds for $x > 0$, where the equality holds if and only if $x=1$,  we have $\sum_i p_i \ln ( p_i / q_i) \geq \sum_i p_i  (1 -  q_i / p_i)  = 0$ and the equality condition.

The entropy and the divergence are related as
\begin{equation}
S_1 (p) = \ln d - S_1 (p \| u),
\label{Shannon_KL}
\end{equation}
where $u$ is the uniform distribution.
From this and the non-negativity of the divergence, we obtain $S_1(p) \leq \ln d$.

We note that if $\Delta p := p - q$ is small, i.e., $\varepsilon := \| \Delta p \|_1 \ll 1$, we can expand the KL divergence up to the second order of $\varepsilon$ as
\begin{equation}
S_1 (p  \| p-\Delta p ) =\frac{1}{2} \sum_i \frac{( \Delta p_i)^2}{p_i} + O(\varepsilon^3 ),
\label{KL_expand}
\end{equation}
where we used $\sum_i \Delta p_i = 0$ to drop the term of $O(\varepsilon)$.
The right-hand side of (\ref{KL_expand}) is related to the Fisher information~\cite{Fisher,Amari2000}.

The KL divergence also satisfies the \textit{monotonicity} (or the \textit{data processing inequality}).  Suppose that $p,q,p',q' \in \mathcal P_d$ satisfy  $p' = Tp$ and $q' = Tq$  for a stochastic matrix $T$.  Then, 
\begin{equation}
S_1 (p \|q) \geq S_1 (p' \| q' ),
\label{monotone_classical}
\end{equation}
which represents that two distributions become harder to be distinguished, if they are  ``coarse-grained'' by a stochastic map.
In terms of resource theory, this implies that the KL divergence is a \textit{monotone}, where in general a monotone is a quantity that does not increase (or decrease) under free operations.  
Here, all stochastic maps are (rather formally) regarded to be free, while in the thermodynamic setup, only Gibbs-preserving maps are supposed to be free (see also Section~\ref{sec:d_majorization} and Chapter \ref{chap:classical_thermodynamics}). 

A straightforward way to show inequality (\ref{monotone_classical}) is that
\begin{eqnarray}
&{}& S_1 (p \|q) - S_1 (p'  \| q' )  =  \sum_i p_i \ln \frac{p_i}{q_i} -  \sum_j p_j' \ln \frac{p_j'}{q_j'} \\
&=&  \sum_{ij} T_{ji} p_i \ln \frac{T_{ji}p_i}{T_{ji}q_i} -  \sum_j p_j' \ln \frac{p_j'}{q_j'} = \sum_j p_j' S_1 (\tilde p^{(j)} \| \tilde q^{(j)})  \geq 0,
\end{eqnarray}
where $S_1 (\tilde p^{(j)} \| \tilde q^{(j)}):=  \sum_i \frac{T_{ji}p_i}{p_j'}\ln \frac{T_{ji}p_i / p_j'}{T_{ji}q_i / q_j'} $ is the KL divergence between the conditional distributions $\tilde p^{(j)}_i := T_{ji}p_i / p_j'$ and $\tilde q^{(j)}_i  := T_{ji}q_i / q_j'$.
Another proof of (\ref{monotone_classical}) will be shown in Section~\ref{sec:classical_general_divergence} as a special case of the monotonicity of more general divergence-like quantities.
We will discuss the significance of the monotonicity in thermodynamics of classical systems in Section~\ref{sec:classical_second_law}.

We note that the converse of the monotonicity~(\ref{monotone_classical}) is not true in general:
There are pairs of states $(p,q)$ and $(p',q')$ satisfying inequality~(\ref{monotone_classical}) such that  any stochastic map $T$  satisfying $p' = Tp$ and $q' = Tq$ does not exist.
In other words, the monotonicity of the KL divergence is only a necessary condition, but not a sufficient condition for the convertibility of pairs of states.
In the resource theory terminology, this implies that the KL divergence  is a monotone but not a \textit{complete} monotone under stochastic maps, because a complete monotone must provide a sufficient condition for state convertibility~\cite{Takagi2019}.
This observation brings us to the concept of majorization (more generally, d-majorization), which will be discussed  in the next chapter.
In particular, we will show that there exists a \textit{complete set} of monotones, which provides a sufficient condition with a collection of infinitely many divergence-like quantities.

We next consider the case of doubly stochastic matrices.
By noting Eq.~(\ref{Shannon_KL}) and the fact that a doubly stochastic matrix $T$ does not change $u$, the monotonicity (\ref{monotone_classical}) of the KL divergence implies that 
\begin{equation}
S_1 (p) \leq S_1 (Tp) ,
\label{monotone_DSM}
\end{equation}
which represents that a doubly stochastic matrix makes distributions more ``random.''
In terms of resource theory, this implies that the Shannon entropy is a monotone under doubly stochastic maps (but again, is not a complete monotone).

The Shannon entropy and the KL divergence also have other fundamental properties such as the subadditivity, the strong subadditivity, the concavity of the Shannon entropy, and the joint concavity of the KL divergence. 
However, we will postpone to discuss them to Chapter~\ref{chap:quantum_entropy} for the quantum setup, from which the classical counterpart immediately follows.

We here briefly remark on a quantity called mutual information.
We consider two systems A and B.
Let $p_{\rm AB}$ be a distribution of AB, and $p_{\rm A}$ and $p_{\rm B}$ be its marginal distributions of A and B, respectively.  Then, the mutual information between A and B in distribution $p_{\rm AB}$ is defined as
\begin{equation}
I_1 (p_{\rm AB} )_{\rm A: B} := S_1 (p_{\rm A} ) + S_1 ( p_{\rm B} ) - S_1 (p_{\rm AB} ) = S_1 (p_{\rm AB} \| p_{\rm A}  \otimes p_{\rm B}  ) \geq 0.
\end{equation}
This quantifies the correlation between A and B, and $I_1 (p_{\rm AB} )_{\rm A: B}  = 0$ holds if and only if A and B are statistically independent, i.e., $p_{\rm AB} = p_{\rm A}  \otimes p_{\rm B} $.
The monotonicity (\ref{monotone_classical}) of the KL divergence implies the data processing inequality of the mutual information: $I_1 (p_{\rm AB} )_{\rm A: B} \geq I_1 (T_{\rm A} \otimes T_{\rm B} p_{\rm AB} )_{\rm A: B}$, where $T_{\rm A} \otimes T_{\rm B}$ represents a stochastic map independently acting on A and B.

\section{R\'enyi entropy and divergence}
\label{sec:Renyi}

We next  discuss generalized entropies:  the \textit{R\'enyi} $\alpha$\textit{-entropies} for $0 \leq \alpha \leq \infty$.  For a  distribution $p \in \mathcal P_d$, the R\'enyi $\alpha$-entropy is defined as 
\begin{equation}
S_\alpha (p) := \frac{1}{1-\alpha} \ln \left( \sum_{i=1}^d p_i^\alpha \right).
\end{equation}
Here, $S_\alpha (p)$ for $\alpha = 0, 1, \infty$ is defined by taking the limit: $S_1 (p )$ is indeed the Shannon entropy, and
\begin{equation}
S_0 (p) := \ln \left( {\rm rank} [ p] \right),
\end{equation}
\begin{equation}
S_\infty (p) := -\ln \left( \max_i \{ p_i \} \right).
\end{equation}
We also denote these quantities by  $S_{\rm min} (p) :=S_\infty (p)$ and $S_{\rm max} (p) := S_0(p)$, which are referred to as the min and the max entropies, respectively.
It is obvious that 
\begin{equation}
S_\alpha (p) \geq 0.
\end{equation}

We next consider the \textit{R\'enyi} $\alpha$\textit{-divergence} for distributions $p,q \in \mathcal P_d$   with $0 \leq \alpha \leq \infty$~\cite{Erven_Harremoes2010,Erven_Harremoes}.  It is defined as
\begin{equation}
S_\alpha (p \| q ) := \frac{1}{\alpha -1} \ln \left( \sum_{i=1}^d \frac{p_i^\alpha}{q_i^{\alpha -1}} \right).
\end{equation}
For $\alpha = 0, 1, \infty$, we again take the limit: $S_1 (p \| q )$ is the KL divergence, and 
\begin{equation}
S_0 (p \| q ) := - \ln \left( \sum_{i: p_i > 0} q_i \right),
\end{equation}
\begin{equation}
S_\infty (p \| q ) := \ln \left( \max_i \left\{ \frac{p_i}{q_i} \right\} \right).
\end{equation}
We also denote these quantities by $S_{\rm min}(p\| q) :=S_0 (p \| q)$ and $S_{\rm max} (p \| q) := S_\infty (p \| q)$, where on the contrary to the entropy case,  the min and the max divergences correspond to $\alpha = 0$ and $\infty$, respectively.
We again note that, in the definition of these divergence-like quantities, we always assume that the support of  $p$ is included in that of $q$.

An important property of the R\'enyi divergence is non-negativity, while we  postpone the proof to Section~\ref{sec:classical_general_divergence}.

\begin{proposition}[Non-negativity of the R\'enyi divergence]
\begin{equation}
S_\alpha (p \| q) \geq 0.
\end{equation}
For $0 < \alpha \leq \infty$,  the equality $S_\alpha (p \| q) = 0$ holds if and only if $p=q$.
For $\alpha = 0$, the equality holds if and only if the supports of $p$ and $q$ are the same.
\label{prop:classical_Renyi_nonnegative}
\end{proposition}

The R\'enyi  $\alpha$-divergence satisfies the monotonicity: it does not increase (and thus is a monotone) under stochastic maps.  
(We again postpone the proof to Section~\ref{sec:classical_general_divergence}.)

\begin{proposition}[Monotonicity of R\'enyi divergence, Theorem 1 of ~\cite{Erven_Harremoes}]
For any stochastic matrix $T$, the R\'enyi $\alpha$-divergence with $0 \leq \alpha \leq \infty$ satisfies
\begin{equation}
S_\alpha (p \| q ) \geq S_\alpha (T p \| Tq ). 
\label{Renyi_divergence_monotonicity}
\end{equation}
\label{Renyi_divergence_monotonicity_lemma}
\end{proposition}

With  the uniform distribution $u$,  the $\alpha$-divergence and the $\alpha$-entropy are related as
\begin{equation}
S_\alpha (p) = \ln d - S_\alpha (p \| u),
\end{equation}
from which we have 
\begin{equation}
S_\alpha (p) \leq \ln d.
\end{equation}
From the monotonicity of the  $\alpha$-divergence, we have, for any doubly stochastic matrix $T$,
\begin{equation}
S_\alpha (p  ) \leq S_\alpha (T p). 
\label{Renyi_entropy_monotonicity}
\end{equation}

We also note the following property.
(Again, we will prove it in Section~\ref{sec:classical_general_divergence}.)

\begin{proposition}[Theorem 3 of \cite{Erven_Harremoes2010}]
\begin{equation}
S_\alpha (p \| q ) \leq S_{\alpha'} (p \| q ) \ \ \rm{for} \ \ \alpha \leq \alpha',
\end{equation}
and thus
\begin{equation}
S_\alpha (p ) \geq S_{\alpha'} (p  ) \ \ \rm{for} \ \ \alpha \leq \alpha'.
\end{equation}
\label{prop:classical_alpha_monotone}
\end{proposition}

In addition, we note that a property called the joint convexity of the R\'enyi $\alpha$-divergence is true for $0 \leq \alpha \leq 1$~\cite{Erven_Harremoes}.  See also Corollary~\ref{cor:joint_convexity_Renyi} in Appendix~\ref{apx:general_monotonicity} for the quantum case.

We finally remark on the R\'enyi $\alpha$-divergence for  negative $\alpha$.
It can be defined for  $- \infty \leq \alpha < 0$ by 
\begin{equation}
S_\alpha ( p \| q ) := \frac{\rm{sgn} (\alpha ) }{\alpha - 1} \ln \left( \sum_{i=1}^d \frac{p_i^\alpha}{q_i^{\alpha -1}} \right),
\label{negative_alpha_divergence}
\end{equation}
 where $\rm{sgn} (\alpha ) := 1$ for $\alpha > 0$ and $\rm{sgn} (\alpha ) := -1$ for $\alpha < 0$.
It is straightforward to check that $S_\alpha ( p \| q ) = \frac{\alpha}{\alpha - 1}  S_{1-\alpha} ( q \| p )$ for $\alpha < 0$.
Correspondingly, the R\'enyi $\alpha$-entropy for negative $\alpha$ can be defined as
\begin{equation}
S_\alpha (p) := \frac{\rm{sgn} (\alpha ) }{1- \alpha} \ln \left( \sum_{i=1}^d  p_i^\alpha \right) = {\rm sgn} (\alpha ) \ln d - S_\alpha (p \| u).
\label{negative_alpha_entropy}
\end{equation}
In the following, however, we only consider positive $\alpha$ unless stated otherwise.

\section{General classical divergences}
\label{sec:classical_general_divergence}

We now give the proofs of the properties of the R\'enyi $\alpha$-divergence  discussed  in Section~\ref{sec:Renyi}.
Moreover, in order to provide a more general perspective, here we discuss  general divergence-like quantities including the $f$-divergence.
See Appendix~\ref{apx:general_monotonicity} for the quantum counterpart.
We start with the formal definition of convexity/concavity.

\begin{definition}[Convexity and concavity]
Let $I \subset \mathbb R$ be an interval.  A function $ f : I \to \mathbb R$ is convex, if for any $x,y \in I$ and any $\lambda \in [0,1]$,
\begin{equation}
  f(\lambda x + (1-\lambda )y) \leq \lambda f(x) + (1-\lambda )f(y).
\end{equation}
Moreover, $f$ is strictly convex at $z \in I$, if for any $x,y \in I$ with $x \neq y$ and any $\lambda \in (0,1)$ such that  $\lambda x + (1-\lambda)y = z$,
\begin{equation}
  f(\lambda x + (1-\lambda )y) < \lambda f(x) + (1-\lambda )f(y).
  \end{equation}
  If $\leq$ and $<$ above are  replaced by $\geq$ and $>$, $f$ is concave and strictly concave, respectively.
\end{definition}

We now state the main lemma of this section, which holds true for vectors in $\mathbb R^d$ that are not necessarily probability distributions.

\begin{lemma}
Let $f$ be a convex function and let $p,q, p', q' \in \mathbb R^d$.
Suppose that all the components of $q,q'$ are positive.
If $p' = Tp$ and $q' = Tq$ hold for a stochastic matrix $T$, then
\begin{equation}
\sum_{i=1}^d q_i' f\left( \frac{p'_i}{q_i'} \right) \leq \sum_{i=1}^d q_i f \left( \frac{p_i}{q_i} \right).
\label{f_divergence_monotone}
\end{equation}
If $f$ is concave, we have the opposite inequality.
\label{f_monotonicity_lemma}
\end{lemma}

\begin{proof}
Let $f$ be convex. By noting that 
\begin{equation}
\frac{p'_j}{q'_j} = \sum_{i=1}^d \frac{T_{ji}q_i}{q'_j} \frac{p_i}{q_i}, \ \ \ \sum_{i=1}^d \frac{T_{ji}q_i}{q'_j} = 1,
\end{equation}
we have from the Jensen inequality
\begin{equation}
\sum_{j=1}^d q'_j  f\left( \frac{p'_j}{q'_j} \right) \leq \sum_{j=1}^d \sum_{i=1}^d q'_j \frac{T_{ji}q_i}{q'_j}  f \left( \frac{p_i}{q_i} \right) = \sum_{i=1}^d q_i  f \left( \frac{p_i}{q_i} \right). 
\end{equation}
We apply the same proof for the concave case.
$\Box$
\end{proof}

Let $p,q \in \mathcal P_d$.
If we take $f(x) :=  x\ln x$ in the above lemma, we have $S_1 (p \| q) = \sum_{i=1}^d q_i f (p_i / q_i )$.
Thus inequality (\ref{f_divergence_monotone}) implies the monotonicity of the KL divergence (\ref{monotone_classical}).
If we take  $f(x) = |x - 1|/2$, the trace distance is written as $D(p ,   q ) = \sum_{i=1}^d q_i f (p_i / q_i )$, from which we obtain the monotonicity of the trace distance~(\ref{monotonicity_trace_norm}).

\begin{corollary}
Let $f$ be a convex function and let $p, p' \in \mathbb R^d$.
If $p' = Tp$ holds for a doubly stochastic matrix $T$,
\begin{equation}
\sum_{i=1}^d f(p'_i) \leq \sum_{i=1}^d f( p_i ).
\label{f_monotone}
\end{equation}
If $f$ is concave, we have the opposite inequality.
\label{f_monotonicity_corollary}
\end{corollary}

\begin{proof}
Set $q_i = 1$ in Lemma~\ref{f_monotonicity_lemma}.
$\Box$
\end{proof}

We now prove the properties of the R\'enyi $\alpha$-divergence, which are stated in Section~\ref{sec:Renyi}.

\

\noindent\textbf{Proof of Proposition~\ref{prop:classical_Renyi_nonnegative}.}
Let $f_\alpha (x) := x^\alpha$.
For $1 < \alpha < \infty$, $f_\alpha$ is convex.  From the Jensen inequality, we have
\begin{equation}
\sum_{i=1}^d q_i f \left( \frac{p_i}{q_i} \right) \geq  f \left(\sum_{i=1}^d q_i \frac{p_i}{q_i} \right)  = f(1) = 1,
\label{classical_f_positive}
\end{equation}
and then take the logarithm of this.
The equaity holds if and only if  $p_i / q_i =1$ for all $i$, because $f_\alpha$ is strictly convex at $x= 1$.

For $0 < \alpha < 1$, $f_\alpha$ is concave, and thus we have the opposite inequality to the above; then take the logarithm of it, by noting the sign of $\alpha - 1$. 

For $\alpha = 0,1, \infty$, we can take the limit to show the non-negativity, but can also easily show it directly.
The equality condition can be confirmed directly for these cases.
In particular, for the case of  $\alpha = 0$, the equality holds if and only if the support of $p$ includes that of $q$;
but under our assumption that the support of $q$ always includes that of $p$,  this condition implies that the supports of them are the same.
$\Box$

\

\noindent\textbf{Proof of Proposition~\ref{Renyi_divergence_monotonicity_lemma}.}
For $0 < \alpha < 1$ and $1 < \alpha < \infty$, we apply Lemma \ref{f_monotonicity_lemma} to  $f_\alpha (x) = x^\alpha$ and take the logarithm of it, by noting the sign of $\alpha -1$.   For $\alpha = 0, 1, \infty$, we can take the limit.
For $\alpha = 1$, we can also directly take $f(x) = x \ln x$ as mentioned before. 
$\Box$

\

\noindent\textbf{Proof of Proposition~\ref{prop:classical_alpha_monotone}.}
Let $\alpha < \alpha '$.
$f(x) := x^{(\alpha - 1) / (\alpha' - 1)}$ is concave for $1 < \alpha < \alpha' < \infty$, while is convex for $0 < \alpha < \alpha' < 1$ and $0 < \alpha < 1 < \alpha'$.
From the Jensen's inequality, and noting the sign of $\alpha - 1$, we obtain
\begin{eqnarray}
S_\alpha ( p \| q ) 
&=& \frac{1}{\alpha -1} \ln \left( \sum_i p_i \left( \frac{p_i}{q_i} \right)^{\alpha-1} \right) \\
&=& \frac{1}{\alpha -1} \ln \left( \sum_i p_i \left( \frac{p_i}{q_i} \right)^{(\alpha ' - 1 )\frac{\alpha - 1}{\alpha'-1}} \right) \\
&\leq& \frac{1}{\alpha' -1} \ln \left( \sum_i p_i \left( \frac{p_i}{q_i} \right)^{\alpha'-1} \right) \\
&=& S_{\alpha'} ( p \| q ).
\end{eqnarray}
For $\alpha, \alpha' = 0,1, \infty$, we take the limit. $\Box$

\

We next  discuss  the concept called the $f$-divergence introduced in Refs.~\cite{Ali,Csiszar} (see also Ref.~\cite{Liese2006}).
Let $f: (0,\infty) \to \mathbb R$ be a convex function.
Suppose that  $f(x)$ is strictly convex at $x=1$ and $f(1) = 0$.
Then,  the $f$-divergence is defined for $p,q \in \mathcal P_d$ by
\begin{equation}
D_f (p\| q) := \sum_{i=1}^d q_i f \left( \frac{p_i}{q_i} \right).
\label{f-divergence}
\end{equation}
The $f$-divergence is non-negative:
\begin{equation}
D_f (p\| q) \geq 0,
\end{equation}
where the equality $D_f (p\| q) = 0$ holds if and only if $p=q$.
This follows from the Jensen inequality (\ref{classical_f_positive}) along with $f(1) = 0$,
where the equality condition follows from the assumption that $f(x)$ is strictly convex at $x=1$.
The $f$-divergence also satisfies the monotonicity for $0 \leq \alpha \leq \infty$, which is nothing but Lemma \ref{f_monotonicity_lemma}.

The KL divergence is the $f$-divergence with $f(x) = x \ln x$.
On the other hand, the R\'enyi $\alpha$-divergence with $\alpha \neq 1$ is not in the form of $f$-divergence.
We note that there is another concept also called $\alpha$-divergence~\cite{Amari2000}, which is the $f$-divergence with $f_\alpha (x) := \frac{1}{\alpha ( \alpha - 1 )} (x^{\alpha} -1)$ for $\alpha \neq 0,1, \infty$, that is,
\begin{equation}
D_{f_\alpha} (p \| q ) := \frac{1}{\alpha ( \alpha - 1 )} \left(  \sum_i  \frac{p_i^\alpha}{q_i^{\alpha - 1}} -1  \right). 
\label{classical_alt_divergence}
\end{equation}
It can be shown that $\lim_{\alpha \to 1}D_{f_\alpha} (p \| q ) = S_1 (p \| q )$.
In this book, however, whenever we simply mention ``the $\alpha$-divergence,'' it indicates the R\'enyi $\alpha$-divergence in the sense of Section~\ref{sec:Renyi}.

The trace distance is  the $f$-divergence with $f(x) := |x-1| / 2$.
The $f$-divergence is also related to the fidelity defined as $F(p,q) := \sum_i \sqrt{p_iq_i}$.  
To see this, let $f(x) := 1 - \sqrt{x}$, which is equivalent to above-mentioned $f_{\alpha}$ with $\alpha = 1/2$ up to normalization. 
Then we have $D_f ( p \| q ) = 1 - F(p,q)$, from which the monotonicity of the fidelity follows, i.e., $F(p,q) \leq F(Tp, Tq)$.  It can be also rewritten as $D_f ( p \| q ) = \frac{1}{2} \sum_i (\sqrt{p_i} - \sqrt{q_i})^2$, and $\sqrt{D_f ( p \| q )}$ is called the  Hellinger distance.

We finally note that, in general, a \textit{divergence} $D(p \| q)$ is defined as a quantity satisfying $D(p \| q) \geq 0$ where the equality holds if and only if $p=q$.
The monotonicity under stochastic maps is also often required for the definition of divergence.
The $f$-divergence and the  R\'enyi $\alpha$-divergence with $0 < \alpha \leq \infty$  are both divergences in this sense, while the R\'enyi $0$-divergence does not satisfy the above equality condition.

\section{Fisher information}
\label{sec:classical_Fisher}

We here briefly discuss the Fisher information~\cite{Fisher,Amari2000} (see also Ref.~\cite{Cover_Thomas}), which has a fundamental connection to divergences.
The quantum counterpart will be discussed in Appendix~\ref{sec:quantum_Fisher}.

A main practical application of the Fisher information is found in the theory of parameter estimation, where one wants to estimate unknown parameters of probability distributions from observed data of random events.
Here, the Fisher information gives a fundamental bound of the accuracy of parameter estimation, which is known as the Cramer-Rao bound.
In this section, instead of going into details of estimation theory, we will focus on the general mathematical properties of the Fisher information.

We consider smooth parametrization of probability distributions, written as $p(\theta ) \in \mathcal P_d$ with parameters $\theta := (\theta^1, \theta^2, \cdots, \theta^m) \in \mathbb R^m$, where  the domain of $\theta$ is an open subset of $\mathbb R^m$.
We suppose that $p(\theta)$ has the full rank, i.e., $p_i (\theta ) > 0$ for any $i$ and $\theta$.
We denote $\partial_k := \partial / \partial \theta^k$.
We note that  the parametrization must satisfy $m \leq d-1$.

\begin{definition}[Fisher information]
Let $p(\theta) \in \mathcal P_d$ have full rank and $\theta \in \mathbb R^m$ be the parameter.
The (classical) Fisher information matrix is an $m \times m$ matrix, whose $(k,l)$-component is defined as
\begin{equation}
J_{p( \theta), kl} := \sum_{i=1}^d p_i (\theta ) \partial_k [\ln p_i(\theta)] \partial_l [\ln p_i(\theta)] = \sum_{i=1}^d \frac{\partial_k p_i (\theta ) \partial_l p_i (\theta )}{p_i (\theta )}.
\end{equation}
\end{definition}

The Fisher information is obtained as the infinitesimal limit of the $f$-divergence.
Suppose that $f''(1) > 0$ exists and $f$ is sufficiently smooth around $1$.
Then it is easy to check that
\begin{equation}
D_f ( p(\theta ) \| p (\theta - \Delta \theta ) ) = \frac{f''(1)}{2} \sum_{kl} \Delta \theta^k J_{p(\theta), kl} \Delta \theta^l + O(\varepsilon^3),
\label{f_divergence_expand}
\end{equation}
where $\varepsilon := \| \Delta \theta \|$.
We note that the term of $O(\varepsilon)$ vanishes because of $\sum_i \partial_k p_i (\theta ) = 0$.
Eq.~(\ref{f_divergence_expand}) is a generalization of Eq.~(\ref{KL_expand}); remarkably, for all the $f$-divergences, we obtain the same Fisher information up to normalization.

The Fisher information satisfies the monotonicity under stochastic maps.  While this may be regarded as a trivial consequence of the monotonicity of the $f$-divergence, we provide a direct proof as follows.

\begin{proposition}[Monotonicity of the Fisher information]
For any stochastic map $T$ that is independent of $\theta$,
\begin{equation}
J_{p( \theta)} \geq J_{Tp( \theta)}.
\end{equation}
\label{thm:classical_Fisher_monotone}
\end{proposition}

\begin{proof}
We omit the argument $\theta$ for simplicity of notations.  Let $p' := Tp$.  Let $c = (c^1, \cdots, c^m ) \in \mathbb R^m$ be a column vector and define $\partial := \sum_k c^k \partial_k$.  Then,
\begin{equation}
c^{\rm T} J_{p} c = \sum_i p_i \left( \frac{\partial p_i}{p_i} \right)^2,  \ \ \ c^{\rm T} J_{p'} c = \sum_i p'_i \left( \frac{ \partial p'_i}{p'_i} \right)^2.
\end{equation}
Then, we apply inequality~(\ref{f_divergence_monotone}) with $f(x) := x^2$.
$\Box$
\end{proof}

An operational meaning of the Fisher information is highlighted by  the Cramer-Rao bound, which states that the accuracy of any unbiased estimation of $\theta$ is bounded by the Fisher information.
Consider a task that one estimates unknown $\theta$ from observed data $i$.  Let $\theta_{\rm est} (i)$ be an estimator and suppose that it satisfies the unbiasedness condition: $\sum_i p_i (\theta ) \theta_{\rm est} (i) =  \theta$ for all $\theta$.  The accuracy of such unbiased estimation can be characterized by the covariance matrix ${\rm Cov}_\theta (\theta_{\rm est}) $, whose $(k,l)$-component is given by
\begin{equation}
{\rm Cov}_\theta^{kl} (\theta_{\rm est} ) := \sum_{i} p_i (\theta) (\theta_{\rm est}^k (i)  - \theta^k )(\theta_{\rm est}^l (i)  - \theta^l ).
\end{equation}
Then, the Cramer-Rao bound states that
\begin{equation}
{\rm Cov}_\theta (\theta_{\rm est} ) \geq J_{p(\theta )}^{-1}.
\label{Cramer_Rao}
\end{equation}
We omit the proof of this, which is not difficult (e.g., Theorem 11.10.1 of Ref.~\cite{Cover_Thomas}).

As an example, we consider a family of probability distributions called the exponential family.
For simplicity, we consider a single parameter $\theta \in \mathbb R$ and the parameterized distribution of the form
\begin{equation}
p_i (\theta ) := h_i \exp( \theta T_i - A(\theta ) ),
\end{equation}
where $A(\theta )$ is a smooth function of $\theta$.
By straightforward computation, we have 
\begin{equation}
\sum_i T_i p_i( \theta) = A' (\theta ), \ \ \sum_i T_i^2 p_i( \theta) = A'' (\theta ) + A'(\theta )^2,
\label{exponential_family_eq}
\end{equation}
where $A' (\theta ) := dA(\theta ) / d \theta$.
Thus, the Fisher information is given by $J_{p (\theta )} = A'' (\theta )$.

In terms of thermodynamics, we can interpret that $p_i (\theta )$ is a Gibbs state, where $T_i$ is the energy, $-\theta$ is the inverse temperature $\beta$, and $\theta^{-1} A(\theta )$ is the free energy $F(\beta )$ (we set $h_i = 1$).
With this correspondence, Eqs.~(\ref{exponential_family_eq}) are  well-known formulas in equilibrium statistical mechanics.

We note that Eqs.~(\ref{exponential_family_eq}) imply that $T_i$ is regarded as an unbiased estimator of $A'(\theta)$ with the variance $A'' (\theta )$.
Then,  we replace the parameter $\theta$ with $\theta' := A' (\theta)$.
The corresponding Fisher information is given by  $J_{p (\theta' )} = A'' (\theta )^{-1}$ because of $d/d\theta  ' = (A''(\theta ) )^{-1} d / d \theta$.
Therefore, $T_i$ as an unbiased estimator of $\theta'$ attains the Cramer-Rao bound~(\ref{Cramer_Rao}).

\

We finally remark that the Fisher information can be regarded as a \textit{metric} on the space of probability distributions (or on the parameter space), which is the perspective of \textit{information geometry}~\cite{Amari2000}
(see Ref.~\cite{Ito2020} for its relation to stochastic thermodynamics).
Here, we introduce the concept called \textit{monotone metric}.

\begin{definition}[Monotone metric]
Suppose that $G_{p} : \mathbb R^d  \times \mathbb R^d \to \mathbb R$ is defined for  distributions $p \in \mathcal P_d$  with full support.
We call $G_{p}$ a monotone metric, if it satisfies the following.
\begin{itemize}
\item $G_{p}$ is bilinear.
\item  $G_{p} (a, a) \geq 0$ holds for any $p$, where the equality is achieved if and only if $a = 0$.
\item $p \mapsto G_p ( a, a) $ is continuous for any $a$.
\item The monotonicity 
\begin{equation}
G_{p} (a,a) \geq G_{Tp}(Ta, Ta)
\label{classical_metric_monotone}
\end{equation} 
holds for any stochastic matrix $T$ and for any $p$, $a$.
\end{itemize}
\end{definition}

In particular, the \textit{Fisher information metric} is defined by
\begin{equation}
G_{p} (a,b) := \sum_{i=1}^d \frac{a_i b_i}{p_i},
\label{classical_Fisher_metric}
\end{equation}
where $a = (a_1, \cdots, a_d )^{\rm T}$ and  $b = (b_1, \cdots, b_d )^{\rm T}$.
This metric is related to the Fisher information matrix as
\begin{equation}
J_{p(\theta ), kl} = G_{p(\theta )} (\partial_k p(\theta ), \partial_l p (\theta )).
\end{equation}
In completely the same manner as the proof of Proposition~\ref{thm:classical_Fisher_monotone}, we can show that the Fisher information metric~(\ref{classical_Fisher_metric}) satisfies the monotonicity~(\ref{classical_metric_monotone}).
Thus, the Fisher information metric is a monotone metric.
Conversely, any monotone metric is the Fisher information metric; in this sense, the Fisher information is unique.
This is known as the Chentsov's theorem.

\begin{theorem}[Chentsov's Theorem~\cite{Chentsov1983}]
Any monotone metric is the Fisher information metric up to normalization.
\label{thm:Chentsov}
\end{theorem}


\chapter{Classical majorization}
\label{chap:classical_majorization}

Majorization is a useful tool in various fields of information theory~\cite{Bhatia,Marshall}.  
In the context of thermodynamics, a generalized majorization plays a crucial role for characterizing a necessary and sufficient condition for state conversion by thermodynamically feasible transformations  described by Gibbs-preserving maps at finite temperature.

In Section~\ref{sec:majorization}, we start with a simplest case: (ordinary) majorization  related to infinite-temperature thermodynamics.
In Secion~\ref{sec:d_majorization}, we consider a generalized version of majorization,  called thermo-majorization, or more generally, d-majorization (relative majorization).
In Section~\ref{sec:catalytic_majorization}, we briefly mention majorization in the presence of ``catalyst.''
In Section~\ref{sec:continuous_majorization}, we roughly discuss the continuous-variable majorization, without going into mathematical details.  
In Section~\ref{sec:majorization_proof}, we provide rigorous proofs of main theorems of this chapter.

\section{Majorization}
\label{sec:majorization}

We first consider \textit{majorization} for classical probability distributions,
which characterizes state convertibility in a simplest thermodynamic setup where the Gibbs state is just the uniform distribution (i.e.,  all the energy levels are degenerate or the temperature is infinite, $\beta = 0$).

Let us illustrate a motivation to introduce majorization (see also Section~\ref{sec:Shannon_KL}).
Remember that if a classical distribution $p$ is mapped to $p'$ by a doubly stochastic map, the Shannon entropy increases (or does not change), $S_1(p) \leq S_1(p')$.
Then, the central problem of this section is related to the converse: If the Shannon entropy increases, does there always exist a doubly stochastic map that converts  $p$ to $p'$? The answer is negative;
The Shannon entropy does not give a sufficient condition for such state conversion. 
In other words, the Shannon entropy is a monotone but not a complete monotone.
Instead, the full characterization of state convertibility in a necessary and sufficient way is given by majorization.

Let $p$ and $p'$ be classical probability distributions.   We define $p^\downarrow$ by rearranging the components of $p= (p_1, p_2, \cdots, p_d)^{\rm T}$ in the decreasing order: $p_1^\downarrow \geq p_2^\downarrow \geq \cdots \geq p_d^\downarrow$.  We also define $p_i^{\prime \downarrow}$ in the same manner.  Majorization is then defined as follows.

\begin{definition}[Majorization]
Let $p, p' \in \mathcal P_d$.
We say that $p$ majorizes $p'$, written as $p' \prec p$,  if for all $k=1,2, \cdots,  d$,
\begin{equation}
\sum_{i=1}^k p^{\prime \downarrow}_i \leq \sum_{i=1}^kp^{\downarrow}_i.
\label{majorization}
\end{equation}
\label{majorization_definition}
\end{definition}

This definition implies that $p'$ is ``more random'' or ``more uniformly distributed'' than $p$.
To visualize this, we can rephrase Definition~\ref{majorization_definition} by the \textit{Lorenz curve} as follows.
As shown in Fig.~\ref{fig:majorization1},
we plot $1/d, 2/d , \cdots, 1$ (at equal spaces) on the horizontal axis, and plot $p_1^\downarrow, p_1^\downarrow+p_2^\downarrow, \cdots,  p_1^\downarrow + \cdots + p_d^\downarrow (=1)$ on the vertical axis.  
We obtain a concave polyline by connecting these points, which is called the Lorenz curve of $p$.
It is obvious from the definition that $p' \prec p$ holds if and only if  the Lorenz curve of $p$ lies above that of $p'$.

\begin{figure}[t]
\begin{center}
\includegraphics[width=7cm]{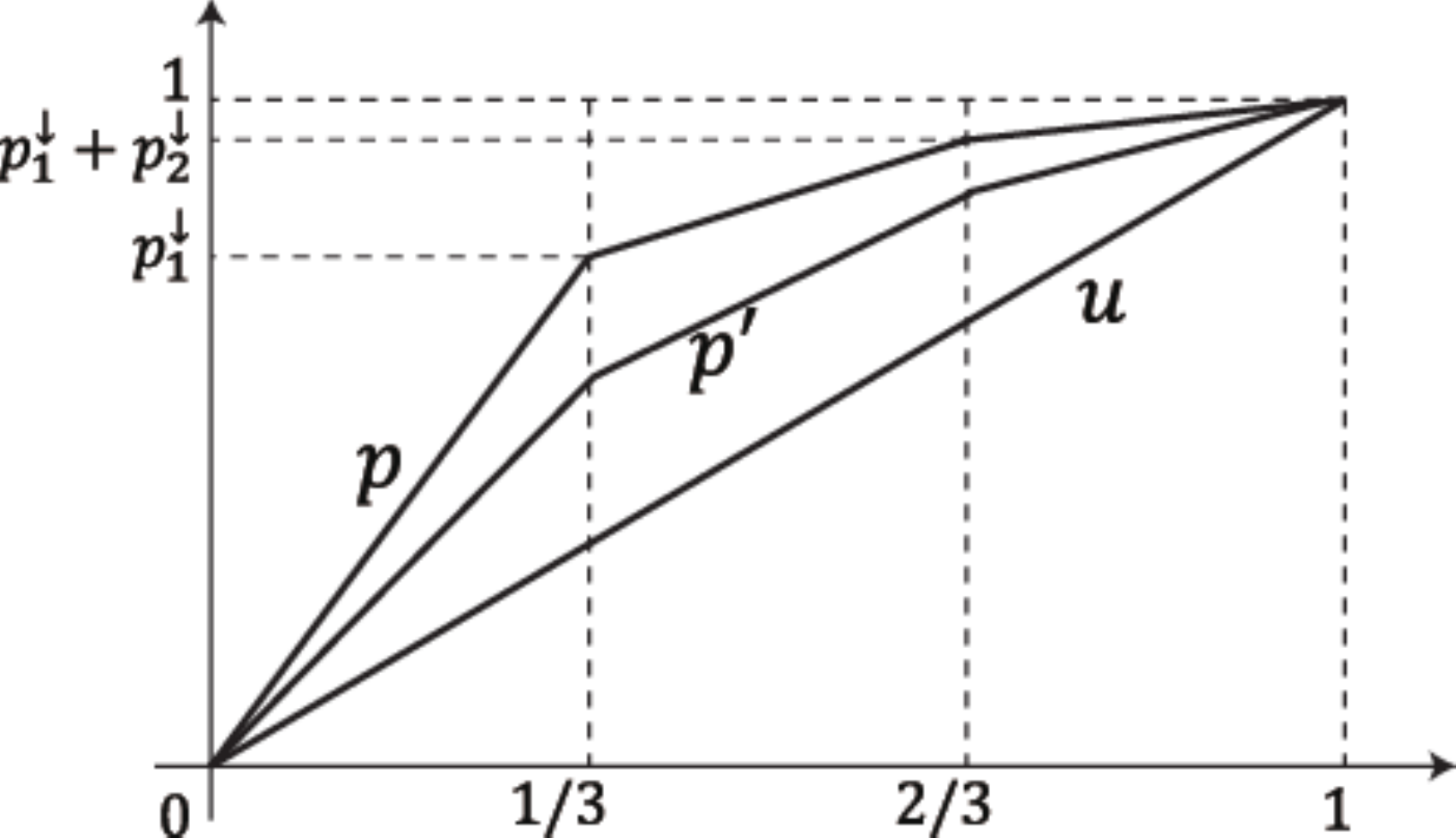}
 \end{center}
 \caption{An example of the Lorenz curve ($d=3$), where $p' \prec p$. The uniform distribution $u$ is represented by the diagonal line, which is majorized by any other distributions.} 
\label{fig:majorization1}
\end{figure}

The Lorenz curve of the uniform distribution $u$ is the diagonal line between $(0,0)$ and $(1,1)$.  Thus, $u \prec p$ holds for all $p$.  
We remark that $\prec$ is not a total order but a preorder. 
In fact,  $p \prec p$ holds for any $p$, and if $p'' \prec p'$ and $p' \prec p$ then $p'' \prec p$ holds.  However, there is a pair of distributions $p,p'$ such that neither $p' \prec p$ nor $p \prec p'$ holds; this happens when the Lorentz curves of $p$ and $p'$ cross with each other.

The definition of majorization can be equivalently rephrased in several ways as follows (see also, e.g., Ref.~\cite{Bhatia}):

\begin{theorem}
Let $p,p' \in \mathcal P_d$. The following are equivalent.
\begin{enumerate}
\item $p' \prec p$.
\item For all $t \in \mathbb R$,
\begin{equation}
\sum_{i=1}^d | p'_i -t | \leq \sum_{i=1}^d | p_i -t |.
\end{equation}
\item For all convex functions $f$,
\begin{equation}
\sum_{i=1}^d f(p'_i) \leq \sum_{i=1}^d f( p_i ).
\label{majorization_monotone}
\end{equation}
\item There exists a doubly stochastic matrix $T$ such that $p' = Tp$.
\end{enumerate}
\label{thm:majorization}
\end{theorem}

We will present the proof of this theorem in Section~\ref{sec:majorization_proof}.
Here we only remark on the following point.
Given (iv) above,  (iii) is just the monotonicity  in the form of Corollary~\ref{f_monotonicity_corollary}.
In particular, by taking  $f(x) :=  x\ln x$, inequality (\ref{majorization_monotone}) implies that the Shannon entropy does not decrease under  doubly stochastic maps.
Thus, Theorem~\ref{thm:majorization} implies that  the Shannon entropy $S_1 (p)$ does not provide a sufficient  condition for state convertibility under doubly stochastic maps; we need to take into account all convex functions $f$ to obtain a sufficient condition.
In terms of resource theory, 
 $\sum_{i=1}^d f( p_i )$ with all convex functions constitute a complete set of monotones for doubly stochastic maps.

We note an explicit example that $S_1(p) \leq S_1(p')$ is not sufficient for the existence of a doubly stochastic matrix $T$ such that $p' = Tp$~\cite{Gour}. 
Let $p=(2/3,1/6,1/6)^{\rm T}$ and $p' = (1/2,1/2,0)^{\rm T}$. 
In this case, it is easy to check that $S_1(p) < S_1(p' )$ holds, while $p' \prec  p$ does not.

We note that the uniform distribution $u$ is the fixed point of doubly stochastic matrices: $u =Tu$.
In terms of thermodynamics, $u$ is regarded as the Gibbs state of a system with all the energy levels being degenerate or at infinite temperature $\beta = 0$;
a doubly stochastic map represents  a thermodynamic process that does not change such a special Gibbs state $u$.

We remark the following theorem as a fundamental characterization of doubly stochastic matrices in the classical case.

\begin{theorem}[Birkhoff's theorem]
Every extreme point of the set of doubly stochastic matrices is a permutation matrix.  That is, the following are equivalent.
\begin{enumerate}
\item $T$ is a doubly stochastic matrix.
\item $T$ can be written as a convex combination of permutation matrices: There exist permutation matrices $P_k$ and coefficients $r_k \geq 0$ with $\sum_k r_k = 1$ such that $T=\sum_k r_k P_k$.
\end{enumerate}
\label{thm:Birkhoff}
\end{theorem}

\begin{proof}
(ii) $\Rightarrow$ (i) is trivial.  (i) $\Rightarrow$ (ii) is Theorem II.2.3 of Ref.~\cite{Bhatia}, whose proof is not very easy. $\Box$
\end{proof}

So far, we have considered majorization of probability distributions in $\mathcal P_d$. 
More generally, we can define majorization of vectors in $\mathbb R^d$:
For $p,p' \in \mathbb R^d$, we write $p' \prec p$ if inequality~(\ref{majorization})  and $\sum_{i=1}^d p_i = \sum_{i=1}^d p_i'$ are satisfied.
Theorem~\ref{thm:majorization} still holds under this definition, where the proof goes in completely the same manner.
From this viewpoint, we remark the following proposition for characterization of doubly stochastic matrices:

\begin{proposition}[Theorem II.1.9 of \cite{Bhatia}]
A matrix $T$ is doubly stochastic if and only if $Tp \prec p$ for all $p \in \mathbb R^d$.
\label{prop:characterize_DSM}
\end{proposition}

\begin{proof}
The ``only if'' part is (iv) $\Rightarrow$ (i) of Theorem~\ref{majorization_monotone}. We can prove the ``if'' part by choosing $p$ to be the uniform distribution $u$ and the distributions of the form $(0, \cdots, 0, 1, 0, \cdots, 0)^{\rm T}$.
$\Box$
\end{proof}

We finally remark on the concept called \textit{Schur-convexity}.
A function $F : \mathbb R^d \to \mathbb R$ is called a Schur-convex function, if $p' \prec p$ implies $F(p') \leq F(p)$ for any $p,p' \in \mathbb R^d$.
In other words, a Schur-convex function is a monotone of majorization. 
For example, $F(p) := \sum_{i=1}^d f(p_i)$ with $f$ being a convex function is Schur-convex from Theorem~\ref{thm:majorization} (iii).
The following proposition is known as a characterization of Schur-convex functions.

\begin{proposition}[Theorem II.3.14 of~\cite{Bhatia}]
A differentiable function $F: \mathbb R^d \to \mathbb R$ is Schur-convex, if and only if $F$ is permutation invariant (i.e., $F(Pp) = F(p)$ for all permutation $P$) and for all $p \in \mathbb R^d$ and $i,j$,
\begin{equation}
( p_i - p_j ) \left( \frac{\partial F}{\partial p_i} - \frac{\partial F}{\partial p_j} \right) \geq 0.
\end{equation}
\end{proposition}

\section{d-majorization and thermo-majorization}
\label{sec:d_majorization}

We next consider a generalization of majorization, called \textit{d-majorization} (or \textit{relative majorization}) for classical distributions (see also, e.g., Refs.~\cite{Brandao2015,Marshall,Renes2015}).
This concept was originally introduced as a generalization of majorization in the context of statistical comparison~\cite{Blackwell,Torgersen,Ruch}.
While ordinary majorization concerns convertibility from a  distribution to another distribution,  d-majorization concerns convertibility between pairs of  distributions.
A special class of d-majorization is called \textit{thermo-majorization}~\cite{Horodecki2013}, which characterizes a necessary and sufficient condition for state convertibility by thermodynamic processes at finite temperature $\beta > 0$.

We first define d-majorization by using the (generalized or relative) Lorenz curve.
We consider classical probability distributions $p=(p_1, \cdots, p_d)^{\rm T}$ and $q=(q_1, \cdots, q_d)^{\rm T}$.
We define $p^\ast$ and $q^\ast$ by rearranging their components such that $p_1^\ast / q_1^\ast \geq p_2^\ast / q_2^\ast \geq \cdots \geq p_d^\ast / q_d^\ast$ holds, where the ways of rearranging the components are the same for $p$ and $q$.
As shown in Fig.~\ref{fig:majorization2} (a), we plot $q_1^\ast, q_1^\ast+q_2^\ast, \cdots, q_1^\ast + \cdots + q_d^\ast (=1)$ on the horizontal axis and plot $p_1^\ast, p_1^\ast+p_2^\ast, \cdots, p_1^\ast + \cdots + p_d^\ast (=1)$ on the vertical axis.
By connecting these points, we obtain a concave polyline, which is the (generalized) Lorenz curve of  the pair $(p,q)$.

\begin{definition}
Let $p,q,p',q' \in \mathcal P_d$.
We say that $(p,q)$ d-majorizes $(p', q')$, written as $(p' , q' ) \prec (p, q)$, if the Lorenz curve of $(p,q)$ lies above that of $(p', q')$.
\label{def:d_majorization}
\end{definition}

Figure \ref{fig:majorization2} (b) shows an example of the Lorenz curves with $(p', q') \prec (p,q)$.  The diagonal line describes the Lorenz curve of $(q'',q'')$ for any $q''$, which is d-majorized by all other $(p,q)$.

\begin{figure}[t]
\begin{center}
\includegraphics[width=12cm]{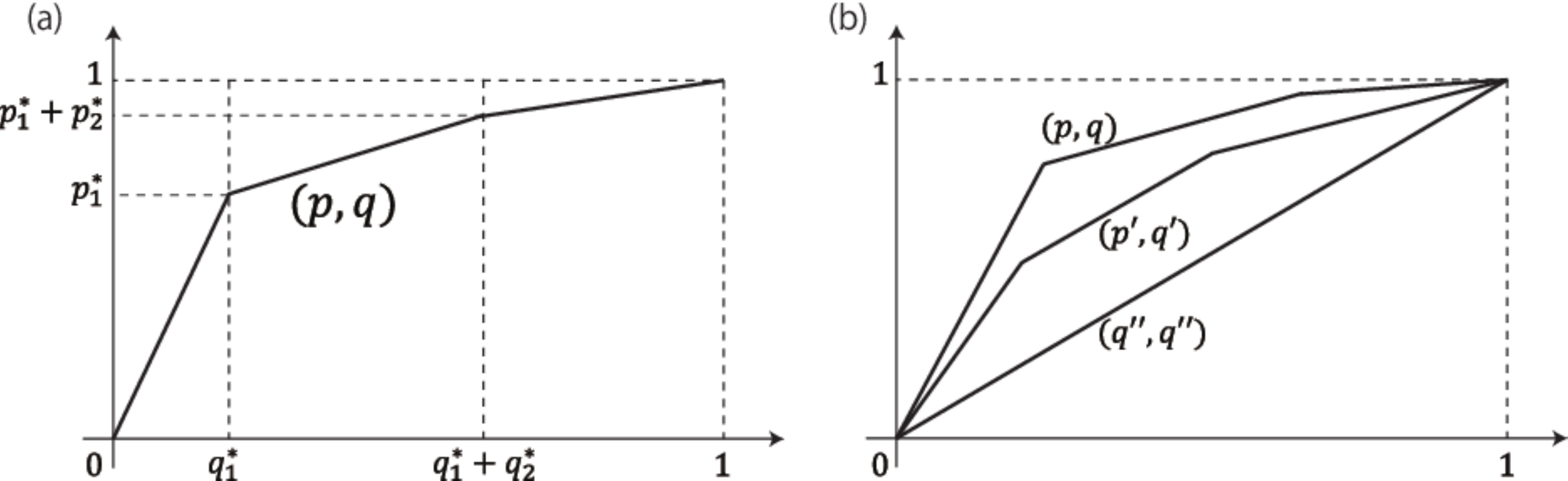}
 \end{center}
 \caption{Lorenz curves for d-majorization ($d=3$).  (a) An example of the Lorenz curve of $(p,q)$.  (b)  An example of a situation where $(p', q') \prec (p,q)$.  The diagonal line describes $(q'',q'')$.} 
\label{fig:majorization2}
\end{figure}

As is the case for ordinary majorization, the definition of d-majorization can be equivalently rephrased in several ways as follows, which is referred to as the Blackwell's theorem~\cite{Blackwell} (see also Ref.~\cite{Torgersen}):

\begin{theorem}[Blackwell's theorem]
Let $p,q,p',q' \in \mathcal P_d$ and suppose that $q,q'$ have full rank.
Then, the following are equivalent.
\begin{enumerate}
\item $(p', q' ) \prec (p,q)$.
\item For all $t \in \mathbb R$,
\begin{equation}
\sum_{i=1}^d | p'_i -t q_i'| \leq \sum_{i=1}^d | p_i -t q_i|.
\label{d_majorization_inequality}
\end{equation}
\item For all convex functions $f$,
\begin{equation}
\sum_{i=1}^d q_i' f\left( \frac{p'_i}{q_i'} \right) \leq \sum_{i=1}^d q_i f \left( \frac{p_i}{q_i} \right).
\label{d_majorization_monotone}
\end{equation}
\item There exists a stochastic matrix $T$ such that $p' = Tp$ and $q' = Tq$.
\end{enumerate}
\label{thm:d_majorization}
\end{theorem}

Clearly,  Theorem~\ref{thm:majorization} for majorization is regarded as a  special case of Theorem~\ref{thm:d_majorization} for d-majorization, by letting $q=q' = u$. 
While we will postpone the proof of Theorem~\ref{thm:d_majorization} to Section~\ref{sec:majorization_proof}, we here make some remarks on the proof.

The most nontrivial part of the above theorem is (i) $\Rightarrow$ (iv).
A way to obtain an intuition about it is to consider majorization for continuous variables as discussed  in Ref.~\cite{Ruch}.
In fact,  by applying a variable transformation to the Lorentz curve of a continuous variable, we can see that  d-majorization is a special case of continuous (ordinary) majorization.
We will discuss this idea in Section~\ref{sec:continuous_majorization} in detail, by which our rigorous proof in  Section~\ref{sec:majorization_proof} is inspired.
We also note that there is an alternative  direct and ``graphical'' proof of  (i) $\Rightarrow$ (iv) of Theorem~\ref{thm:d_majorization}, where one does not even need to invoke  (i) $\Rightarrow$ (iv) of Theorem~\ref{thm:majorization}, as shown in Ref.~\cite{Shiraishi2020}.

Given (iv) above,  (iii)  is just the monotonicity in the form of Lemma~\ref{f_monotonicity_lemma}.
The KL divergence with $f(x) = x \ln x$ does not provide a sufficient condition for convertibility of distributions; instead, we need all convex functions $f$.
In terms of resource theory, 
$\sum_{i=1}^d q_i f( p_i / q_i )$ with all convex functions constitute  a complete set of monotones  for  stochastic maps.
We note that $\sum_{i=1}^d q_i f(p_i / q_i)$  is the $f$-divergence defined in Eq.~(\ref{f-divergence}), if  $f(1)=0$ is satisfied  (this is always possible by adding a constant) and $f(x)$ is strictly convex at $x=1$.

If $q=q'$, d-majorization is called thermo-majorization, where the condition (iv) of Theorem~\ref{thm:d_majorization} reduces to $p' = Tp$ and $q = Tq$, that is, $q$ is a fixed point of $T$.

\begin{definition}[Thermo-majorization]
Let $p,q,p' \in \mathcal P_d$.
We say that $p$ thermo-majorizes $p'$ with respect to $q$, if $(p' , q) \prec (p, q)$.
\end{definition}

Because the Lorenz curve of $(q,q)$ is the diagonal line,  $q$ is thermo-majorized by any other distributions with respect to $q$ itself: $(q,q) \prec (p,q)$ for all $p$.
Therefore, the thermo-majorization relation characterizes how ``close'' a distribution $p$ is to  $q$.

By the naming of thermo-majorization, we have in mind that $q$ is a Gibbs state  of a Hamiltonian, written as $q=p^{\rm G}$.
We note that any distribution $q$ of full rank is regarded as a Gibbs state of some Hamiltonian.
In this case, $q = Tq$ implies that $T$ does not change the Gibbs state, which is called a Gibbs-preserving map (see also Chapter~\ref{chap:classical_thermodynamics}).
In terms of resource theory of thermodynamics, Gibbs-preserving maps are specified as free operations (and Gibbs states are free states).
Note that at infinite temperature, the Gibbs state is $u$ and a Gibbs-preserving map is a doubly-stochastic map.

\

We next remark that the min and the max divergences, $S_0(p \| q)$ and $S_\infty (p \| q)$, can be visualized by the Lorenz curve of $(p,q)$, as illustrated in Fig.~\ref{fig:majorization3}.  
The following theorem straightforwardly follows from this graphical representation.

\begin{figure}[t]
\begin{center}
\includegraphics[width=7cm]{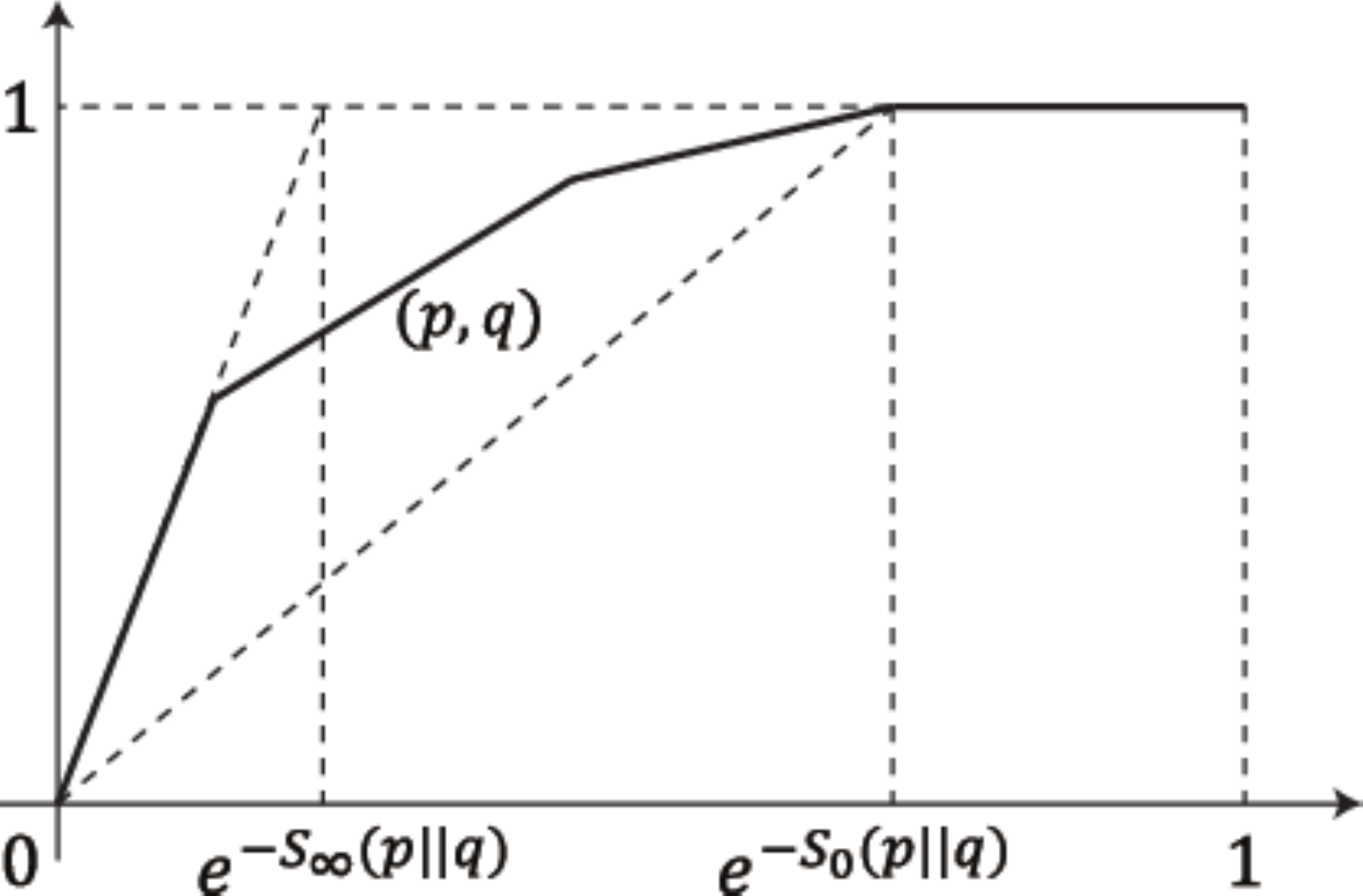}
 \end{center}
 \caption{Graphical representation of the relationship between the Lorenz curve of $(p,q)$ and the divergences $S_0(p \| q)$ and $S_\infty (p \| q)$, 
where  $e^{-S_0(p \| q)} = \sum_{i: p_i > 0} q_i$ and $e^{-S_\infty(p \| q)} =\left( \max_i \left\{ p_i / q_i \right\} \right)^{-1}$.} 
\label{fig:majorization3}
\end{figure}

\begin{theorem}[Conditions for state conversion]
\
\begin{description}
\item[(a) Necessary conditions:] If $(p', q') \prec (p,q)$, then
\begin{equation}
S_0(p \| q ) \geq S_0 (p' \| q' ), \ \ \ S_\infty(p \| q ) \geq S_\infty (p' \| q' ).
\end{equation}
\item[(b) Sufficient condition:] $(p', q' ) \prec (p,q)$ holds, if (but not only if)
\begin{equation}
S_\infty (p' \| q' ) \leq S_0 (p \| q ).
\end{equation}
\end{description}
\label{thm:asymp0}
\end{theorem}

\begin{proof}
These are obvious from Fig.~\ref{fig:majorization3}.  We note that (a) is nothing but the monotonicity of  $S_0(p \| q )$ and $S_\infty(p \| q )$, given Theorem~\ref{thm:d_majorization}. $\Box$
\end{proof}

In the above theorem, the necessary condition and the sufficient condition are distinct, suggesting that there is no single necessary and sufficient condition for d-majorization in terms of the R\'enyi divergences (see also Fig.~\ref{fig:quantum_d_majorization} in  Section~\ref{sec:quantum_d_majorization}).
However, if we take the asymptotic limit, the R\'enyi $0$- and $\infty$-divergences can collapse to a single value for a broad class of distributions, giving a complete monotone for d-majorization.  This is the main topic of Section~\ref{sec:asymptotic}.


Finally, we note that, as is the case for ordinary majorization, the concept of d-majorization can be extended to general $p,q,p',q' \in \mathbb R^d$, where $q_i > 0$, $q'_i > 0$ for all $i$ and  $\sum_{i=1}^d q_i = \sum_{i=1}^d q'_i$ are always assumed.
Then, we write $(p', q') \prec (p, q)$ if the Lorenz curve of $(p,q)$ lies above that of $(p',q')$ and $\sum_{i=1}^d p_i = \sum_{i=1}^d p'_i$ is satisfied.
Under this definition, Theorem~\ref{thm:d_majorization} still holds.

\section{Catalytic majorization}
\label{sec:catalytic_majorization}

It is useful to introduce a \textit{catalyst} for state transformations, where a catalyst means an auxiliary system whose states are the same before and after the transformation. 
Here we briefly introduce some important results about majorization in the presence of catalyst, without going into the proofs.
This section is logically independent of the subsequent sections.

We first define catalytic majorization, where, as mentioned later, the condition that the initial and final states of the catalyst are exactly the same is crucial.

\begin{definition}[Catalytic majorization]
Let $p, p' \in \mathcal P_d$.
We say that $p$ majorizes $p'$ with the aid of a catalyst, or $p$ can be \textit{trumped} into $p'$, if there exists a finite-dimensional distribution $r \in \mathcal P_N$ such that $p' \otimes r \prec p \otimes r$. 
\label{def:catalytic_majorization}
\end{definition}

An example that $p' \prec p$ does not hold but $p$ can be trumped into $p'$ is given by $p=(4/10,4/10,1/10,1/10)^{\rm T}$ and $p'=(1/2,1/4,1/4,0)^{\rm T}$, along with $r =(6/10,4/10)^{\rm T}$~\cite{Jonathan1999}.

A necessary and sufficient condition for catalytic majorization was proved in Refs.~\cite{Turgut2007,Klimesh2007}, which is stated as follows.

\begin{theorem}[Theorem 1 of \cite{Turgut2007}; Theorem 2 of \cite{Klimesh2007}]
Let $p, p' \in \mathcal P_d$ have full rank and suppose  that $p^\downarrow \neq p'^\downarrow$.  Then, $p$ can be trumped into $p'$, if and only if $f_\alpha (p') < f_\alpha (p)$ for all $\alpha \in (-\infty, \infty)$, where
\begin{equation}
f_\alpha (p) := \left\{
\begin{array}{cc}
\ln \sum_i p_i^\alpha &  (\alpha > 1), \\
\sum_i p_i \ln p_i & (\alpha = 1), \\
-\ln \sum_i p_i^\alpha & (0<\alpha < 1), \\
-\sum_i \ln p_i & (\alpha = 0), \\
\ln \sum_i p_i^\alpha & (\alpha < 0).
\end{array}
\right.
\end{equation}
\label{catalytic_theorem1}
\end{theorem}

It is worth noting that there is a pair of $(p,p')$ that satisfies $f_\alpha (p') < f_\alpha (p)$ for all $\alpha \in (-\infty, 1)$ and $\alpha \in (1, \infty)$ but satisfies $f_1 (p') = f_1 (p)$ (see Ref.~\cite{Turgut2007} for an explicit example).  In this case, $p$ cannot be trumped into $p'$.

In  Theorem~\ref{catalytic_theorem1}, $f_\alpha$ is the same as the R\'enyi $\alpha$-entropy (\ref{negative_alpha_entropy}) (including negative $\alpha$) up to a negative coefficient, except for $\alpha = 0$.
By allowing an  infinitesimal error in the final state, the $\alpha = 0$ case turns into the R\'enyi $0$-entropy  as follows.

\begin{lemma}[Proposition 4 of \cite{Brandao2015}]
Let $p, p' \in \mathcal P_d$.  The following are equivalent.
\begin{enumerate}
\item  For any $\varepsilon > 0$, there exists a distribution $p_\varepsilon'$ such that $p$ can be trumped into $p_\varepsilon'$ and $D(p', p_\varepsilon' ) \leq \varepsilon$. 
\item $S_\alpha (p) \leq S_\alpha (p')$ holds for all $\alpha \in (-\infty, \infty)$.
\end{enumerate}
\label{Brandao_lemma}
\end{lemma}

A slightly different definition of approximate trumping from the above has been considered in Ref.~\cite{Aubrun2008}, where the necessary and sufficient condition is given by  $S_\alpha (p) \leq S_\alpha (p')$ only with $\alpha \in (1, \infty)$ (Theorem 1 of Ref.~\cite{Aubrun2008}; see also discussion in Ref.~\cite{Brandao2015}).

We next consider the catalytic d-majorization.

\begin{definition}[Catalytic d-majorization]
Let $p, q, p', q' \in \mathcal P_d$.
We say that $(p,q)$ d-majorizes $(p', q')$ with the aid of catalysts, or $(p,q)$ can be \textit{d-trumped} into $(p', q')$, if there exist finite-dimensional distributions $r, s \in \mathcal P_N$ such that $( p' \otimes r, q' \otimes s )    \prec  ( p \otimes r, q \otimes s)$.
\label{def:catalytic_d}
\end{definition}

In the above definition, it is often possible to take $s$ as the uniform distribution $u$.  In the case of thermo-majorization, this implies that the Hamiltonian of the catalyst can be taken trivial.
In fact, Lemma~\ref{Brandao_lemma} can be further generalized to the following form.

\begin{theorem}[Theorem 17 of \cite{Brandao2015}]
Let $p, q, p', q' \in \mathcal P_d$ have full rank.  Then, the following are equivalent.
\begin{enumerate}
\item For any $\varepsilon > 0$, there exists a distribution  $p_\varepsilon' \in \mathcal P_d$ 
such that
$(p,q)$ can be d-trumped into $(p_\varepsilon', q')$ and 
$D(p',p_\varepsilon') \leq \varepsilon$.
\item $S_\alpha ( p \| q ) \geq S_\alpha (p' \| q' )$ for all $\alpha \in (-\infty, \infty)$.
\end{enumerate}
Moreover,  the catalyst distribution $s$ of Definition~\ref{def:catalytic_d} for the d-trumping of  (i) can be taken as the uniform distribution $u$.
\end{theorem}

The above theorem implies that the R\'enyi $\alpha$-divergences with $-\infty < \alpha < \infty$ constitute a complete set of monotones in the presence of catalyst.

\

Once we remove the requirement that the catalyst state exactly goes back to the initial state without any correlation with the system, the characterization of state convertibility becomes drastically different.
In fact, any state conversion becomes possible in the absence of this requirement~\cite{vanDam2003},
 even if the error in the final state is arbitrarily small, as long as it is independent of the sizes of the system and the catalyst.  
In other words, the majorization structure becomes trivial in the presence of such non-exact catalyst,
 which  is called  the embezzling phenomenon.
  In the context of thermodynamics~\cite{Brandao2015}, this implies that 
one can extract more work than the standard thermodynamic bound by ``embezzling'' an auxiliary system whose operation is not exactly cyclic.

\begin{theorem}[Main result of ~\cite{vanDam2003}]
For any $\varepsilon > 0$ and for any  $p, p' \in \mathcal P_d$, there exist a distribution $r \in \mathcal P_N$ of a catalyst  and a distribution $r' \in \mathcal P_{dN}$ such that $r' \prec p \otimes r$  and $\| r' - p' \otimes r\|_1 < \varepsilon$.
\end{theorem}

On the other hand, it is known that there are several ways to make the non-exact majorization structure nontrivial by considering slightly different setups from the case of the embezzling phenomenon; We now suppose that the change of the catalyst state is ``modestly'' non-exact. 
Remarkably, the KL divergence often serves as a single ``complete monotone'' for such setups, as described below.
Therefore, in the presence of such modestly non-exact catalyst, thermodynamic transformation can be completely characterized by the standard nonequilbrium free energy (\ref{1_free_energy_classical}) introduced in Section~\ref{sec:classical_second_law} based on the KL divergence.

First, in Ref.~\cite{Brandao2015}, it has been shown that  if the approximation error scales modestly in the dimension of the catalyst,
then the majorization condition is only given by the Shannon entropy up to small correction terms. 

\begin{theorem}[Theorem 23 of ~\cite{Brandao2015}]
Let $p, p' \in \mathcal P_d$. 
\begin{description}
\item[(a)] Let $\varepsilon \geq 0$.  If there exist a distribution $r \in \mathcal P_N$ of a catalyst  and a distribution $r' \in \mathcal P_{dN}$ such that $r' \prec p \otimes r$ and $\| r' - p' \otimes r \|_1 \leq \varepsilon / \ln N$, then
\begin{equation}
S_1 (p) \leq S_1 (p') - \varepsilon - \varepsilon \frac{\ln d}{\ln N} - h \left(  \frac{\varepsilon}{\ln N} \right),
\end{equation}
where $h(x) := - x \ln x - (1-x ) \ln (1-x)$.
\item[(b)]  If $S_1 (p ) < S_1 (p')$,
then for any sufficiently large $N$, there exist a distribution $r \in \mathcal P_N$ of a catalyst and a distribution $r' \in \mathcal  P_{dN}$ such that $r' \prec p \otimes r$ and $\| r' - p' \otimes r \|_1 \leq \exp ( - c \sqrt{\ln N} )$ for some constant $c>0$.  
\end{description}
\end{theorem}

We next consider the situation where several subsystems in the catalyst can form correlations in the final state~\cite{Muller2016,Lostaglio2015}.
This implies that independence of subsystems in the catalyst is a resource of state transformation.
The case of majorization is given by the following theorem (see Theorem 1 of Ref.~\cite{Lostaglio2015} for the thermo-majorization case).

\begin{theorem}[Theorem 1 of \cite{Muller2016}]
Let $p, p' \in \mathcal P_d$  with $p^\downarrow \neq p'^\downarrow$.
The following are equivalent.
\begin{enumerate}
\item There exist distributions  $r_1, r_2, \cdots, r_k \in \mathcal P_{N}$ of a $k$-partite catalyst and a distribution $r \in \mathcal P_{kN}$ with marginals $r_1, r_2, \cdots, r_k$,  such that $p' \otimes r \prec p \otimes r_1 \otimes r_2 \otimes \cdots \otimes r_k$.
\item $S_0 (p ) \leq S_0 (p')$ and $S_1 (p) < S_1 (p')$ hold.
\end{enumerate}
Moreover, we can choose $k=3$.
\label{correlated_catalytic_theorem0}
\end{theorem}

Finally, we consider the case that a small amount of correlation between the system and the catalyst is allowed in the final state, while the marginal state of the catalyst should be exactly the same as the initial state.
This is referred to as correlated-catalytic transformation~\cite{Gallego2016,Wilming2017}.

\begin{theorem}[Main theorem of \cite{Muller2018}]
Let $p, p' \in \mathcal P_d$  with $p^\downarrow \neq p'^\downarrow$.
The following are equivalent.
\begin{enumerate}
\item There exist a distribution  $r \in \mathcal P_N$ of a catalyst  and a distribution $r' \in \mathcal P_{dN}$,  such that $r' \prec p \otimes r$ and the marginal distributions of $r'$ equal $p'$ and $r$.
\item $S_0 (p ) \leq S_0 (p')$ and $S_1 (p) < S_1 (p')$ hold.
\end{enumerate}
Moreover, for any $\delta > 0$,  we can take $r'$ with which the mutual information between the system and the catalyst is smaller than $\delta$.
\label{correlated_catalytic_theorem}
\end{theorem}


The above theorem can be generalized to d-majorization (Theorem 1 of Ref.~\cite{Rethinasamy2019}).
Furthermore, it can be generalized to approximate d-majorization, as proved in Theorem 7 of Ref.~\cite{Muller2018}  and Theorem 2 of Ref.~\cite{Rethinasamy2019}, which is characterized only by the KL divergence.

\begin{theorem}
Let $p, p', q, q' \in \mathcal P_d$  and suppose that $q, q'$ have full rank.  Then, the following are equivalent.
\begin{enumerate}
\item For any $\varepsilon > 0$, there exist
a distribution $p_\varepsilon' \in \mathcal P_d$,
distributions  $r_\varepsilon, s \in \mathcal P_N$ of a catalyst,
and a distribution $r_\varepsilon' \in \mathcal P_{dN}$,
such that 
$( r_\varepsilon', q' \otimes s ) \prec ( p \otimes r_\varepsilon,  q \otimes s)$,
the marginals of $r_\varepsilon'$ equal $p_\varepsilon'$ and $r_\varepsilon$,
and $D(p',p_\varepsilon') \leq \varepsilon$.
\item $S_1 ( p \| q ) \geq S_1 (p' \| q' )$  holds.
\end{enumerate}
Moreover, in (i),  we can take $s$  as the uniform distribution $u$. Also,  for any $\delta > 0$,  we can take $r_\varepsilon'$ (for a given $\varepsilon$) with which  the mutual information between the system and the catalyst  is smaller than $\delta$.
\label{cor:correlating_catalyst}
\end{theorem}

We note that  Theorem~\ref{correlated_catalytic_theorem} straightforwardly applies to the quantum case (i.e., quantum majorization defined in Section~\ref{sec:quantum_majorization}) as discussed in Ref.~\cite{Muller2018}.
The quantum version of Theorem~\ref{cor:correlating_catalyst} has been proved in Ref.~\cite{Shiraishi2020a}, except that the catalyst state, written as $s$ in the classical case above, can be taken to be the uniform distribution.
The proof in Ref.~\cite{Shiraishi2020a} is based on the asymptotic theory discussed in Section~\ref{sec_condition_information}, especially Theorem~\ref{thm:iid_divergence}, where the KL divergence appears as a single complete monotone as well.

\section{Continuous variable case}
\label{sec:continuous_majorization}

We next consider majorization and d-majorization for continuous variables, i.e., for infinite-dimensional spaces (see also Refs.~\cite{Ryff65,Ryff63} for mathematical details).
In this case,  random variable $x$ takes continuous values, and correspondingly, we write the probability \textit{density} as $p(x)$.

For simplicity, let $x \in [0,1]$. We denote the set of $L^1$-functions on $[0,1]$ just by $L^1$.
Let $\mu$ be the Lebesgue measure on $[0,1]$, and the following arguments should be read as ``almost everywhere'' with respect to $\mu$, if needed.
We also abbreviate $\int_{[0,1]}\mu(dx)$ as $\int_0^1 dx$.
Any probability density $p(x)$ is a $L^1$-function  (i.e., $p \in L^1$) that satisfies $\int_0^1 p(x)dx = 1$ and $p(x) \geq 0$. 
We note that the finite-dimensional setup  is regarded as a special case of this continuous setup by letting $p(x) := p_id$ for $(i-1)/d \leq x < i/d$.

We  consider the way of ``rearranging'' of $p \in L^1$ in order to define $p^\downarrow \in L^1$.
Let $m_p(y):= \mu [x : p(x) > y]$.
It is straightforward to see that 
\begin{equation}
p^{\downarrow} (x) := \sup \{ y : m_p(y) > x \}
\label{rearrange_c}
\end{equation}
gives a proper generalization of $p^\downarrow_i$ to the continuous case.
We note that $p^\downarrow \in L^1$ and $\int_0^1 dx p(x) = \int_0^1 dx p^\downarrow (x)$.
The Lorenz curve in the continuous case is then given by the graph of $l_p(x) := \int_0^x p^{\downarrow} (x') dx'$, which can be a smooth curve, while the Lorenz curve was a polyline in the discrete case.
Given the definition of majorization for general vectors in $\mathbb R^d$ (see the second last paragraph of Section~\ref{sec:majorization}),
we define majorization in the continuous case:

\begin{definition}[Majorization for continuous variables]
Let $p,p' \in L^1$.  We say that $p$ majorizes $p'$, written as $p' \prec p$,  if
 \begin{equation}
\int_0^y dx p'^\downarrow (x)  \leq \int_0^y dx p^\downarrow (x), \ \  \forall y \in [0,1 )
\end{equation}
and $\int_0^1 dx p'(x) = \int_0^1 dx p(x)$ are satisfied.
\end{definition}

We next define doubly stochastic maps of the continuous case as follows, given Proposition~\ref{prop:characterize_DSM} of the discrete case.

\begin{definition}[Doubly stochastic maps]
A linear map $T: L^1 \to L^1$ is called doubly stochastic, if   $Tp\prec p$ for all $p \in L^1$.
\end{definition}


The following  theorem guarantees that the above definition is indeed reasonable (see also  Ref.~\cite{Ryff63} for mathematical details).

\begin{theorem}
$T: L^1 \to L^1$ is a doubly stochastic map, if and only if
\begin{equation}
(Tp) (x') = \frac{d}{dx'} \int_0^1 K(x',x) p(x)dx,
\label{continuous_doubly_stochastic}
\end{equation}
where $K(x',x)$ is monotonically increasing in $x'$, and satisfies $K(0,x)=0$, $K(1,x)=1$, and $\int_0^1 K(x',x)dx = x'$.
Note that there are also additional technical conditions on $K(x',x)$: the essential supremum of  $V_x$ is finite where $V_x$ is the total variation of $K(x',x)$ as a function of $x'$, and $\int_0^1 K(x',x) p(x) dx$ is absolutely continuous with respect to $x'$ for any $p \in L^1$.
\end{theorem}

Here, the reason why $d/dx'$ is put outside of the integral in Eq.~(\ref{continuous_doubly_stochastic}) is to avoid the explicit appearance of generalized functions (i.e., Schwartz's distributions) such as the delta function.
However, it is intuitively easier to look at Eq.~(\ref{continuous_doubly_stochastic})  by rewriting it as
\begin{equation}
(Tp) (x') = \int_0^1 K'(x',x) p(x)dx,
\label{continuous_doubly_stochastic2}
\end{equation}
where  $K'(x',x) := \partial K(x',x) /\partial x'$ is a generalized function and satisfies 
\begin{equation}
K'(x',x) \geq 0, \ \  \int_0^1 K'(x',x) dx' = \int_0^1 K'(x',x)dx = 1.
\end{equation}
Equality~(\ref{continuous_doubly_stochastic2}) is regarded as a ``matrix'' representation of a continuous map, where $K'(x',x)$ is the integral kernel.
For example, if $T$ is the identity map, we have $K(x',x)= \theta (x'-x)$ (the step function) and $K'(x',x) = \delta (x'-x)$ (the delta function).

Under these definitions, the following theorem holds, which is the continuous version of (i) $\Leftrightarrow$ (iv) of Theorem~\ref{thm:majorization}.

\begin{theorem}[Theorem 3 of~\cite{Ryff63}]
For $p,p' \in L^1$, $p' \prec p$ holds if and only if there exists a doubly stochastic map $T$ such that $p'= Tp$.
\label{thm:majorization_continuous}
\end{theorem}

The ``if'' part is now trivial by definition of doubly stochastic maps.
The essential idea of the proof of the ``only if'' part is the same as the proof of Theorem~\ref{thm:majorization} (i) $\Rightarrow$ (iv) presented in Section~\ref{sec:majorization_proof}, while we now need functional analysis for the rigorous proof; the hyperplane separation theorem should be replaced by the Hahn-Banach separation theorem.

\

We next consider d-majorization for continuous variables in line with Ref.~\cite{Ruch} in a non-rigorous manner.
As mentioned before, we can intuitively derive  Theorem~\ref{thm:d_majorization}  for d-majorization from  Theorem~\ref{thm:majorization}  for majorization, by going through majorization of continuous variables and using a variable transformation.
Especially, we focus on (i) $\Leftrightarrow$ (iv) of Theorem~\ref{thm:d_majorization} in the following.

Let $p, q \in L^1$ be probability densities.
For simplicity, we assume that $q(x) > 0$ for all $x \in [0,1]$, while the generalization to the case with $q(x)=0$ for some intervals of $x$ is straightforward.
We consider the following variable transformation.
Let $dz:=q(x)dx$, or equivalently,
\begin{equation}
z(x) := \int_0^x q(x'')dx''.
\label{q_z}
\end{equation}
We then define $x(z)$ as the inverse function of $z(x)$.  Since $\int_0^1 q(x'')dx'' = 1$, we have $z \in [0,1]$.  With this new variable $z$, we define a probability density
\begin{equation}
p_q(z) := \frac{p(x(z))}{q(x(z))},
\label{p_q}
\end{equation}
which satisfies $\int_0^1p_q(z) dz = \int_0^1 p(x)dx = 1$.
Figure~\ref{fig:majorization_c} illustrates  this transformation.

\begin{figure}[t]
\begin{center}
\includegraphics[width=10cm]{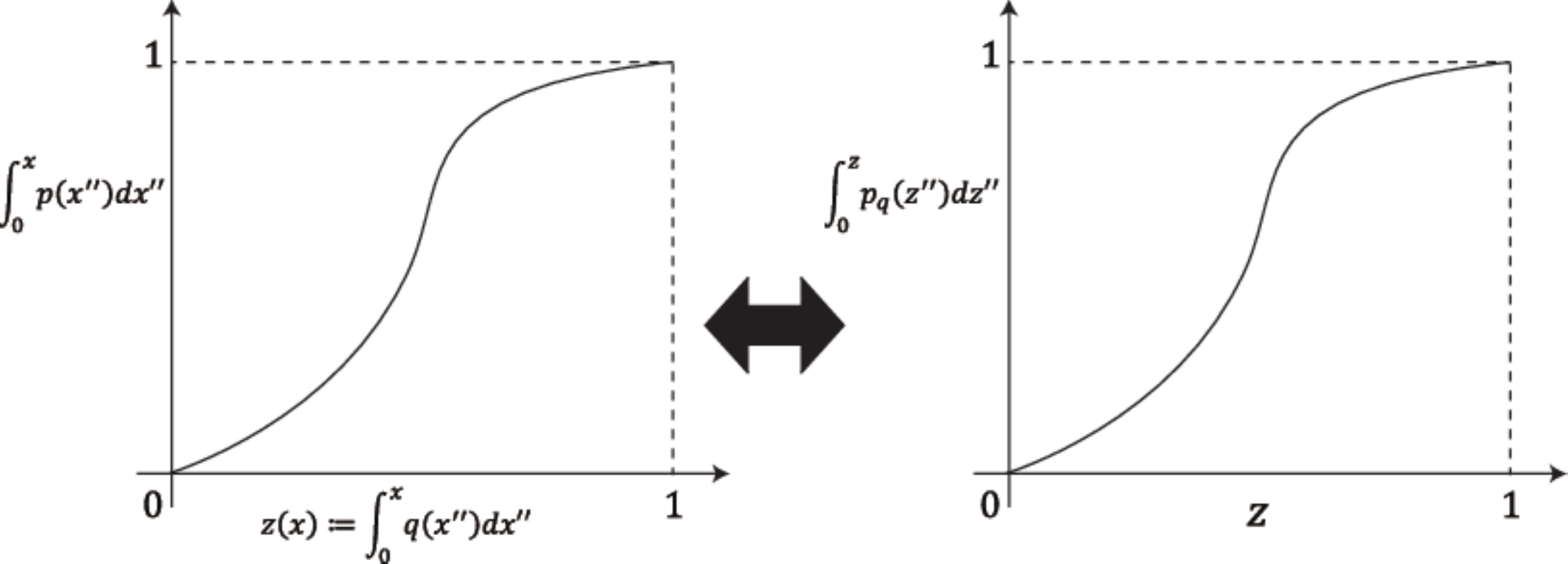}
 \end{center}
 \caption{Schematic of the variable transformation (\ref{q_z}) and (\ref{p_q}).  
We obtain the concave Lorenz curve by $l_{(p,q)}(z) := \int_0^zp_q^{\downarrow}(z'')dz''$ with Eq.~(\ref{rearrange_c}).} 
\label{fig:majorization_c}
\end{figure}

The Lorenz curve of  continuous d-majorization is obtained by rearranging $p(x)/q(x)$ in the decreasing order.
That is, we can define $p_q^{\downarrow}(z)$ with Eq.~(\ref{rearrange_c}), and obtain the Lorenz curve $l_{(p,q)}(z) := \int_0^zp_q^{\downarrow}(z'')dz''$.
In other words, d-majorization $(p',q') \prec (p,q)$ is equivalent to majorization $p'_{q'} \prec p_q$, which, from Theorem \ref{thm:majorization_continuous}, is further equivalent to that there exists a doubly stochastic map $\bar T$ such that $p'_{q'} =\bar Tp_q$.
This doubly stochastic map can be written as, with the notation of Eq.~(\ref{continuous_doubly_stochastic2}),
\begin{equation}
p'_{q'}(z') = \int_0^1 \bar K' (z',z) p_q(z)dz,
\end{equation}
where $\int_0^1 \bar K' (z',z) dz = \int_0^1 \bar K' (z',z) dz' = 1$.  We then restore the variable to the original one by defining 
\begin{equation}
K'(x',x) := q'(x') \bar K' (z'(x'),z(x)),
\end{equation}
where $z'(x')$ is defined by Eq.~(\ref{q_z}) with $q'$.
By noting that $dz' = q'(x')dx'$ and $dz=q(x)dx$, we obtain $\int_0^1 K'(x',x)dx' =1$,
$
p'(x') = \int_0^1 K'(x',x) p(x) dx,
$
and
$
q'(x') = \int_0^1 K'(x',x) q(x) dx.
$
We finally obtain the desired  stochastic map $T$, which is defined by the form of Eq.~(\ref{continuous_doubly_stochastic2}) with the integral kernel $K'(x',x)$ constructed above.

To summarize, $(p',q') \prec (p,q)$ holds if and only if there exists a stochastic map $T$ such that $p'=Tp$ and $q'=Tq$.  As a special case of this argument, we obtain (i) $\Leftrightarrow$ (iv) of Theorem \ref{thm:d_majorization} of the finite-dimensional case.

The above (non-rigorous) argument for continuous variables suggests a rigorous proof for the finite-dimensional case.
The benefit of the above argument is that we can perform the variable transformation of the horizontal axis of the Lorenz curve, which is in general impossible if the variable is discrete.    However,  if all of the components of $q$ and  $q'$ are rational numbers, we can make the same argument for the finite-dimensional case in a rigorous manner without referring to the continuous variable case.  We can then take a limit  if there are irrational components.  We will prove Theorem~\ref{thm:d_majorization} in this line in Section~\ref{sec:majorization_proof}.

\section{Proofs}
\label{sec:majorization_proof}

We here present the proofs of the two main theorems of this chapter (Theorem~\ref{thm:majorization} and Theorem~\ref{thm:d_majorization}).

\

\noindent\textbf{Proof of Theorem~\ref{thm:majorization}.}
We prove this theorem for general $p, p' \in \mathbb R^d$.

(iv) $\Rightarrow$ (iii) follows from Corollary~\ref{f_monotonicity_corollary}.

(iii) $\Rightarrow$ (ii) is trivial because $| x|$ is convex.

(i) $\Leftrightarrow$ (ii) is Theorem II.1.3 of Ref.~\cite{Bhatia}. 
 Suppose (i).  Let  $p_{k+1}'{}^\downarrow \leq t \leq p_k'{}^\downarrow$.
 By noting that  $\sum_{i=1}^d p_i' = \sum_{i=1}^d p_i$,
\begin{equation}
\sum_{i=1}^d | p'_i - t| = \sum_{i=1}^k  ( p_i'{}^\downarrow - t ) - \sum_{i=k+1}^d ( p_i'{}^\downarrow - t )  \leq \sum_{i=1}^k ( p_i^\downarrow - t ) - \sum_{i=k+1}^d ( p_i^\downarrow  - t ) \leq   \sum_{i=1}^d | p_i - t|,
\end{equation}
which implies (ii). 
Next, suppose (ii).
By choosing $t$ to be sufficiently large and sufficiently small, we obtain $\sum_{i=1}^d p_i' = \sum_{i=1}^d p_i =: C$.
Let $t=p_k^\downarrow$.  Then,
\begin{equation}
\sum_{i=1}^d | p_i - t | =  \sum_{i=1}^k ( p_i^\downarrow - t) - \sum_{i=k+1}^d ( p_i^\downarrow - t) = 2\sum_{i=1}^k  p_i^\downarrow  + (d-2k)t - C.
\end{equation}
On the other hand,
\begin{equation}
\sum_{i=1}^d  | p'_i - t| \geq \sum_{i=1}^k ( p_i'{}^\downarrow - t) - \sum_{i=k+1}^d ( p_i'{}^\downarrow - t) = 2\sum_{i=1}^k  p_i'{}^\downarrow  + (d-2k)t - C.
\end{equation}
We thus obtain $\sum_{i=1}^k p'{}_i^\downarrow \leq \sum_{i=1}^k p_i^\downarrow$, which is nothing but (i).

(i) $\Rightarrow$ (iv) has several proofs.  Here, we prove it by using the  hyperplane separation theorem in line with Ref.~\cite{Ryff65}.
Let $\mathcal D$ be the set of doubly stochastic matrices and define $\mathcal D(p) := \{ Tp: T \in \mathcal D\}$.
We consider a proof by contradiction; Suppose that $p' \prec p$ but $p' \not\in \mathcal D (p)$.
Since $\mathcal D (p)$ is a convex set, if $p' \not\in \mathcal D (p)$, then $p'$ and $\mathcal D (p)$ can be separated by a hyperplane.
That is, there exists a vector $r \in \mathbb R^d$ such that for all $T \in \mathcal D$,
\begin{equation}
(Tp, r) < (p',r)
\label{plane_product}
\end{equation}
holds, where $(\cdot, \cdot )$ describes the ordinary inner product of $\mathbb R^d$.
Since $\sum_{i=1}^d(Tp)_i  = \sum_{i=1}^dp_i'$, inequality (\ref{plane_product}) does not change by adding a constant vector $(c, c, \cdots, c)^{\rm T}$ to $r$.
Therefore, without loss of generality  we can assume that  the components of $r$  are all non-negative.
Let $P,R$ be permutation matrices such that $p = P p^\downarrow$, $r = Rr^\downarrow$.
Then we have $(p^\downarrow, r^\downarrow ) = (R P^\dagger p, r )$.
Let $T := R P^\dagger$, which is also a permutation matrix, and thus $T \in \mathcal D$.
From inequality (\ref{plane_product}), we have
\begin{equation}
(p^\downarrow, r^\downarrow ) = (T p, r ) < (p', r) = (R^\dagger p', r^\downarrow).
\label{plane_product2}
\end{equation}
We note that $R^\dagger p' \prec p$.  On the other hand, in general, if  $v',v \in \mathbb R^d$ satisfy $v' \prec v$,  we have for any  $w \in \mathbb R^d$ with non-negative components
\begin{equation}
(v^\downarrow , w^\downarrow ) \geq (v', w^\downarrow ). 
\label{plane_product3}
\end{equation}
In fact, $0 \leq \sum_{k=1}^d (w_k^\downarrow - w_{k+1}^\downarrow) \sum_{i=1}^k (v_i^\downarrow - v'_i) = \sum_{k=1}^d w_k^\downarrow (v_k^\downarrow - v'_k )$ with $w_{d+1}^\downarrow := 0$.
Then, inequalities (\ref{plane_product3}) and (\ref{plane_product2}) contradict.  Therefore, $p' \in \mathcal D (p)$.
$\Box$

\

\noindent\textbf{Proof of Theorem~\ref{thm:d_majorization}.}

The following proof is parallel to that of Theorem~\ref{thm:majorization}, except that (ii) $\Rightarrow$ (iv) is only a nontrivial part.
For simplicity, we here assume that $p,q,p',q' \in \mathcal P_d$, while the generalization to the more general case is straightforward.

(iv) $\Rightarrow$ (iii) follows from Lemma~\ref{f_monotonicity_lemma}.

(iii) $\Rightarrow$ (ii) is trivial.

(i) $\Leftrightarrow$ (ii) is also proved in essentially the same manner as (i) $\Leftrightarrow$ (ii) of Theorem~\ref{thm:majorization}.
In general, $[0,q_1^\ast), [q_1^\ast, q_1^\ast+q_2^\ast ), \cdots, [q_1^\ast+ \cdots + q_{d-1}^\ast, 1]$ and $[0,q'{}_1^\ast), [q'{}_1^\ast, q'{}_1^\ast+q'{}_2^\ast ), \cdots, [q'{}_1^\ast+\cdots + q'{}_{d-1}^\ast, 1]$ give different partitions of the interval $[0,1]$.
Then, take $\tilde q = (\tilde q_1 , \tilde q_2, \cdots, \tilde q_{d'})$ such that $[0,\tilde q_1^\ast), [\tilde q_1^\ast, \tilde q_1^\ast + \tilde q_2^\ast ), \cdots, [\tilde q_1^\ast + \cdots + \tilde q_{d'-1}^\ast, 1]$ gives a refinement of both of the partitions.
Correspondingly, we can take $\tilde p$ and $\tilde p'$ such that the Lorenz curves of $(p,q)$ and $(p',q')$ are respectively the same as those of $(\tilde p, \tilde q)$ and $(\tilde p', \tilde q)$, where the ways of rearranging $\tilde q$ to $\tilde q^\ast$ are the same for these Lorenz curves.
Therefore, it is sufficient to prove (i) $\Leftrightarrow$ (ii) for $(\tilde p, \tilde q)$ and $(\tilde p', \tilde q)$, but this can be proved in the same manner as (i) $\Leftrightarrow$ (ii) of Theorem~\ref{thm:majorization}.

We now prove the nontrivial part, (ii) $\Rightarrow$ (iv).
We first suppose that all the components of $q,q'$ are rational numbers, and let $q_i = m_i / M$, $q'_i = m'_i / M$ ($m_i,m'_i, M \in \mathbb N$) with $\sum_i m_i = \sum_i m'_i = M$.  We denote the $m_i$-dimensional uniform distribution by $u_{m_i} := (1/m_i, 1/m_i, \cdots, 1/m_i)$.
We then define the following $M$-dimensional probability vectors:
\begin{equation}
p_q := ( p_1 u_{m_1}, p_2 u_{m_2}, \cdots,  p_d u_{m_d} )^{\rm T},  \ p'_{q'} := ( p'_1  u_{m'_1},  p'_2 u_{m'_2}, \cdots, p'_d  u_{m'_d} )^{\rm T}.
\end{equation}
Then, by replacing $tM$ by $t$, inequality (\ref{d_majorization_inequality}) is written as 
\begin{equation}
\sum_{i=1}^M | p'_{q',i} -t | \leq \sum_{i=1}^M | p_{q,i} -t |.
\label{d_majorization_inequality2}
\end{equation}
Therefore, from Theorem~\ref{thm:majorization}, there exists an $M \times M$ doubly stochastic matrix $\bar T$ such that $p'_{q'} = \bar T p_q$.
We partition  $\bar T$ into several blocks, where each block is an $m'_i \times m_j$ matrix.
We sum up all the elements of each block, and obtain a $d \times d$ matrix whose elements are those sums.  We denote this $d \times d$ matrix by $\bar T'$.
We then define a matrix $T$, which is obtained by dividing all the elements of the $j$th column of $\bar T'$ by $m_j$. 
Since the sum of the $j$th column of $\bar T'$ is $m_j$, $T$ is a stochastic matrix.
Furthermore, since $\bar T$ is a doubly stochastic matrix, the sum of all the elements of the $i$th row is $m'_i$.
Therefore, we obtain
\begin{equation}
\left( 
\begin{array}{c}
m'_1 \\
m'_2 \\
\vdots \\
m'_d
\end{array}
\right) = T
\left( 
\begin{array}{c}
m_1 \\
m_2 \\
\vdots \\
m_d
\end{array}
\right),
\end{equation}
or equivalently, $q'=Tq$.  Also, $p'=Tp$ holds by construction.  Therefore, (ii) $\Rightarrow$ (iv) is proved for the case that all the components of $q,q'$ are rational.

Finally, we prove (ii) $\Rightarrow$ (iv) for the case where there are irrational components in $q,q'$ by approximating them.
Let $\{ q^{(n)} \}_{n \in \mathbb N}, \{ q'{}^{(n)} \}_{n \in \mathbb N}$ be sequences of probability vectors that converges to $q,q'$, respectively.
We suppose that all the components of these sequences are rational and that $q_i^\ast{}^{(n)} \leq q_i^\ast$C$q'{}_i^\ast{}^{(n)} \geq  q'{}_i^\ast$ for $i=1,\cdots, d-1$ and $q_d^\ast{}^{(n)} \geq q_d^\ast$C$q'{}_d^\ast{}^{(n)} \leq  q'{}_d^\ast$.
Here, $q^{(n)}$, $q'{}^{(n)} $ should be taken sufficiently close to $q$, $q'$ such that the orders of rearranging to define the Lorenz curves do not change;
If there is $i$ such that $p'{}_i^\ast / q'_i{}^\ast = p_{i+1}'{}^\ast / q'{}_{i+1}^\ast$, we need to approximate $p'$ as well in order to keep the order of rearranging.

Then, 
the Lorenz curve of $(p,q)$ is shifted to the left and that of $(p',q')$ is shifted to the right, while fixing $(0,0)$ and $(1,1)$.
Correspondingly, we can show that
\begin{equation}
\sum_{i=1}^d | p'{}_i -t q'{}_i^{(n)}| \leq \sum_{i=1}^d | p'_i -t q'_i| \leq \sum_{i=1}^d | p_i -t q_i| \leq \sum_{i=1}^d | p_i -t q_i^{(n)} |.
\end{equation}
Therefore, $(p',q') \prec (p,q)$ implies $(p',q'^{(n)}) \prec (p,q^{(n)})$.
Therefore, there exists a  stochastic matrix $T^{(n)}$ such that $p' = T^{(n)}p$C$q'{}^{(n)}=T^{(n)}q^{(n)}$.
Since the set of stochastic matrices is sequential compact, $\{ T^{(n)} \}_{n \in \mathbb N}$ has a convergent subsequence.  Let $T$ be its limit.
Then, we have $p' =Tp$ and $q'=Tq$, which proves (ii) $\Rightarrow$ (iv). $\Box$

\

As discussed in Ref.~\cite{Horodecki2013}, the above proof of (ii) $\Rightarrow$ (iv)  has a clear thermodynamic interpretation in the case of Gibbs-preserving maps.
We here briefly discuss this point of view (without caring about rational vs irrational etc.).
First, in addition to the present system (named S), we consider a heat bath B.
Suppose that for any energy level $E_i$ of S, there exists an energy level $E_i^{\rm B} := E - E_i$ of B, where $E$ is a constant.
Also suppose that the energy level $E_i^{\rm B}$ of B is $d_i$-fold degenerate with $d_i$ being  proportional to $e^{-\beta E_i}$.
Then, as in the standard argument in statistical mechanics, the Gibbs state of S is regarded as the marginal distribution of the microcanonical distribution of SB with the total energy $E$.
Thus, a Gibbs-preserving map of S corresponds to a doubly stochastic map of SB (that preserves the microcanonical distribution of SB), which is equivalent to the construction of $\bar T$ in the above proof.


\chapter{Classical thermodynamics}
\label{chap:classical_thermodynamics}
We  discuss the thermodynamic implications of classical information theory  developed in Chapter~\ref{chap:classical_entropy} and Chapter~\ref{chap:classical_majorization}.
(See also  Chapter~\ref{chap:quantum_thermodynamics} for a more comprehensive formulation of the quantum case.)

In  Section~\ref{sec:classical_second_law}, we show that the second law of thermodynamics is derived as a straightforward consequence of the monotonicity (\ref{monotone_classical}) of the KL divergence.
Here, work is a fluctuating quantity and satisfies the second law  only at the level of the ensemble average, which is  a perspective of stochastic thermodynamics.

In Section~\ref{sec:classical_work_bound}, we focus on the resource-theory approach to thermodynamics, where the  work does not fluctuate even when dynamics of the system (and thus heat) are stochastic.
This is referred to as single-shot (or one-shot) thermodynamics~\cite{Horodecki2013,Aberg2013}. 
In this setup, any entropic contribution from the work storage can be excluded, and thus work can be regarded as a purely ``mechanical'' quantity.
Specifically, we show that the work bounds are given by the R\'enyi $0$- and $\infty$-divergences in the single-shot situations on the basis of thermo-majorization.

\section{Second law and the KL divergence}
\label{sec:classical_second_law}

First, we formulate thermodynamic processes in terms of classical probability theory, and  discuss the relationship between thermodynamics and the KL divergence in the spirit of stochastic thermodynamics.
In particular, we will show that the second law of thermodynamics at the level of the ensemble average is  an immediate consequence of the monotonicity (\ref{monotone_classical}) of the KL divergence.

Let $E_i$ be the energy of level $i$ of the system; the classical Hamiltonian $H$ can be identified with $( E_1, E_2, \cdots, E_d ) \in \mathbb R^d$. 
The Gibbs state $p^{\rm G} \in \mathcal P_d$ is defined as $p^{\rm G}_i := e^{-\beta E_i}/Z$, where $Z:= \sum_i e^{-\beta E_i}$ is the partition function.  The equilibrium free energy is defined as $F := -\beta^{-1} \ln Z$.

A stochastic map $T$ is called a \textit{Gibbs-preserving map}, if it does not change the Gibbs state, i.e., 
\begin{equation}
Tp^{\rm G} = p^{\rm G}.
\label{classical_GMP}
\end{equation}
We note that if the detailed balance condition
\begin{equation}
T_{ji}e^{-\beta E_i} = T_{ij}e^{-\beta E_j}
\label{classical_detailed_balance}
\end{equation}
is satisfied, then $T$ is a Gibbs-preserving map.  (The converse is not true in general.)
We  emphasize that $\beta \geq 0$ is interpreted as the inverse temperature of the heat bath.

By substituting  $q = p^{\rm G}$ to the monotonicity (\ref{monotone_classical}) of the KL divergence, we obtain for any Gibbs-preserving map $T$ and for any $p \in \mathcal P_d$
\begin{equation}
S_1 ( p \| p^{\rm G}) \geq S_1 (Tp \|  p^{\rm G}).
\label{classlcal_Gibbs_second0} 
\end{equation}
This can be rewritten as
$S_1 (Tp) - S_1 (p) \geq \beta \left( \sum_i E_i (Tp)_i - \sum_i E_i p_i  \right)$, where the right-hand side represents the average energy change in the system by $T$.
Because we do not consider the time-dependence of the Hamiltonian  (and we do not explicitly consider the work storage) at present, no external work is performed on or extracted from the system, and thus the energy change equals the heat absorption from the heat bath.
Thus,  we can identify  $Q := \sum_i E_i (Tp)_i - \sum_i E_i p_i$ to the average  heat absorption.
We note that  the stochastic heat is given by  $Q_{ji} :=E_j - E_i $  during transition from $i$ to $j$, whose ensemble average  equals $Q$, i.e., $Q= \sum_{ij}T_{ji}p_i Q_{ji}$.

Let $\Delta S_1 :=  S_1 (Tp) - S_1 (p)$ be the change in the Shannon entropy of the system.
Inequality~(\ref{classlcal_Gibbs_second0}) is now rewritten as
\begin{equation}
\Delta S_1 \geq \beta Q,
\label{classlcal_Gibbs_second1} 
\end{equation}
which is the second law of thermodynamics in the present setup~\cite{Crooks1999,Seifert2005,Esposito2011}.
This is in the same form as the conventional Clausius inequality of equilibrium thermodynamics~\cite{Callen},
while  inequality~(\ref{classlcal_Gibbs_second1}) includes an informational entropy (i.e., the Shannon entropy) instead of the Boltzmann entropy defined only for equilibrium states. 
We note that we will derive the same inequality for the quantum case in Section~\ref{sec:quantum_entropy} (see also Section~\ref{sec:fluctuating_work}).

In stochastic thermodynamics, 
\begin{equation}
\Sigma := S_1 ( p \| p^{\rm G}) - S_1 (Tp \|  p^{\rm G}) = \Delta S_1 - \beta Q \geq 0
\label{def:entropy_production}
\end{equation}
is often called the (average) \textit{entropy production}~\cite{Seifert2012,Sagawa2012}. 
Here, $ - \beta Q $ is interpreted as the entropy increase in the heat bath~\cite{Prigogine}, which can be justified if the state of the bath is sufficiently close to the Gibbs state because of Eq.~(\ref{KL_expand}).
We note that if $T$ does not preserve the Gibbs state but preserves a nonequilibrium state $q$ (i.e., a nonequilibrium steady state), $S_1 ( p \| q ) - S_1 (Tp \| q)$ is called the excess (or non-adiabatic) entropy production~\cite{Hatano2001,Esposito2010} (see also Ref.~\cite{Parrondo2013,Horowitz2014} for the quantum extension).

Inequality~(\ref{classlcal_Gibbs_second1}) is also a starting point of thermodynamics of information~\cite{Parrondo2015}.
For example, inequality~(\ref{classlcal_Gibbs_second1}) can be regarded as a generalized Landauer's principle~\cite{Landauer}.
The original form of the Landauer's principle states that the erasure of one bit ($=\ln 2$) of information must be accompanied by $\beta^{-1} \ln 2$ of heat emission to the environment.
Correspondingly, if $\Delta S_1 = -\ln 2$ on the left-hand side of ~(\ref{classlcal_Gibbs_second1}), the heat emission is bounded as $-Q \geq \beta^{-1} \ln 2$.
Inequality~(\ref{classlcal_Gibbs_second1}) also sets a fundamental upper bound of the work extraction by utilizing information, as is the case for the thought experiments of the Szilard engine and Maxwell's demon~\cite{Sagawa2012a,Parrondo2015,Sagawa2019r,Sagawa2012l,Sagawa2013}.

We note that a doubly stochastic  map is regarded as a Gibbs-preserving map (with any Hamiltonian) at infinite temperature or a Gibbs-preserving map  (at any temperature) of the trivial Hamiltonian  whose energy levels are all degenerate; in these cases, the Gibbs state is given by the uniform distribution  $p^{\rm G} = u$, 
where the Shannon entropy does not decrease as represented by inequality (\ref{monotone_DSM}).
This may be related to the second law for adiabatic processes, stating that the entropy does not decrease if the system does not exchange the heat with the environment.

\

We next consider the second law in terms of the work and the free energy.
We now suppose that the Hamiltonian of the system is time-dependent, and  the work is performed on the system  through the time-dependence of the Hamiltonian.
Here, we do not explicitly take into account the work storage and suppose that the Hamiltonian of the system  is driven by an external agent.
This is contrastive to another formulation discussed  in Section~\ref{sec:classical_work_bound} and Chapter~\ref{chap:quantum_thermodynamics}, where the entire system is supposed to be autonomous by including the work storage and the ``clock'' degrees of freedom as a part of the total system~\cite{Horodecki2013}.

First, as a simplest situation, we suppose that the work is induced by  a quench (i.e., an instantaneous change) of the Hamiltonian from $H$ to $H' = (E_1', \cdots, E_d')$.
Let $E := \sum_i E_i p_i$ and $E' := \sum_i E_i' p_i$ be the average energies immediately before and after the quench, respectively.
Because the quench is supposed to be very quick and the heat exchange between the system and the bath is ignored during the quench, the average  work performed on the system by the quench is given by
\begin{equation}
W := E' - E.
\end{equation}
We emphasize that, in this setup, the work is fluctuating;
If the system is in $i$ immediately before the quench, the stochastic work is given by $w_i := E_i' - E_i$, which is a random variable.  Then, the average work $W$ is given by the ensemble average of $w_i$, that is, $W = \sum_i w_i p_i$.

We next consider  more general thermodynamic  processes.
Suppose that the entire process consists of multiple quench steps and relaxation processes 
between  the quenches (see  Fig.~\ref{fig:quench_and_relax} for a schematic), which we refer to as a ``quench-and-relax'' process.
The Hamiltonian is fixed during the period between the quenches.
We assume that each relaxation process is stochastically independent of other relaxation processes  and is described by a single Gibbs-preserving  map with respect to the Hamiltonian of that period.
This assumption means that stochastic dynamics is essentially Markovian and no memory effect is present.

We denote the initial and final distributions of the entire process by $p$, $p'$.
Let $H$, $H'$ be the initial and final Hamiltonians,
$p^{\rm G}$, $p^{\rm G}{}'$ be the corresponding Gibbs states, and $F$, $F'$ be the corresponding equilibrium free energies.
We consider the average work $W$ and the average heat  $Q$ during the above-mentioned entire process.
Let $\Delta E := E' - E$ be the change in the average energy, where $E := \sum_i E_i p_i$ and $E' := \sum_i E_i' p_i'$.
These energetic quantities satisfy the first law of thermodynamics, i.e., the energy conservation:
\begin{equation}
\Delta E = W + Q.
\label{classical_first_law}
\end{equation}

\begin{figure}[t]
\begin{center}
\includegraphics[width=7cm]{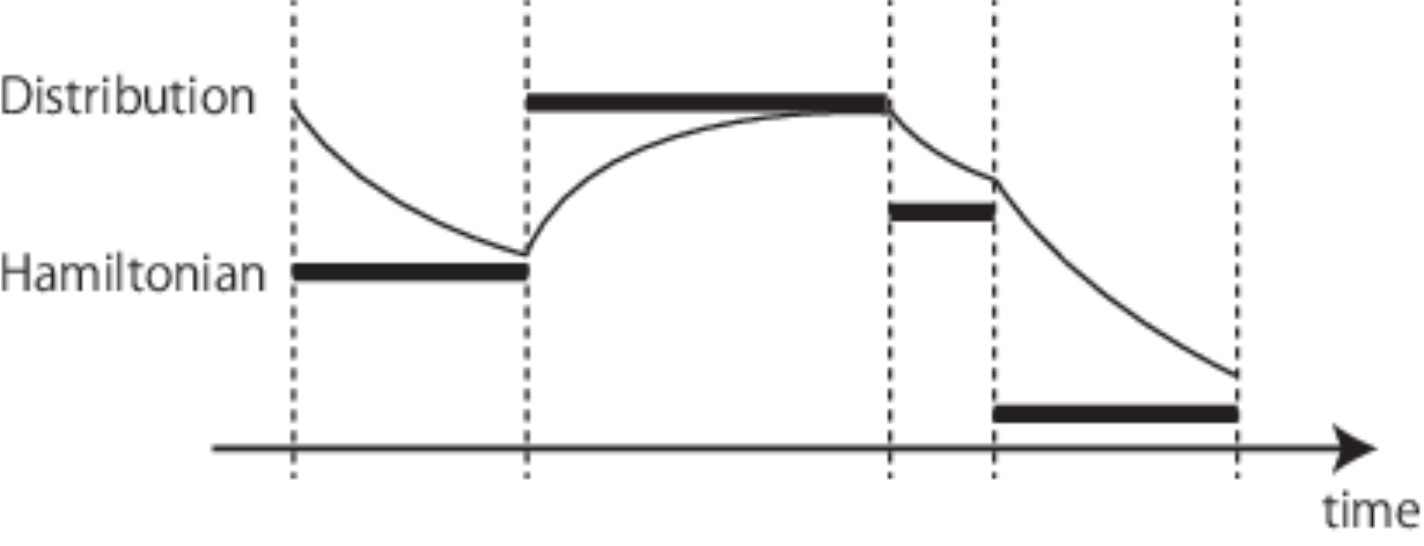}
 \end{center}
 \caption{A schematic of the quench-and-relax protocol.  The bold horizontal lines symbolically represent the Hamiltonians, and the curves symbolically represent the time evolution of the probability distribution.  The Hamiltonian is quenched at the times of the dashed vertical lines.} 
\label{fig:quench_and_relax}
\end{figure}

Again because a quench is instantaneous, the probability distribution does not change and the heat exchange is zero during it.
Thus, the second law of the form (\ref{classlcal_Gibbs_second1}) is unaffected, i.e., $\Delta S_1 \geq \beta Q$ holds with $\Delta S_1 := S_1 (p') - S_1(p)$, even in the presence of the quenches.
From the first law~(\ref{classical_first_law}), we can rewrite the second law as
\begin{equation}
W \geq \Delta E - \beta^{-1} \Delta S_1.
\label{classical_work_SQ}
\end{equation}

Now we define the nonequilibrium (R\'enyi $1$-) free energy by 
\begin{equation}
F_1 (p; H) := E - \beta^{-1} S_1 (p) = \beta^{-1} S_1 (p \| p^{\rm G}) + F.
\label{1_free_energy_classical}
\end{equation}
This quantity equals the equilibrium free energy, i,e,  $F_1 (p;H) = F$ holds, if and only if $p=p^{\rm G}$.
See also Section~\ref{sec:thermodynamic_entropy} for the $\alpha$-free energy of the quantum case.
By using the nonequilibrium free energy,  the second law~(\ref{classical_work_SQ}) is rewritten as
\begin{equation}
W \geq F_1 (p', H') - F_1 (p; H),
\label{classical_work_bound_F1}
\end{equation}
which is the work bound for the case where we allow work fluctuations.

We consider the following  special cases of the above work bound.
First, we consider the necessary work for the formation of a nonequilibrium state $p'$ starting from the Gibbs state $p^{\rm G}$ by fixing the Hamiltonian  $H = H'$.
Then,  inequality~(\ref{classical_work_bound_F1}) reduces to $ W \geq F_1(p'; H) - F = \beta^{-1} S_1 (p' \| p^{\rm G})$.

Second, we consider the work extraction from a nonequilibrium state  $p$ again  by fixing the Hamiltonian  $H = H'$.
Suppose that $p' = p^{\rm G}{}' (= p^{\rm G})$.
Then, inequality~(\ref{classical_work_bound_F1}) reduces to
$- W \leq F_1(p; H) - F = \beta^{-1} S_1 (p \| p^{\rm G})$,
where $-W$ is the extracted work.

Finally, for a transition between equilibrium states, i.e., if $p= p^{\rm G}$ and $p' = p^{\rm G}{}'$ (and in general $H \neq H'$), inequality~(\ref{classical_work_bound_F1}) reduces to the conventional work bound:
\begin{equation}
W \geq \Delta F,
\label{work_classical_equilibrium}
\end{equation}
where $\Delta F:= F' - F$ is the change in the equilibrium free energy.

We note that in the single-shot setup where the work does not fluctuate,
the above inequalities are replaced by 
inequalities (\ref{single_shot_work_classical_infty}), (\ref{single_shot_work_classical_0}), (\ref{single_shot_work_classical_equilibrium})
in Section~\ref{sec:classical_work_bound},
respectively.

\

We can construct quasi-static processes by taking the limit of the foregoing quench-and-relax processes, which achieves the equality of~(\ref{classical_work_bound_F1}).
Let $H$ and $H'$ be  the initial and final Hamiltonians, respectively.
We divide the total process into $N$ steps of quench-and-relax with the same time intervals.
Let $H^{(n)} = ( E_1^{(n)}, \cdots, E_d^{(n)} )$ be the Hamiltonian after the $n$th quench ($n= 0, 1, 2, \cdots, N$), where $H^{(0)} = H$ and $H^{(N)} = H'$.
We suppose that the energy change in each quench is of the order of $\varepsilon := 1/N$, i.e., 
$| E_i^{(n)} - E_i^{(n+1)} | = O (\varepsilon )$ for all $i$ and $n$.
Let $p^{{\rm G},(n)}$ be the Gibbs state of $H^{(n)}$.
Suppose that the initial state is the Gibbs state $p^{{\rm G},(0)}$, and that each relaxation process is long enough so that the system reaches the Gibbs state $p^{{\rm G},(n)}$ before the $(n+1)$th quench.
The quasi-static limit is then given by $N \to \infty$ or equivalently $\varepsilon \to 0$.

In this setup, $\Delta S_1 - \beta Q$ in the $n$th relaxation process ($n=1,2, \cdots, N$) is given by $S_1 ( p^{{\rm G},(n)} \| p^{{\rm G},(n+1)} )$.
Because $\| p^{{\rm G},(n)} - p^{{\rm G},(n+1)} \|_1= O (\varepsilon )$, we have  $S_1 ( p^{{\rm G},(n)} \| p^{{\rm G},(n+1)} ) = O (\varepsilon^2)$ from Eq.~(\ref{KL_expand}).
By summing up these terms over $n=1,2, \cdots, N$, we obtain for the total process
\begin{equation}
\Delta S_1 - \beta Q = O(N\varepsilon ^ 2) = O(\varepsilon).
\end{equation}
Thus, the equality in the second law (\ref{classical_work_bound_F1}) is achieved, or equivalently the average entropy production becomes zero, in the quasi-static limit $\varepsilon \to 0$.
From the law of large numbers, we can also show that 
the fluctuation of the work during the total process vanishes,
which is consistent with the single-shot scenario  discussed in Section~\ref{sec:classical_work_bound}.

The quasi-static process discussed above requires that the system is always close to the Gibbs state of the Hamiltonian at that moment.
If the initial and final states are not the Gibbs states in general, the equality in the second law (\ref{classical_work_bound_F1}) can be achieved by the following protocol (see  Fig.~\ref{fig:equality_protocol_PHS})~\cite{Parrondo2015,Sagawa2019r,Aberg2013}.
Let $p$ and $p'$ be the initial and final states, and again let $H$ and $H'$ be the initial and final Hamiltonians.
Suppose that $p$ and $p'$ have full rank.
\begin{enumerate}
\item At the initial time, we instantaneously quench the Hamiltonian from $H$ to  $\tilde H = ( \tilde E_1, \cdots, \tilde E_d )$ such that $p$ is the Gibbs state of $\tilde H$.  Explicitly, we take $\tilde E_i := - \beta^{-1} \ln p_i$ up to constant.
\item We next change the Hamiltonian from $\tilde H$ to $\tilde H'$ quasi-statically.  Here,  $\tilde H' = ( \tilde E_1', \cdots, \tilde E_d' )$ is the Hamiltonian such that $p'$ is the Gibbs state of $\tilde H'$, i.e., $\tilde E_i' := - \beta^{-1} \ln p_i'$ up to constant.
\item Finally, we instantaneously quench the Hamiltonian $\tilde H'$ to the final one $H'$.
\end{enumerate}
In (i) and (iii), the state does not change and thus the energy change (work) equals the free energy change, leading to the equality of (\ref{classical_work_bound_F1}).
In (ii), the equality is achieved because the process is quasi-static.
Therefore, the equality of   (\ref{classical_work_bound_F1}) is  achieved in the entire process. 
We  note that this protocol can be straightforwardly generalized to the quantum case~\cite{Jacobs2009}. 

The fact that the equality of  (\ref{classical_work_bound_F1}) is achievable implies that the KL divergence provides a necessary and sufficient condition of state conversion in the present setup.
This might seem different from the argument in Section~\ref{sec:d_majorization}, where we emphasized that the KL divergence does not provide a necessary and sufficient condition of state conversion.
However, of course, this is not a contradiction.
A crucial point here is that the work inevitably fluctuates in the quench steps (in contrast to the single-shot case).
This would be   reminiscent  of the case of ``modestly non-exact'' catalytic majorization  discussed in Section~\ref{sec:catalytic_majorization};
To elaborate this point of view, we need to  consider the work storage as a ``catalyst,'' which we did not explicitly take into account in this section.
See Section~\ref{sec:fluctuating_work} for a related argument.

\begin{figure}[t]
\begin{center}
\includegraphics[width=12cm]{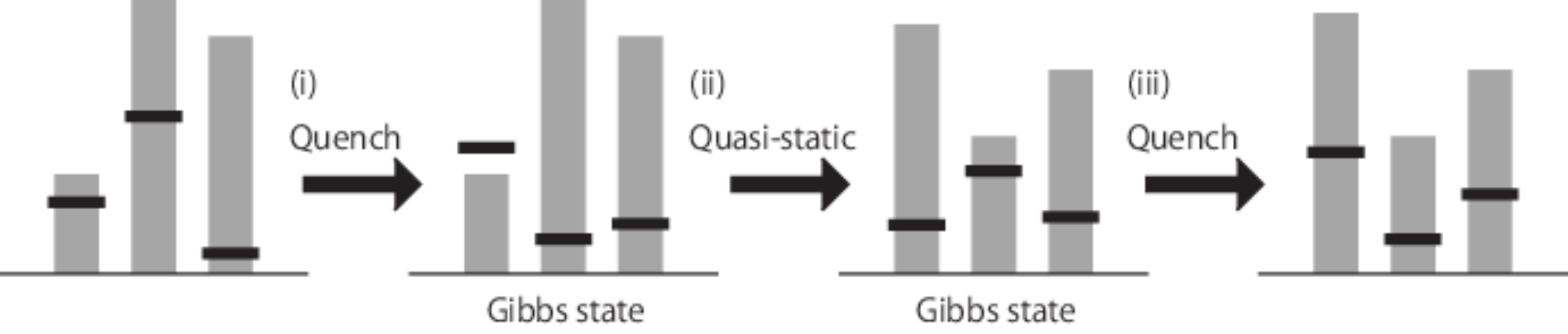}
 \end{center}
 \caption{A schematic of the protocol that achieves the equality of the second law (\ref{classical_work_bound_F1}). The bold horizontal lines represent the energy levels, and the shaded vertical bars represent the probability distributions.} 
\label{fig:equality_protocol_PHS}
\end{figure}

\section{Single-shot work bound}
\label{sec:classical_work_bound}

We next consider the single-shot situations where work does not fluctuate, and derive the work bounds as an important application of thermo-majorization.
Especially, we focus  on  the necessary work for creating a nonequilibrium state and the extractable  work  from a given nonequilibrium state~\cite{Horodecki2013,Aberg2013}.
In this section, we graphically derive the work bounds for these setups by using the Lorenz curves.

A key idea of the present setup, in contrast to the setup of the previous section,  is that we include a work storage (or a mechanical ``weight'') as a part of the entire system, and also include a ``clock''~\cite{Horodecki2013}
 that can effectively simulate the time-dependence of  the Hamiltonian of the system.
By doing so, the Hamiltonian of the entire system becomes time-independent, and thus the Gibbs-preserving map on the entire system becomes relevant, even when the effective Hamiltonian of the system is time-dependent.
The entire system now consists of the work storage W, the clock C,  and the system S, and obeys a Gibbs-preserving map of the total Hamiltonian of SCW.
See also Section~\ref{sec:clock_work} for more details of the clock and the work storage in the quantum case.

Let $p^{\rm G}$ be the Gibbs state of S with $p^{\rm G}_i := e^{-\beta E_i}/Z$, where $E_i$ is the energy of level $i$ and $Z$ is the partition function, as in Section~\ref{sec:classical_second_law}.
Let $p$ be an arbitrary distribution of S.
Suppose that W has only two energy levels, $0$ and $w$.
 We write the Gibbs state of W as $r^{\rm G} := (1/(1+e^{-\beta w}), e^{-\beta w}/(1+e^{-\beta w}))$ and write an arbitrary distribution of W as $r = (r_0, r_w)^{\rm T}$ with $r_0 + r_w = 1$.
We assume that the initial and final energies of W are  given by $0$ or $w$,  and the energy change of W is always given by $w$ or $-w$ (depending on the situation that we are considering) with unit probability, which is the characteristic of the single-shot scenario.

\

\noindent\textit{State formation.}
First, we consider the minimum work that is needed for creating nonequilibrium distribution $p$ of S  from the Gibbs state $p^{\rm G}$.
We here suppose that the initial and final Hamiltonians of S are the same, and do not consider  C explicitly.
Let $w> 0$.

The initial distribution of S is  $p^{\rm G}$ and that of W is  $r^{\rm up}:=(0,1)^{\rm T}$, 
and  the initial distribution of SW is given by $p^{\rm G} \otimes r^{\rm up}$, whose Lorenz curve consists of a straight line and a horizontal line (see Fig.~\ref{fig:majorization4}).
Then, the final distribution of S is $p$ and that of W is $r^{\rm down}:= (1,0)^{\rm T}$.
The final distribution of SW is given by $p \otimes r^{\rm down}$, where  $w > 0$ is the work performed on S by W.

Figure~\ref{fig:majorization4} shows the Lorenz curves of the initial and the final distributions of SW, along with the relevant $\infty$-divergences.
Then we see that the state conversion is possible, if and only if
\begin{equation}
e^{-S_\infty (p^{\rm G} \otimes r^{\rm up} \| p^{\rm G}\otimes r^{\rm G})} \leq  e^{- S_\infty (p \otimes r^{\rm down} \| p^{\rm G}\otimes r^{\rm G})},
\end{equation}
which reduces to
\begin{equation}
\frac{e^{-\beta w}}{1+e^{-\beta w}} \leq \frac{1}{1+e^{-\beta w}}e^{-S_\infty(p\| p^{\rm G})}. 
\end{equation}
Thus, we obtain the necessary and sufficient condition for the state conversion as 
\begin{equation}
w \geq \beta^{-1} S_\infty (p \| p^{\rm G}),
\label{single_shot_work_classical_infty}
\end{equation}
which gives the lower bound of the necessary work~\cite{Horodecki2013}.

\begin{figure}[t]
\begin{center}
\includegraphics[width=7cm]{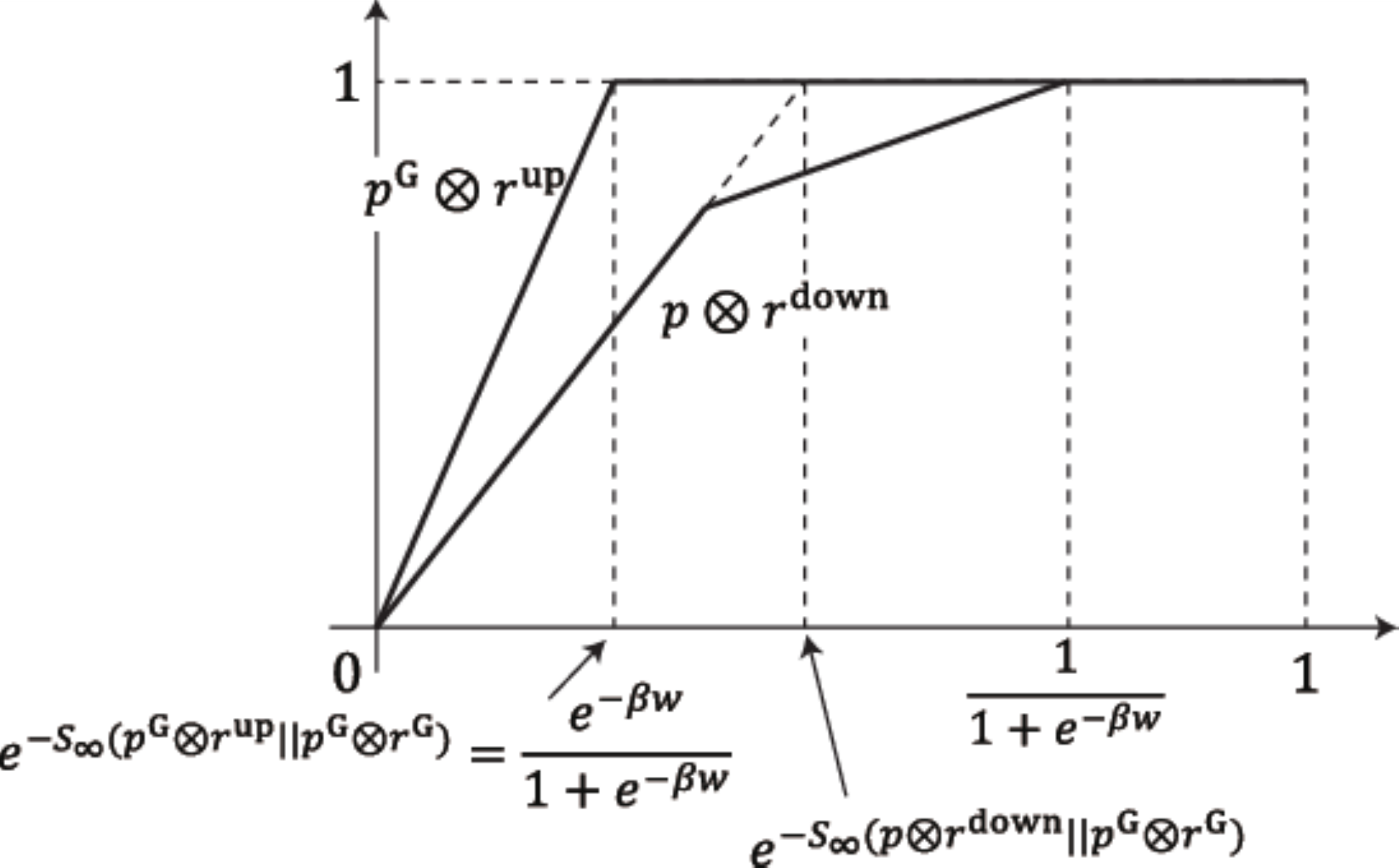}
 \end{center}
 \caption{Lorenz curves of the state conversion from the Gibbs state $p^{\rm G}$ of S to a  nonequilibirum state $p$ with work cost $w > 0$.  The state conversion is possible, if and only if $p^{\rm G} \otimes r^{\rm up}$ thermo-majorizes $p \otimes r^{\rm down}$.} 
\label{fig:majorization4}
\end{figure}

\

\noindent\textit{Work extraction.}
We next consider the maximum work that is extractable from an initial nonequilibrium distribution $p$ of S.
We here again suppose that the initial and final Hamiltonians  of S are the same, and do not consider C explicitly.
Let $w< 0$.

The initial state of W is $r^{\rm down}:=(0,1)^{\rm T}$  and that of SW is $p \otimes r^{\rm down}$, where ``up'' and ``down'' are exchanged from the previous setup.
The final distribution of S is $p^{\rm G}$ and that of W is $r^{\rm up}:= (1,0)^{\rm T}$, and the final distribution of SW is $p^{\rm G} \otimes r^{\rm up}$.
In this setup, $-w>0$ is the work that is extracted from S and finally stored in W.

Figure~\ref{fig:majorization5} shows the Lorenz curves of the initial and the final distributions of SW, along with the relevant $0$-divergences.
From this, we see that the state conversion is possible, if and only if 
\begin{equation}
e^{-S_0 (p \otimes r^{\rm down} \| p^{\rm G}\otimes r^{\rm G})} \leq  e^{- S_0 (p^{\rm G} \otimes r^{\rm up} \| p^{\rm G}\otimes r^{\rm G})},
\end{equation}
which reduces to
\begin{equation}
\frac{e^{-\beta w}}{1+e^{-\beta w}}e^{-S_0(p\| p^{\rm G})} \leq \frac{1}{1+e^{-\beta w}}.
\end{equation}
Therefore, we obtain the necessary and sufficient condition for the state conversion
\begin{equation}
-w \leq \beta^{-1} S_0 (p \| p^{\rm G}),
\label{single_shot_work_classical_0}
\end{equation}
which gives the upper bound of the extracted work $-w$~\cite{Horodecki2013,Aberg2013}.

\begin{figure}[t]
\begin{center}
\includegraphics[width=7.5cm]{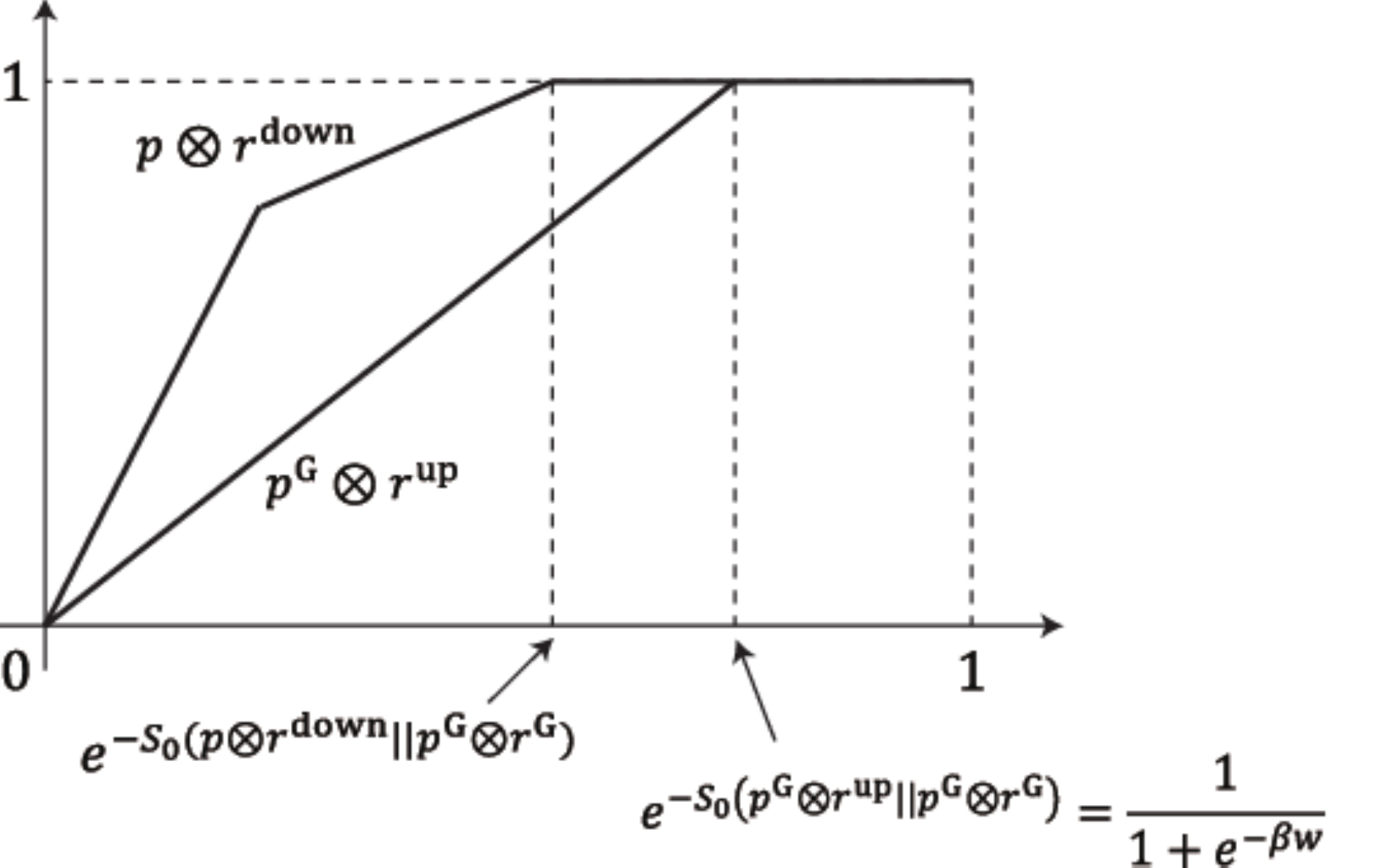}
 \end{center}
 \caption{Lorenz curves of the state conversion from an arbitrary initial state $p$ of S to its Gibbs state $p^{\rm G}$ with work extraction $-w>0$.
 The state conversion is possible, if and only if  $p \otimes r^{\rm down}$ thermo-majorizes $p^{\rm G} \otimes r^{\rm up}$.} 
\label{fig:majorization5}
\end{figure}

The equality in (\ref{single_shot_work_classical_0}) can be achieved by the following protocol proposed in Ref.~\cite{Aberg2013}, which is illustrated in Fig.~\ref{fig:equality_protocol_Renyi0}.
Here, we do not take W and C into account explicitly.

\begin{enumerate}
\item We instantaneously push the energy levels of $i$ with $p_i=0$ up to $+ \infty$, by keeping the other energy levels  fixed.
During this quench, no work is performed on S.
\item After waiting for the relaxation of S, we quasi-statically restore the pushed-up energy levels to the original values.
\end{enumerate}
The work extraction during the quasi-static process equals  the change in the equilibrium free energies, which is given by
\begin{equation}
 -w = - \beta^{-1} \ln \frac{Z_0}{Z} = -\beta^{-1} \ln \left( \sum_{i: p_i > 0} p_i^{\rm G} \right) = \beta^{-1} S_0 (p \| p^{\rm G}),
\end{equation}
where $Z_0 := \sum_{i: p_i > 0}e^{-\beta E_i}$.
The work does not fluctuate in the entire process, and thus it is a single-shot protocol.

\begin{figure}[t]
\begin{center}
\includegraphics[width=12cm]{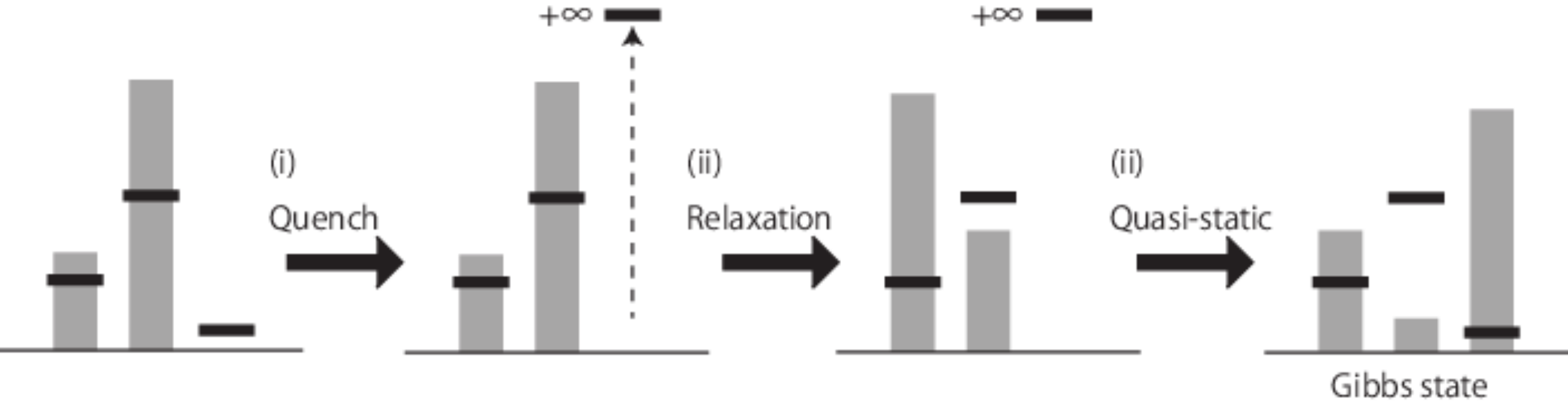}
 \end{center}
 \caption{A schematic of the protocol that achieves the equality in  (\ref{single_shot_work_classical_0}). The bold horizontal lines represent the energy levels, and the shaded vertical bars represent the probability distributions.} 
\label{fig:equality_protocol_Renyi0}
\end{figure}

\

\noindent\textit{Equilibrium transition.}
Finally, we consider transitions between Gibbs states of S, where  the Hamiltonian of S is changed from the initial one to the final one by using the clock C. 
We will see that this setup reproduces the conventional form of the second law (\ref{work_classical_equilibrium}), which is characterized by the equilibrium free energies.

 We suppose that  Hamiltonian of S is given by the energy levels $E_i$ if the clock indicates ``$0$'', and by $E_i'$ if ``$1$''.
The initial and the final distributions of S, written as $p$ and $p'$, are given by the Gibbs states of the initial and final Hamiltonians: $p_i := e^{-\beta E_i}/Z$ and $p_i' := e^{-\beta E_i'}/Z'$ with $Z$, $Z'$ being the partition functions (here we dropped the superscript ``G'' for simplicity).
The corresponding free energies are given by $F:= -\beta^{-1} \ln Z$ and $F' := -\beta^{-1} \ln Z'$.
 
For simplicity, we suppose  that C only has these two states $0$ and $1$.
Let $(c_0,c_1)^{\rm T}$ be the probability distribution of C in general.
Specifically, we suppose that the initial distribution of C is  $c=(1,0)^{\rm T}$ and the final one is  $c'=(0,1)^{\rm T}$.
By noting that C is coupled to S such that C induces the change of the Hamiltonian of S,
the total Gibbs state of SC should be given by
\begin{equation}
p_{\rm SC}^{\rm G} = \frac{Z}{Z+Z'} p \otimes  c + \frac{Z'}{Z+Z'} p' \otimes  c'.
\end{equation}
See  Eq.~(\ref{Hamiltonian_SCW}) and Eq.~(\ref{Gibbs_SCW}) of Section~\ref{sec:work_bound_single} for a more explicit argument.

Suppose that the initial and final states of W are respectively given by  $r^{\rm down}:=(0,1)^{\rm T}$ and $r^{\rm up}=(1,0)^{\rm T}$.
Here, we supposed that $- w> 0$ (i.e., the case of work extraction), while this is not actually necessary for the following argument.

\begin{figure}[t]
\begin{center}
\includegraphics[width=7cm]{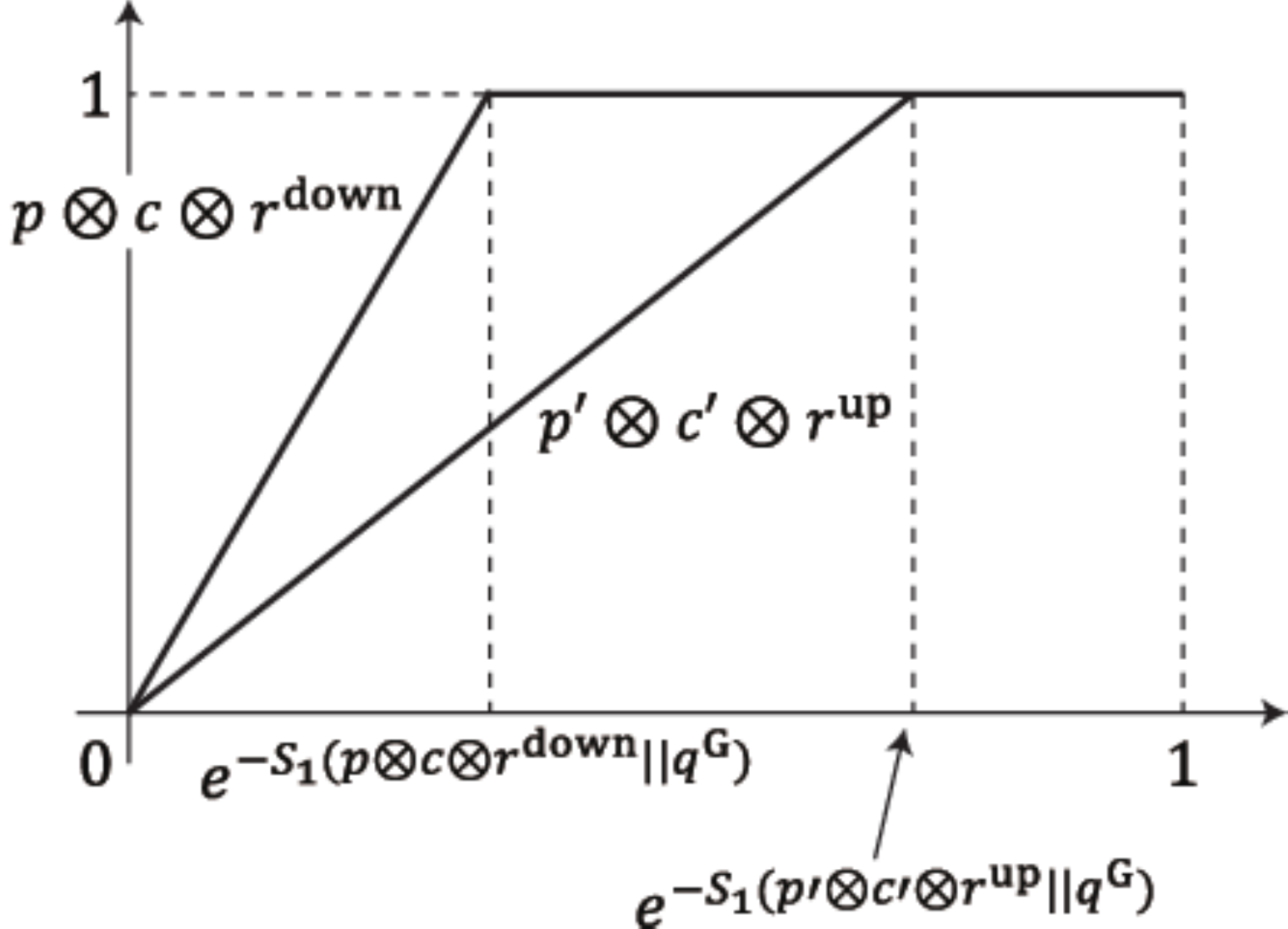}
 \end{center}
 \caption{Lorenz curves for a transition between Gibbs states of S.  The state conversion is possible, if and only if $p \otimes c \otimes r^{\rm down}$ thermo-majorizes $p' \otimes c' \otimes r^{\rm up}$.
} 
\label{fig:majorization6}
\end{figure}

Then, the Gibbs state of SCW is given by $q^{\rm G} := p_{\rm SC}^{\rm G} \otimes r^{\rm G}$.
The entire system SCW obeys a Gibbs-preserving map with respect to $q^{\rm G}$.
In this setup, each of the initial and the final Lorenz curves consists of a straight line and a horizontal line (see Fig.~\ref{fig:majorization6}).
Thus, $S_{\alpha} (p \otimes c \otimes r^{\rm down} \| q^{\rm G} )$ collapses to a single value for all $0 \leq \alpha \leq \infty$. In particular, the KL divergence ($\alpha = 1$) coincides with all the other R\'enyi divergences.
Thus, from the Lorenz curves shown in Fig.~\ref{fig:majorization6}, the state conversion is possible if and only if
\begin{equation}
e^{- S_{1} (p \otimes c \otimes r^{\rm down} \| q^{\rm G} ) } \leq e^{- S_{1} (p' \otimes c' \otimes r^{\rm up} \|  q^{\rm G})}.
\end{equation}
We also have
$e^{- S_{1} (p \otimes c \otimes r^{\rm down} \| q^{\rm G} ) } =  Z / ( Z+Z') \cdot e^{-\beta w} / (1+e^{-\beta w})$,
$e^{- S_{1} (p' \otimes c' \otimes r^{\rm up} \|  q^{\rm G})} =Z'/ (Z+Z' ) \cdot  1/ ( 1+e^{-\beta w})$.
Thus, we obtain the necessary and sufficient condition for the state conversion as
\begin{equation}
- w \leq -( F'-F),
\label{single_shot_work_classical_equilibrium}
\end{equation}
which is nothing but the second law  (\ref{work_classical_equilibrium}) in the conventional form.
This is also consistent with the observation in Section~\ref{sec:classical_second_law} that the work fluctuation vanishes in the quasi-static limit, which implies that equilibrium transitions can be realized by single-shot protocols.


\chapter{Quantum entropy and divergence}
\label{chap:quantum_entropy}

In this chapter, we consider the basic concepts of quantum information theory.
In Section~\ref{sec:quantum_state}, we briefly summarize the basic concepts of quantum states and quantum dynamics.
We then consider the von Neumann entropy and the quantum KL divergence (quantum relative entropy) in Section~\ref{sec:quantum_entropy}, and the quantum R\'enyi $0$- and $\infty$-divergences (min and max divergences) in Section~\ref{sec:quantum_Renyi}.
We postpone the proof of the monotonicity of these divergences to Appendix~\ref{apx:general_monotonicity}.

\section{Quantum state and dynamics}
\label{sec:quantum_state}
 
As a preliminary, we overview the basic concepts of quantum states and quantum dynamics very briefly.  See also, for example, Ref.~\cite{Nielsen} for details.

A quantum system is described by a Hilbert space $\mathcal H$.
Throughout this book, we assume that the Hilbert space is finite-dimensional, unless stated otherwise.
We write the identity on any Hilbert space by $\hat I$.
We denote the set of linear operators acting on $\mathcal H$ by $\mathcal L (\mathcal H)$.

Let $d$ be the dimension of the Hilbert space $\mathcal H$.
A pure state of the system is represented by a normalized vector $| \varphi \rangle \in \mathcal H$ with $\langle \varphi | \varphi \rangle = 1$.
In general, a \textit{quantum state} is represented by a density operator $\hat \rho$ acting on  $\mathcal H$, which is normalized, ${\rm tr}[\hat \rho] = 1$, and  is positive, $\hat \rho \geq 0$.  The trace is defined as ${\rm tr}[\hat \rho] := \sum_{i=1}^d \langle \varphi_i | \hat \rho | \varphi_i \rangle$ with $\{ | \varphi_i \rangle \}_{i=1}^d$ being an orthonormal basis of $\mathcal H$.   

Here, an operator $\hat \rho \in \mathcal L (\mathcal H)$ is \textit{positive} (or more precisely, positive semidefinite), if  $\langle \varphi | \hat \rho | \varphi \rangle \geq 0$ for all $| \varphi \rangle \in \mathcal H$.
In addition, an operator $\hat \rho$ is \textit{positive definite}, if  $\langle \varphi | \hat \rho | \varphi \rangle > 0$ for all $| \varphi \rangle \in \mathcal H$ with $| \varphi \rangle \neq 0$.
We denote $\hat \rho \geq 0$ and $\hat \rho > 0$ if $\hat \rho$ is positive and positive definite, respectively.
Also, $\hat A \geq \hat B$ and $\hat A > \hat B$ represent $\hat A - \hat B \geq 0$ and $\hat A - \hat B > 0$, respectively.

We denote the set of (normalized) quantum states on $\mathcal H$ by $\mathcal S (\mathcal H)$.
This set is convex; a state is pure if and only if it is  an extreme point of $\mathcal S (\mathcal H)$.  
(In this book, we do not consider a superselection rule that makes the definition of pure states more subtle.) 
We note that if an operator  $\hat \rho$ is  positive and satisfies ${\rm tr}[\hat \rho] \leq 1$, it is referred to as a \textit{subnormalized state}.
An operator that is only positive is referred to an \textit{unnormalized state}.
In this book, states are always normalized unless otherwise stated.

Since any positive operator is Hermitian, a quantum state $\hat \rho$ has the spectrum decomposition of the form
\begin{equation}
\hat \rho = \sum_{i=1}^d p_i | \varphi_i \rangle \langle \varphi_i |,
\end{equation}
where $| \varphi_i \rangle \langle \varphi_i |$ is a projection.
We note that $p_i \geq 0$ and $\sum_{i=1}^d p_i = 1$.
Here, $p := (p_1, p_2, \cdots, p_d)^{\rm T}$ is regarded as a classical probability distribution, which is referred to as the \textit{diagonal distribution} of $\hat \rho$.
If a state is pure, the density operator is given by $\hat \rho = | \varphi \rangle \langle \varphi |$.

The \textit{support} of $\hat \rho$ is the subspace spanned by the eigenvectors with nonzero eigenvalues of $\hat \rho$, which is denoted as ${\rm supp}[\hat \rho] \subset \mathcal H$.
The rank of $\hat \rho$, denoted as ${\rm rank}[\hat \rho]$, is the dimension of ${\rm supp}[\hat \rho]$.
The support of any positive-definite operator equals the entire space $\mathcal H$ (i.e., has full rank).

We next consider a composite system AB that consists of subsystems A and B.
Let $\mathcal H_{\rm A}$ and $\mathcal H_{\rm B}$ be the Hilbert spaces of subsystems A and B with dimensions $d$ and $d'$, respectively.
Then, the Hilbert space of AB is given by their tensor product  $\mathcal H_{\rm A} \otimes \mathcal H_{\rm B}$ with dimension $dd'$.

If the partial states of A and B are independent and written as $\hat \rho_{\rm A} \in \mathcal S (\mathcal H_{\rm A})$ and $\hat \rho_{\rm B} \in \mathcal S (\mathcal H_{\rm B})$,
the corresponding state of AB is given by $\hat \rho_{\rm A} \otimes \hat \rho_{\rm B} \in \mathcal S (\mathcal H_{\rm A} \otimes \mathcal H_{\rm B})$, which is called a product state.
If the partial states are both pure and written as $| \varphi_{\rm A} \rangle \in \mathcal H_{\rm A}$, $| \varphi_{\rm B} \rangle \in \mathcal H_{\rm B}$, the state vector of AB is given by $| \varphi_{\rm A} \rangle \otimes | \varphi_{\rm B} \rangle  \in \mathcal H_{\rm A} \otimes \mathcal H_{\rm B}$.
In the following, we abbreviate $| \varphi_{\rm A} \rangle \otimes | \varphi_{\rm B} \rangle$ as   $| \varphi_{\rm A} \rangle | \varphi_{\rm B} \rangle$.
In general, a state of AB is a density operator $\hat \rho \in  \mathcal S (\mathcal H_{\rm A} \otimes \mathcal H_{\rm B})$, whose partial states are given by the partial trace; for example, the partial state of A is given by
\begin{equation}
\hat \rho_{\rm A} = {\rm tr}_{\rm B} [ \hat \rho ] := \sum_{i=1}^{d'} \langle \varphi_{{\rm B}, i} | \hat \rho  | \varphi_{{\rm B}, i} \rangle \in \mathcal S (\mathcal H_{\rm A}),
\end{equation}
where $\{ | \varphi_{{\rm B}, i} \rangle \}_{i=1}^{d'}$ is an orthonormal basis of $\mathcal H_{\rm B}$.

We note that a pure state of AB is called \textit{entangled}, if it is not in the form of  a product state $| \varphi_{\rm A} \rangle | \varphi_{\rm B} \rangle$.
A mixed state $\hat \rho_{\rm AB}$ is called a \textit{product} state if it is written as $\hat \rho_{\rm AB} = \hat \rho_{\rm A} \otimes \hat \rho_{\rm B}$, and it is called a \textit{separable} state if it can be written as $\hat \rho_{\rm AB} = \sum_k r_k \hat \rho_{\rm A}^{(k)} \otimes \hat \rho_{\rm B}^{(k)}$ with some classical probability distribution $r = (r_1, r_2, \cdots, r_K )^{\rm T}$ with $K<\infty$.  Any separable state has only classical correlations between A and B.  If $\hat \rho_{\rm AB}$ is not separable,  it is called entangled.


\

We next consider quantum dynamics, which is represented by a linear map  $\mathcal E : \ \mathcal L(\mathcal H)  \to \mathcal L(\mathcal H')$.
Such a linear map acting on an operator space $\mathcal L (\mathcal H)$ is sometimes called a superoperator.
In general,  we allow different Hilbert spaces $\mathcal H$, $\mathcal H'$ for the input and output systems.

The simplest case is unitary dynamics that is given by, with a unitary operator $\hat U \in \mathcal L (\mathcal H )$,
\begin{equation}
\mathcal E(\hat \rho) = \hat U \hat \rho \hat U^\dagger,
\label{E_unitary}
\end{equation}
which maps a pure state to a pure state, i.e., $| \varphi \rangle \mapsto \hat U | \varphi \rangle$.
In general, any physically-realizable dynamics must be completely-positive (CP) as defined as follows:

\begin{definition}[Completely positive maps]
Let  $\mathcal E : \mathcal L (\mathcal H) \to \mathcal L (\mathcal H')$ be a linear map. 
\begin{itemize}
\item $\mathcal E$ is positive, if any positive operator is mapped to  a positive operator, that is, $\hat X \geq 0$ implies $\mathcal E (\hat X) \geq 0$.
\item $\mathcal E$ is $n$-positive if $\mathcal E \otimes \mathcal I_{n}$ is positive,
where  $\mathcal I_n$ is the identity superoperator on  $\mathcal L ( \mathbb C^n)$.
\item $\mathcal E$ is completely positive (CP) if it is $n$-positive for all $n \in \mathbb N$.
\end{itemize}
\end{definition}

An example of a map that is positive but not CP is the transpose of matrices with a given basis.
For example, if we transpose the density matrix of a partial state of an entangled state, the total density operator can have negative eigenvalues. Based on this property, the transpose can be used for quantifying entanglement~\cite{Peres96,Plenio2005}.

We also define the concept of trace-preserving (TP), which represents the conservation of probability; If quantum dynamics is TP, the output state is normalized for any normalized input state.

\begin{definition}[Trace-preserving]
A linear map $\mathcal E$ is called trace-preserving (TP), if it satisfies
\begin{equation}
{\rm tr}[\mathcal E(\hat X) ] = {\rm tr}[\hat X]
\end{equation}
for all $\hat X$.
If ${\rm tr}[\mathcal E(\hat X) ] \leq {\rm tr}[\hat X]$ holds for all $\hat X \geq 0$, $\mathcal E$ is called trace-nonincreasing.
\end{definition}

To summarize, any physically realizable map that conserves  probability must be completely-positive and trace-preserving (CPTP).
It is  known that a superoperator$\mathcal E$  is CP if and only if it is written as 
\begin{equation}
\mathcal E(\hat \rho) = \sum_{k} \hat M_k \hat \rho \hat M_k^\dagger,
\label{CP_Kraus}
\end{equation}
which is called the Kraus representation ($\hat M_k$'s are called the Kraus operators).
We omit the proof of this (see, e.g., Theorem 8.1 of \cite{Nielsen}, Theorem 1 of \cite{Choi75}).
If $\mathcal E$ is CPTP, the Kraus operators satisfy $\sum_k \hat M_k^\dagger \hat M_k = \hat I$,
which guarantees the conservation of the trace.
If $\mathcal E$ is CP and trace-nonincreasing, we have $\sum_k \hat M_k^\dagger \hat M_k \leq \hat I$.

We further remark that any CPTP map $\mathcal E$ on $\mathcal H$ can be represented as a unitary map on an extended space: If $\mathcal E$ is CPTP, there exists a quantum state $\hat \sigma \in \mathcal S (\mathcal H_{\rm A})$ of an auxiliary system A  and a unitary operator on $\mathcal H \otimes \mathcal H_{\rm A}$ such that
\begin{equation}
\mathcal E (\hat \rho ) = {\rm tr}_{\rm A'} [\hat U \hat \rho \otimes \hat \sigma \hat U^\dagger],
\label{Naimark}
\end{equation}
where ${\rm tr}_{\rm A'}$ is  the partial trace over an output auxiliary system $\rm A'$ such that $\mathcal H \otimes \mathcal H_{\rm A} \simeq \mathcal H' \otimes \mathcal H_{\rm A'}$.
Eq.~(\ref{Naimark}) is called the Naimark extension (see Ref.~\cite{Nielsen} for details).
Conversely, any map of the form (\ref{Naimark}) is CPTP.
This implies that any quantum dynamics can be regarded as unitary dynamics of a larger system including the environment.

We here remark some examples of CP maps.

First, unitary dynamics (\ref{E_unitary}) is CPTP with  a single Kraus operator $\hat U$.

The partial trace itself is CPTP:
for $\mathcal H \simeq \mathcal H' \otimes \mathcal H_{\rm A'}$, 
 $\mathcal E (\hat \rho ) := {\rm tr}_{\rm A'} [\hat \rho]$ is CPTP.
 In particular, the trace itself is CPTP, where the output space is the one-dimensional Hilbert space $\mathbb C$.

A quantum measurement with a particular (post-selected) outcome is an example of CP and trace-nonincreasing maps.
Let $k$ be a measurement outcome.
For the measured state $\hat \rho$, the post-measurement state with outcome $k$ is given by $\mathcal E_k (\hat \rho)$, where the probability of outcome $k$ is given by $p_k = {\rm tr}[\mathcal E_k (\hat \rho)] $.
$\mathcal E_k$ must be CP, while it does not preserve the trace if $p_k < 1$.
The normalized post-measurement state is given by $\mathcal E_k (\hat \rho)/ p_k$.
If we average the post-measurement states over all outcomes, the averaged state is given by $\mathcal E (\hat \rho) := \sum_k \mathcal E_k (\hat \rho)$.  Now $\mathcal E$ is CPTP, because of $\sum_k p_k = 1$. 
We note that the map $\hat \rho \mapsto p_k$ is also  CP where the output Hilbert space is one dimensional.
Correspondingly, we can embed $\{ p_k \}$ to the diagonal elements of a density operator with some  fixed basis, as $\hat p := \sum_k p_k | k\rangle \langle k |$. Then, $\hat \rho \mapsto \hat p$ is a CPTP map.

\

A map $\mathcal E : \mathcal L (\mathcal H ) \to \mathcal L (\mathcal H)$ is called \textit{unital}, if it preserves the identity of $\mathcal H$, i.e., $\mathcal{E}( \hat I ) = \hat I$.
Unital should not be confused with unitary, but any unitary map is unital.
The concept of CPTP unital is a quantum analogue of doubly stochastic, as seen in Section~\ref{sec:quantum_majorization}.

We next consider an inner product of the operator space, called  the Hilbert-Schmidt inner product defined as
\begin{equation}
\langle \hat Y, \hat X \rangle_{\rm HS} := {\rm tr}[\hat Y^\dagger \hat X].
\end{equation}
We can define the Hermitian conjugate of $\mathcal E : \mathcal L (\mathcal H) \to \mathcal L (\mathcal H')$ with respect to the Hilbert-Schmidt inner product, which is a superoperator  $\mathcal E^\dagger : \mathcal L (\mathcal H') \to \mathcal L (\mathcal H)$ satisfying 
\begin{equation}
\langle \hat Y, \mathcal E (\hat X ) \rangle_{\rm HS} = \langle \mathcal E^\dagger ( \hat Y ), \hat X \rangle_{\rm HS}
\end{equation}
for all $\hat X$, $\hat Y$.
Clearly, $\mathcal E^{\dagger \dagger} = \mathcal E$.

$\mathcal E$ is TP if and only if $\mathcal E^\dagger$ is unital.  In fact, by taking $\hat Y = \hat I$ in the above relation, we have 
${\rm tr}[\mathcal E (\hat X )] = \langle \hat I, \mathcal E (\hat X ) \rangle_{\rm HS} = \langle \mathcal E^\dagger ( \hat I ), \hat X \rangle_{\rm HS} = {\rm tr} [\mathcal E^\dagger ( \hat I ) \hat X ]$.
Thus,  ${\rm tr}[\mathcal E (\hat X )] = {\rm tr} [\hat X ]$ holds for all $\hat X$, if and only if $\mathcal E^\dagger ( \hat I ) = \hat I$.

We note that $\mathcal E$ is positive if and only if $\mathcal E^\dagger$ is positive.
To see this, let $\hat Y$ be positive and take $\hat X = | \varphi \rangle \langle \varphi |$.
Suppose that $\mathcal E$ is positive.  Then
$\langle \varphi | \mathcal E^\dagger (\hat Y ) | \varphi \rangle  = {\rm tr}[\hat Y \mathcal E ( \hat X)] =  {\rm tr}[\mathcal E ( \hat X)^{1/2} \hat Y \mathcal E ( \hat X)^{1/2}] \geq 0$.
Moreover, $\mathcal E$ is $n$-positive (resp. CP) if and only if $\mathcal E^\dagger$ is $n$-positive (resp. CP).
In fact, $\mathcal E \otimes  \mathcal I_n$ is positive if and only if $\mathcal E^\dagger \otimes  \mathcal I_n$ is positive, because $\mathcal I_n^\dagger = \mathcal I_n$.

\

We finally remark on norms of operators. Let $\hat X$  be an arbitrary operator acting on $\mathcal H$.
First, the trace norm is defined as 
\begin{equation}
\| \hat X \|_1 := {\rm tr}[| \hat X |],
\end{equation}
where $| \hat X | := \sqrt{\hat X^\dagger \hat X}$.
This is indeed a norm in the mathematics sense, because it satisfies the triangle inequality $\| \hat X + \hat X' \|_1 \leq \| \hat X \|_1 + \| \hat X' \|_1$ and $\| \hat X \|_1 = 0$ holds if and only if $\hat X = 0$. 
The trace distance is then defined as
\begin{equation}
D (\hat \rho , \hat \sigma ) := \frac{1}{2} \| \hat \rho - \hat \sigma \|_1,
\end{equation}
which is commonly used for measuring a distance between quantum states $\hat \rho, \hat \sigma$, and satisfies the monotonicity under CPTP map $\mathcal E$ (Theorem 9.2 of \cite{Nielsen}):
\begin{equation}
D (\hat \rho , \hat \sigma ) \geq D (\mathcal E (\hat \rho ), \mathcal E(\hat \sigma) ).
\end{equation}
Another important norm is the operator norm defined by
\begin{equation}
\| \hat X \|_\infty := \max_{\| | \varphi \rangle \| = 1} \sqrt{\langle \varphi | \hat X^\dagger \hat X | \varphi \rangle},
\end{equation}
which is nothing but the largest singular value of $\hat X$.

\section{von Neumann entropy and the quantum KL divergence}
\label{sec:quantum_entropy}

The quantum analogue of the Shannon entropy is the \textit{von Neumann entropy}, which is defined for a density operator $\hat \rho$ of dimension $d$ by
\begin{equation}
S_1(\hat \rho) := -{\rm tr} [\hat \rho \ln \hat \rho].
\end{equation}
Let $p$ be the diagonal distribution of $\hat \rho$. Then $S_1 (p) = S_1 (\hat \rho)$ holds, implying that the von Neumann entropy characterizes ``classical randomness'' of quantum states.
We thus have
\begin{equation}
0 \leq S_1 (\hat \rho ) \leq \ln d.
\end{equation}
It is  obvious that the von Neumann entropy is invariant under unitary transformation: $S_1 (\hat U \hat \rho \hat U^\dagger ) = S_1 (\hat \rho )$ holds for any unitary $\hat U$.

The KL divergence is also generalized to the quantum case, which is referred to as the quantum KL divergence or the quantum relative entropy~\cite{Nielsen}. It is defined for two quantum states $\hat \rho$ and $\hat \sigma$ of dimension $d$ as
\begin{equation}
S_1(\hat \rho \| \hat \sigma ) := {\rm tr}[\hat \rho \ln \hat \rho - \hat \rho \ln \hat \sigma].
\end{equation}
If the support of $\hat \rho$ is not included in that of $\hat \sigma$, we define $S_1(\hat \rho \| \hat \sigma ) = + \infty$.
As in the classical case, however, we assume that the support of  $\hat \rho$ is included in that of $\hat \sigma$ throughout this book, whenever we consider quantum divergence-like quantities (including those in Appendix~\ref{apx:general_monotonicity}).
We note that 
\begin{equation}
 S_1(\hat \rho)=   \ln d - S_1(\hat \rho \| \hat I /d),
\label{vonNeumann_KL}
\end{equation}
where $\hat I / d$ is the maximally mixed state (i.e., the uniform distribution).
If $\hat \rho$ and $\hat \sigma$ are simultaneously diagonalizable such that $\hat \rho = \sum_{i=1}^d p_i | \varphi_i \rangle \langle \varphi_i |$ and $\hat \sigma = \sum_{i=1}^d q_i | \varphi_i \rangle \langle \varphi_i |$ with the same basis, $S_1 (\hat \rho \| \hat \sigma)$ reduces to the classical KL divergence of the diagonal distributions of $\hat \rho$ and $\hat \sigma$, i.e., $S_1(p \| q) = S_1 (\hat \rho \| \hat \sigma)$.
We note that $S_1(\hat U \hat \rho \hat U^\dagger \| \hat U \hat \sigma \hat U^\dagger ) = S_1(\hat \rho \| \hat \sigma )$ holds for any unitary $\hat U$.

As in the classical case, the quantum KL divergence is an asymmetric ``distance'' between two quantum states.
In fact, the quantum KL divergence is non-negative:
\begin{equation}
S_1(\hat \rho \| \hat \sigma )  \geq 0,
\label{quantum_KL_positive}
\end{equation} 
where the equality $S_1(\hat \rho \| \hat \sigma )  = 0$ holds if and only if $\hat \rho = \hat \sigma$.
The proof of this is straightforward (e.g., Theorem 11.7 of Ref~\cite{Nielsen}),  but  we will prove it as Corollary \ref{cor:quantum_KL_positive} in Appendix~\ref{apx:general_monotonicity} in a special case of a more general statement. 
We here note that if two states are unnormalized, the non-negativity~(\ref{quantum_KL_positive}) is generalized as
\begin{equation}
S_1(\hat \rho \| \hat \sigma )  \geq {\rm tr}[\hat \rho - \hat \sigma],
\end{equation}
which is called the Klein's inequality.

Another fundamental property of the quantum KL divergence is the monotonicity (or the data processing inequality):  

\begin{theorem}[Monotonicity]
For any CPTP map $\mathcal E$, 
\begin{equation}
S_1(\hat \rho \| \hat \sigma ) \geq S_1(\mathcal{E}( \hat \rho ) \| \mathcal{E}( \hat \sigma )). 
\label{monotone}
\end{equation}
\label{thm:monotone}
\end{theorem}

This implies that, as in the classical case discussed in Section~\ref{sec:Shannon_KL}, the KL divergence is a monotone under CPTP maps.  Note that, however,  the KL divergence is not a complete monotone, as it does not provide a sufficient condition for state convertibility, again as in the classical case discussed in Chapter~\ref{chap:classical_majorization}.
 
The proof of the monotonicity in the quantum case is nontrivial, in contrast to the classical counterpart (\ref{monotone_classical}) whose proof  was easy.
Historically, Lieb and Ruskai~\cite{Lieb2,Lieb3} proved the strong subadditivity of the von Neumann entropy (inequality (\ref{strong_subadditivity})) on the basis of the Lieb's theorem~\cite{Lieb1} (see Ref.~\cite{Ruskai} for a self-contained review), which was then rephrased as the monotonicity of the quantum KL divergence~\cite{Lindblad2,Lindblad3,Uhlmann1977}.
Later, Petz gave a simple proof of the monotonicity~\cite{Petz1986} (see also Refs.~\cite{Petz3,Nielsen2005,Hiai2011f}), on which we focus in this book.
In Appendix~\ref{apx:general_monotonicity}, we will prove the monotonicity of general divergence-like quantities called the Petz's quasi-entropies~\cite{Petz1985,Petz1986};
Theorem~\ref{thm:monotone} above is an immediate consequence of  Corollary~\ref{cor:KL_monotonicity}.

We consider the change in the von Neumann entropy by CPTP unital maps. Let  $\mathcal E$ be CPTP unital. By noting Eq.~(\ref{vonNeumann_KL}) and the monotonicity (\ref{monotone}), we have
\begin{equation}
S_1(\hat \rho) \leq S_1(\mathcal E (\hat \rho )),
\label{unital_entropy}
\end{equation}
which implies that a CPTP unital map makes quantum states more ``random,'' as in the classical counterpart (\ref{monotone_DSM}) with doubly-stochastic maps.
In other words, the von Neumann entropy is a monotone under CPTP unital maps.

\

We discuss some important properties of the von Neumann entropy and the quantum KL divergence.
We consider two systems A and B and their  composite system AB.
Let $\hat \rho_{\rm AB}$ be a density operator of AB and $\hat \rho_{\rm A}$, $\hat \rho_{\rm B}$ be its partial states.
Then, the \textit{subadditivity} of  the von Neumann entropy is expressed as
\begin{equation}
S_1(\hat \rho_{\rm AB}) \leq S_1(\hat \rho_{\rm A}) + S_1(\hat \rho_{\rm B}),
\label{subadditivity}
\end{equation}
where the equality holds  if and only if $\hat \rho_{\rm AB} = \hat \rho_{\rm A}\otimes \hat \rho_{\rm B}$.
This is a straightforward consequence of the non-negativity of the quantum KL divergence:
\begin{equation}
S_1(\hat \rho_{\rm A}) + S_1(\hat \rho_{\rm B}) - S_1 (\hat \rho_{\rm AB}) = S_1(\hat \rho_{\rm AB} \| \hat \rho_{\rm A} \otimes \hat \rho_{\rm B} ) \geq 0.
\end{equation}
The left-hand side above is called the mutual information between A and B, written as
\begin{equation}
I_1 (\hat \rho_{\rm AB})_{\rm A: B} :=S_1(\hat \rho_{\rm A}) + S_1(\hat \rho_{\rm B}) - S_1 (\hat \rho_{\rm AB}) \geq 0.
\label{mutual_information_1}
\end{equation}
Note that for any CPTP map of the form $\mathcal E_{\rm A} \otimes \mathcal E_{\rm B}$ acting independently on A and B, the monotonicity of the KL divergence (\ref{monotone}) implies the data processing inequality of mutual information: $I_1 (\hat \rho_{\rm AB})_{\rm A: B} \geq I_1 (\mathcal E_{\rm A} \otimes \mathcal E_{\rm B} ( \hat \rho_{\rm AB}))_{\rm A: B}$.

A much stronger property than the subadditivity is the strong subadditivity.
We consider three subsystems A, B, C. Then,
\begin{equation}
S_1(\hat \rho_{\rm ABC}) + S_1(\hat \rho_{\rm B}) \leq S_1(\hat \rho_{\rm AB}) + S_1(\hat \rho_{\rm BC}).
\label{strong_subadditivity}
\end{equation}
This can be proved from the monotonicity~(\ref{monotone}) of the quantum KL divergence.
Let $\hat \sigma_{\rm A} := \hat I_{\rm A} / d_{\rm A}$, where $\hat I_{\rm A}$ is the identity of $\mathcal H_{\rm A}$ and $d_{\rm A}$  is its dimension.  We then have $ [ S_1(\hat \rho_{\rm AB}) + S_1(\hat \rho_{\rm BC})] - [ S_1(\hat \rho_{\rm ABC}) + S_1(\hat \rho_{\rm B}) ] = [ S_1(\hat \rho_{\rm BC}) -  S_1(\hat \rho_{\rm ABC})] - [ S_1(\hat \rho_{\rm B}) - S_1(\hat \rho_{\rm AB})] = S_1(\hat \rho_{\rm ABC} \| \hat \sigma_{\rm A} \otimes \hat \rho_{\rm BC}) - S_1(\hat \rho_{\rm AB} \| \hat \sigma_{\rm A} \otimes \hat \rho_{\rm B}) \geq 0$,
where we used the monotonicity~(\ref{monotone}) for the partial trace
$\mathcal E (\hat \rho_{\rm ABC}) := {\rm tr}_{\rm C}[\hat \rho_{\rm ABC}] =  \hat \rho_{\rm AB}$ that  is CPTP as mentioned before.

As seen in the above proof, the strong subadditivity~(\ref{strong_subadditivity}) is  equivalent to the monotonicity of the partial trace.
Conversely, the monotonicity of any CPTP map follows from the monotonicity of the partial trace  because of 
the Naimark extension (\ref{Naimark}).
Thus, the strong subadditivity~(\ref{strong_subadditivity}) is essentially equivalent to the monotonicity~(\ref{monotone})  of CPTP maps.

\

We next discuss the convex/concave properties.
The quantum KL divergence satisfies the following property called the joint convexity.

\begin{theorem}[Joint convexity]
Let $\hat \rho = \sum_k p_k \hat \rho_k$ and $\hat \sigma = \sum_k p_k \hat \sigma_k$, where $\hat \rho_k$ and $\hat \sigma_k$ are quantum states and $p_k$'s represent a classical distribution with $p_k > 0$.
Then,
\begin{equation}
S_1 (\hat \rho \| \hat \sigma ) \leq \sum_k p_k S_1 (\hat \rho_k \| \hat \sigma_k) .
\label{joint_convexity_KL}
\end{equation}
The equality holds if $\mathcal P_k$'s are orthogonal to each other, where $\mathcal P_k$ is the subspace spanned by the supports of $\hat \rho_k$ and $\hat \sigma_k$. 
\label{thm:joint_convexity_KL}
\end{theorem}

\begin{proof}
We start with the case of the equality. By assumption, $\hat \rho$ and $\hat \sigma$ are simultaneously block-diagonal with subspaces $\mathcal P_k$'s, and thus we have 
\begin{eqnarray}
S_1 (\hat \rho \| \hat \sigma ) 
&=& {\rm tr} \left[ \left( \sum_k p_k \hat \rho_k \right) \ln \left( \sum_k p_k \hat \rho_k \right)   - \left( \sum_k p_k \hat \rho_k \right) \ln \left( \sum_k p_k \hat \sigma_k \right)  \right] \\
&=& {\rm tr} \left[ \sum_k p_k \hat \rho_k \ln  ( p_k \hat \rho_k )    - \sum_k p_k \hat \rho_k  \ln  (  p_k \hat \sigma_k )   \right] \\
&=&\sum_k p_k S_1 (\hat \rho_k \| \hat \sigma_k).
\end{eqnarray}
In general, let $\{ | k \rangle \}$ be an orthonormal basis of an auxiliary system A, and define
\begin{equation}
\hat \rho' := \sum_k p_k \hat \rho_k \otimes | k \rangle \langle k |, \ \ \hat \sigma' := \sum_k p_k \hat \sigma_k \otimes | k \rangle \langle k |.
\end{equation}
From the equality case, we have
\begin{equation}
S_1 (\hat \rho' \| \hat \sigma') = \sum_k p_k S_1 (\hat \rho_k  \otimes | k \rangle \langle k | \| \hat \sigma_k \otimes | k \rangle \langle k |) = \sum_k p_k S_1 (\hat \rho_k \| \hat \sigma_k).
\end{equation}
Note hat $\hat \rho = {\rm tr}_{\rm A} [\hat \rho']$, $\hat \sigma = {\rm tr}_{\rm A} [\hat \sigma']$.
Because the partial trace is CPTP, the monotonicity of the quantum KL divergence implies that 
\begin{equation}
S_1 (\hat \rho \| \hat \sigma ) \leq S_1 (\hat \rho' \| \hat \sigma') =  \sum_k p_k S_1 (\hat \rho_k \| \hat \sigma_k). 
\end{equation}
$\Box$
\end{proof}

By letting $\hat \sigma_k = \hat I / d$ for all $k$ in the above theorem,  we obtain the concavity of the von Neumann entropy:
\begin{equation}
\sum_k p_k S_1 (\hat \rho_k ) \leq S_1 (\hat \rho ).
\label{von_Neumann_concave}
\end{equation}
The equality is achieved if we have a single $k$, i.e., $p_1 = 1$ and $\hat \rho = \hat \rho_1$.
We note that this itself is easily provable, without invoking the joint convexity of the quantum KL divergence: 
Consider $\hat \rho' := \sum_k p_k \hat \rho_k \otimes | k \rangle \langle k |$ and apply the subadditivity~(\ref{subadditivity}):
\begin{equation}
S_1 (\hat \rho ) + S_1 (p) \leq  S_1 (\hat \rho' ) = \sum_k p_k S_1 (\hat \rho_k ) + S_1 (p), 
\end{equation} 
where, to obtain the right equality, we used that $\hat \rho_k \otimes | k \rangle \langle k |$'s are mutually orthogonal and $S_1 (\hat \rho_k \otimes | k \rangle \langle k | ) =  S_1 (\hat \rho_k )$  holds.

Related to this, an upper bound of the von Neumann entropy is also known:
\begin{equation}
S_1 (\hat \rho ) \leq \sum_k p_k S_1 (\hat \rho_k ) + S_1 (p),
\label{concave_upper_bound}
\end{equation}
where the equality   holds if and only if the supports of $\hat \rho_k$'s are mutually orthogonal.
We omit the proof of this (see, e.g., Theorem 11.10 of \cite{Nielsen}).

\

The relation between the quantum KL divergence and the second law of thermodynamics is completely parallel to the classical case discussed in Section~\ref{sec:classical_second_law}.
We here briefly discuss the second law in the quantum case, while a more detailed argument will be provided in Chapter~\ref{chap:quantum_thermodynamics}.

Let $\hat H$ be the Hamiltonian of the system.
The corresponding Gibbs state is defined as $\hat \rho^{\rm G} := e^{-\beta \hat H}/ Z$, where $Z := {\rm tr}[e^{-\beta \hat H}]$ is the partition function.
The equilibrium free energy is  given by $F:=-\beta^{-1} \ln Z$.
As in the classical case~(\ref{classical_GMP}), a CPTP map $\mathcal E$ is called a \textit{Gibbs-preserving map}, if it satisfies 
\begin{equation}
\mathcal E (\hat \rho^{\rm G} ) = \hat \rho^{\rm G}.
\label{quantum_GMP0}
\end{equation}
From the monotonicity~(\ref{monotone}) of the quantum KL divergence, we have for any Gibbs-preserving map $\mathcal E$
\begin{equation}
S(\hat \rho \| \hat \rho^{\rm G}) \geq S( \mathcal{E}(\hat \rho )\| \hat \rho^{\rm G} ).
\end{equation}
In parallel to inequality~(\ref{classlcal_Gibbs_second1}) of  the classical case, we can rewrite the above inequality as
\begin{equation}
\Delta S_1 \geq \beta Q,
\label{quantum_Gibbs_second1}
\end{equation}
where $\Delta S_1 := S_1 (\mathcal E (\hat \rho)) - S_1 (\hat \rho)$ is the change in the von Neumann entropy and $Q := {\rm tr} [(\mathcal{E}(\hat \rho) - \hat \rho ) \hat H]$ is the heat absorption.
This is a Clausius-type representation of the second law or a generalized Landauer principle in the quantum case~\cite{Sagawa2012}.
If the Gibbs state is given by $\hat \rho^{\rm G} = \hat I/d$ (in the case of $\beta = 0$ or $\hat H \propto \hat I$), the Gibbs-preserving map is unital and  the second law~(\ref{quantum_Gibbs_second1})  reduces to  inequality~(\ref{unital_entropy}).

In the quench-and-relax processes  discussed in Section~\ref{sec:classical_second_law} where the Hamiltonian is driven by an external agent, we can also derive the work bound in the same form as the classical counterpart~(\ref{classical_work_bound_F1}). 
In Section~\ref{sec:fluctuating_work}, we will derive the second law of this form,  where we  adopt the setup where the total Hamiltonian is time-independent by including the ``clock'' degrees of freedom.

\section{Quantum R\'enyi entropy and divergence}
\label{sec:quantum_Renyi}

We next consider quantum analogues of the R\'enyi entropies and divergences.
The R\'enyi $\alpha$-entropy can be straightforwardly defined  in the same manner as the classical case:
\begin{equation}
S_\alpha (\hat \rho) := \frac{1}{1-\alpha} \ln \left( {\rm tr}[\hat \rho^\alpha ] \right),
\end{equation}
where $\hat \rho$ is a density operator of dimension $d$.
In particular, $S_1 (\hat \rho)$ is the von Neumann entropy, and 
\begin{equation}
S_0 (\hat \rho ) :=  \ln ({\rm rank}[\hat \rho]),
\end{equation}
\begin{equation}
S_\infty (\hat \rho) := - \ln \| \hat \rho \|_\infty
\end{equation}
are the max and the min R\'enyi entropies.  

The quantum R\'enyi entropy equals the classical one of the diagonal distribution: $S_\alpha (\hat \rho) = S_\alpha ( p)$ with $p$ being the diagonal distribution of $\hat \rho$.
We thus have 
\begin{equation}
0 \leq S_\alpha  (\hat \rho ) \leq \ln d
\end{equation}
and
\begin{equation}
S_\alpha ( \hat \rho ) \geq S_{\alpha'} (\hat \rho) \ \ \rm{for} \ \ \alpha \leq \alpha'.
\end{equation}
It is also obvious that $S_\alpha (\hat U \hat \rho \hat U^\dagger) = S_\alpha (\hat \rho )$ holds for any unitary $\hat U$.

On the other hand, the quantum versions of the quantum R\'enyi divergence are not unique due to the non-commutability of density operators $\hat \rho, \hat \sigma$~\cite{Hiai2011f,Tomamichel,Wilde2014,Lennert2013,Frank2013,Beigi2013} (see also Appendix~\ref{apx:general_monotonicity}).  Here, we mainly focus  on the following two limiting cases with $\alpha = 0,\infty$~\cite{Datta2009}:
\begin{equation}
S_0 (\hat \rho \| \hat \sigma ) := -\ln \left(  {\rm tr} [ \hat P_{\hat \rho} \hat \sigma ] \right),
\end{equation} 
where $\hat P_{\hat \rho}$ is the projection onto the support of  $\hat \rho$, and 
 \begin{equation}
 S_\infty (\hat \rho \| \hat \sigma ) := \ln \left( \min \{ \lambda : \hat \rho \leq \lambda \hat \sigma \} \right) = \ln \| \hat \sigma^{-1/2} \hat \rho \hat \sigma^{-1/2} \|_\infty.
 \label{def_Renyi_infty}
\end{equation}
We again note that we always assume that the support of $\hat \rho$ is included in that of $\hat \sigma$ for these divergence quantities.
We remark on some properties:
$S_{\alpha } (\hat \rho ) = \ln d - S_{\alpha} (\hat \rho \| \hat I / d)$,
$S_\alpha(\hat U \hat \rho \hat U^\dagger \| \hat U \hat \sigma \hat U^\dagger ) = S_\alpha (\hat \rho \| \hat \sigma )$  for any unitary $\hat U$, and
\begin{equation}
S_\alpha (\hat \rho \otimes \hat \rho' \| \hat \sigma \otimes \hat \sigma' ) = S_\alpha (\hat \rho \| \hat \sigma ) + S_\alpha (\hat \rho' \| \hat \sigma' ).
\label{quantum_Renyi_additive}
\end{equation}
Also, these divergences are non-negative and satisfy the following: 

\begin{theorem}[Lemma 5 and Lemma 6 of \cite{Datta2009}]
\begin{equation}
0 \leq S_0 (\hat \rho \| \hat \sigma )  \leq S_\infty (\hat \rho \| \hat \sigma ).
\end{equation}
Here,  $S_0 (\hat \rho \| \hat \sigma ) = 0$ holds if and only if the supports of $\hat \rho$ and $\hat \sigma$ are the same, and $S_\infty (\hat \rho \| \hat \sigma ) = 0$ holds if and only if $\hat \rho = \hat \sigma$.
\end{theorem}

\begin{proof}
$0 \leq S_0 (\hat \rho \| \hat \sigma )$ is obvious from the definition.
 $S_0 (\hat \rho \| \hat \sigma ) \leq S_\infty (\hat \rho \| \hat \sigma )$ can be shown as follows.
Let $\lambda := e^{S_\infty (\hat \rho \| \hat \sigma )}$.  
Then by definition, $0 \leq \lambda \hat \sigma - \hat \rho$ holds.
Thus $0 \leq {\rm tr}[( \lambda \hat \sigma - \hat \rho ) \hat P_{\hat \rho}] = \lambda e^{-S_0 (\hat \rho \| \hat \sigma )} -1$, which implies $S_0 (\hat \rho \| \hat \sigma ) \leq S_\infty (\hat \rho \| \hat \sigma )$.

The equality $S_0 (\hat \rho \| \hat \sigma ) = 0$ holds if and only if the support of $\hat \rho$ includes that of $\hat \sigma$, but in our setup, the latter always includes the former.

Suppose that $S_\infty (\hat \rho \| \hat \sigma ) = 0$ holds.  By definition,  $\hat \Delta := \hat \sigma - \hat \rho$ is positive.
If $\hat \Delta \neq 0$, there exists a positive eigenvalue and any other eigenvalues are non-negative, and thus ${\rm tr}[\hat \Delta] > 0$.
But  this contradicts  the normalization of $\hat \rho$ and $\hat \sigma$.
$\Box$
\end{proof}

The following relation is also known.

\begin{proposition}
\begin{equation}
S_0 (\hat \rho \| \hat \sigma ) \leq  S_1 (\hat \rho \| \hat \sigma ) \leq S_\infty (\hat \rho \| \hat \sigma ).
\label{quantum_0_1_infty}
\end{equation}
\end{proposition}

\begin{proof}
The following proof is based on \cite{Mitsuhashi2020p}.

We first prove the right inequality.  By definition, $e^{S_\infty (\hat \rho \| \hat \sigma )}\hat \sigma \geq \hat \rho$ holds, and thus $S_\infty (\hat \rho \| \hat \sigma ) - S_1 ( \hat \rho \| \hat \sigma ) = {\rm tr} [  \hat \rho ( \ln (e^{S_\infty (\hat \rho \| \hat \sigma )}\hat \sigma ) - \ln \hat \rho  )] \geq 0$.  Here, to obtain the last inequality, we used the fact that $\ln x$ is operator monotone (see Appendix~\ref{sec:operator_convex}).

To prove the left inequality,
let $\hat \rho = \sum_i p_i | \varphi_i \rangle \langle \varphi_i |$ and $\hat \sigma = \sum_j q_j | \psi_j \rangle \langle \psi_j |$.
Define $\hat \tau := \sum_{i : p_i > 0} e^{\langle \varphi_i | \ln \hat \sigma | \varphi_i \rangle} | \varphi_i \rangle \langle \varphi_i | / Z$, where $Z := \sum_{i : p_i > 0} e^{\langle \varphi_i | \ln \hat \sigma | \varphi_i \rangle}$.  Note that $Z = \sum_{i : p_i > 0} e^{\sum_j |  \langle \varphi_i | \psi_j \rangle |^2 \ln q_j} \leq \sum_{i : p_i > 0} \sum_j | \langle \varphi_i | \psi_j \rangle |^2 q_j = {\rm tr}[\hat P_{\rm \hat \rho} \hat \sigma]$.
We then have $0 \leq S_1 (\hat \rho \| \hat \tau) = S_1 (\hat \rho \| \hat \sigma ) + \ln Z \leq S_1 (\hat \rho \| \hat \sigma )  + \ln {\rm tr}[\hat P_{\rm \hat \rho} \hat \sigma] =S_1 (\hat \rho \| \hat \sigma ) - S_0 (\hat \rho \| \hat \sigma )$.
$\Box$
\end{proof}

An alternative way to understand  (\ref{quantum_0_1_infty}) is to consider two kinds of R\'enyi $\alpha$-divergences, both of which are defined for $0 < \alpha < 1$ and $1 < \alpha < \infty$ (see Appendix~\ref{apx:general_monotonicity} for details).
One is the simple R\'enyi  divergence  $\tilde S_\alpha (\hat \rho \| \hat \sigma)$, defined in Eq.~(\ref{simple_Renyi_divergence})~\cite{Hiai2011f,Tomamichel}.
The other is the sandwiched R\'enyi divergence $S_\alpha (\hat \rho \| \hat \sigma)$, defined in Eq.~(\ref{def_sandwiched})~\cite{Wilde2014,Lennert2013,Frank2013,Beigi2013}.
Both are non-decreasing functions in $\alpha$ as in the classical case (Proposition~\ref{prop:classical_alpha_monotone}), and satisfy $\lim_{\alpha \to 1} \tilde S_\alpha (\hat \rho \| \hat \sigma) = \lim_{\alpha \to 1} S_\alpha (\hat \rho \| \hat \sigma) = S_1   (\hat \rho \| \hat \sigma)$.
It is also known that  $\lim_{\alpha \to +0} \tilde S_\alpha (\hat \rho \| \hat \sigma)  = S_0   (\hat \rho \| \hat \sigma)$ and $\lim_{\alpha \to \infty}  S_\alpha (\hat \rho \| \hat \sigma)  = S_\infty   (\hat \rho \| \hat \sigma)$.  Thus we have inequalities~(\ref{quantum_0_1_infty}).

\

 It is also known that these divergences satisfy the monotonicity (and thus are monotones) under CPTP maps, as in the case for $\alpha = 1$:

\begin{theorem}[Monotonicity of the R\'enyi divergence, Lemma 7 of \cite{Datta2009}]
For  $\alpha = 0, \infty$, if $\mathcal E$ is CPTP, 
\begin{equation}
S_\alpha (\hat \rho \| \hat \sigma) \geq S_\alpha (\mathcal E (\hat \rho ) \|  \mathcal E ( \hat \sigma ) ).
\end{equation}
\label{thm:quantum_Renyi_monotonicity}
\end{theorem}

\begin{proof}
We first prove the case of $\alpha = \infty$~\cite{Datta2009}.
From the positivity of $\mathcal E$, $\lambda \hat \sigma - \hat \rho \geq 0$ implies $\lambda \mathcal E (\hat \sigma) - \mathcal E (\hat \rho) \geq 0$, and therefore $\{ \lambda : \lambda \hat \sigma - \hat \rho \geq 0 \} \subset \{ \lambda :  \lambda \mathcal E (\hat \sigma) - \mathcal E (\hat \rho) \geq 0 \}$.
We thus have $\min \{ \lambda : \lambda \hat \sigma - \hat \rho \geq 0 \} \geq \min \{ \lambda :  \lambda \mathcal E (\hat \sigma) - \mathcal E (\hat \rho) \geq 0 \}$, which implies $S_\infty (\hat \rho \| \hat \sigma ) \geq S_\infty (\mathcal E (\hat \rho) \| \mathcal E (\hat \sigma))$.
As seen from the proof, the positivity of $\mathcal E$ is enough to show the monotonicity.

We next prove the case of  $\alpha = 0$~\cite{Mitsuhashi2020p}.
Because any CPTP $\mathcal E$ can be written as $\mathcal E ( \hat \rho ) = {\rm tr}_{\rm A} [\hat U \hat \rho \otimes \hat \gamma \hat U^\dagger]$ with a state $\hat \gamma$ of an auxiliary system A and a unitary $\hat U$.
From Eq.~(\ref{quantum_Renyi_additive}) and the unitary invariance, it is enough to prove the monotonicity under the partial trace.  
Let $\hat \rho' := \hat U \hat \rho \otimes \hat \gamma \hat U^\dagger$ and $\hat \sigma' := \hat U \hat \sigma \otimes \hat \gamma \hat U^\dagger$.
We have $ {\rm tr}[\hat P_{{\rm tr}_{\rm A}[\hat \rho']}{\rm tr}_{\rm A}[\hat \sigma']] =  {\rm tr}[\hat P_{{\rm tr}_{\rm A}[\hat \rho'] \otimes \hat I}\hat \sigma']  \geq {\rm tr}[\hat P_{\hat \rho'} \hat \sigma']$,
where the right inequality follows from ${\rm supp}[{\rm tr}_{\rm A}[\hat \rho'] \otimes \hat I] \supset  {\rm supp}[\hat \rho']$.
In fact, any $| \varphi \rangle$ with $\langle \varphi |{\rm tr}_{\rm A}[\hat \rho'] | \varphi \rangle = 0$ satisfies $\langle \varphi | \hat \rho' | \varphi \rangle = 0$.
$\Box$
\end{proof}

\begin{theorem}[Joint convexity of R\'enyi 0-divergence]
We use the same notation as Theorem~\ref{thm:joint_convexity_KL} with $p_k > 0$.
Then,
\begin{equation}
S_0  (\hat \rho \| \hat \sigma ) \leq \sum_k p_k S_0 (\hat \rho_k \| \hat \sigma_k) .
\label{joint_convexity_0}
\end{equation}
The equality holds if the same equality condition as in Theorem~\ref{thm:joint_convexity_KL} is satisfied and $S_0 (\hat \rho_k \| \hat \sigma_k)$'s are the same for all $k$.
\label{thm:joint_convexity_0}
\end{theorem}

\begin{proof}
From $\hat P_{\hat \rho} \geq \hat P_{\hat \rho_k}$, we have 
\begin{equation}
 -\ln \left(  {\rm tr} [ \hat P_{\hat \rho} \hat \sigma ] \right)  
=   -\ln \left(  {\rm tr} \left[ \hat P_{\hat \rho} \sum_k p_k \hat \sigma_k \right] \right)   
\leq  -\ln \left(  {\rm tr} \left[ \sum_k p_k  \hat P_{\hat \rho_k} \hat \sigma_k \right] \right).   
\end{equation}
By noting the convexity of $-\ln x$, 
\begin{equation}
 -\ln \left(  {\rm tr} \left[ \sum_k p_k  \hat P_{\hat \rho_k} \hat \sigma_k \right] \right) \leq  - \sum_k p_k  \ln \left(  {\rm tr} [ \hat P_{\hat \rho_k} \hat \sigma_k ] \right),
\end{equation}
and thus we obtain inequality~(\ref{joint_convexity_0}).
The equality  in (\ref{joint_convexity_0}) holds if the equality conditions for the above two inequalities are satisfied. 
$\Box$
\end{proof}

The above proof is essentially the same as the proof of  Theorem 11 of Ref.~\cite{Erven_Harremoes} for the classical case.
We remark that the joint convexity does not hold for $\alpha = \infty$  (see also Corollary~\ref{cor:joint_convexity_Renyi} in Appendix~\ref{apx:general_monotonicity}).

\

For $\alpha = 0,1, \infty$, 
we can define $S_\alpha (\hat \rho \| \hat \sigma)$  for unnormalized $\hat \sigma$,
which is useful to define the nonequilibrium free energy (see also Chapter~\ref{chap:quantum_thermodynamics}).
For $Z>0$, it is easy to check the scaling property
\begin{equation}
S_\alpha (\hat \rho \| \hat \sigma / Z ) = S_\alpha (\hat \rho \| \hat \sigma ) + \ln Z.
\label{unnormalized_sigma_scaling}
\end{equation}
In general, if $\hat \sigma \leq \hat \sigma'$, then
\begin{equation}
 S_\alpha (\hat \rho \| \hat \sigma' ) \leq  S_\alpha (\hat \rho \| \hat \sigma ).
 \label{unnormalized_sigma_inequality}
\end{equation}
Again it is easy to check this for $\alpha = 0, \infty$.
For $\alpha = 1$, we use the fact that $\ln x$ is operator monotone (see Appendix~\ref{apx:general_monotonicity}).
The monotonicity of the $\alpha$-divergence with  $\alpha = 0,1, \infty$ under CPTP maps applies to unnormalized states (see Theorem~\ref{quantum_f_monotonicity} that is true for unnormalized states).



\chapter{Quantum majorization}
\label{chap:quantum_majorization}

In this chapter, we consider what majorization means in quantum systems.
The ordinary majorization can be directly generalized to the quantum case, as discussed in Section~\ref{sec:quantum_majorization}.
In Section~\ref{sec:noisy_operation}, we discuss a subclass of CPTP unital maps, called noisy operations.
On the other hand, the quantum analogue of d-majorization is not very straightforward, while we still have a simple characterization of state conversion by the R\'enyi $0$- and $\infty$-divergences as discussed in Section~\ref{sec:quantum_d_majorization}.
In Section~\ref{sec:entanglement}, we briefly remark on resource theory of entanglement and discuss its relation to infinite-temperature thermodynamics.

In this chapter, Chapter~\ref{chap:approximate_asymptotic}, and Chapter~\ref{chap:quantum_thermodynamics}, we assume that the dimensions of the input and output Hilbert spaces of CPTP maps are the same, i.e., $\mathcal E : \mathcal L (\mathcal H ) \to \mathcal L ( \mathcal H )$, unless stated otherwise.

\section{Quantum majorization}
\label{sec:quantum_majorization}

We first consider the quantum version of ordinary majorization, which  can be formulated in a parallel manner to the classical case (see also Ref.~\cite{Gour}).
Let $\hat \rho$ and $\hat \rho'$ be quantum states of dimension $d$, whose spectral decompositions are given by
\begin{equation}
\hat \rho = \sum_{i=1}^d p_i | \varphi_i \rangle \langle \varphi_i |, \ \hat \rho' = \sum_{i=1}^d p_i' | \varphi_i' \rangle \langle \varphi_i' |,
\label{diagonal_majorization}
\end{equation}
where their bases are not necessarily the same.  $p:= (p_1 ,\cdots, p_d)^{\rm T}$ and $p' :=(p_1', \cdots,  p_d')^{\rm T}$ are their diagonal distributions.  Then, we define quantum majorization as follows.

\begin{definition}[Quantum majorization]
We say that a quantum state $\hat \rho$ majorizes another state $\hat \rho'$, written as  $\hat \rho' \prec \hat \rho$, if their diagonal distributions satisfy  $p' \prec p$.
\end{definition}

Based on this definition, we have the following theorem.

\begin{theorem}
Let  $\hat \rho$, $\hat \rho'$ be quantum states.  The following are equivalent.
\begin{enumerate}
\item $\hat \rho' \prec \hat \rho$.
\item There exists a CPTP unital map $\mathcal{E}$ such that $\hat \rho' = \mathcal{E} (\hat \rho)$.
\item There exists a CPTP map that is a mixture of unitaries $\mathcal{E}$ such that $\hat \rho' = \mathcal{E} (\hat \rho)$.  Here, a mixture of unitaries is given in the form 
\begin{equation}
\mathcal E(\hat \rho) = \sum_k r_k  \hat U_k \hat \rho \hat U_k^\dagger,
\label{mixture_unitaries}
\end{equation} where $\hat U_k$ is unitary and $\sum_k r_k = 1$, $r_k \geq 0$.
\end{enumerate}
\label{thm:majorization_q}
\end{theorem}

\begin{proof}
In the following proof, we use the notations of Eq.~(\ref{diagonal_majorization}).

We first show (i) $\Rightarrow$ (ii).  Since $\hat \rho' \prec \hat \rho$ implies $p' \prec p$ by definition, Theorem~\ref{thm:majorization} for classical majorization implies that there exists a doubly stochastic matrix $T$ such that $p' = Tp$. We then define a map
\begin{equation}
\mathcal E(\hat \rho ) := \sum_{ji} T_{ji} | \varphi'_j \rangle \langle \varphi_i | \hat \rho | \varphi_i \rangle  \langle \varphi'_j |.
\end{equation}
Since $\sum_{ji} \hat M_{ji}^\dagger \hat M_{ji} = \hat I$ with $\hat M_{ji}:= \sqrt{T_{ji}} | \varphi_j' \rangle \langle \varphi_i |$, $\mathcal E$ is CPTP.
Also, since $T$ is doubly stochastic, $\mathcal E(\hat I) = \hat I$ holds and thus  $\mathcal E$ is unital. Finally, it is obvious that $\hat \rho' = \mathcal E(\hat \rho)$.

We next show (ii) $\Rightarrow$ (i).
Suppose that $\mathcal E$ is CPTP unital, and  define
\begin{equation}
T_{ji} := \langle \varphi'_j | \mathcal E( | \varphi_i \rangle \langle \varphi_i |) | \varphi'_j \rangle. 
\end{equation}
Then, since $\sum_{i}T_{ji} = \sum_{j}T_{ji} = 1$ and $T_{ji} \geq 0$, $T_{ji}$ is doubly stochastic.
Also, it is obvious that $p'_j = \sum_{i}T_{ji}p_i$.
Therefore, $p' \prec p$ holds, and thus $\hat \rho' \prec \hat \rho$.

(i) $\Leftrightarrow$ (iii) (or (ii) $\Leftrightarrow$ (iii)) is called the  Uhlmann's theorem (Theorem~12.13 of Ref.~\cite{Nielsen}).  Since (iii)  $\Rightarrow$ (ii) is obvious, here we only show (i) $\Rightarrow$ (iii).
Suppose that $\hat \rho' \prec \hat \rho$.  Then, there exists a doubly stochastic matrix $T$ such that $p' = Tp$ for the diagonal distributions. From the Birkhoff's theorem (Theorem~\ref{thm:Birkhoff}), $T$ can be written as $T=\sum_k r_k P_k$ with  permutation matrices $P_k$.
Let $\hat P_k$ be the operator that acts as $P_k$ in the basis $\{ | \varphi_k \rangle \}$, and let $\hat V$ be the unitary operator that converts the basis $\{ | \varphi_k \rangle \}$ into $\{ | \varphi'_k \rangle \}$.
Then, we have $\hat \rho' = \sum_k r_k \hat V \hat P_k \hat \rho \hat P_k^\dagger \hat V^\dagger$.
By noting that $\hat U_k := \hat V \hat P_k$  is unitary, we have $\hat \rho' = \sum_k r_k  \hat U_k \hat \rho \hat U_k^\dagger$.
$\Box$
\end{proof}


While (ii) $\Leftrightarrow$ (iii) is true in Theorem~\ref{thm:majorization_q}, the following fact is known: The Birkhoff's theorem (Theorem~\ref{thm:Birkhoff}) is not true in the quantum case.

\begin{theorem}
For $d \geq 3$ with $d$ being the dimension of the system,  the set of mixtures of unitaries is a strict subset of CPTP unital maps.  That is, there are CPTP unital maps that cannot be written as any mixtures of unitaries (\ref{mixture_unitaries}). In other words, an extreme point of CPTP unital maps is not necessarily unitary.
\end{theorem}

\begin{proof}
We show an example of an operator that is CPTP unital but not a mixture of unitaries~\cite{Landau}.   Let $\hat S_k$ ($k=x,y,z$) be the spin-$l$ operator (that is a representation of ${\mathfrak su}(2)$) with $l=1/2, 1, 3/2, \cdots$. 
We define a CPTP map
\begin{equation}
\mathcal E(\hat \rho) := \frac{1}{l(l+1)}\sum_{k=x,y,z}\hat S_k \hat \rho \hat S_k, 
\label{unital_extreme}
\end{equation}
which is shown to be unital.
For $l \geq 1$ (or equivalently $d \geq 3$ because of $d=2l+1$), $\hat S_k$ is not unitary.
On the other hand, $\mathcal E$ is an extreme point of CPTP unital maps, because of the Choi's theorem (Theorem~5 of Ref.~\cite{Choi75}).
Therefore, $\mathcal E $ in Eq.~(\ref{unital_extreme}) cannot be written as any mixture of unitaries.

We note that Choi's theorem is stated as follows.
Let $\mathcal E$ be  a CPTP unital map that has a Kraus representation $\mathcal E(\hat \rho) := \sum_k \hat M_k \hat \rho \hat M_k^\dagger$ with $\sum_k \hat M_k \hat M_k^\dagger = \sum_k \hat M_k^\dagger \hat M_k = \hat I$.
Then, $\mathcal E$ is an extreme point of CPTP unital maps, if and only if $\hat M_k \hat M_l^\dagger$'s are linearly independent in the operator space. $\Box$
\end{proof}

\section{Noisy operations}
\label{sec:noisy_operation}

We next consider a subclass of CPTP unital maps, called \textit{noisy operations}.  
A noisy operation is also a special case of thermal operations introduced in Section~\ref{sec:thermodynamic_operations}.

\begin{definition}[Noisy operations]
A CPTP map $\mathcal E$ is an exact noisy operation, if there exists an auxiliary system B with Hilbert space dimension $d_{\rm B} < \infty$ and exists a unitary operator $\hat U$ acting on the composite system, such that
\begin{equation}
\mathcal E (\hat \rho ) = {\rm tr}_{\rm B} [ \hat U \hat \rho \otimes ( \hat I_{\rm B} / d_{\rm B} ) \hat U^\dagger ],
\label{noisy_operation}
\end{equation}
where $\hat I_{\rm B}$ is the identity operator of B.
Furthermore, a CPTP map $\mathcal E$ is a noisy operation, if there exists a sequence of exact noisy operations $\{ \mathcal E_n \}_{n=1}^\infty$ such that $\mathcal E_n$ converges to $\mathcal E$ in $n \to \infty$.
\label{def:noisy_operation}
\end{definition}

In Ref.~\cite{Shor}, exact noisy operations are called exactly factorizable maps,  and noisy operations are called strongly factorizable maps.
There is also a concept called factorizable maps~\cite{Haagerup}, for which  infinite-dimensional auxiliary systems are allowed with the use of von Neumann algebras.
We note that the set of factorizable maps coincide with the set of strongly factorizable maps if and only if  Connes' embedding conjecture is true~\cite{Shor} (see  Ref.~\cite{Ji2020} for a recent work on the conjecture). 

We next discuss the relationship between unital maps and noisy operations.

\begin{proposition}[Lemma 5 of \cite{Gour}]
Any noisy operation is a unital map.   Also, any mixture of unitaries is a noisy operation.
\label{noisy_unital}
\end{proposition}

\begin{proof}
The former is obvious from Eq.~(\ref{noisy_operation}).  To show the latter, let $\mathcal E(\hat \rho) = \sum_k r_k \hat U_k \hat \rho \hat U_k^\dagger$ be the mixture of unitaries with $r_k > 0$.
We first consider the case that all of $r_k$'s are rational numbers written as $l_k/m$ with positive  integers $l_k$ and $m$.
Then we take an auxiliary Hilbert space in dimension $d_{\rm B} := m$, and divide it into subspaces labeled by $k$ with dimensions $l_k$.  Then we can take $\hat U$ in Eq.~(\ref{noisy_operation}) as a controlled unitary, which acts on the system as $\hat U_k$ if the auxiliary state is in the $k$th subspace.
Finally, if some of $r_k$'s are irrational, we can approximate them by rational numbers and then take the limit of $m \to \infty$.
$\Box$
\end{proof}

The above inclusions are strict.

\begin{proposition}
There are CPTP unital maps that cannot be written as any noisy operations.  Also, there are noisy operations that cannot be written as any mixtures of unitaries.
\label{strict_inclusions_noisy}
\end{proposition}

\begin{proof}
An example of the former is shown in Example 3.1 of Ref.~\cite{Haagerup}, which is the same as Eq.~(\ref{unital_extreme}) (with $l=1$).  An example of the latter is discussed in Ref.~\cite{Shor} based on Ref.~\cite{Mendl2009}.
$\Box$
\end{proof}

Although the above inclusions are strict, we have the following because any mixture of unitaries is a noisy operation.

\begin{corollary}
The conditions of Theorem~\ref{thm:majorization_q} is further equivalent to:\\
(iv) There exists a noisy operation $\mathcal E$ such that $\hat \rho' = \mathcal E (\hat \rho)$.
\label{noisy_corollary}
\end{corollary}

We finally remark on the classical case, where the set of noisy operations and the set of doubly stochastic matrices are equivalent.

\begin{definition}[Classical noisy operations]
A stochastic matrix $T$ is an exact classical noisy operation, if there exists an auxiliary system B  in dimension $d_{\rm B} < \infty$ and exists a permutation matrix $P$ acting on the composite system such that
\begin{equation}
(T p)_i= \sum_{jkl} P_{ij; kl} p_k u_l,
\label{noisy_operation_c}
\end{equation}
where $u:=(1/d_{\rm B}, \cdots, 1/d_{\rm B})^{\rm T}$ is the uniform distribution of B. 
Furthermore, a stochastic matrix $T$ is a classical noisy operation, if there exists a sequence of exact classical noisy operations $\{ T_n \}_{n=1}^\infty$ such that $T_n$ converges to $T$ in $n \to \infty$.
\label{def:noisy_operation_c}
\end{definition}

\begin{proposition}[Lemma 6 of \cite{Gour}]
For any classical  stochastic map $T$, the following are equivalent. \\
(i) $T$ is a doubly stochastic map.\\
(ii) $T$ is a noisy operation.
\end{proposition}

\begin{proof}
Because (ii) $\Rightarrow$ (i) is obvious, we only show (i) $\Rightarrow$ (ii).
The proof is parallel to the latter part of Proposition~\ref{noisy_unital}.
Let  $T=\sum_k r_k P_k$ be the decomposition of the Birkhoff's theorem (Theorem~\ref{thm:Birkhoff}).
If all of $r_k$'s are rational, we divide the sample space (i.e., the set of labels of components of probability vectors) into subspaces labeled by $k$ with dimensions proportional to $r_k$.  Then, apply the controlled permutation acting on the system as $P_k$ if the auxiliary state is in the $k$th subspace.  For general $r_k$'s, take the limit.
$\Box$
\end{proof}

\section{Quantum d-majorization}
\label{sec:quantum_d_majorization}

We now consider quantum d-majorization (or quantum relative majorization).
It is not very straightforward to formulate the quantum counterpart of the classical case discussed in Section~\ref{sec:d_majorization}, 
especially the  ``Lorenz curve-like'' characterization of the necessary and sufficient condition for the existence of a CPTP map $\mathcal E$ satisfying $\hat \rho'=\mathcal{E} (\hat \rho)$ and $\hat \sigma' = \mathcal{E} (\hat \sigma)$ for given pairs of quantum states $(\hat \rho, \hat \sigma)$, $(\hat \rho', \hat \sigma')$
(i.e.,  the full quantum version of Theorem.~\ref{thm:d_majorization} or the d-majorization version of Theorem~\ref{thm:majorization_q}).
A difficulty lies in the fact that $\hat \rho$ and $\hat \sigma$ are not necessarily commutable (and thus not simultaneously diagonalizable) in the quantum case.

There are several interesting approaches to formulate full quantum d-majorization on the basis of,
e.g., a  generalized Lorenz curve~\cite{Buscemi2017} and matrix majorization~\cite{Gour2017}.
We will mention the former later in this section.
Here, however, we adopt a most naive approach to quantum d-majorization: 
we simply \textit{define} quantum d-majorization in the following form, given (iv) of Theorem~\ref{thm:d_majorization} of the classical case.

\begin{definition}[Quantum d-majorization]
Let $\hat \rho, \hat \sigma, \hat \rho', \hat \sigma'$ be quantum states.
We say that a pair $(\hat \rho, \hat \sigma)$ d-majorizes another pair $(\hat \rho', \hat \sigma') $, written as
$(\hat \rho', \hat \sigma') \prec (\hat \rho, \hat \sigma)$, if  there exists a CPTP map $\mathcal E$ such that 
\begin{equation}
\hat \rho' = \mathcal E (\hat \rho), \ \ \hat \sigma' = \mathcal E(\hat \sigma).
\end{equation}
In particular, we say that $\hat \rho$ thermo-majorizes $\hat \rho'$ with respect to $\hat \sigma$, if $(\hat \rho', \hat \sigma) \prec (\hat \rho, \hat \sigma)$.
\label{def:quantum_d_majorization}
\end{definition}

Under this  definition, the characterization of state convertibility by the R\'enyi $0$- and $\infty$-divergences still works, as described by the following theorem.
This is the quantum counterpart of Theorem~\ref{thm:asymp0} of the classical case.
While Theorem~\ref{thm:asymp0} was obvious if one looks at the classical Lorenz curves, the following theorem is proved without invoking such graphical representation.

\begin{theorem}[Conditions for state conversion] 
\
\begin{description}
\item[(a) Necessary conditions:] If $( \hat \rho', \hat \sigma' ) \prec (\hat \rho, \hat \sigma)$ holds, then 
\begin{equation}
S_0 (\hat \rho' \| \hat \sigma') \leq S_0 (\hat \rho \| \hat \sigma ), \ \ \ S_\infty (\hat \rho' \| \hat \sigma') \leq S_\infty (\hat \rho \| \hat \sigma ).
\end{equation} 
\item[(b) Sufficient condition:]  $( \hat \rho', \hat \sigma' ) \prec (\hat \rho, \hat \sigma)$ holds, if (but not only if)
\begin{equation}
S_\infty (\hat \rho' \| \hat \sigma') \leq S_0 (\hat \rho \| \hat \sigma ).
\end{equation}
\end{description}
\label{thm:asymp0_q}
\end{theorem}

\begin{proof}
(a) This is the monotonicity of  $S_0 (\hat \rho \| \hat \sigma )$ and $S_\infty (\hat \rho \| \hat \sigma )$ (Theorem~\ref{thm:quantum_Renyi_monotonicity}).

(b) The following proof is based on Refs.~\cite{Faist2018,Mitsuhashi2020p}.
We explicitly construct a  CPTP map $\mathcal E$ that transforms  $( \hat \rho , \hat  \sigma )$ to $(\hat \rho', \hat \sigma')$ by a ``measure-and-prepare'' method.

First, we perform a projection measurement 
$\{ \hat P_{\hat \rho}, \hat I - \hat P_{\hat \rho} \}$, where  $\hat P_{\hat \rho}$ is the projection onto the support of $\hat \rho$.
We label the measurement outcome corresponding to $ \hat P_{\hat \rho}$ and $\hat I - \hat P_{\hat \rho} $ by ``0'' and ``1'', respectively.  
Then, the conditional probabilities of getting these outcomes are given by 
\begin{equation}
P( 0 | \hat \rho) = 1, \ P(1 | \hat \rho ) = 0, \ P(0 | \hat \sigma )= c, \ P(1 | \hat \sigma ) = 1 - c,
\label{SH_S0_eq}
\end{equation}
where $c:=e^{-S_0 (\hat \rho \| \hat \sigma )  }$.
We note that this measurement corresponds to the hypothesis testing with $\eta = 1$ (see Appendix \ref{appx:hypothesis_testing}),
and Eq.~(\ref{SH_S0_eq}) implies that  $S_0 (\hat \rho \| \hat \sigma )$ can be identified with the hypothesis testing divergence  $S_{\rm H}^\eta  (\hat \rho \| \hat \sigma )$ with $\eta = 1$ (see also Eq.~(\ref{SH_S0_1})).

Next, we prepare a state depending on the outcome.
Suppose that $c \neq 1$.
When the outcome is ``0'', we prepare $\hat \rho'$.
When the outcome is ``1'', we prepare 
\begin{equation}
\hat \sigma'' := \frac{\hat \sigma' - c \hat \rho'}{1-c}.
\end{equation}
Here, if  $S_\infty (\hat \rho' \| \hat \sigma') \leq S_0 (\hat \rho \| \hat \sigma )$, then
$\hat \sigma'' \geq  ( \hat \sigma' - e^{-S_\infty (\hat \rho' \| \hat \sigma')} \hat \rho' ) / ( 1-c ) \geq 0$,
and thus $\hat \sigma''$ is a (normalized) quantum state.
By this state preparation, we obtain the final state $c \hat \rho' + (1-c) \hat \sigma'' = \hat \sigma'$ from the initial state $\hat \sigma$.
If $c=1$, $S_\infty (\hat \rho' \| \hat \sigma') = 0$ or equivalently $\hat \rho' = \hat \sigma'$ holds, and thus the state preparation is trivial.

In summary, the constructed CPTP map is given by (for $c \neq 1$)
\begin{equation}
\mathcal E (\hat \tau ) := {\rm tr}[\hat P_{\hat \rho} \hat \tau] \hat \rho' + \left( 1- {\rm tr}[\hat P_{\hat \rho} \hat \tau] \right) \frac{\hat \sigma' - c \hat \rho'}{1-c}.
\end{equation}
$\Box$
\end{proof}

Figure~\ref{fig:quantum_d_majorization} schematically illustrates the sufficient condition of Theorem~\ref{thm:asymp0_q}.
In the case of Fig.~\ref{fig:quantum_d_majorization} (a), the $0$- and $\infty$-divergences of $(\hat \rho, \hat \sigma)$ and those of $(\hat \rho', \hat \sigma')$ are completely separated, and the sufficient condition of Theorem~\ref{thm:asymp0_q} is satisfied.
On the other hand, in the case of Fig.~\ref{fig:quantum_d_majorization} (b), we cannot judge whether state transformation is possible from $(\hat \rho, \hat \sigma)$ to $(\hat \rho', \hat \sigma')$, because their divergences are not separated.
This illustrates that a necessary and sufficient condition cannot be given only by the $0$- and $\infty$-divergences in general.
In the next chapter, however, we will see that the $0$- and $\infty$-divergences can approximately collapse to a single value under certain conditions, if we take the asymptotic limit; In such a case, a complete monotone emerges and provides a necessary and sufficient condition for state conversion by d-majorization.

\begin{figure}[htbp]
 \begin{center}
  \includegraphics[width=70mm]{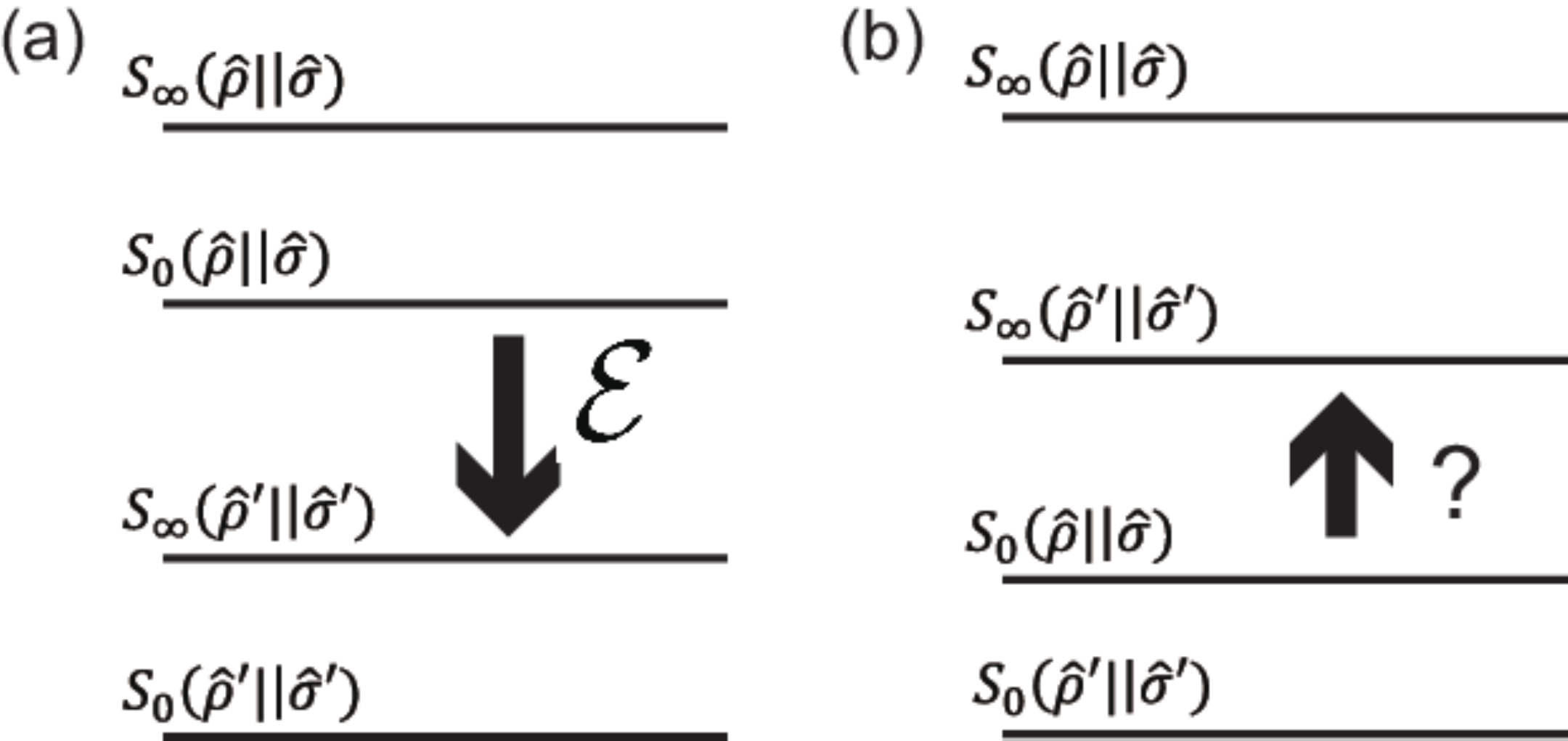}
 \end{center}
 \caption{Schematic of situations where (a) the sufficient condition $S_0 (\hat \rho \| \hat \sigma ) \geq S_\infty (\hat \rho' \| \hat \sigma') $ works and (b) does not work.} 
\label{fig:quantum_d_majorization}
\end{figure}

\

We  next briefly discuss a quantum analogue of the Lorenz curve in line with Ref.~\cite{Buscemi2017}, which, however, cannot characterize state convertibility in general.

\begin{definition}[Quantum Lorenz curve~\cite{Buscemi2017}]
Let $\hat \rho, \hat \sigma$ be quantum states.  The quantum Lorenz curve of $( \hat \rho, \hat \sigma)$ is defined as the upper boundary of the set $\{ (x,y) \ : \ x={\rm tr}[\hat Q \hat \rho],  y= {\rm tr}[\hat Q \hat \sigma],  0 \leq \hat Q \leq \hat I \} \subset \mathbb R^2$.
Then, we write  $(\hat \rho', \hat \sigma' ) \prec_{\rm q} (\hat \rho, \hat \sigma)$, if the  quantum Lorenz curve of $(\hat \rho, \hat \sigma)$ lies above that of $(\hat \rho', \hat \sigma')$.  
\end{definition}

If $[ \hat \rho, \hat \sigma ] = 0$ and $[ \hat \rho', \hat \sigma' ] = 0$, the above definition reduces to the classical case.
In the quantum case, however,  $(\hat \rho', \hat \sigma' ) \prec_{\rm q} (\hat \rho, \hat \sigma)$ is \textit{not} equivalent to $(\hat \rho', \hat \sigma' ) \prec (\hat \rho, \hat \sigma)$ of Definition~\ref{def:quantum_d_majorization}~\cite{Heinosaari2012,Matsumoto2014}.
Instead, it is worth noting the following theorem.

\begin{theorem}[Theorem 2 of~\cite{Buscemi2017}]
Let  $\hat \rho$, $\hat \rho'$, $\hat \sigma$, $\hat \sigma'$ be quantum states.  The following are equivalent.
\begin{enumerate}
\item $(\hat \rho', \hat \sigma') \prec_{\rm q} (\hat \rho, \hat \sigma)$.
\item For all $t \in \mathbb R$, $\| \hat \rho' - t \hat \sigma' \|_1 \leq \| \hat \rho - t \hat \sigma \|_1$. 
\item For all $0 \leq \eta \leq 1$, $Q^\eta_{\rm H} ( \hat \rho \|  \hat \sigma ) \leq Q^\eta_{\rm H} (\hat \rho' \| \hat \sigma')$ holds, where $Q^\eta_{\rm H} (\hat \rho \| \hat \sigma ) := \min_{0 \leq \hat Q \leq \hat I, {\rm tr}[\hat \rho \hat Q] \geq \eta } {\rm tr}[\hat \sigma \hat Q ] $.
\end{enumerate}
\end{theorem}

Here,  the classical counterpart of  (i) $\Leftrightarrow$ (ii)  is that of Theorem~\ref{thm:d_majorization}.
In (iii) above, $Q^\eta_{\rm H} (\hat \rho \| \hat \sigma ) = \eta e^{- S_{\rm H}^\eta  (\hat \rho \| \hat \sigma ) }$ holds for $0 < \eta < 1$, where  $S_{\rm H}^\eta  (\hat \rho \| \hat \sigma )$ is the hypothesis testing divergence defined by Eq.~(\ref{def_SH}) of Appendix~\ref{appx:hypothesis_testing}.

\section{Remark on resource theory of entanglement}
\label{sec:entanglement}

Let us very briefly mention the relationship between majorization and entanglement.
(See, e.g., Ref.~\cite{Nielsen} for details. We also note Refs.~\cite{Oppenheim2002,Horodecki2002} as seminal works that made an analogy between entanglement and thermodynamics.)

Suppose that a bipartite system AB is shared by Alice and Bob.
They are allowed to perform any quantum operations on their own systems and to communicate by classical channels, which is called local operations and classical communications (LOCC).
For simplicity, we suppose that the state of AB is a pure state $| \Psi \rangle$.
In this case, product states are free states and LOCC are free operations.

By the Schmidt decomposition, the pure state can be always written as $| \Psi \rangle = \sum_i \sqrt{p_i} | \varphi_{{\rm A},i} \rangle | \varphi_{{\rm B},i} \rangle$ with a classical probability distribution $p_i$ and orthonormal bases  $\{ | \varphi_{{\rm A},i} \rangle \}$, $\{ | \varphi_{{\rm B},i} \rangle \}$ of A, B.
Then, the reduced density operator of A and B are respectively given by $\hat \rho_{\rm A} =  \sum_i p_i | \varphi_{{\rm A},i} \rangle \langle \varphi_{{\rm A},i} |$ and $\hat \rho_{\rm B} =  \sum_i p_i | \varphi_{{\rm B},i} \rangle \langle \varphi_{{\rm B},i} |$.  Since these reduced density operators are the same up to a unitary, we will only  focus on  $\hat \rho_{\rm A}$, and write it as $\hat \rho_\Psi$ to explicitly show its dependency on the global state $| \Psi \rangle$.

Entanglement of pure states can be quantified by the von Neumann entropy $S_1 (\hat \rho_\Psi)$, which is called entanglement entropy.  While the von Neumann entropy is a monotone under LOCC, it is not a complete monotone.
On the other hand, majorization can completely characterize state convertibility under LOCC, which is shown in the following theorem.

\begin{theorem}[Theorem 1 of \cite{Nielsen1999}]
$| \Psi \rangle$ can be converted into $| \Psi' \rangle$ by LOCC, if and only if $\hat \rho_\Psi \prec \hat \rho_{\Psi'}$.
\end{theorem}

This is reminiscent of Theorem~\ref{thm:majorization_q} related to infinite-temperature thermodynamics.
However, there is a crucial difference between  entanglement and thermodynamics.
In fact, the direction of state conversion is opposite to that of majorization in the case of entanglement, in contrast to the case of thermodynamics at infinite temperature.
This is because the completely mixed state (the uniform distribution) corresponds to the maximally entangled state that has the maximum resource of entanglement, while in thermodynamics at infinite temperature, the completely mixed state is just a free state.


\chapter{Approximate and asymptotic majorization and divergence}
\label{chap:approximate_asymptotic}

In this chapter, we consider  an approximate (or ``smoothed'') version of the R\'enyi $0$- and $\infty$-divergences, called the smooth divergences.
We then consider  the asymptotic limit of these smooth divergences, leading to the concept called \textit{information spectrum} (i.e.,   the upper and lower spectral divergence rates corresponding to the  $0$- and $\infty$-divergences).
This formalism naturally leads to  characterization of asymptotic state convertibility of many-body systems even beyond  independent and identically distributed (i.i.d.) situations.

In Section~\ref{sec:smooth_divergence}, we introduce the smooth $0$- and $\infty$-divergences and the corresponding approximate majorization.
In Section~\ref{sec:asymptotic}, we consider the asymptotic limit and introduce the spectral divergence rates, and show their relation to asymptotic d-majorization.
In Section~\ref{sec_condition_information}, we investigate the condition with which the upper and lower spectral divergence rates collapse to a single value.
This is the quantum version of the asymptotic equipartition property (AEP).
In such a case, asymptotic state convertibility by d-majorization can be characterized by a single complete  monotone.
In Section~\ref{sec:misc}, we make some technical remarks on the smoothing of the divergences.

Throughout this chapter, we consider the general setup: quantum d-majorization and the corresponding quantum divergences.
The classical counterpart is immediately obtained as a special case of the quantum formulation.

\

The central concepts introduced in this chapter are summarized as follows.
First, for $\varepsilon > 0$, we introduce smooth divergences $S_0^\varepsilon (\hat \rho \| \hat \sigma)$ and $S_\infty^\varepsilon (\hat \rho \| \hat \sigma)$, which are $\varepsilon$-approximations of the $0$- and $\infty$-divergences, respectively.
We next consider the asymptotic limit, where $\hat \rho_n$ and $\hat \sigma_n$ are regarded as states of a large system of size $n$.
A simplest example is the i.i.d. case, where $\hat \rho_n := \hat \rho^{\otimes n}$ and $\hat \sigma_n := \hat \sigma^{\otimes n}$.
We emphasize, however, that our formulation here is very general and not restricted to the i.i.d. case.
Then, we take the asymptotic limit of $S_0^\varepsilon (\hat \rho_n \| \hat \sigma_n)$ and $S_\infty^\varepsilon (\hat \rho_n \| \hat \sigma_n)$ by taking $n \to \infty$ first and then $\varepsilon \to + 0$, and we obtain  the lower and the upper  spectral divergence rates, respectively.
We show that the spectral divergence rates give a characterization of state convertibility in the asymptotic regime.
Under certain conditions including ergodicity, we show that $S_0^\varepsilon (\hat \rho_n \| \hat \sigma_n) \simeq S_\infty^\varepsilon (\hat \rho_n \| \hat \sigma_n) \simeq S_1 (\hat \rho_n \| \hat \sigma_n) $ holds with  $n \to \infty$ and $\varepsilon \to + 0$, which is a consequence of the quantum AEP and implies that the KL divergence rate serves as a complete monotone.

We note that quantum hypothesis testing is a useful tool in this chapter, which gives a divergence-like quantity called the hypothesis testing divergence $S_{\rm H}^\eta (\hat \rho \| \hat \sigma)$.
While we postpone the details of hypothesis testing to Appendix \ref{appx:hypothesis_testing}, here we only remark that, roughly speaking,
$S_{\rm H}^{\eta \simeq 1} (\hat \rho \| \hat \sigma) \simeq S_0^{\varepsilon \simeq 0} (\hat \rho \| \hat \sigma)$
and
$S_{\rm H}^{\eta \simeq 0} (\hat \rho \| \hat \sigma) \simeq S_\infty^{\varepsilon \simeq 0} (\hat \rho \| \hat \sigma)$ hold (see Proposition~\ref{Faist_prop} for a rigorous formulation).
Thus, the hypothesis testing divergence has essentially the same information as the smooth $0$- and $\infty$-divergences; this is particularly the case when we consider the asymptotic limit, because the correction terms are independent of $n$.


\section{Smooth divergence and approximate majorization}
\label{sec:smooth_divergence}

We consider  approximate state conversion and smooth divergences in line with the framework developed in Refs.~\cite{Renner2005,Renner2004}. 
Specifically, we consider approximate quantum d-majorization, and correspondingly, the smooth quantum R\'enyi $0$- and $\infty$-divergences introduced in Ref.~\cite{Datta2009}.

Let $\varepsilon \geq 0$.
We define the $\varepsilon$-neighborhood of a normalized state $\hat \rho$ by  
\begin{equation}
B^\varepsilon (\hat \rho ) := \{ \hat \tau  : \  D(  \hat \tau , \hat \rho ) \leq \varepsilon,  {\rm tr}[\hat \tau] = 1, \hat \tau \geq 0 \}.
\label{normalized_B}
\end{equation}
Here we adopted the trace distance as a norm for smoothing, while this choice is not essential and other norms can be adopted equivalently.
Also, only normalized states are allowed in the above definition, while this is again not essential.
See also Section~\ref{sec:misc} for alternative ways of smoothing.

We now define approximate d-majorization and smooth divergences.

\begin{definition}[Approximate d-majorization]
Let $\varepsilon \geq 0$.  We say that a pair $(\hat \rho, \hat \sigma )$ $\varepsilon$-approximately d-majorizes another pair $(\hat \rho', \hat \sigma')$, written as $(\hat \rho', \hat \sigma') \prec^\varepsilon (\hat \rho , \hat \sigma)$, if  there exists $\hat \tau' \in B^\varepsilon (\hat \rho' )$ such that $(\hat \tau', \hat \sigma') \prec (\hat \rho ,\hat \sigma )$.
\label{approx_majorization}
\end{definition}

\begin{definition}[Smooth R\'enyi $0/\infty$-divergence]
\begin{equation}
S_\infty^\varepsilon (\hat \rho \| \hat \sigma ) := \min_{\hat \tau  \in B^\varepsilon (\hat \rho )} S_\infty ( \hat \tau \|  \hat \sigma),
\end{equation}
\begin{equation}
S_0^\varepsilon(\hat \rho \| \hat \sigma ) := \max_{\hat \tau  \in B^\varepsilon (\hat \rho )} S_0 ( \hat \tau \|  \hat \sigma).
\end{equation}
\end{definition}

The smooth divergences satisfy the monotonicity, which is  the smoothed version of Theorem~\ref{thm:quantum_Renyi_monotonicity} or Theorem~\ref{thm:asymp0_q} (a).

\begin{theorem}[Monotonicity of the smooth divergences]
$ (\hat \rho' , \hat \sigma' ) \prec  (\hat \rho , \hat \sigma )$ implies that for any $\varepsilon > 0$,
\begin{description}
\item[(a)]  $
S_\infty^\varepsilon (\hat \rho \| \hat \sigma)  \geq  S_\infty^\varepsilon (\hat \rho'  \| \hat \sigma' ) 
$.
\item[(b)] 
$
S_0^\varepsilon (\hat \rho \| \hat \sigma) \geq  S_0^{\varepsilon^2 /6} ( \hat \rho ' \| \hat \sigma' ) + \ln (\varepsilon^2 / 6) 
$.
\end{description}
\label{Smooth_monotone}
\end{theorem}

\begin{proof}
Let $\mathcal E$ be a CPTP map such that $\hat \rho' = \mathcal E (\hat \rho)$ and $\hat \sigma' = \mathcal E  ( \hat \sigma)$.

(a)
The proof is parallel to that of Lemma 59 of Ref.~\cite{Gour} for ordinary majorization:
\begin{eqnarray}
S_\infty^\varepsilon (\hat \rho'  \| \hat \sigma') 
&:=& \min_{\hat \tau'  \in B^\varepsilon (\hat \rho' )} S_\infty ( \hat \tau' \|  \hat \sigma') \\
&\leq& \min_{\hat \tau  \in B^\varepsilon (\hat \rho )} S_\infty ( \mathcal E( \hat \tau) \|  \hat \sigma') \\
&\leq& \min_{\hat \tau  \in B^\varepsilon (\hat \rho )} S_\infty ( \hat \tau \|  \hat \sigma) \\
&=:& S_\infty^\varepsilon (\hat \rho  \| \hat \sigma).
\end{eqnarray}
Here, to obtain the second line, we used that  $\{ \mathcal E (\hat \tau ) :  \hat \tau \in B^\varepsilon (\hat \rho) \} \subset  B^\varepsilon (\hat \rho')$ holds from the monotonicity of the trace distance.
We also used the monotonicity of $S_\infty (\hat \rho \| \hat \sigma ) $ 
 to obtain the third line.

(b)
From the monotonicity of the hypothesis testing divergence $S_{\rm H}^\eta (\hat \rho \| \hat \sigma)$ (Proposition~\ref{prop:H_monotone}) and Proposition \ref{Faist_prop} in Appendix~\ref{appx:hypothesis_testing}, we have 
\begin{eqnarray}
S_0^\varepsilon (\hat \rho \| \hat \sigma) 
&\geq&   S_{\rm H}^{1-\varepsilon'}(\hat \rho \| \hat \sigma ) - \ln \frac{1-\varepsilon'}{\varepsilon'} \\
&\geq& S_{\rm H}^{1-\varepsilon'}(\hat \rho' \| \hat \sigma') - \ln \frac{1-\varepsilon'}{\varepsilon'} \\
&\geq& S_0^{\varepsilon'} (\hat \rho'  \| \hat \sigma') + \ln \varepsilon',
\end{eqnarray}
where $\varepsilon' := \varepsilon^2 / 6$. 
$\Box$
\end{proof}

The above theorem can be generalized to the case where state conversion is not exact.

\begin{theorem}[Necessary conditions for approximate state conversion]
For any $\varepsilon \geq 0$, $(\hat \rho', \hat \sigma')  \prec^\varepsilon  (\hat \rho , \hat \sigma)$ implies that:
\begin{description}
\item[(a)]
 For any $\delta \geq 0$,
 \begin{equation}
 S_\infty^\delta  (\hat \rho \| \hat \sigma) \geq S_\infty^{\varepsilon + \delta}(\hat \rho' \| \hat \sigma')  .
 \end{equation}
\item[(b)]  For any $\delta$ satisfying $\delta^2/6 > \varepsilon$,
\begin{equation}
S_0^\delta (\hat \rho \| \hat \sigma )
\geq 
S_0^{\delta^2/6 - \varepsilon} (\hat \rho' \| \hat \sigma' )
+ \ln ( \delta^2/6 ) .
\end{equation}
\end{description}
\label{thm:asymp1_q_a}
\end{theorem}

\begin{proof}
If $(\hat \rho', \hat \sigma')  \prec^\varepsilon  (\hat \rho , \hat \sigma)$, there exists a CPTP map $\mathcal E$ such that $\hat \sigma' = \mathcal E(\hat \sigma)$ and  $D (\hat \rho', \hat \tau') \leq \varepsilon$ with $\hat \tau' := \mathcal E (\hat \rho)$.

(a) 
The proof of this case is parallel to that of Lemma 64 of Ref.~\cite{Gour}.
Theorem~\ref{Smooth_monotone} (a) implies $S_\infty^\delta (\hat \rho \| \hat \sigma) \geq S_\infty^\delta (\hat \tau' \| \hat \sigma')$. We can take $\hat \tau''$ such that $S_\infty (\hat \tau''  \| \hat \sigma' ) = S_\infty^\delta ( \hat \tau' \| \hat \sigma')$, and then 
$D (\hat \tau'',  \hat \rho' ) \leq D( \hat \tau'' , \hat \tau' ) + D( \hat \tau' , \hat \rho' ) \leq  \delta+ \varepsilon$.
Therefore, $ S_\infty^\delta (\hat \tau' \| \hat \sigma') = S_\infty (\hat  \tau'' \| \hat \sigma' )  \geq S_\infty^{\varepsilon + \delta} (\hat \rho' \| \hat \sigma ' )$.

(b) 
The following proof is based on Ref.~\cite{Mitsuhashi2020p}.
Let $\delta' := \delta^2 /6$.
From Theorem~\ref{Smooth_monotone} (b), we have $S_0^\delta (\hat \rho \| \hat \sigma ) \geq S_0^{\delta'} (\hat \tau' \| \hat \sigma' ) + \ln (\delta' )$.
Meanwhile, let $\hat \tau''$ be the optimal choice such that   $S_0^{\delta' - \varepsilon} (\hat \rho' \| \hat \sigma' ) = S_0 (\hat \tau'' \| \hat \sigma' )$ with $D( \hat \tau'', \hat \rho' ) \leq \delta' - \varepsilon$.
From $D(\hat \tau', \hat \tau'') \leq D(\hat \tau', \hat \rho') + D(\hat \rho', \hat \tau'') \leq \delta'$, $\hat \tau''$ is a candidate for maximization in $S_0^{\delta'} (\hat \tau' \| \hat \sigma' )$.
Thus, $S_0^{\delta'} (\hat \tau' \| \hat \sigma' ) \geq S_0^{\delta' - \varepsilon} (\hat \rho' \| \hat \sigma' )$.
$\Box$
\end{proof}

We next consider the approximate version of Theorem~\ref{thm:asymp0_q} (b).

\begin{theorem}[Sufficient condition for approximate state conversion]
For $\varepsilon \geq 0$,  $(\hat \rho', \hat \sigma')  \prec^\varepsilon  (\hat \rho , \hat \sigma)$ holds if (but not only if)
\begin{equation}
S_\infty^{\varepsilon /2} (\hat \rho' \| \hat \sigma') \leq S_0^{\varepsilon /2} (\hat \rho \| \hat \sigma).
\end{equation}
\label{thm:asymp1_q_b}
\end{theorem}

\begin{proof}
The proof is parallel to that of Lemma 63 of Ref.~\cite{Gour}.
We can take $\hat \tau'$, $\hat \tau$ such that 
$S_\infty^{\varepsilon /2} (\hat \rho' \| \hat \sigma' ) = S_\infty( \hat \tau' \| \hat \sigma')$, $S_0^{\varepsilon /2} (\hat \rho \| \hat \sigma) = S_0 (\hat \tau  \| \hat \sigma )$, and 
$D( \hat \tau' , \hat \rho' ) \leq \varepsilon /2$, $D( \hat \tau , \hat \rho )   \leq \varepsilon /2$. From Theorem~\ref{thm:asymp0_q} (b), we have $(\hat \tau', \hat \sigma') \prec (\hat \tau, \hat \sigma)$, that is, there exists $\mathcal E$ such that $\hat \tau'= \mathcal E( \hat \tau)$C$\hat \sigma' = \mathcal E (\hat \sigma)$.  By defining $\hat \tau'' := \mathcal E(\hat \rho)$, we have
\begin{equation}
D( \hat \tau'' , \hat \rho' ) \leq D(  \hat \tau'' , \hat \tau' ) + D( \hat \tau' , \hat \rho' ) \leq  D( \hat \rho , \hat \tau ) + D( \hat \tau' , \hat \rho' ) \leq  \varepsilon / 2 + \varepsilon / 2 = \varepsilon,
\end{equation}
which implies $(\hat \rho', \hat \sigma')  \prec^\varepsilon  (\hat \rho , \hat \sigma)$. $\Box$
\end{proof}

\section{Information spectrum and asymptotic majorization}
\label{sec:asymptotic}

We next consider information spectrum and the asymptotic limit of majorization.
The concept of information spectrum has been introduced by Han and Verd\'u~\cite{Han1993,Han2003} and generalized to the quantum regime by Nagaoka and Hayashi~\cite{Nagaoka2007}.
It has been rewritten in terms of the smooth entropies and divergences by Datta and Renner~\cite{Datta_Renner,Datta2009}, on which the argument of this section is based.

To take the asymptotic limit of d-majorization, we consider a sequence of quantum states,   $\widehat{P}:= \{ \hat \rho_n \}_{n \in \mathbb N}$, where $\hat \rho_n \in \mathcal S (\mathcal H^{\otimes n})$.
This can be an arbitrary sequence, and is not restricted to an i.i.d. sequence.
We note that, in the  case of i.i.d., we can write $\hat \rho_n := \hat \rho^{\otimes n}$ with $\hat \rho \in \mathcal S ( \mathcal  H)$.
First, we define the asymptotic limit of the von Neumann entropy rate and the KL divergence rate:

\begin{definition}
Let $\widehat{P} := \{ \hat \rho_n \}$, $\widehat{\Sigma} := \{ \hat \sigma_n \}$ be sequences of quantum states.  The von Neumann entropy rate is defined as
\begin{equation}
S_1 (\widehat{P}) := \lim_{n \to \infty} \frac{1}{n} S_1 (\hat \rho_n ),
\label{entropy_rate_q}
\end{equation}
and the KL divergence rate is defined as
\begin{equation}
S_1(\widehat{P} \|\widehat{\Sigma} ) := \lim_{n \to \infty} \frac{1}{n} S_1(\hat \rho_n \| \hat \sigma_n).
\label{divergence_rate_q}
\end{equation}
We note that these limits do not necessarily exist.
\end{definition}

We next define asymptotic state conversion for sequences of quantum states  in terms of d-majorization.

\begin{definition}[Asymptotic quantum d-majorization]
Let $\widehat{P} := \{ \hat \rho_n \}$, $\widehat{\Sigma} := \{ \hat \sigma_n \}$, $\widehat{P}' := \{ \hat \rho'_n\}$, $\widehat{\Sigma}' := \{ \hat \sigma'_n \}$ be sequences of quantum states.
Then, $(\widehat{P}, \widehat{\Sigma})$ asymptotically d-majorizes $(\widehat{P}', \widehat{\Sigma}')$, written as $(\widehat{P}', \widehat{\Sigma}') \prec^{\rm a} (\widehat{P}, \widehat{\Sigma})$, if there exists a sequence of CPTP maps $\{ \mathcal E_n \}$ such that 
\begin{equation}
\lim_{n \to \infty} D( \mathcal E_n ( \hat \rho_n ),   \hat \rho'_n ) = 0, \ \ \mathcal E_n (\hat \sigma_n) = \hat \sigma'_n.
\end{equation}
Or equivalently, for any $\varepsilon >0$, there exists $N_\varepsilon \in \mathbb N$ such that $(\hat \rho'_n, \hat \sigma'_n) \prec^\varepsilon (\hat \rho_n, \hat \sigma_n)$ holds for all $n \geq N_\varepsilon$.
\label{quantum_d_majorization_asymp}
\end{definition}

We emphasize that in the above definition the state conversion of $\hat \sigma_n$ should be exact.  
In terms of thermodynamics, this means that Gibbs-preserving maps should preserve Gibbs states exactly (see also Section~\ref{sec:work_asymptotic}).

We now introduce the upper and lower  spectral divergence rates for sequences of quantum states, which are respectively given by the asymptotic limit of $\infty$- and $0$-divergences~\cite{Datta2009}.

\begin{definition}[Quantum spectral divergence rates] 
The upper and lower spectral divergence rates are respectively defined as
\begin{equation}
\overline{S}(\widehat{P} \| \widehat  \Sigma ) := \lim_{\varepsilon \to +0}\limsup_{n \to \infty} \frac{1}{n} S_\infty^\varepsilon (\hat \rho_n \| \hat \sigma_n),
\end{equation}
\begin{equation}
\underline{S}(\widehat{P} \| \widehat{\Sigma} ) := \lim_{\varepsilon \to +0}\liminf_{n \to \infty} \frac{1}{n} S_0^\varepsilon (\hat \rho_n \| \hat \sigma_n).
\end{equation}
\label{def:spectral_divergence}
\end{definition}

These quantities are also called information spectrum.
We note that $\lim_{\varepsilon \to +0}$ always exists, because $S_0^\varepsilon (\hat \rho \| \hat \sigma)$ and $S_\infty^\varepsilon (\hat \rho \| \hat \sigma)$ monotonically change in $\varepsilon$.
Also, the order of the limits of $n$ and $\varepsilon$ is crucial in the above definition.
In fact, the convergence in $n$ is not uniform on $\varepsilon > 0$ in many cases.

These quantities  have been originally introduced in Ref.~\cite{Nagaoka2007} in the form
\begin{equation}
\overline{S}(\widehat{P} \| \widehat  \Sigma ) = \inf \left\{ a \ : \ \limsup_{n \to \infty} {\rm tr} \left[  \hat P \left\{ \hat \rho_n - e^{na} \hat \sigma_n \geq 0 \right\} \hat \rho_n  \right] = 0 \right\},
\label{spectral_Nagaoka1}
\end{equation}
\begin{equation}
\underline{S}(\widehat{P} \| \widehat  \Sigma ) = \sup \left\{ a \ : \ \liminf_{n \to \infty} {\rm tr} \left[ \hat P \left\{ \hat \rho_n - e^{na} \hat \sigma_n \geq 0 \right\} \hat \rho_n \right]= 1 \right\},
\label{spectral_Nagaoka2}
\end{equation}
where $\hat P \left\{ \hat X \geq 0 \right\}$ denotes 
the projector onto the eigenspace of $\hat X$ with non-negative eigenvalues.
The equivalence between the above expression and Definition~\ref{def:spectral_divergence} has been proved in Theorem 2 and Theorem 3 of  Ref.~\cite{Datta2009}.

We note that (Proposition of \cite{Bowen_Datta}) 
\begin{equation}
\underline{S}(\widehat{P} \| \widehat{\Sigma} ) \leq \overline{S}(\widehat{P} \| \widehat  \Sigma ).
\end{equation}
In fact,  from Proposition~\ref{Faist_prop} and inequality~(\ref{eta_SH}), we have for $0 < \varepsilon < 1/3$
\begin{equation}
\frac{1}{n} S_0^\varepsilon ( \hat \rho_n \| \hat \sigma_n ) \leq \frac{1}{n} S_\infty^\varepsilon  ( \hat \rho_n \| \hat \sigma_n ) + \frac{1}{n} \ln \frac{2 }{1 - \varepsilon }.
\end{equation}
Then, take  $\liminf_{n \to \infty}$ and $\limsup_{n \to \infty}$, and  next $\varepsilon \to +0$.

\

We next show that the spectral divergence rates satisfy the monotonicity (and thus are monotones) under asymptotic state conversion.
This is regarded as a necessary condition for state conversion as the  asymptotic limit of  Theorem~\ref{thm:asymp1_q_a}.
We note that the following theorem is slightly general than the monotonicity proved in Proposition 4 of Ref.~\cite{Bowen_Datta} where state conversion is assumed to be exact.

\begin{theorem}[Monotonicity of the spectral divergence rates]
If  $(\widehat{P}', \widehat{\Sigma}') \prec^{\rm a} (\widehat{P}, \widehat{\Sigma})$, then
\begin{equation}
\overline{S}(\widehat{P}' \| \widehat{\Sigma}') \leq \overline{S}
(\widehat{P} \| \widehat{\Sigma}), \ \ \ 
\underline{S}(\widehat{P}' \| \widehat{\Sigma}') \leq \underline{S} (\widehat{P} \| \widehat{\Sigma}).
\end{equation} 
\end{theorem}

\begin{proof}

From the assumption, for any $\varepsilon >0$, there exists $N_\varepsilon \in \mathbb N$  such that  $(\hat \rho'_n, \hat \sigma'_n) \prec^\varepsilon (\hat \rho_n, \hat \sigma_n)$ for all $n \geq N_\varepsilon$.

We first prove the former inequality.
  Theorem~\ref{thm:asymp1_q_a} (a) implies that for any $\delta > 0$, $S_\infty^{\varepsilon + \delta} (\hat \rho'_n \| \hat \sigma'_n) \leq S_\infty^\delta (\hat \rho_n \| \hat \sigma_n)$ holds for all $n \geq N_\varepsilon$.  We first take  $\limsup_{n \to \infty}$  by dividing the above inequality by $n$ and by fixing $\varepsilon$ and $\delta$, and then take the limit $\varepsilon \to +0$ and $\delta \to + 0$.
Then we have $\overline{S}(\widehat{P}' \| \widehat{\Sigma}' ) \leq \overline{S}(\widehat{P} \| \widehat{\Sigma} )$.

For the latter inequality, the proof goes in a similar way.
By applying  Theorem~\ref{thm:asymp1_q_a} (b),  we have for any  $\delta$ satisfying $\delta^2/6 > \varepsilon$,
$S_0^\delta (\hat \rho_n \| \hat \sigma_n ) \geq S_0^{\delta^2/6 - \varepsilon} (\hat \rho_n' \| \hat \sigma_n' ) + \ln ( \delta^2/6 ) $ for all $n \geq N_\varepsilon$.
We first take $\liminf_{n \to \infty}$ by dividing the above inequality by $n$  and by fixing $\varepsilon$ and $\delta$, and then take the limit $\varepsilon \to +0$ and $\delta \to + 0$ by keeping  $\delta^2/6 > \varepsilon$.
Then we have $\underline{S}(\widehat{P}' \| \widehat{\Sigma}') \leq \underline{S} (\widehat{P} \| \widehat{\Sigma})$.
$\Box$
\end{proof}

We next consider a sufficient condition for asymptotic state conversion, which is obtained as the limit of Theorem~\ref{thm:asymp1_q_b}.

\begin{theorem}[Sufficient condition for asymptotic state conversion]
$(\widehat{P}', \widehat{\Sigma}') \prec^{\rm a} (\widehat{P}, \widehat{\Sigma})$ holds, if (but not only if)
\begin{equation}
\overline{S}(\widehat{P}' \| \widehat{\Sigma}') < \underline{S}(\widehat{P} \| \widehat{\Sigma}).
\label{spectral_sufficient}
\end{equation}
\label{thm:spectral_sufficient}
\end{theorem}

\begin{proof}
Suppose that  $\overline{S}(\widehat{P}' \| \widehat{\Sigma}') < \underline{S}(\widehat{P} \| \widehat{\Sigma})$.
For any sufficiently small $\varepsilon > 0$, 
there exists $N_\varepsilon \in \mathbb N$ such that $S_\infty^{\varepsilon /2} (\hat \rho'_n \| \hat \sigma'_n ) \leq S_0^{\varepsilon /2} (\hat \rho_n \| \hat \sigma_n)$ holds for all $n \geq N_\varepsilon$.  Thus, from Theorem~\ref{thm:asymp1_q_b}, we have$(\hat \rho'_n, \hat \sigma'_n) \prec^\varepsilon (\hat \rho_n,\hat \sigma_n)$  for all $n \geq N_\varepsilon$.
This implies $(\widehat{P}', \widehat{\Sigma}') \prec^{\rm a} (\widehat{P}, \widehat{\Sigma})$.
$\Box$
\end{proof}

We note that the equality is excluded from inequality~(\ref{spectral_sufficient}).
In fact, the above proof does not work in the equality case.
Moreover, it is known that there indeed exists an example that the equality $\overline{S}(\widehat{P}' \| \widehat{\Sigma}') = \underline{S}(\widehat{P} \| \widehat{\Sigma})$ is satisfied but asymptotic state conversion from $(\widehat{P}, \widehat{\Sigma})$ to $(\widehat{P}', \widehat{\Sigma}')$ is impossible, which has been shown in Refs.~\cite{Ito2015,Kumagai2013} (but we need to slightly change the setup there).
This is a topic of the second-order asymptotics~\cite{Hayashi2008,Hayashi2009,Datta2015,Tomamichel2016}.

We now consider a special case that the upper and lower spectral divergence rates collapse to a single value $S(\widehat{P}\| \widehat{\Sigma})$.  
In such a case, we obtain a (almost) necessary and sufficient characterization for state conversion (where ``almost'' means that the equality case mentioned above is excluded). 
The characterization is given by a single scalar entropy-like function $S(\widehat{P}\| \widehat{\Sigma})$ (i.e., a (almost) complete monotone).  


\begin{theorem}
Suppose that the upper and lower spectral divergence rates coincide: $\overline{S}(\widehat{P}\| \widehat{\Sigma} ) = \underline{S}(\widehat{P} \| \widehat{\Sigma}) =:  S(\widehat{P}\| \widehat{\Sigma})$ and $\overline{S}(\widehat{P}' \| \widehat{\Sigma}')  = \underline{S}(\widehat{P}' \|  \widehat{\Sigma}') =:  S(\widehat{P}' \| \widehat{\Sigma}')$.  Then,
\begin{description}
\item[(a)] $(\widehat{P}', \widehat{\Sigma}') \prec^{\rm a} (\widehat{P}, \widehat{\Sigma})$ implies $S(\widehat{P}' \|  \widehat{\Sigma}') \leq S(\widehat{P} \|  \widehat{\Sigma})$.
\item[(b)]  $S(\widehat{P}' \|  \widehat{\Sigma}') < S(\widehat{P} \|  \widehat{\Sigma})$ implies $(\widehat{P}', \widehat{\Sigma}') \prec^{\rm a} (\widehat{P}, \widehat{\Sigma})$.
\end{description}
\label{thm:d_asym_q}
\end{theorem}

The condition that the upper and lower spectral divergence rates collapse will be examined in Section~\ref{sec_condition_information}.
We note that the above theorem has been shown in Ref.~\cite{Sagawa2019}, while some special cases have been discussed in Ref.~\cite{Matsumoto2010} for i.i.d. states and in Ref.~\cite{Jiao2017} for unital maps.

\

We next consider asymptotic (ordinary) majorization as a special case of d-majorization.  We first define the following.

\begin{definition}[Spectral entropy rate]
Let $\widehat{P}:= \{ \hat \rho_n \}_{n \in \mathbb N}$ with $\hat \rho_n \in \mathcal S (\mathcal H^{\otimes n})$ be a sequence of quantum states, and let $\widehat{ID} := \{ \hat I^{\otimes n}/ d^{n} \}_{n \in \mathbb N}$ be the sequence of the maximally mixed states.  We then define
\begin{equation}
\underline{S} (\widehat{P}) := \ln d - \overline{S} (\widehat{P}\| \widehat{ID}),
\end{equation}
\begin{equation}
\overline{S} (\widehat{P}) := \ln d - \underline{S} (\widehat{P}\| \widehat{ID}).
\end{equation}
\end{definition}

Here, we note that  ``upper'' and ``lower'' are the opposite to those of the divergence case.
We can define $S_\infty^\varepsilon (\hat \rho )$ and $S_0^\varepsilon (\hat \rho )$ correspondingly.
Under these definitions, the foregoing argument can reduce to the ordinary majorization case.
In addition, we note the following theorem without a proof:

\begin{proposition}[Lemma 3 of Ref.~\cite{Bowen_Datta}]
\begin{equation}
\underline{S}(\widehat{P} ) \leq \liminf_{n \to \infty} \frac{1}{n} S_1 (\hat \rho_n ) \leq \limsup_{n \to \infty}  \frac{1}{n} S_1(\hat \rho_n) \leq \overline{S} (\widehat{P}).
\end{equation}
Thus, if $\underline{S}(\widehat{P} ) = \overline{S} (\widehat{P}) $ is satisfied,
the von Neumann entropy rate $S_1 (\widehat{P})$ exists and
\begin{equation}
\underline{S}(\widehat{P} ) = \overline{S} (\widehat{P}) = S_1 (\widehat{P}) 
\end{equation}
holds.
\end{proposition}

On the other hand, the above proposition with the spectral entropy rates replaced by the spectral divergence rates is not valid.
In fact, there exists a counterexample even in the classical case (see Section IV.4.3. of Ref.~\cite{Sagawa2019}), for which
$\overline{S}(\widehat{P}\| \widehat{\Sigma} ) = \underline{S}(\widehat{P} \| \widehat{\Sigma}) =: S(\widehat{P} \| \widehat{\Sigma})$ holds and $S_1(\widehat{P} \| \widehat{\Sigma})$ exists, but  $S(\widehat{P} \| \widehat{\Sigma}) \neq S_1(\widehat{P} \| \widehat{\Sigma})$.

We finally note the following general property of the spectral divergence rates.
This plays an important role in  ergodicity-broken systems as discussed at the end of the next section.

\begin{proposition}[Proposition 11 of \cite{Sagawa2019}]
Suppose that $\widehat{P} := \{ \hat \rho_n \}$ is given by the mixture of $\widehat{P}^{(k)} := \{ \hat \rho_n^{(k)} \}$ with probability $r_k > 0$ ($k=1,2, \cdots, K < \infty$), i.e., $\hat \rho_n = \sum_k r_k \hat \rho_n^{(k)}$ with $r_k$ being independent of $n$.
Then, 
\begin{equation}
\underline{S}(\widehat{P} \| \widehat{\Sigma}) = \min_k \{ \underline{S} (\widehat{P}^{(k)} \| \widehat{\Sigma})  \}, \ \overline{S}(\widehat{P} \| \widehat{\Sigma}) = \max_k \{ \overline{S} (\widehat{P}^{(k)} \| \widehat{\Sigma})  \}.
\label{non_ergodic}
\end{equation}
\label{prop:non_ergodic}
\end{proposition}

\section{Quantum asymptotic equipartition property}
\label{sec_condition_information}

We next examine  under what conditions  the assumption of Theorem~\ref{thm:d_asym_q} is satisfied and a (almost) complete monotone emerges.
That is, a goal of this section is to establish a  condition that  the upper and lower spectral divergence rates collapse to a single value.

The asymptotic equipartition property (AEP) plays a central role.
In the classical case~\cite{Cover_Thomas}, the AEP characterizes  typical asymptotic (large deviation) behavior of probability distributions, and is closely related to various concepts in classical probability theory: the law of large numbers, the ergodic theorem, and the Shannon-McMillan theorem.
The AEP states that almost all events have almost the same probability, as represented by ``equipartition'' (see also Proposition~\ref{prop_c_AEP0} in Appendix~\ref{app_classical}).
 Therefore, if the AEP is satisfied,  the Lorenz curve asymptotically consists of an  almost  straight line (and a horizontal line), implying the collapse of the upper and lower spectral entropy rates.
Furthermore, we  can consider the relative version of the AEP (see also Proposition~\ref{prop_c_AEP} in Appendix~\ref{app_classical}), which is equivalent to the collapse of the upper and lower spectral divergence rates to a single value given by the KL divergence rate.
The relative AEP is also equivalent to the Stein's lemma for hypothesis testing.
See Appendix~\ref{app_classical} for details of the classical case.

In the following, we consider the quantum version of the (relative) AEP, which is equivalent to the quantum Stein's lemma for quantum hypothesis testing (see also Appendix~\ref{appx:hypothesis_testing}).
Then, the goal of this section is rephrased as follows: Clarify under what conditions the quantum relative AEP holds.

We first note the  i.i.d. case, where it is well-known that the relative AEP holds and the upper and lower spectral divergence rates collapse.  

\begin{theorem}[Theorem 2 of Ref.~\cite{Nagaoka2007}]
If $\widehat{P}:= \{ \hat \rho^{\otimes n} \}$ and $\widehat{\Sigma} := \{ \hat \sigma^{\otimes n} \}$ are both i.i.d.,
\begin{equation}
\overline{S}(\widehat{P}\| \widehat{\Sigma} ) = \underline{S}(\widehat{P} \| \widehat{\Sigma}) =  S_1(\widehat{P}\| \widehat{\Sigma} ).
\end{equation}
\label{thm:iid_divergence}
\end{theorem}

To go beyond the i.i.d. situation, \textit{ergodicity} plays a significant role, with which our answer to the above-mentioned question is as follows: If $\widehat{P}$ is translation invariant  and ergodic and $\widehat{\Sigma}$ is the Gibbs state of a local and translation-invariant Hamiltonian, then $\overline{S}(\widehat{P}\| \widehat{\Sigma} )$ and $\underline{S}(\widehat{P} \| \widehat{\Sigma})$ coincide. 
In the following, we clarify the definitions of these concepts.

We consider a quantum many-body system with many local spins on a lattice.
Concretely, consider a quantum spin system on an infinite lattice $\mathbb Z^d$ in any spatial dimension ($d=1,2,3, \cdots$), where a finite-dimensional Hilbert space $\mathcal H_i$ is attached to each lattice site $i \in \mathbb Z^d$.
Let $\hat \rho$ represent a state on this infinite lattice system, and $\hat \rho_\Lambda$ be its reduced density operator on a bounded region $\Lambda \subset \mathbb Z^d$.

In this book,  we do not go into the mathematical details of infinite spin systems  with $C^\ast$-algebras.
In this paragraph, however, let us just briefly remark on the rigorous formulation, which the readers can skip.
We consider the $C^\ast$-algebra written as  $\mathcal A := \overline{\otimes_{i \in \mathbb Z^d}\mathcal A_i}$ with $\mathcal A_i := \mathcal L (\mathcal H_i)$.
Here, $\mathcal A$ is obtained  as the $C^\ast$-inductive limit of the algebra of local operators $\mathcal A_{\rm loc} := \cup_\Lambda \otimes_{i \in \Lambda} \mathcal A_i$ with $\Lambda \subset \mathbb Z^d$ being bounded.
Then, $\hat \rho$ is the density operator (in the infinite-dimensional trace-class operator space) corresponding to a normal state $\rho : \mathcal A \to \mathbb C$.
The expectation value of an operator $\hat X$ is given by $\rho (\hat X) = {\rm tr} [\hat \rho \hat X ]$.
See Ref.~\cite{Sagawa2019} for details of the rigorous formulation, and  Refs.~\cite{Bratteli1979,Bratteli1981,Ruelle1999} for the general mathematical theory of infinite spin systems.

We next introduce the shift superoperator $\mathcal{T}_i$ with $i \in \mathbb Z^d$, which maps any operator on site $j \in \mathbb Z^d$ to the same operator on site $j+i$.   
Then, we  define  translation invariance and ergodicity.

\begin{definition}[Translation-invariant states]
A state $\hat \rho$ is translation invariant, if it is invariant under the action of the shift operator:  For any $i \in \mathbb Z^d$ and any $\hat X \in \mathcal A$, 
\begin{equation}
{\rm tr}[\hat \rho \hat X] = {\rm tr}[\hat \rho \mathcal{T}_i (\hat X)].
\end{equation}
\end{definition}

\begin{definition}[Ergodic states]
Let $\hat X \in \mathcal A$ be a self-adjoint operator and define
\begin{equation}
\hat X_\Lambda := \frac{1}{| \Lambda |} \sum_{i \in \Lambda} \mathcal{T}_i (\hat X),
\label{ergodic_macro_operator}
\end{equation} 
where $\Lambda \subset \mathbb Z^d$ is a bounded region and $| \Lambda |$ is the number of its elements.
Then, a translation-invariant state $\hat \rho$ is ergodic, if the variance of any operator $\hat X_\Lambda$ of the form (\ref{ergodic_macro_operator}) vanishes in the limit of $\Lambda \to \mathbb Z^d$.  Here, we take $\Lambda$ as a hypercube to specify the meaning of this limit.
\label{def:q_ergodic}
\end{definition}

This definition implies that any  macroscopic (i.e., extensive) observable of the form (\ref{ergodic_macro_operator})   (e.g., the total magnetization) 
does not exhibit macroscopic fluctuations, that is, any macroscopic observable has a definite value in the thermodynamic limit.  
We note that this is ``spatial'' ergodicity in the sense that the spatial average equals the ensemble average, while in standard statistical mechanics temporal ergodicity often plays significant roles.
 The above definition also implies that the state is in a ``pure thermodynamic phase'' without any phase mixture, which is formalized as follows.

\begin{proposition}
A state $\hat \rho$ is translation invariant and ergodic, if and only if it is an extreme point of the set of translation-invariant states.
\end{proposition}

In literature, the above expression of ergodicity is often adopted as a definition of it.
See, for example, Theorem 6.3.3, Proposition 6.3.5, and Lemma 6.5.1 of Ref.~\cite{Ruelle1999} for other equivalent expressions.

We next formulate the Hamiltonian of the system.

\begin{definition}[Local and translation-invariant Hamiltonians]
A Hamiltonian of the system is local, if its truncation on any bounded region $\Lambda \subset \mathbb Z^d$, written as $\hat H_\Lambda$, is given of the form
\begin{equation}
\hat H_\Lambda = \sum_{i \in \Lambda} \hat h_i, 
\end{equation} 
where  $\hat h_i$ is a Hermitian operator acting on a bounded region around $i \in \mathbb Z^d$.
Moreover, a local Hamiltonian is  translation invariant, if $\hat h_i = \mathcal{T}_i (\hat h_0 )$ holds.
\end{definition}

The corresponding truncated Gibbs state is given by $\hat \rho_\Lambda^{\rm G} := \exp (\beta (F_\Lambda  - \hat H_\Lambda ) )$ with $F_\Lambda := - \beta^{-1} {\rm tr}[e^{-\beta \hat H_\Lambda}]$.
Note that if the global Gibbs state (strictly speaking, the Kubo-Martin-Schwinger (KMS) state) is unique, it is obtained by the limit of $\hat \rho_\Lambda^{\rm G}$ in the weak-$\ast$ (or ultraweak) topology  (cf. Proposition 6.2.15 of Ref.~\cite{Bratteli1981}).

It is known that the KMS state of a local and translation-invariant Hamiltonian is ergodic, for any $\beta > 0$ in one dimension~\cite{Araki1969}, and for sufficiently small $\beta$ (i.e., for sufficiently high temperature) in higher dimensions (see, e.g., Ref.~\cite{Tasaki2018} and references therein).

We note that i.i.d. is ergodic.
In addition, an i.i.d. state is regarded as the Gibbs state of a non-interacting Hamiltonian.  On the other hand, the foregoing general formulation includes non-i.i.d. cases, in the sense that  general ergodic states can be non-i.i.d. and Hamiltonians can be interacting.

We consider our setup of the asymptotic majorization in the previous section along with the foregoing  infinite-lattice formulation.
Let $\Lambda_n \subset \mathbb Z^d$ be a bounded region with $| \Lambda_n | = n$.
We consider sequences $\widehat{P} = \{ \hat \rho_n \}_{n=1}^\infty$ and $\widehat{\Sigma} = \{ \hat \sigma_n \}_{n=1}^\infty$, where $\hat \rho_n$ is the reduced state  on $\Lambda_n$ of an ergodic state $\hat \rho$, and $\hat \sigma_n := \rho_{\Lambda_n}^{\rm G}$ is the truncated Gibbs state of a local and translation-invariant Hamiltonian. 

Before going to the relative quantum AEP with $\widehat{P}$ and $\widehat{\Sigma}$, we consider the quantum version of the ordinary AEP of $\widehat{P}$, which is referred to as the  quantum  Shannon-McMillan theorem~\cite{Bjelakovic2002,Bjelakovic2003,Ogata2013}.

\begin{theorem}[Quantum Shannon-McMillan Theorem]
Suppose that $\widehat{P}$ is ergodic.
Then, for any $\varepsilon > 0$, there exists a sequence of projectors $\hat \Pi_{\widehat{P},n}^\varepsilon$ (called typical projectors) that satisfy, for sufficiently large $n$,
\begin{equation}
e^{-n (s+ \varepsilon )} \hat \Pi_{\widehat{P},n}^\varepsilon \leq \hat \Pi_{\widehat{P},n}^\varepsilon \hat \rho_n  \hat \Pi_{\widehat{P},n}^\varepsilon \leq e^{-n(s-\varepsilon)} \hat \Pi_{\widehat{P},n}^\varepsilon,
\label{rho_typical1}
\end{equation}
\begin{equation}
e^{n(s-\varepsilon)} < {\rm tr}[\hat \Pi_{\widehat{P},n}^\varepsilon] < e^{n(s+\varepsilon)},
\label{rho_typical2}
\end{equation}
\begin{equation}
\lim_{n \to \infty} {\rm tr}[\hat \Pi_{\widehat{P},n}^\varepsilon \hat \rho_n] = 1,
\label{rho1_1}
\end{equation}
where $s:= S_1 (\widehat{P})$.
\label{ergodic_prop}
\end{theorem}

We note that, comparing the above quantum theorem with the classical case, Eq.~(\ref{rho1_1}) is a stronger statement than Proposition~\ref{prop_c_AEP0} (b), but the former is indeed provable (see Theorem 1.4 of Ref.~\cite{Ogata2013}).
Then, the following is a consequence of the quantum Shannon-McMillan theorem and is itself regarded as a representation of the quantum AEP.

\begin{theorem}[Quantum AEP, Theorem 7 of ~\cite{Schoenmakers2007}] If $\widehat{P}$ is translation invariant  and ergodic,
\begin{equation}
\underline{S} (\widehat{P}) = \overline{S} (\widehat{P}) = S_1(\widehat{P}).
\end{equation}
\end{theorem}

We now consider the relative version of the above theorem, which is a goal of this section.
It is a quantum version of the relative Shannon-McMillan theorem or the relative AEP.
We note that it is further equivalent to the quantum Stein's lemma for non-i.i.d. situations (see also Appendix \ref{appx:hypothesis_testing}).

\begin{theorem}[Relative quantum AEP, Theorem 3 of \cite{Sagawa2019}]
If $\widehat{P}$ is translation invariant and ergodic, and $\widehat{\Sigma}$ is the Gibbs state of a local and translation invariant Hamiltonian, then
\begin{equation}
\overline{S}(\widehat{P}\| \widehat{\Sigma} ) = \underline{S}(\widehat{P} \| \widehat{\Sigma}) =  S_1(\widehat{P}\| \widehat{\Sigma} ).
\label{eq_relative_quantum_AEP}
\end{equation}
\label{thm:main_ergodic}
\end{theorem}

Combining Theorem~\ref{thm:main_ergodic} with Theorem~\ref{thm:d_asym_q}, we find that the KL divergence rate provides a (almost) necessary and sufficient condition for asymptotic state conversion:

\begin{corollary}
Suppose that $\widehat{P}$ is translation invariant and ergodic, and $\widehat{\Sigma}$ is the Gibbs state of a local and translation-invariant Hamiltonian. Then:
\begin{description}
\item[(a)] $(\widehat{P}', \widehat{\Sigma}') \prec^{\rm a} (\widehat{P}, \widehat{\Sigma})$ implies $S_1(\widehat{P}' \|  \widehat{\Sigma}') \leq S_1 (\widehat{P} \|  \widehat{\Sigma})$.
\item[(b)]  $S_1(\widehat{P}' \|  \widehat{\Sigma}') < S_1(\widehat{P} \|  \widehat{\Sigma})$ implies $(\widehat{P}', \widehat{\Sigma}') \prec^{\rm a} (\widehat{P}, \widehat{\Sigma})$.
\end{description}
\label{cor:d_asym_q}
\end{corollary}

This corollary implies that under the above conditions, the KL divergence rate serves as  a (almost) complete monotone in the asymptotic limit.
This enables us to introduce a single thermodynamic potential that fully characterizes macroscopic state conversion, as will be discussed in Section~\ref{sec:work_asymptotic}.

We here remark on the history of the relative quantum AEP (or equivalently, the quantum Stein's lemma).
Hiai and Petz~\cite{Hiai_Petz}  provided a half of the proof  for a (completely) ergodic quantum state with respect to an i.i.d. state.
Ogawa and Nagaoka~\cite{Ogawa2000} completed the proof  for the case that both states are i.i.d..
Bjelakovi\'c and Siegmund-Schultze~\cite{Bjelakovic2004} proved a  more general form for an ergodic (not necessarily completely ergodic)  state with respect to an i.i.d. state (corresponding to the Gibbs state of a  non-interacting Hamiltonian).
Faist, Sagawa, \textit{et al.}~\cite{Sagawa2019,Faist2019} further generalized it for an ergodic state with respect to the Gibbs state of an interacting, local and translation-invariant Hamiltonian  (Theorem~\ref{thm:main_ergodic} above).

\

If the ergodicity is broken, the upper and lower divergence rates do not coincide in general.
In particular, we consider a mixture of different ergodic states, where the entire state is not ergodic.
With the notation of Proposition~\ref{prop:non_ergodic} in the previous section,
let $\widehat{P}^{(k)}$ be a translation-invariant ergodic state and $\widehat{\Sigma}$ be the Gibbs state of a local and translation-invariant Hamiltonian.
Then, Eq.~(\ref{non_ergodic}) reduces to
\begin{equation}
\underline{S}(\widehat{P} \| \widehat{\Sigma}) = \min_k \{ S_1 (\widehat{P}^{(k)} \| \widehat{\Sigma})  \}, \ \overline{S}(\widehat{P} \| \widehat{\Sigma}) = \max_k \{ S_1 (\widehat{P}^{(k)} \| \widehat{\Sigma})  \}.
\label{non_ergodic1}
\end{equation}
Thus, except for the special case that all of $S_1 (\widehat{P}^{(k)} \| \widehat{\Sigma}) $'s happen to take the same value,  state convertibility can no longer be characterized by a single complete monotone (in particular  by the KL divergence rate).

We note that, in the above setup,
\begin{equation}
S_1 (\widehat{P} \| \widehat{\Sigma}) = \sum_k r_k S_1 (\widehat{P}^{(k)} \| \widehat{\Sigma}).
\label{mix_KL}
\end{equation}
In fact,  $-{\rm tr}[\hat \rho^{(k)}_n \ln \hat \sigma_n ]$ is  additive with respect to $k$.
Also because $\sum_k r_k S_1 (\hat \rho_n^{(k)}) \leq S_1 (\hat \rho_n) \leq \sum_k r_k S_1 (\hat \rho_n^{(k)}) + S_1 (r)$ [Eqs.~(\ref{von_Neumann_concave}) and (\ref{concave_upper_bound})], we obtain $S_1 (\widehat{P}) = \sum_k r_k S_1 (\widehat{P}^{(k)})$.
From Eq.~(\ref{non_ergodic1}) and Eq.~(\ref{mix_KL}), we have 
\begin{equation}
\underline{S}(\widehat{P} \| \widehat{\Sigma}) \leq S_1 (\widehat{P} \| \widehat{\Sigma})  \leq \overline{S}(\widehat{P} \| \widehat{\Sigma}).
\label{S_inequality}
\end{equation}
However, as mentioned in Section~\ref{sec:asymptotic}, inequality~(\ref{S_inequality}) is not necessarily true in general, even in the classical case~\cite{Sagawa2019}.

\section{More on the smoothing}
\label{sec:misc}

As a side remark, we discuss some other ways of smoothing.
First, the trace distance can be replaced by another distance for the definition of neighborhood $B^\varepsilon (\hat \rho)$.   A commonly used one is the purified distance  defined as~\cite{Faist2018,Dupuis2012}
\begin{equation}
P(\hat \rho, \hat \sigma) := \sqrt{1 - F(\hat \rho, \hat \sigma)^2},
\label{def_purified}
\end{equation}
where the fidelity is given by
\begin{equation}
F (\hat \rho, \hat \sigma) := {\rm tr} [ ( \hat \rho^{1/2} \hat \sigma \hat \rho^{1/2} )^{1/2} ] = \| \hat \rho^{1/2} \hat \sigma^{1/2} \|_1.
\label{quantum_fidelity}
\end{equation}
Here, $F (\hat \rho, \hat \sigma) = F ( \hat \sigma, \hat \rho)$ and thus $P (\hat \rho, \hat \sigma) = P ( \hat \sigma, \hat \rho)$.
We note that $F (\hat \rho, \hat \sigma)$ is sometimes called the square root fidelity and $F (\hat \rho, \hat \sigma)^2$ is called the fidelity.
The fidelity satisfies the monotonicity: it does not decrease under CPTP map $\mathcal E$:
\begin{equation}
F(\hat \rho, \hat \sigma) \leq F(\mathcal E ( \hat \rho ) , \mathcal E (\hat \sigma ))
\label{fidelity_monotonicity}
\end{equation}
 (see Theorem 9.6 of Ref.~\cite{Nielsen} for the proof).
We note that the monotonicity is also satisfied under positive and TP maps (Corollary A.5 of Ref.~\cite{Mosonyi2015}).
Correspondingly,  the purified distance satisfies the monotonicity $P(\hat \rho, \hat \sigma) \geq P(\mathcal E ( \hat \rho ) , \mathcal E (\hat \sigma))$.
The purified distance and the trace distance are related with each other as
\begin{equation}
D(\hat \rho, \hat \sigma) \leq P(\hat \rho, \hat \sigma) \leq \sqrt{2D(\hat \rho, \hat \sigma)},
\label{ineq_purified}
\end{equation}
and thus they give asymptotically the same smoothing of the $0$- and $\infty$-divergences.

\

Meanwhile, it is sometimes convenient to define the lower spectral divergence rate in an alternative manner based on a R\'enyi divergence other than $S_0 (\hat \rho \| \hat \sigma)$~\cite{Dupuis2012}.
It is the sandwiched R\'enyi $1/2$-divergence defined by the Fidelity  as
\begin{equation}
S_{1/2}  (\hat \rho \| \hat \sigma ) := - 2 \ln F(\hat \rho, \hat \sigma).
\label{quantum_fidelity_Renyi}
\end{equation}
See Appendix~\ref{sec:general_monotonicity} for details of  the sandwiched R\'enyi divergences.
We then define the smooth $1/2$-divergence as
\begin{equation}
S_{1/2}^\varepsilon(\hat \rho \| \hat \sigma ) := \max_{\hat \tau  \in B^\varepsilon (\hat \rho )} S_{1/2} ( \hat \tau \|  \hat \sigma).
\end{equation}
The following proposition relates $S_{1/2}^\varepsilon (\hat \rho \| \hat \sigma)$ to $S_0^\varepsilon (\hat \rho \| \hat \sigma)$.

\begin{proposition}[Proposition 2 of \cite{Sagawa2019}]
\begin{equation}
S_{1/2}^\varepsilon (\hat \rho \| \hat \sigma) - 6 \ln \frac{3}{\varepsilon} \leq S_0^{2\varepsilon} (\hat \rho \| \hat \sigma) \leq S_{1/2}^{2\varepsilon} (\hat \rho \| \hat \sigma) - \ln (1 - 2 \varepsilon ).
\end{equation}
\end{proposition}

We note that  the above inequality is slightly different from that of Ref.~\cite{Sagawa2019}, because the way of smoothing is slightly different (see Lemma~\ref{subnormalized_prop} later).
From this proposition,  we see that $S_{1/2} (\hat \rho \| \hat \sigma)$ has the same information as $S_0 (\hat \rho \| \hat \sigma)$ in the asymptotic limit.
In this sense, $S_{1/2}  (\hat \rho \| \hat \sigma ) $  is sometimes  called  the min divergence.

\

In the foregoing formulation of smooth divergences, we only adopted  normalized states for the definition of neighborhood $B^\varepsilon (\hat \rho)$. 
On the other hand, we can also adopt subnormalized states for smoothing as in Refs.~\cite{Datta_Renner,Datta2009}.
The corresponding neighborhood is given by
\begin{equation}
\bar B^\varepsilon (\hat \rho ) := \{ \hat \tau  : \ \bar D( \hat \tau , \hat \rho ) \leq \varepsilon,  {\rm tr}[\hat \tau] \leq 1, \hat \tau \geq 0 \},
\label{subnormalized_B}
\end{equation}
where the trace distance is replaced by a generalized one:
\begin{equation}
\bar D( \hat \tau , \hat \rho ) := \frac{1}{2} \| \hat \tau - \hat \rho \|_1 + \frac{1}{2} | {\rm tr}[\hat \tau] - {\rm tr}[\hat \rho ] |.
\end{equation}

Suppose that  $\hat \rho$ is normalized and $\hat \tau$ is subnormalized. Then, $\| \hat \tau / {\rm tr}[\hat \tau] - \hat \rho \|_1 \leq \| \hat \tau / {\rm tr}[\hat \tau] - \hat \tau \|_1 + \| \hat \tau - \hat \rho \|_1 =  | 1- {\rm tr}[\hat \tau] | + \| \hat \tau - \hat \rho \|_1$, and thus $D(\hat \tau / {\rm tr}[\hat \tau], \hat \rho) \leq \bar D(\hat \tau, \hat \rho)$.
Also, by noting that $\| \hat \rho \|_1 - \| \hat \tau \|_1 \leq \| \hat \tau - \hat \rho \|_1$, $\bar D( \hat \tau , \hat \rho ) \leq \varepsilon$ implies ${\rm tr}[\hat \tau ] \geq 1-\varepsilon$.

We define $\bar{S}^\varepsilon_\alpha (\hat \rho \| \hat  \sigma)$ with $\alpha = 0,1/2, \infty$ as the smooth $\alpha$-divergences with subnormalized neighborhood~(\ref{subnormalized_B}).
(We note that $S_\alpha (\hat \rho \| \hat  \sigma)$ can be defined for subnormalized $\hat \rho$.)
Then, the two ways of smoothing with normalized or subnormalized states are essentially equivalent~\cite{Faist2018}, that is, $\bar{S}^\varepsilon_\alpha (\hat \rho \| \hat  \sigma)$ and $S^\varepsilon_\alpha (\hat \rho \| \hat  \sigma)$ give the same asymptotic limit.

\begin{lemma}
\begin{description}
\item[(a)] $\bar{S}^\varepsilon_{0} (\hat \rho \| \hat \sigma) = {S}^\varepsilon_{0} (\hat \rho \| \hat \sigma)$.
\item[(b)] ${S}^\varepsilon_{1/2} (\hat \rho \| \hat \sigma) \leq \bar {S}^\varepsilon_{1/2} (\hat \rho \| \hat \sigma) \leq {S}^\varepsilon_{1/2} (\hat \rho \| \hat \sigma) - \ln (1-\varepsilon )$.
\item[(c)] $\bar{S}^\varepsilon_\infty (\hat \rho \| \hat \sigma) \leq  {S}^\varepsilon_\infty (\hat \rho \| \hat \sigma) \leq \bar{S}^\varepsilon_\infty (\hat \rho \| \hat \sigma) - \ln (1-\varepsilon)$.
\end{description}
\label{subnormalized_prop}
\end{lemma}

\begin{proof}
(a) $\bar{S}^\varepsilon_{0} (\hat \rho \| \hat \sigma) \geq S^\varepsilon_{0} (\hat \rho \| \hat \sigma)$ is obvious from the definition.
On the other hand, there exists  $\hat \tau \in \bar B^\varepsilon (\hat \rho )$ such that  $\bar{S}^\varepsilon_{0} (\hat \rho \| \hat \sigma) = S_{0} (\hat \tau \| \hat \sigma)$.  Here, $\hat \tau / {\rm tr}[\hat \tau]$ is normalized and has the same support as $\hat \tau$, and is a candidate for maximization of $S^\varepsilon_{0} (\hat \rho \| \hat \sigma)$.  Thus, $\bar{S}^\varepsilon_{0} (\hat \rho \| \hat \sigma) \leq S^\varepsilon_{0} (\hat \rho \| \hat \sigma)$.

(b) $\bar{S}^\varepsilon_{1/2} (\hat \rho \| \hat \sigma) \geq S^\varepsilon_{1/2} (\hat \rho \| \hat \sigma)$ is obvious from the definition.
On the other hand, there exists  $\hat \tau \in \bar B^\varepsilon (\hat \rho )$ such that  $\bar{S}^\varepsilon_{1/2} (\hat \rho \| \hat \sigma) = S_{1/2} (\hat \tau \| \hat \sigma)$.
Since $\hat \tau / {\rm tr}[\hat \tau]$ is a candidate for  maximization of $S^\varepsilon_{1/2} (\hat \rho \| \hat \sigma)$, we have
$S^\varepsilon_{1/2} (\hat \rho \| \hat \sigma) \geq  S_{1/2} (\hat \tau / {\rm tr}[\hat \tau] \| \hat \sigma) \geq S_{1/2} (\hat \tau \| \hat \sigma)+ \ln (1-\varepsilon )$.
Thus we have the right inequality.

(c) $\bar{S}^\varepsilon_\infty (\hat \rho \| \hat \sigma) \leq  {S}^\varepsilon_\infty (\hat \rho \| \hat \sigma)$ is obvious from the definition.
On the other hand, there exists  $\hat \tau \in \bar B^\varepsilon (\hat \rho )$ such that  $\bar{S}^\varepsilon_\infty (\hat \rho \| \hat \sigma) = S_\infty (\hat \tau \| \hat \sigma)$. 
Since $\hat \tau / {\rm tr}[\hat \tau ]$ is a candidate for minimization of ${S}^\varepsilon_\infty (\hat \rho \| \hat \sigma)$, we have ${S}^\varepsilon_\infty (\hat \rho \| \hat \sigma) \leq S_\infty (\hat \tau / {\rm tr}[\hat \tau ] \| \hat \sigma ) \leq S_\infty (\hat \tau  \| \hat \sigma )   - \ln (1-\varepsilon)$.
Thus we have the right inequality.
$\Box$
\end{proof}


Finally, we note that the fidelity for subnormalized states is given by
$\bar F (\hat \rho, \hat \sigma) := F(\hat \rho, \hat \sigma) + \sqrt{(1 - {\rm tr}[\hat \rho])(1- {\rm tr}[\hat \sigma])}$~\cite{Faist2018,Dupuis2012}.
The corresponding purified distance is defined in the same manner as Eq.~(\ref{def_purified}): 
$\bar P(\hat \rho, \hat \sigma) := \sqrt{1 - \bar F(\hat \rho, \hat \sigma)^2}$,
and satisfies the same inequality as (\ref{ineq_purified}):
$\bar D(\hat \rho, \hat \sigma) \leq \bar P(\hat \rho, \hat \sigma) \leq \sqrt{2\bar D(\hat \rho, \hat \sigma)}$.


\chapter{Quantum thermodynamics}
\label{chap:quantum_thermodynamics}

We consider a general setup of quantum thermodynamics
and  discuss the thermodynamic implications of quantum information theory  developed in Chapter~\ref{chap:quantum_entropy},
Chapter \ref{chap:quantum_majorization}, and Chapter~\ref{chap:approximate_asymptotic}.

As a preliminary, we discuss the relation between the free energy and the divergences in Section~\ref{sec:thermodynamic_entropy}.
In  Section~\ref{sec:thermodynamic_operations}, we formulate Gibbs-preserving maps and thermal operations, which are free operations in resource theory of thermodynamics.
In Section~\ref{sec:clock_work}, we introduce the work storage (weight) and the clock.

In Section~\ref{sec:fluctuating_work}, we consider the second law of thermodynamics at the level of ensemble average, where the work can fluctuate.
This is the setup of stochastic thermodynamics and is a quantum counterpart of Section~\ref{sec:classical_second_law}.
A characteristic of the present formulation is that we explicitly take into account the work storage and the clock.
We will also discuss the optimal protocol that saturates the average work bound.

From  Section~\ref{sec:work_bound_single} to Section~\ref{sec:trace_nonincreasing}, 
we focus on the single-shot (or one-shot) scenario without allowing work fluctuations.
This is a resource-theoretic perspective and is a quantum counterpart of Section~\ref{sec:classical_work_bound}.
In Section~\ref{sec:work_bound_single}, we derive the fundamental work bound of the single-shot scenario.
We then consider the approximate and asymptotic versions of the single-shot work bound in Section~\ref{sec:work_approximate} and Section~\ref{sec:work_asymptotic}, respectively.
In particular, as a consequence of the quantum AEP discussed in Section~\ref{sec_condition_information},
we see that a single complete thermodynamic potential (i.e., a complete monotone) emerges in the asymptotic limit and reduces to the KL divergence rate under certain conditions.
In addition, we discuss an alternative formulation of the single-shot scenario with trace-nonincreasing maps in Section~\ref{sec:trace_nonincreasing}.


\section{Nonequilibrium free energy}
\label{sec:thermodynamic_entropy}

First of all, we clarify the definitions of the free energies of nonequilibrium quantum states.
We consider a finite-dimensional quantum system  with the Hamiltonian  $\hat H$.  
Let $\beta \geq 0$.
The Gibbs state is  written as $\hat \rho^{\rm G} := e^{-\beta \hat H} / Z$, where $Z := {\rm tr}[e^{-\beta \hat H}]$ is the partition function and $F:= -\beta^{-1} \ln Z$ is the equilibrium free energy.

For $\alpha = 0,1,\infty$, we  introduce the nonequilibrium $\alpha$-free energy of state $\hat \rho$ by
\begin{equation}
F_\alpha (\hat \rho; \hat H) := \beta^{-1} S_\alpha (\hat \rho \| \hat \rho^{\rm G}) + F.
\label{alpha_free_energy}
\end{equation}
Here, $\beta$ equals the inverse temperature of the reference Gibbs state $\hat \rho^{\rm G}$.
Eq.~(\ref{alpha_free_energy}) can be rewritten as
\begin{equation}
F_\alpha(\hat \rho; \hat H) = \beta^{-1} S_\alpha (\hat \rho \| e^{-\beta \hat H}),
\label{alpha_free_energy_eH}
\end{equation}
where we used the  scaling property~(\ref{unnormalized_sigma_scaling}).
This definition is a generalization of the classical $1$-free energy defined in Eq.~(\ref{1_free_energy_classical}).
Clearly,
\begin{equation}
F_\alpha(\hat \rho; \hat H)  \geq F.
\end{equation}

The $1$-free energy can also be rewritten as
\begin{equation}
F_1(\hat \rho; \hat H) = E - \beta^{-1} S_1 (\hat \rho),
\label{quantum_1_free_energy}
\end{equation}
where $E := {\rm tr}[\hat H \hat \rho]$ is the average energy.
As mentioned above,
$F_1(\hat \rho; \hat H) \geq  F$, or equivalently
\begin{equation}
S_1 (\hat \rho) \leq  \beta (E-  F)
\label{von_Neumann_Boltzmann}
\end{equation}
holds, where the equality holds  if and only if $\hat \rho = \hat \rho^{\rm G}$.
From this, we see that, under a fixed average energy $E$, the von Neumann entropy $S_1 (\hat \rho)$ takes the maximum value if and only if $\hat \rho = \hat \rho^{\rm G}$.
This guarantees that the Gibbs state is the maximum entropy state under the average energy constraint.

The Gibbs states are \textit{free states} in resource theory of thermodynamics, which can be freely added to other states.
This is based on the property called \textit{complete passivity}~\cite{Pusz1978,Lenard1978}, which means that one cannot extract a positive amount of work from any number of copies of the state, $\hat \rho^{\otimes n}$, by any unitary operation with the initial and final Hamiltonians being the same (i.e., the operation is cyclic).
It is known that a state is completely passive if and only if it is Gibbs.

\section{Thermal operations and Gibbs-preserving maps}
\label{sec:thermodynamic_operations}

We now generally and precisely formulate the classes of dynamics  relevant to quantum thermodynamics.
They are Gibbs-preserving maps and thermal operations, which are regarded as \textit{free operations} in resource theory of thermodynamics.
Correspondingly, the nonequilibrium free energies introduced in the previous section are monotones under these free operations, because of their monotonicity properties.

We consider the system S in contact with a heat bath B.
Hereafter, we use subscripts S or B for states and Hamiltonians of them.
Both of S and B are described by finite-dimensional Hilbert spaces.
Mathematically,  we do not necessarily suppose that B is a ``large'' system, while physically it is natural to suppose so.
In the following argument, $\beta \geq 0$ is interpreted as the inverse temperature of B.

We first reproduce  the definition of the Gibbs-preserving map, which was already introduced in Eq.~(\ref{classical_GMP}) of Section~\ref{sec:classical_second_law}   for the classical case and Eq.~(\ref{quantum_GMP0}) of Section~\ref{sec:quantum_entropy} for the quantum case.

\begin{definition}[Gibbs-preserving maps]
A CPTP map $\mathcal E_{\rm S}$ on S  is a Gibbs-preserving map  at $\beta \geq 0$ with Hamiltonian $\hat H_{\rm S}$, if
 \begin{equation}
 \mathcal E_{\rm S} (e^{-\beta \hat H_{\rm S}} ) = e^{-\beta \hat H_{\rm S}},
 \end{equation}
 or equivalently, $\mathcal E (\hat \rho_{\rm S}^{\rm G} ) = \hat \rho_{\rm S}^{\rm G}$.
\label{def:GMP}
\end{definition}

 We note that a CPTP unital map is a Gibbs-preserving map (with any $\hat H$) at $\beta = 0$ and also  a Gibbs-preserving map of the trivial Hamiltonian $\hat H \propto \hat I$ (at any $\beta$). 

We next introduce another important class of thermodynamic operations, called thermal operations.

\begin{definition}[Thermal operations]
A CPTP map $\mathcal E_{\rm S}$ on S  is an exact thermal operation at $\beta \geq 0$ with  Hamiltonian $\hat H_{\rm S}$, if there exists a heat bath B with Hamiltonian $\hat H_{\rm B}$ and the corresponding Gibbs state $\hat \rho^{\rm G}_{\rm B} := e^{-\beta \hat H_{\rm B}}/ Z_{\rm B}$, and exists a unitary operator $\hat U$ acting on SB, such that
\begin{equation}
\mathcal E_{\rm S} (\hat \rho_{\rm S}) = {\rm tr}_{\rm B} \left[ \hat U \hat \rho_{\rm S} \otimes \hat \rho_{\rm B}^{\rm G} \hat U^\dagger  \right]
\end{equation}
and
\begin{equation}
[ \hat U,  \hat H_{\rm S} + \hat H_{\rm B} ]  = 0.
\label{thermal_operation_conserve}
\end{equation}
Furthermore, a CPTP map $\mathcal E_{\rm S}$ on S is a thermal operation at $\beta$ with  Hamiltonian $\hat H_{\rm S}$, if there exists a sequence of exact thermal operations $\{ \mathcal E_{{\rm S},n} \}_{n=1}^\infty$ with Hamiltonian  $\hat H_{\rm S}$ such that $\mathcal E_{{\rm S},n}$ converges to $\mathcal E_{\rm S}$ in $n \to \infty$.
\label{def:thermal_operation}
\end{definition}

In the definition of thermal operation, the condition (\ref{thermal_operation_conserve}) implies that the sum of the free Hamiltonians of S and B, represented by $\hat H_{\rm S} + \hat H_{\rm B}$, is strictly conserved under the unitary operation $\hat U$.
This is an idealized representation of the resource-theoretic viewpoint that the energy is a resource and is not free.
If  $\hat U$ is written as $\hat U = \exp (- {\rm i} \hat H_{\rm tot})$ with the total Hamiltonian
\begin{equation}
\hat H_{\rm tot} = \hat H_{\rm S} + \hat H_{\rm int} + \hat H_{\rm B},
\label{total_Hamiltonian_general}
\end{equation}
where $\hat H_{\rm int}$ is the interaction Hamiltonian, then Eq.~(\ref{thermal_operation_conserve})  can be expressed as
\begin{equation}
[ \hat H_{\rm tot}, \hat H_{\rm S} + \hat H_{\rm B} ]  = 0 \ \ \ \Leftrightarrow \ \ \ [ \hat H_{\rm int}, \hat H_{\rm S} + \hat H_{\rm B} ]  = 0.
\label{thermal_operation_conserve_H}
\end{equation}

The condition~(\ref{thermal_operation_conserve}) or (\ref{thermal_operation_conserve_H}) is exactly satisfied for special systems such as the Jaynes-Cummings model at the resonant condition, whose Hamiltonian is given by
\begin{equation}
\hat H_{\rm JC} = \frac{1}{2} \omega \hat \sigma_z + g(\hat \sigma_+ \hat a + \hat \sigma_- \hat a^\dagger ) + \omega \hat a^\dagger \hat a.
\end{equation}
Here, $\hat \sigma_z$, $\hat \sigma_{\pm}$ are the Pauli $z$ and ladder operators of a two-level atom and $\hat a$ is the annihilation operator of a photon, and the second term on the right-hand side represents the interaction.
More generally, the rotating wave approximation  in the derivation of the Lindblad equation~\cite{Breuer2002} approximately leads to  the condition (\ref{thermal_operation_conserve}) or (\ref{thermal_operation_conserve_H}) (see also Ref.~\cite{Funo2018}). 
This is particularly relevant to the long time regime, in which one can expect that thermal operations can be realized at least approximately in a broad class of dynamics.

We note that a noisy operation (Definition~\ref{def:noisy_operation}) is a special thermal operation where both of the Hamiltonians, $\hat H_{\rm S}$ and $\hat H_{\rm B}$, are trivial (i.e., proportional to the identities).
As a generalization of the fact that any noisy operation is unital (Proposition~\ref{noisy_unital}), we have the following.

\begin{lemma}
If $\mathcal E_{\rm S}$ is a  thermal operation, then $\mathcal E_{\rm S}$ is also a Gibbs-preserving map.
\label{lemma:TO_GMP}
\end{lemma}

\begin{proof}
If $\mathcal E_{\rm S}$ is an exact thermal operation,
$
\mathcal E_{\rm S} (e^{-\beta \hat H_{\rm S}})
= Z_{\rm B}^{-1} {\rm tr}_{\rm B} [ \hat U e^{-\beta (\hat H_{\rm S} + \hat H_{\rm B})} \hat U^\dagger ] 
= Z_{\rm B}^{-1} {\rm tr}_{\rm B} [  e^{-\beta \hat U (\hat H_{\rm S} + \hat H_{\rm B})\hat U^\dagger}  ]  
= Z_{\rm B}^{-1} {\rm tr}_{\rm B} [  e^{-\beta  (\hat H_{\rm S} + \hat H_{\rm B})}  ]  
= e^{-\beta \hat H_{\rm S}}.
$
For the non-exact case,  take the limit. $\Box$
\end{proof}

The converse of the above lemma is not true: There are  Gibbs-preserving maps that cannot be written as thermal operations, as in the case for noisy operations (Proposition~\ref{strict_inclusions_noisy}).
More strongly, we have the following proposition, which is now contrastive to the case of noisy operations (Corollary~\ref{noisy_corollary}).
The crucial observation is that any thermal operation cannot create coherence between energy eigenstates  of $\hat H_{\rm S}$.

\begin{proposition}[Main claim of ~\cite{Faist2015}]
For any $\beta > 0$ and any Hamiltonian $\hat H_{\rm S}$ that is not proportional to the identity, there exists a pair of states $(\hat \rho_{\rm S}, \hat \rho_{\rm S}')$ such that there is a Gibbs-preserving map $\mathcal E_{\rm S}$ satisfying $\hat \rho_{\rm S}' = \mathcal E_{\rm S} (\hat \rho_{\rm S})$ but there is no such a thermal operation.
\end{proposition}

\begin{proof}
We first consider the case of  exact thermal operations.
By noting that the initial state of B is Gibbs in this case, we consider an eigenstate $| E_i^{\rm S} \rangle | E_j^{\rm B} \rangle$ of $\hat H_{\rm S} + \hat H_{\rm B}$.  Then, from the condition (\ref{thermal_operation_conserve}), $\hat U | E_i^{\rm S} \rangle | E_j^{\rm B} \rangle$ should be given by a  superposition of $| E_k^{\rm S} \rangle | E_l^{\rm B} \rangle$'s satisfying $E_i^{\rm S} + E_j^{\rm B} = E_k^{\rm S} + E_l^{\rm B}$.
By tracing out the state of B, we find that the state of S cannot have any superposition of eigenstates with different energies.
This property is kept unchanged even when one takes the limit  for non-exact thermal operations.

On the other hand, such superposition can be created by a Gibbs-preserving map as shown below.
Without loss of generality, let us focus on the situation that S is a qubit system and the Hamiltonian is given by $\hat H_{\rm S} = E_0 | 0 \rangle \langle 0 | + E_1 | 1 \rangle \langle 1 |$ with $E_0 < E_1$. (We omitted the superscript S in $E_i$.)
As shown above, it is impossible to implement the transformation $| 1 \rangle \mapsto | + \rangle := (| 0 \rangle + | 1\rangle ) / \sqrt{2}$ by any thermal operation.
We then construct a Gibbs-preserving map that enables this transformation.
Let $p_i := e^{-\beta E_i} / Z_{\rm S}$ with $Z_{\rm S} := \sum_{i=0,1}e^{-\beta E_i}$.
Then it is easy to show that $\hat \sigma := p_0^{-1}(\hat \rho_{\rm S}^{\rm G} - p_1 | + \rangle \langle + | )$ is a normalized positive operator (i.e., a density operator).  Then define
\begin{equation}
\mathcal E_{\rm S}( \hat \rho_{\rm S} ) := \langle 0 | \hat \rho_{\rm S} | 0 \rangle \hat \sigma + \langle 1 | \hat \rho_{\rm S} | 1 \rangle | + \rangle  \langle + |,
\end{equation}
which satisfies $\mathcal E_{\rm S} (\hat \rho_{\rm S}^{\rm G} ) = \hat \rho_{\rm S}^{\rm G}$ and  $\mathcal E_{\rm S} (| 1 \rangle \langle 1 | ) =  | + \rangle  \langle + |$.
$\Box$
\end{proof}

We next remark on classical thermal operations.

\begin{definition}[Classical thermal operations]
A stochastic matrix $T$ is an exact classical thermal operation with respect to a Hamiltonian of S, if there exists a heat bath B in finite dimension with some Hamiltonian  and exists a permutation matrix $P$ acting on the composite system such that
\begin{equation}
(T p)_i= \sum_{jkl} P_{ij; kl} p_k p_{{\rm B}, l}^{\rm G},
\label{noisy_operation_c}
\end{equation}
where $ p_{{\rm B}, l}^{\rm G}$ is the Gibbs state of B and $P$ preserves the sum of the energies of S and B. 
Furthermore, a stochastic matrix $T$ is a classical thermal operation with respect to a Hamiltonian of S, if there exists a sequence of exact thermal operations $\{ T_n \}_{n=1}^\infty$ with the same Hamiltonian  such that $T_n$ converges to $T$ in $n \to \infty$.
\label{def:thermal_operation_c}
\end{definition}

Based on the above definition, we can show that classical thermal operations are equivalent to classical Gibbs-preserving maps. A rigorous proof of this is given in Ref.~\cite{Shiraishi2020}; see also Ref.~\cite{Horodecki2013}.

\begin{proposition}
For any classical stochastic map $T$ and a given Hamiltonian, the following are equivalent. \\
(i) $T$ is a  Gibbs-preserving map.\\
(ii) $T$ is a  thermal operation.
\end{proposition}

We also remark on the semiclassical case.
We refer to a thermal operation as \textit{semiclassical}, if the system is quantum but the initial state $\hat \rho_{\rm S}$ and the final state $\mathcal E_{\rm S} (\hat \rho_{\rm S})$ of S are both (block-)diagonal  in the eigenbasis of the Hamiltonian $\hat H_{\rm S}$.
In the semiclassical case, thermal operations are basically equivalent to Gibbs-preserving maps, and therefore, most of the results in the subsequent sections are also valid even if  the Gibbs-preserving maps are replaced with thermal operations.
 We notice, however, that this equivalence is in general not at the level of maps but only at the level of state convertibility in the sense of Theorem~\ref{thm:majorization_q} and Corollary~\ref{noisy_corollary} (see also Proposition~\ref{strict_inclusions_noisy}).

\

Going back to the quantum case, we consider the Lindblad equation as an illustrative example~\cite{Breuer2002}, which describes quantum Markovian dynamics.
Let $\hat H_{\rm S} = \sum_i E_i | E_i \rangle \langle E_i |$ be the Hamiltonian, which is assumed to be non-degenerate for simplicity.  The Lindblad equation is given by
\begin{equation}
\frac{d \hat \rho_{\rm S} (t)}{d t} = -{\rm i} [\hat H_{\rm S}, \hat \rho_{\rm S} (t)]  +\sum_{ij} \left(  \hat L_{ij} \hat \rho_{\rm S} (t) \hat L_{ij}^\dagger -   \frac{1}{2}  \{ \hat L_{ij}^\dagger \hat L_{ij}, \hat \rho_{\rm S} (t) \}  \right).
\label{Lindblad_eq}
\end{equation}
The solution of this gives a CPTP map: the map from $\hat \rho_{\rm S} (0)$ to $\hat \rho_{\rm S} (t)$ under  Eq.~(\ref{Lindblad_eq}) is a CPTP map.

We here assume that the Lindblad operators $\hat L_{ij}$ are  labeled by $(i,j)$ with both $i$ and $j$ representing the labels of the eigenvalues of $\hat H_{\rm S}$, and that $[\hat L_{ij}, \hat H_{\rm S} ]= (E_j - E_i ) \hat L_{ij}$,
which represents that $\hat L_{ij}$ describes the quantum jump from $|E_j \rangle$ to $| E_i \rangle$.
The detailed balance condition is given by
\begin{equation}
\hat L_{ji}^\dagger e^{-\beta E_i / 2}  = \hat L_{ij} e^{-\beta E_j / 2},
\label{q_detailed_balance}
\end{equation}
which is a quantum counterpart of the classical detailed balance condition (\ref{classical_detailed_balance}).
It is straightforward to show that the detailed balance condition (\ref{q_detailed_balance}) implies that the Gibbs state $\hat \rho_{\rm S}^{\rm G}$ of $\hat H_{\rm S}$ is the stationary solution of Eq.~(\ref{Lindblad_eq}).
In other words,  the detailed balance condition (\ref{q_detailed_balance}) implies that the map from $\hat \rho_{\rm S} (0)$ to $\hat \rho_{\rm S} (t)$ is a Gibbs-preserving map for any $t$.
See also Refs.~\cite{Funo2018,Parrondo2013,Horowitz2014} for quantum stochastic thermodynamics based on the Lindblad equations.

Moreover, by looking at the standard derivation of the Lindblad equation from unitary dynamics of the system and the bath (see, e.g., Ref.~\cite{Breuer2002}), we can see that the Lindblad dynamics with the foregoing assumptions may be regarded as an approximate thermal operation for each infinitesimal time step~\cite{Funo2018}.
In fact, as already mentioned before,  the condition (\ref{thermal_operation_conserve}) is guaranteed by the rotating wave approximation.
Also, we usually suppose that  the bath is initially in the Gibbs state, which leads to the detailed balance condition (\ref{q_detailed_balance}).

\section{Clock and work storage}
\label{sec:clock_work}

We next introduce the work storage (or the weight) W and the clock C, which are already discussed in Section~\ref{sec:classical_work_bound} for the classical case. 
Hereafter, we call the original part of the system just as the ``system'' and denote it by S.
That is, the entire system consists of S, C, W, apart from  the heat bath B.

The work storage W has its own Hamiltonian, in which the work is stored. Let $\hat H_{\rm W}$ be the Hamiltonian of W.  At this stage, we do not specify its explicit form.

The change of the Hamiltonian of S by  external driving, from $\hat H_{\rm S}$ to $\hat H_{\rm S}'$, can be implemented in an autonomous way by introducing the clock C such that the Hamiltonian of SC becomes time-independent~\cite{Horodecki2013,Aberg2014,Skrzypczyk2014,Malabarba2015,Korzekwa2016,Woods2019,Frenzel2014}.
For example, suppose that the clock evolves from a pure state $| 0 \rangle$ to another pure state $| 1 \rangle$ with these states being orthogonal, and that the Hamiltonian of SC is given by the form $\hat H_{\rm SC} = \hat H_{\rm S} \otimes |  0 \rangle \langle  0 | + \hat H_{\rm S}' \otimes |  1  \rangle \langle  1  |$ as adopted in Ref.~\cite{Horodecki2013}.
In this case, the Hamiltonian of S effectively changes from $\hat H_{\rm S}$ to $\hat H_{\rm S}'$.

The free Hamiltonian of SCW is now  given by
\begin{equation}
\hat H_{\rm SCW} = \hat H_{\rm S} \otimes |  0 \rangle \langle  0 | + \hat H_{\rm S}' \otimes |  1  \rangle \langle  1  | +  \hat H_{\rm W}.
\label{Hamiltonian_SCW}
\end{equation}
In the subsequent sections, we will consider Gibbs-preserving maps and thermal operations with respect to the Hamiltonian (\ref{Hamiltonian_SCW}).
The condition of thermal operation on SCW (Eq.~(\ref{thermal_operation_conserve}) in the absence of CW) is written as
\begin{equation}
[ \hat U, \hat H_{\rm SCW} + \hat H_{\rm B} ] = 0,
\label{thermal_operation_SCW}
\end{equation}
where $\hat H_{\rm B}$ is the Hamiltonian of B and $\hat U$ is the unitary operator acting on SCWB.
The condition (\ref{thermal_operation_SCW}) represents the strict energy conservation of the free Hamiltonian of SCWB.

In most of the remaining part of this chapter, we use the following notations.
Let  $\hat \rho_{\rm SCW}^{\rm G}$, $\hat \rho_{\rm S}^{\rm G}$, $\hat \rho_{\rm S}^{\rm G}{}'$, $\hat \rho_{\rm W}^{\rm G}$ be the Gibbs states of the Hamiltonians $\hat H_{\rm SCW}$, $\hat H_{\rm S}$, $\hat H_{\rm S}'$, $\hat H_{\rm W}$,
and $Z_{\rm SCW}$, $Z_{\rm S}$, $Z_{\rm S}'$, $Z_{\rm W}$ and  $F_{\rm SCW}$, $F_{\rm S}$, $F_{\rm S}'$, $F_{\rm W}$ be the corresponding partition functions and the free energies, respectively.
Let $\Delta F_{\rm S} := F_{\rm S}' - F_{\rm S}$.
 From Eq.~(\ref{Hamiltonian_SCW}), we have
\begin{equation}
\hat \rho_{\rm SCW}^{\rm G} = \hat \rho_{\rm SC}^{\rm G} \otimes \hat \rho_{\rm W}^{\rm G}, 
\label{Gibbs_total}
\end{equation}
where 
\begin{equation}
 \hat \rho_{\rm SC}^{\rm G} = \frac{Z_{\rm S}}{Z_{\rm S} + Z_{\rm S}'}  \hat \rho_{\rm S}^{\rm G} \otimes | 0 \rangle \langle 0 | + \frac{Z_{\rm S}'}{Z_{\rm S} + Z_{\rm S}'}\hat \rho_{\rm S}^{\rm G}{}' \otimes | 1 \rangle \langle 1 |.
 \label{Gibbs_SCW}
\end{equation}

Concrete constructions of the work storage and the clock  are discussed in Refs.~\cite{Aberg2014,Skrzypczyk2014,Malabarba2015,Korzekwa2016,Woods2019,Frenzel2014}.
Here we illustrate the function of the clock by  a (almost trivial) toy example.
For simplicity, we only consider the composite SC.  Let $\hat H_{\rm SC} = \hat \sigma_z \otimes | 0 \rangle \langle 0 | + \hat \sigma_x \otimes | 1 \rangle \langle 1 |$ be the Hamiltonian and $\hat U = \hat h \otimes \hat \sigma_x$ be the unitary operation, where $\hat \sigma_i$ ($i=x,y,z$) is the Pauli operator with $\{ | 0 \rangle, |1 \rangle \}$ being the eigenbasis of $\hat \sigma_z$, and $\hat h$ is the Hadamard operator satisfying $\hat h^\dagger = \hat h$, $\hat h \hat \sigma_x \hat h = \hat \sigma_z$, $\hat h \hat \sigma_z \hat h = \hat \sigma_x$.
Then, $\hat U^\dagger \hat H_{\rm SC} \hat U = \hat H_{\rm SC}$, or equivalently $[ \hat U, \hat H_{\rm SC} ] = 0$, holds.
For any input state of the form $| \psi \rangle | 0 \rangle$, we have $\hat U | \psi \rangle | 0 \rangle = (\hat h | \psi \rangle ) | 1 \rangle$.
Thus in this case, the clock works perfectly for all initial states of S, while satisfying condition~(\ref{thermal_operation_SCW}).

In general, however,
not all the initial states of S are consistent with the evolution of C from $| 0 \rangle$ to $ |1 \rangle$,
under the constraint of the strict energy conservation~(\ref{thermal_operation_SCW}).
This is because $\hat H_{\rm S}$ and $\hat H_{\rm S}'$ (or  $\hat H_{\rm S} + \hat H_{\rm W} + \hat H_{\rm B}$ and $\hat H_{\rm S}' + \hat H_{\rm W} + \hat H_{\rm B}$) can have different energy spectra.
Thus, in the formulation of the subsequent sections, we only require that the clock evolves from $| 0 \rangle$ to $|1\rangle$ for a \textit{given} initial state of S.

Suppose that  the unitary operator is generated by a time-independent Hamiltonian as $\hat U = \exp (-{\rm i} \hat H_{\rm tot} t)$, where $\hat H_{\rm tot} = \hat H_{\rm S} \otimes |  0 \rangle \langle  0 | + \hat H_{\rm S}' \otimes |  1  \rangle \langle  1  | +  \hat H_{\rm W} + \hat H_{\rm B} + \hat H_{\rm int}$
is the total Hamiltonian and  $\hat H_{\rm int}$ is the interaction Hamiltonian.
Here, $\hat H_{\rm int}$ represents not only the interaction between S and B, but also the interaction between S and W, and moreover, driving of C.
The condition of thermal operation on SCW (\ref{thermal_operation_SCW}) is rewritten as $[  \hat H_{\rm tot}, \hat H_{\rm SCW} + \hat H_{\rm B} ]= 0$, or equivalently $[  \hat H_{\rm int}, \hat H_{\rm SCW} + \hat H_{\rm B} ]= 0$.

It is not trivial whether a desired unitary operation $\hat U$ can be implemented by a time-independent Hamiltonian $\hat H_{\rm tot}$ in a physically feasible form without any additional thermodynamic cost.
A way to implement such a unitary is that we introduce a ``hyper-clock'', named C', in addition to SCWB, as discussed in Ref.~\cite{Malabarba2015} (C' is called just a ``clock'' there).
In that case, the total Hamiltonian is given by $\hat H_{\rm tot} = \hat H_{\rm SCW} + \hat H_{\rm B} + \hat H_{\rm int}' + \hat H_{\rm C'}$, where $\hat H_{\rm int}'$ describes the interaction between SCWB and C'.
By imposing $[ \hat H_{\rm int}',  \hat H_{\rm SCW} + \hat H_{\rm B}] = 0$, the strict energy conservation~(\ref{thermal_operation_SCW}) is guaranteed for $\hat U = \exp (-{\rm i} \hat H_{\rm tot} t)$.
We note that a general bound of the quantum coherence cost to implement such (hyper-)clocks has been obtained in Ref.~\cite{Tajima2018}.

\section{Average work bound}
\label{sec:fluctuating_work}

In this section, we consider the work bound  at the level of ensemble average for the situation where the work fluctuates~\cite{Aberg2014,Skrzypczyk2014,Alhambra2016}.
This is contrastive to the single-shot case discussed in the subsequent sections
and is relevant to stochastic thermodynamics in the quantum regime~\cite{Esposito2009,Sagawa2012,Funo2018}.
This also complements the arguments in  Section~\ref{sec:classical_second_law} and Section~\ref{sec:quantum_entropy}.
In the following, we derive the average work bound in two slightly different ways and discuss the protocols that achieve the bound.

First, we derive the second law by considering the Gibbs-preserving map of SCW.
We adopt the setup and the notations of Section~\ref{sec:clock_work}; 
Specifically,  the Hamiltonian $\hat H_{\rm SCW}$ is given by Eq.~(\ref{Hamiltonian_SCW}). 
We suppose that the initial state of SCW is given by $\hat \rho := \hat \rho_{\rm S} \otimes | 0 \rangle \langle 0 | \otimes \hat \rho_{\rm W}$.
(We omit the subscript SCW from $\hat \rho$.)
At this stage, we do not make any assumption about the form of $\hat H_{\rm W}$ and  $\hat \rho_{\rm W}$.

The final state after the Gibbs-preserving map of SCW, written as $\mathcal E_{\rm SCW}$, is given by  $\hat \rho' := \mathcal E_{\rm SCW} (\hat \rho)$.
For a given initial state $\hat \rho_{\rm S}$ of S, we assume that $\hat \rho'$ is given by the product of the final state of SW, written as $\hat \rho_{\rm SW}'$, and the final pure state of C given by $| 1\rangle \langle 1 |$.
This implies that the clock works perfectly for the given initial state.
Let $\hat \rho_{\rm S}' := {\rm tr}_{\rm W} [\hat \rho_{\rm SW}' ]$ and $\hat \rho_{\rm W}' := {\rm tr}_{\rm S} [\hat \rho_{\rm SW}' ]$.
Then, the average work is defined as
\begin{equation}
W := {\rm tr} [( \hat \rho_{\rm W} - \hat \rho_{\rm W}') \hat H_{\rm W}].
\end{equation}

We here make an assumption that $S_1 (\hat \rho_{\rm W } ) = S_1 (\hat \rho_{\rm W } ') $, which excludes the possibility of  ``cheating'' that W works as an entropic source.
This assumption can be satisfied at least approximately;
For example,  it has been shown in Ref.~\cite{Skrzypczyk2014} that the entropy change can be arbitrarily small by appropriately designing the interaction and the initial state of W.

We now derive the work bound.
By noting that the Gibbs states of SCW and SC are  given by Eq.~(\ref{Gibbs_total}) and Eq.~(\ref{Gibbs_SCW}), we have
\begin{equation}
S_1 ( \hat \rho \| \hat \rho^{\rm G}_{\rm SCW} )  = S_1 ( \hat \rho_{\rm S} \| \hat \rho_{\rm S}^{\rm G} )   + \beta F_{\rm S} + \ln ( Z_{\rm S} + Z_{\rm S}' ) +  S_1 ( \hat \rho_{\rm W} \| \hat \rho_{\rm W}^{\rm G} ),
\end{equation}
and
\begin{equation}
S_1 ( \hat \rho' \| \hat \rho^{\rm G}_{\rm SCW} )  = S_1 ( \hat \rho_{\rm SW}' \| \hat \rho_{\rm SW}^{\rm G}{}' ) + \beta F_{\rm S}' + \ln ( Z_{\rm S} + Z_{\rm S}' ),
\end{equation}
where we used the scaling property~(\ref{unnormalized_sigma_scaling}) and
defined $\hat \rho_{\rm SW}^{\rm G}{}' := \hat \rho_{\rm S}^{\rm G}{}' \otimes \hat \rho_{\rm W}^{\rm G}$.
From the subadditivity (\ref{subadditivity}) of the von Neumann entropy, we have
\begin{eqnarray}
S_1 ( \hat \rho_{\rm SW}' \| \hat \rho_{\rm SW}^{\rm G}{}' ) 
&=& - S_1 ( \hat \rho_{\rm SW}' ) -{\rm tr}[ \hat \rho_{\rm S}' \ln  \hat \rho_{\rm S}^{\rm G}{}'] - {\rm tr}[ \hat \rho_{\rm W}' \ln  \hat \rho_{\rm W}^{\rm G}] \\
&\geq&  - S_1 ( \hat \rho_{\rm S}' ) - S_1 ( \hat \rho_{\rm W}' ) -{\rm tr}[ \hat \rho_{\rm S}' \ln  \hat \rho_{\rm S}^{\rm G}{}']  - {\rm tr}[ \hat \rho_{\rm W}' \ln  \hat \rho_{\rm W}^{\rm G}] \\
&=& S_1 ( \hat \rho_{\rm S}' \| \hat \rho_{\rm S}^{\rm G}{}' ) + S_1 ( \hat \rho_{\rm W}' \| \hat \rho_{\rm W}^{\rm G}).
\end{eqnarray}
Also,  $S_1 ( \hat \rho_{\rm W} \| \hat \rho_{\rm W}^{\rm G} ) - S_1 ( \hat \rho_{\rm W}' \| \hat \rho_{\rm W}^{\rm G} ) = \beta W$ holds from  the assumption  $S_1 (\hat \rho_{\rm W } ) = S_1 (\hat \rho_{\rm W } ') $.
By combining these relations, the monotonicity of the quantum KL divergence, i.e., $S_1 ( \hat \rho \| \hat \rho^{\rm G}_{\rm SCW} )  \geq S_1 ( \hat \rho' \| \hat \rho^{\rm G}_{\rm SCW} )$, leads to 
\begin{equation}
\beta (W - \Delta F_{\rm S} ) \geq S_1 ( \hat \rho_{\rm S}' \| \hat \rho_{\rm S}^{\rm G}{}' ) -  S_1 ( \hat \rho_{\rm S} \| \hat \rho_{\rm S}^{\rm G} ).
\label{fluctuating_work_second1}
\end{equation}
This is the bound of the average work in the present setup.
By using the nonequilibrium $1$-free energy introduced in Eq.~(\ref{alpha_free_energy}), we rewrite the obtained work bound as
\begin{equation}
W \geq F_1 (\hat \rho_{\rm S}'; \hat H_{\rm S}') - F_1 (\hat \rho_{\rm S}; \hat H_{\rm S}),
\label{fluctuating_work_second2}
\end{equation}
which is the quantum extension of inequality~(\ref{classical_work_bound_F1}).

It is reasonable to define  the  average heat absorbed by the system through the first law of thermodynamics:
\begin{equation}
Q := {\rm tr}[\hat H_{\rm S}' \hat \rho_{\rm S}'] -  {\rm tr}[\hat H_{\rm S} \hat \rho_{\rm S}] - W.
\label{def_average_quantum_heat}
\end{equation}
Then, we can rewrite the second law (\ref{fluctuating_work_second2}) in the form of the Clausius inequality:
\begin{equation}
S_1 (\hat \rho_{\rm S}') - S_1 (\hat \rho_{\rm S} ) \geq \beta Q,
\end{equation}
which is the quantum counterpart of inequality~(\ref{classlcal_Gibbs_second1}).

To strictly justify the definition of heat  (\ref{def_average_quantum_heat}),  we need an additional assumption about the energy balance between  SCW and the bath B.
A strong assumption is the strict energy conservation~(\ref{thermal_operation_SCW}) that is satisfied if $\mathcal E_{SCW}$ is a thermal operation, with which the heat (\ref{def_average_quantum_heat}) exactly coincides with the energy change in B.
On the other hand, a weaker assumption is often enough to justify Eq.~(\ref{def_average_quantum_heat}):
the average energy conservation, which means that the expectation value of the  free Hamiltonian of SCWB, $\hat H_{\rm SCW} + \hat H_{\rm B}$, is conserved for a \textit{given} initial state, which is the approach adopted in Ref.~\cite{Skrzypczyk2014}.

\

We next derive the work bound~(\ref{fluctuating_work_second1}) in an alternative way without assuming that $S_1 (\hat \rho_{\rm W } ) = S_1 (\hat \rho_{\rm W } ') $~\cite{Malabarba2015}.
We explicitly take the heat bath B into account, and consider the Hamiltonian $\hat H_{\rm SCW} + \hat H_{\rm B}$ and a unitary operation $\hat V$ acting on SCWB.
At this stage, we do not necessarily require the strict energy conservation (i.e., $[\hat V, H_{\rm SCW} + \hat H_{\rm B}] = 0$), but only require the average energy conservation at least for a given initial state $\hat \rho_{\rm SCB}$.
Now we trace out W and define a CPTP map
\begin{equation}
\mathcal E_{\rm SCB} (\hat \rho_{\rm SCB}) := {\rm tr}_{\rm W} [ \hat V \hat \rho_{\rm SCB} \otimes \hat \rho_{\rm W} \hat V^\dagger ].
\end{equation}
A central assumption of the present approach is that $\mathcal E_{\rm SCB}$ is unital,
 which is indeed satisfied in some constructions of $\hat V$ discussed later.
From this assumption,
\begin{equation}
S_1(\hat \rho_{\rm SCB}) \leq S_1(\hat \rho_{\rm SCB}').
\label{SCB_unital}
\end{equation}

Let $\hat \rho_{\rm SCB} := \hat \rho_{\rm S} \otimes | 0 \rangle \langle 0 | \otimes \hat \rho_{\rm B}^{\rm G}$, $\hat \rho_{\rm SCB} ' := \mathcal E_{\rm SCB} (\hat \rho_{\rm SCB} )$.
For a given $\hat \rho_{\rm S}$, we assume that $\hat \rho'_{\rm SCB}$ is given by the product of the final state of SB, written as $\hat \rho_{\rm SB}'$, and the final pure state of C given by $| 1\rangle \langle 1 |$.
We then have
$S_1 (\hat \rho_{\rm SCB} ) = - S_1(\hat \rho_{\rm S} \| \hat \rho^{\rm G}_{\rm S}) - \beta (F_{\rm S} - {\rm tr}[\hat H_{\rm S}\hat \rho_{\rm S}]) - \beta (F_{\rm B} - {\rm tr}[\hat H_{\rm B}\hat \rho_{\rm B}^{\rm G}])$ and
$S_1 (\hat \rho_{\rm SCB}' ) \leq S_1(\hat \rho_{\rm S}') + S_1(\hat \rho_{\rm B}') 
= - S_1(\hat \rho_{\rm S}' \| \hat \rho_{\rm S}^{\rm G}{}') - \beta (F_{\rm S}' - {\rm tr}[\hat H_{\rm S}' \hat \rho_{\rm S}']) 
-S_1 (\hat \rho_{\rm B}' \| \hat \rho_{\rm B}^{\rm G})- \beta (F_{\rm B} - {\rm tr}[\hat H_{\rm B}\hat \rho_{\rm B}'])$.
Then, inequality~(\ref{SCB_unital}) leads to
\begin{equation}
\beta \left(  {\rm tr}[\hat H_{\rm S}' \hat \rho_{\rm S}'] + {\rm tr}[\hat H_{\rm B}\hat \rho_{\rm B}'] - {\rm tr}[\hat H_{\rm S}\hat \rho_{\rm S}] -  {\rm tr}[\hat H_{\rm B}\hat \rho_{\rm B}^{\rm G}] - \Delta F_{\rm S} \right)
\geq  S_1(\hat \rho_{\rm S}' \| \hat \rho_{\rm S}^{\rm G}{}')  - S_1(\hat \rho_{\rm S} \| \hat \rho^{\rm G}_{\rm S}).
\end{equation}
From the (at least average)  energy conservation, we have
\begin{equation}
{\rm tr}[\hat H_{\rm S}' \hat \rho_{\rm S}'] + {\rm tr}[\hat H_{\rm B}\hat \rho_{\rm B}'] - {\rm tr}[\hat H_{\rm S}\hat \rho_{\rm S}] -  {\rm tr}[\hat H_{\rm B}\hat \rho_{\rm B}^{\rm G}] = {\rm tr}[\hat H_{\rm W} \hat \rho_{\rm W}]-  {\rm tr}[\hat H_{\rm W} \hat \rho_{\rm W}'] =: W.
\end{equation}
Therefore, we obtain the same inequality as (\ref{fluctuating_work_second1}):
\begin{equation}
\beta (W - \Delta F_{\rm S} ) \geq S_1 ( \hat \rho_{\rm S}' \| \hat \rho_{\rm S}^{\rm G}{}' ) -  S_1 ( \hat \rho_{\rm S} \| \hat \rho_{\rm S}^{\rm G} ).
\label{fluctuating_work_second0}
\end{equation}

The optimal work (i.e., the equality case of the second law~(\ref{fluctuating_work_second0})) can be achieved by special protocols with certain limit~\cite{Aberg2014,Skrzypczyk2014,Malabarba2015,Korzekwa2016}.
In Refs.~\cite{Aberg2014,Malabarba2015}, it was shown that there is an optimal protocol under the strict energy conservation, where a large amount of coherence is required in  W.
In Ref.~\cite{Skrzypczyk2014}, an optimal protocol was constructed under the average energy conservation, where the work value can be independent of the initial state of W.
We will discuss these protocols in detail later.

The fact that the equality in (\ref{fluctuating_work_second0}) can be saturated implies that a necessary and sufficient condition for state conversion is given by the KL divergence (rigorously speaking, except for some technical issues), 
if the work fluctuation is allowed.
This is contrastive to the single-shot case discussed in Section.~\ref{sec:work_bound_single}, where the KL divergence does not provide such a condition and instead the R\'enyi $0$- and $\infty$-divergences appear.
The role of the work storage here would be reminiscent of  ``modestly non-exact'' catalyst discussed in Section~\ref{sec:catalytic_majorization} where the KL divergence also appears.  In fact, the work storage can be regarded as a kind of catalyst as discussed later in this section.

\

We now explicitly show that the equality in (\ref{fluctuating_work_second0}) can be achieved based on the protocol considered in Refs.~\cite{Aberg2014,Skrzypczyk2014,Malabarba2015}.
We start with a general strategy and ignore the work storage W for a while.
Also for simplicity, we assume that the clock C is absent and $\hat H_{\rm S} = \hat H_{\rm S}'$.

We will construct a unitary operator $\hat U$ acting on SB, which is not necessarily energy conserving.
We suppose that the average energy change of SB equals the average work, and show that the average work saturates the equality in (\ref{fluctuating_work_second0}) by an optimal protocol.
We will construct an energy-conserving unitary acting on SBW from $\hat U$ later.

Let $\hat \rho_{\rm S} = \sum_i p_i | \varphi_i \rangle \langle \varphi_i |$, $\hat \rho_{\rm S}' = \sum_i p_i' | \varphi_i' \rangle \langle \varphi_i' |$, and $\hat H_{\rm S} = \sum_i E_i | E_i \rangle \langle E_i|$.
Suppose that $\hat \rho_{\rm S}'$ is positive definite.
We define $\hat \tau_{\rm S} := \sum_i p_i^{\downarrow} | E_i \rangle \langle E_i |$ and $\hat \tau_{\rm S}' := \sum_i p_i'{}^{\downarrow} | E_i \rangle \langle E_i |$.
We also define  $p^{(n)}$ ($n=0, 1,2, \cdots, N$) satisfying $p^{(n)}= p^{(n) \downarrow}$ such that $p^{(0)} = p^\downarrow$, $D(p^{(n)}, p^{(n+1)} ) = O(1/N)$  ($n=0,1, \cdots, N-1$),  and $p^{(N)} = p'{}^\downarrow$.
Let $\hat \tau_{\rm S}^{(n)} := \sum_i p_i^{(n)} | E_i \rangle \langle E_i |$.
We divide B into $N$ subsystems, and introduce the Hamiltonian $\hat H_{\rm B}^{(n)}$ of the $n$th subsystem such that its Gibbs state $\hat \rho_{\rm B}^{{\rm G}(n)}$ equals $\hat \tau_{\rm S}^{(n)}$.
The total Hamiltonian of B is given by $\hat H_{\rm B} := \sum_{n=1}^N \hat H_{\rm B}^{(n)}$.
Then, we construct $\hat U$ acting on SB as follows.

\begin{enumerate}
\item Rotate the basis of S such that $\hat \rho_{\rm S}$ is mapped to $\hat \tau_{\rm S} = \hat \tau_{\rm S}^{(0)}$.
\item Swap $\hat \tau_{\rm S}^{(n-1)}$ with the $n$th subsystem of B with the state $\hat \rho_{\rm B}^{{\rm G}(n)} = \hat \tau_{\rm S}^{(n)}$ ($n=1, \cdots, N$).  
\item Rotate the basis of S such that $\hat \tau_{\rm S}'$ is mapped to $\hat \rho_{\rm S}'$.
\end{enumerate}

In each step of the above protocol, the change in the energy of SB equals the change in the $1$-free energy  of S at least approximately.
In (i) and (iii), the change in the von Neumann entropy is zero, and thus the energy change of S equals the free energy change of S.
In (ii), let $\Delta S_{\rm S}^{(n)} := S_1 (\hat \tau_{\rm S}^{(n)}) -S_1 (\hat \tau_{\rm S}^{(n-1)} ) =: - \Delta S_{\rm B}^{(n)}$,  $\Delta E_{\rm S}^{(n)} :={\rm tr}[\hat H_{\rm S} \hat \tau_{\rm S}^{(n)}] -{\rm tr}[\hat H_{\rm S} \hat \tau_{\rm S}^{(n-1)} ]$, and 
$\Delta E_{\rm B}^{(n)} :={\rm tr}[\hat H_{\rm B}^{(n)} \hat \tau_{\rm S}^{(n-1)}] -{\rm tr}[\hat H_{\rm B}^{(n)} \hat \tau_{\rm S}^{(n)}]$
($n=1,2, \cdots, N$).
By noting that $\beta \Delta E_{\rm B}^{(n)} - \Delta S_{\rm B}^{(n)} = S_1 (\tau_{\rm S}^{(n-1)} \| \hat \tau_{\rm S}^{(n)} ) = O(1/N^2)$, we have $\Delta E_{\rm S}^{(n)} + \Delta E_{\rm B}^{(n)} = \Delta E_{\rm S}^{(n)} - \beta^{-1}\Delta S_{\rm S}^{(n)} + O(1/N^2)$,
which is the $1$-free energy change of S up to the error term $O(1/N^2)$.
By summing these terms over $n=1,\cdots, N$, the total error is given by $O(1/N)$.
In the entire process (i)-(iii), therefore, the energy change of SB equals the $1$-free energy change of S in the limit of $N \to \infty$.

We note that the above construction of the optimal process is slightly different from the 
quantum version of Fig.~\ref{fig:equality_protocol_PHS}.
In the present construction, the Hamiltonian of S is not quenched but the state of S rotates in (i) and (iii).
Also, the gradual change of the state of S in (ii) is induced by a sequence of swapping processes, instead of ordinary relaxation processes.

If $\hat H_{\rm S} \neq \hat H_{\rm S}'$ in general, we can just insert the unitary change of the state  of C from $| 0 \rangle $ to $| 1 \rangle $ into any stage of (i)-(iii) above.  During this step, the energy change of S equals the $1$-free energy change of S, because the state of S does not change at all. 

\

So far, we have constructed unitary $\hat U$ acting on SB, which is not energy-conserving.  
Our next purpose is to construct an energy-conserving unitary acting on SBW that reproduces the same work value as the average energy change of SB by $\hat U$.
This can be done by both the strict and average energy conservation protocols, as discussed below.
In the following,  we relabel SB just as ``S''.

First, we consider the case of the strict energy conservation.
In line with Ref.~\cite{Aberg2014}, we show that  if W is appropriately designed, 
any unitary $\hat U$ acting on S, which is not necessarily energy-conserving, can be implemented as a strictly energy conserving unitary $\hat V_{\hat U}$ acting on SW.
Such a construction is inevitably accompanied by work fluctuations, and thus is not compatible with the single-shot scenario.

For the sake of simplicity, we first suppose that S is a two-level system spanned by $\{ | 0 \rangle , | 1 \rangle \}$ with Hamiltonian $\hat H_{\rm S} =| 1 \rangle \langle 1 | $.  
We prepare W as an infinite-dimensional ladder system spanned by $\{ | n \rangle \}_{n \in \mathbb Z}$ with Hamiltonian $\hat H_{\rm W} = \sum_{n \in \mathbb Z} n | n \rangle \langle n |$.
Here, we do not go into details of the technical aspect of unbounded operators of infinite-dimensional spaces.
Also, a physical problem is that this Hamiltonian is not bounded from below, but this  can be overcome by a slight modification of the setup~\cite{Aberg2014}.
We now define the unitary operator
\begin{equation}
\hat V_{\hat U} := \sum_{i,j=0,1} | i \rangle \langle i| \hat U | j \rangle \langle j | \otimes \hat \Delta^{j-i},
\label{Aberg_unitary_V}
\end{equation}
where $\hat \Delta := \sum_{n \in \mathbb Z} | n  + 1 \rangle \langle n |$ is the shift operator of the ladder.
It is straightforward to check the strict energy conservation, i.e., $[ \hat V_{\hat U}, \hat H_{\rm S} + \hat H_{\rm W} ] = 0$.
Also, $[\hat V_{\hat U} , \hat \Delta ] = 0$ is obviously satisfied, implying that $\hat V_{\hat U}$ is translation invariant with respect to the ladder of W.

If S has multiple levels in general, W should consist of multiple ladders corresponding to all the energy gaps of S.  The ideal limit is that W is a continuous system with $(\hat x, \hat p )$ satisfying $[ \hat x, \hat p ] = {\rm i}$ and the Hamiltonian of W is given by $\hat H_{\rm W} = \hat x$.
In the case of a general Hamiltonian $\hat H_{\rm S} = \sum_i E_i | E_i \rangle \langle E_j |$, Eq.~(\ref{Aberg_unitary_V}) is replaced by~\cite{Malabarba2015}
\begin{equation}
\hat V_{\hat U} := \sum_{i,j} | E_i \rangle \langle E_i| \hat U | E_j \rangle \langle E_j | \otimes e^{-{\rm i} (E_j-E_i) \hat p}.
\label{Aberg_unitary_V_general}
\end{equation}

Going back to the two-level case, we consider the CPTP map 
\begin{eqnarray}
\mathcal E_{\hat U} (\hat \rho_{\rm S}) &:=& {\rm tr}[\hat V_{\hat U} \hat \rho_{\rm S} \otimes \hat \rho_{\rm W} \hat V_{\hat U}^\dagger]
\label{E_U_def} \\
&=& \sum_{i,j,k,l = 0,1} |  i \rangle \langle i | \hat U | j \rangle \langle j | \hat \rho_{\rm S} | l \rangle \langle l | \hat U^\dagger | k \rangle \langle k | {\rm tr}[\hat \Delta^{k-l+j-i} \hat \rho_{\rm W}].
\end{eqnarray}
If ${\rm tr}[\hat \Delta^{k-l+j-i} \hat \rho_{\rm W}] = 1$ for all $(i,j,k,l)$, we recover the unitary operation on S: $\mathcal E_{\hat U} (\hat \rho_{\rm S})  = \hat U \hat \rho_{\rm S} \hat U^\dagger$.
This can be approximately satisfied if we choose the initial state of W with a large amount of coherence in the energy basis.
In fact, if  the initial state of W is given by a pure state $\sum_{n=0}^{N-1} | n \rangle / \sqrt{N}$, we have ${\rm tr}[\hat \Delta^m \hat \rho_{\rm W}] = (N-| m | )/N$, which goes to $1$ in $N \to \infty$.
In this limit, however, the variance of the energy distribution of W diverges and thus the work values can hardly be distinguished by measuring the energy of W~\cite{Hayashi2017}.

In Ref.~\cite{Aberg2014}, it has been shown that the coherence of W can be used repeatedly in the following sense:
Let $\hat \rho_{\rm W}'$ be the final state of W of the above operation. Then, we have  ${\rm tr}[\hat \Delta^m \hat \rho_{\rm W}] =  {\rm tr}[\hat \Delta^m \hat \rho_{\rm W}']$.
Therefore, W is regarded as a catalyst of coherence.
We also note that, by applying the foregoing argument to SB instead of S, a thermal operation of SWB can create coherence of S, if W has coherence.

It is straightforward to check that $\mathcal E_{\hat U}$ constructed from $\hat V_{\hat U}$ is unital for any $\hat U$ and any initial state of W.  Thus the work bound (\ref{fluctuating_work_second0}) is applicable.
Then, by implementing the above-discussed optimal $\hat U$ that saturates the work bound (\ref{fluctuating_work_second0}), we obtain the optimal protocol including W under the strict energy conservation and with a large amount of initial coherence of W.

As a side remark, we consider the case that the initial state of W does not have any coherence in the foregoing construction.
Suppose that  $\hat \rho_{\rm W}$ in Eq.~(\ref{E_U_def})  is diagonal in the energy basis, which implies that ${\rm tr}[\hat \Delta^m \hat \rho_{\rm W}] = 0$ for $m \neq 0$.
In this case, it is easy to show that 
\begin{equation}
{\rm tr} [ \mathcal E_{\hat U} ( \mathcal D (\hat \rho_{\rm S} ) ) \hat H_{\rm S} ] ={\rm tr}[ \mathcal E_{\hat U} ( \hat \rho_{\rm S} ) \hat H_{\rm S} ],
\label{E_U_same}
\end{equation}
where $ \mathcal D (\hat \rho_{\rm S} ) := \sum_i | E_i \rangle \langle E_i | \hat \rho_{\rm S} | E_i \rangle \langle E_i |$ with $\hat H_{\rm S} = \sum_i E_i | E_i \rangle \langle E_i |$. 
Eq.~(\ref{E_U_same}) implies that coherence of $\hat \rho_{\rm S}$ does not play any role in the energy balance during the process.
When we explicitly consider B, notice that the initial state of B, $\hat \rho_{\rm B}^{\rm G}$, is already diagonal in the energy basis.
Then, in the case that $\hat H_{\rm S} = \hat H_{\rm S}'$ and $\hat \rho_{\rm S}' = \hat \rho_{\rm S}^{\rm G}$, we obtain a tighter work bound with the incoherent work storage~\cite{Skrzypczyk2013}:
\begin{equation}
- \beta W \leq S_1 (\mathcal D (\hat \rho_{\rm S} ) \| \hat \rho_{\rm S}^{\rm G}),
\end{equation}
where the right-hand side is equal to or smaller than $S_1 (\hat \rho_{\rm S}  \| \hat \rho_{\rm S}^{\rm G})$  in  (\ref{fluctuating_work_second0}).
This implies that, in the absence of coherence of W and under the strict energy conservation, one cannot harvest coherence of the initial state of S for the work extraction.

\

Next, we consider the case of the average energy conservation in line with Ref.~\cite{Skrzypczyk2014}.
Our goal is to implement a unitary $\hat U$  acting on S by a unitary $\hat V'_{\hat U}$ acting on SW that only conserves the average energy for a given initial state $\hat \rho_{\rm S}$.

Let $\hat \rho_{\rm S} =\sum_i p_i | \varphi_i \rangle \langle \varphi_i |$ be the initial state of S and $\hat H_{\rm S}$ be its Hamiltonian.
Suppose that W is a continuous system with $( \hat x, \hat p)$ satisfying $[ \hat x, \hat p ] = {\rm i}$ and its Hamiltonian is given by $\hat H_{\rm W} = \hat x$.
Instead of Eq.~(\ref{Aberg_unitary_V}) or (\ref{Aberg_unitary_V_general}),
we consider the following unitary operator that depends on $\hat \rho_{\rm S}$:
\begin{equation}
\hat V'_{\hat U} := \sum_{i}  \hat U | \varphi_i \rangle \langle \varphi_i | \otimes e^{-{\rm i} \epsilon_i \hat p},
\end{equation}
where 
\begin{equation}
\epsilon_i := \langle \varphi_i | \hat H_{\rm S} | \varphi_i \rangle - \langle \varphi_i | \hat U^\dagger \hat H_{\rm S} \hat U | \varphi_i \rangle.
\end{equation}

The corresponding CPTP map on S is given by $\mathcal E'_{\hat U} ( \bullet ) := {\rm tr}_{\rm W}[\hat V'_{\hat U} \bullet \otimes \hat \rho_{\rm W}\hat V'_{\hat U}{}^\dagger]$,
where $\hat \rho_{\rm W}$ is an  initial state of W.
Although $\mathcal E'_{\hat U}$ is not unitary in general, it reproduces the action of unitary $\hat U$ only for the given initial state $\hat \rho_{\rm S}$; it is straightforward to show that
\begin{equation}
{\rm tr}_{\rm W}[\hat V'_{\hat U} \hat \rho_{\rm S}  \otimes \hat \rho_{\rm W}\hat V'_{\hat U}{}^\dagger] = \hat U \hat \rho_{\rm S} \hat U^\dagger.
\end{equation}
Thus, the energy change of S by $\hat V'_{\hat U}$ equals that by $\hat U$, which is given by ${\rm tr}_{\rm S} [\hat H_{\rm S} \hat U \hat \rho_{\rm S}\hat U^\dagger] - {\rm tr}_{\rm S} [\hat H_{\rm S} \hat \rho_{\rm S}]  =  - \sum_i p_i \epsilon_i$.
Let  $\hat \rho_{\rm W}' := {\rm tr}_{\rm S}[\hat V'_{\hat U}  \hat \rho_{\rm S} \otimes \hat \rho_{\rm W}\hat V'_{\hat U}{}^\dagger]$ be the final state of W.
The energy change of W is given by
${\rm tr}_{\rm W} [\hat H_{\rm W} \hat \rho_{\rm W}'] - {\rm tr}_{\rm W} [\hat H_{\rm W} \hat \rho_{\rm W}] = \sum_i p_i \epsilon_i$.
Thus, $\hat V'_{\hat U}$ satisfies the average energy conservation.

We emphasize that $\hat V'_{\hat U}$ simulates the action of $\hat U$ and conserves the average energy only for the given initial state $\hat \rho_{\rm S}$; In fact, $\hat V'_{\hat U}$ itself depends on $\hat \rho_{\rm S}$.
On the other hand, the initial state $\hat \rho_{\rm W}$ of W is arbitrary; In particular,  the initial coherence of W is not required for this protocol.
We note that the translation invariance of W is satisfied:  $[\hat V_{\hat U}' , e^{-{\rm i} \varepsilon \hat p}] = 0$ for all $\varepsilon \in \mathbb R$.

It is straightforward to check that $\mathcal E'_{\hat U}$ constructed from $\hat V_{\hat U}'$ is unital for any $\hat U$ and any initial state of W.  Thus the work bound (\ref{fluctuating_work_second0}) is again applicable.
Then, by constructing $\hat V_{\hat U}'$ from the optimal $\hat U$ that saturates the work bound (\ref{fluctuating_work_second0}),
we obtain the optimal protocol including W under the average energy conservation.

\section{Single-shot work bound: Exact case}
\label{sec:work_bound_single}

In the remaining part of this chapter, we consider the work bounds for the single-shot scenario, where work fluctuation is not allowed.
This is the quantum extension of the resource-theoretic  argument in Section~\ref{sec:classical_work_bound}.

The following argument is again based on the setup and the notations of Section~\ref{sec:clock_work},
especially  the Hamiltonian of SCW (\ref{Hamiltonian_SCW}).
In addition, in the case of the single-shot scenario, it is natural to take the work storage W as a two-level system, whose Hamiltonian  is given by
\begin{equation}
\hat H_{\rm W} = E_{\rm i} | {\rm i} \rangle \langle {\rm i} |+  E_{\rm f}  |{\rm f} \rangle \langle {\rm f} |, \ \ \ E_{\rm i} - E_{\rm f} =w,
\label{Hamiltonian_W}
\end{equation}
where $w$ represents the single-shot work cost.
In the following, we assume this form of $\hat H_{\rm W}$.
Then, we define  thermodynamic processes of the single-shot setup with work cost $w$ as follows.

\begin{definition}[Single-shot work-assisted state transformation]
Let $w \in \mathbb R$.
A state $\hat \rho_{\rm S}$ is $w$-assisted transformable to another state $\hat \rho_{\rm S}'$ with the initial and final Hamiltonians $\hat H_{\rm S}$ and $\hat H_{\rm S}'$,
 if there exists a Gibbs-preserving map of SCW, written as $\mathcal E_{\rm SCW}$, with  respect to the Hamiltonian $\hat H_{\rm SCW}$ of Eq.~(\ref{Hamiltonian_SCW}) with $\hat H_{\rm W}$ satisfying Eq.~(\ref{Hamiltonian_W}), such that 
\begin{equation}
\mathcal{E}_{\rm SCW} \left(\hat \rho_{\rm S} \otimes | 0 \rangle \langle 0 | \otimes  | E_{\rm i} \rangle \langle E_{\rm i} |  \right)  =  \hat \rho_{\rm S}' \otimes | 1 \rangle \langle 1 | \otimes | E_{\rm f} \rangle \langle E_{\rm f}  |.
\label{final_state_product}
\end{equation}
\label{def:single_work_trans}
\end{definition}

In the above definition, we required the existence of $\mathcal E_{\rm SCW}$ satisfying Eq.~(\ref{final_state_product}) only for a given initial state $\hat \rho_{\rm S}$, which represents that the function of the clock and the work transfer are supposed to be both perfect at least for $\hat \rho_{\rm S}$.

The following theorem gives the work bounds of the single-shot scenario, which corresponds to Theorem~\ref{thm:asymp0_q}.
This has been shown in Refs.~\cite{Horodecki2013,Aberg2013,Faist2018} (see also Proposition 16 of Ref.~\cite{Sagawa2019}).

\begin{theorem}[Single-shot work bounds]
Let $w \in \mathbb R$.
\begin{description}
\item[(a) Necessary conditions for state conversion:] If $\hat \rho_{\rm S}$ is  $w$-assisted  transformable to $\hat \rho_{\rm S}'$ with the Hamiltonians $\hat H_{\rm S}$ and $\hat H_{\rm S}'$, then
\begin{equation}
\beta (w -  \Delta F_{\rm S} ) \geq S_\alpha (\hat \rho_{\rm S}' \| \hat \rho_{\rm S}^{\rm G}{}') - S_\alpha (\hat \rho_{\rm S} \| \hat \rho_{\rm S}^{\rm G}),
\label{theorem_single_work1}
\end{equation}
for $\alpha = 0, 1, \infty$.
\item[(b) Sufficient condition for state conversion:] $\hat \rho_{\rm S}$ is  $w$-assisted  transformable to $\hat \rho_{\rm S}'$ with the Hamiltonians $\hat H_{\rm S}$ and $\hat H_{\rm S}'$, if (but not only if)
\begin{equation}
\beta ( w  - \Delta F_{\rm S}  ) \geq S_\infty (\hat \rho_{\rm S}' \| \hat \rho_{\rm S}^{\rm G}{}') - S_0 (\hat \rho_{\rm S} \| \hat \rho_{\rm S}^{\rm G}).
\label{theorem_single_work2}
\end{equation}
\end{description}
\label{general_single_shot_work}
\end{theorem}

\begin{proof}
Let $\hat \rho := \hat \rho_{\rm S} \otimes | 0 \rangle \langle 0 | \otimes  | E_{\rm i} \rangle \langle E_{\rm i} |$  and $\hat \rho' :=  \mathcal{E}_{\rm SCW}  (\hat \rho) = \hat \rho_{\rm S}' \otimes | 1 \rangle \langle 1 | \otimes | E_{\rm f} \rangle \langle E_{\rm f}  |$ by noting Eq.~(\ref{final_state_product}).
The Gibbs states of SCW and SC are respectively given by Eq.~(\ref{Gibbs_total}) and Eq.~(\ref{Gibbs_SCW}), where in the present setup,
\begin{equation}
\hat \rho_{\rm W}^{\rm G} = \frac{e^{-\beta E_{\rm i}}}{Z_{\rm W}} | E_{\rm i} \rangle \langle E_{\rm i} | + \frac{e^{-\beta E_{\rm f}}}{Z_{\rm W}} | E_{\rm f} \rangle \langle E_{\rm f} |
\end{equation}
with $Z_{\rm W} = e^{-\beta E_{\rm i}} + e^{-\beta E_{\rm f}}$.
Thus, for $\alpha = 0,1, \infty$,
\begin{equation}
S_\alpha ( \hat \rho \| \hat \rho_{\rm SCW}^{\rm G} )  = S_\alpha ( \hat \rho_{\rm S} \| \hat \rho_{\rm S}^{\rm G} )  + \beta F_{\rm S} + \ln ( Z_{\rm S} + Z_{\rm S}' ) + \beta E_{\rm i} + \ln Z_{\rm W},
\label{proof_initial}
\end{equation}
and
\begin{equation}
S_\alpha ( \hat \rho' \| \hat \rho_{\rm SCW}^{\rm G} )  = S_\alpha ( \hat \rho_{\rm S}' \| \hat \rho_{\rm S}^{\rm G}{}' ) + \beta F_{\rm S}' + \ln ( Z_{\rm S} + Z_{\rm S}' ) + \beta E_{\rm f} + \ln Z_{\rm W},
\label{proof_final}
\end{equation}
where we used the scaling property~(\ref{unnormalized_sigma_scaling}).

(a) 
From the monotonicity of the R\'enyi $\alpha$-divergence of SCW, we have $S_\alpha ( \hat \rho \| \hat \rho_{\rm SCW}^{\rm G} ) \geq S_\alpha ( \hat \rho' \| \hat \rho_{\rm SCW}^{\rm G} )$.
By noting Eq.~(\ref{proof_initial}) and Eq.~(\ref{proof_final}), the monotonicity leads to inequality~(\ref{theorem_single_work1}).

(b)
From Theorem~\ref{thm:asymp0_q} (b), $S_\infty ( \hat \rho \| \hat \rho_{\rm SCW}^{\rm G} ) \geq S_0 ( \hat \rho' \| \hat \rho_{\rm SCW}^{\rm G} )$ implies that the transformation from $\hat \rho$ to $\hat \rho'$ is possible. 
This is exactly the claim of Theorem~\ref{general_single_shot_work} (b), again because of Eq.~(\ref{proof_initial}) and Eq.~(\ref{proof_final}).
$\Box$
\end{proof}

By using the $\alpha$-free energy introduced in Eq.~(\ref{alpha_free_energy}), we rewrite inequalities~(\ref{theorem_single_work1}) and (\ref{theorem_single_work2}) in Theorem \ref{general_single_shot_work} as, respectively,
\begin{equation}
w \geq F_\alpha (\hat \rho_{\rm S}'; \hat H_{\rm S}') - F_\alpha (\hat \rho_{\rm S}; \hat H_{\rm S}),
\label{work_bound_alpha_F}
\end{equation}
and 
\begin{equation}
w \geq F_\infty (\hat \rho_{\rm S}'; \hat H_{\rm S}') - F_0 (\hat \rho_{\rm S}; \hat H_{\rm S}).
\label{work_bound_alpha_F_sufficient}
\end{equation}

We next consider the special cases of the above theorem, 
which  reproduce the work bounds of the classical case obtained by the Lorenz curve  in Section~\ref{sec:classical_work_bound}.

\begin{corollary}[Work extraction / state formation]
Let $w \in \mathbb R$.\\
(a) $\hat \rho_{\rm S}$ is  $w$-assisted  transformable to $\hat \rho_{\rm S}^{\rm G}$ with a fixed Hamiltonian  $\hat H_{\rm S}$, if and only if 
\begin{equation}
- \beta w \leq S_0 (\hat \rho_{\rm S} \| \hat \rho_{\rm S}^{\rm G}).
\label{main_single_shot_work_0}
\end{equation}
(b)
$\hat \rho^{\rm G}$ is   $w$-assisted transformable to $\hat \rho_{\rm S}'$  with a fixed Hamiltonian  $\hat H_{\rm S}$, if and only if 
\begin{equation}
\beta w \geq S_\infty (\hat \rho_{\rm S}' \| \hat \rho_{\rm S}^{\rm G}).
\label{main_single_shot_work_infty}
\end{equation}
(c) 
$\hat \rho^{\rm G}$ is   $w$-assisted transformable to $\hat \rho_{\rm S}^{\rm G}{}'$ with initial and final Hamiltonians $\hat H_{\rm S}$ and $\hat H_{\rm S}'$, if and only if 
\begin{equation}
w \geq  \Delta F_{\rm S}.
\label{main_single_shot_work_infty}
\end{equation}
\label{main_single_shot_work}
\end{corollary}

\begin{proof}
To prove  (a) and (b) above, apply Theorem~\ref{general_single_shot_work} to the case that $\hat \rho_{\rm S}' = \hat \rho_{\rm S}^{\rm G}$ and  $\hat \rho_{\rm S} = \hat \rho_{\rm S}^{\rm G}$, respectively.
To prove (c), take  $\hat \rho_{\rm S} = \hat \rho_{\rm S}^{\rm G}$ and  $\hat \rho_{\rm S}' = \hat \rho_{\rm S}^{\rm G}{}'$.
$\Box$
\end{proof}

\section{Single-shot work bound: Approximate case}
\label{sec:work_approximate}

We consider the single-shot work bound for approximate (or imperfect) thermodynamic processes.
There can be several sources of ``failure'' in thermodynamic processes: the final state of the system deviates from the target state, the clock works imperfectly, or the work extraction fails with some probability, and so on.
To incorporate such sources altogether, we adopt the following definition of approximate thermodynamic processes.

\begin{definition}[$\varepsilon$-approximate single-shot thermodynamic process]
Let $w \in \mathbb R$ and $\varepsilon \geq 0$.
A state $\hat \rho_{\rm S}$ is  $\varepsilon$-approximate $w$-assisted  transformable to another state $\hat \rho_{\rm S}'$ with Hamiltonians $\hat H_{\rm S}$ and $\hat H_{\rm S}'$, if there exists  a Gibbs-preserving map of SCW, written as $\mathcal E_{\rm SCW}$, with  respect to  the Hamiltonian $\hat H_{\rm SCW}$ of Eq.~(\ref{Hamiltonian_SCW}) with $\hat H_{\rm W}$ satisfying Eq.~(\ref{Hamiltonian_W}),  such that
\begin{equation}
D ( \mathcal E_{\rm SCW} ( \hat \rho ) , \hat \rho')  \leq  \varepsilon,
\label{approximate_process_trace_distance}
\end{equation} 
where 
$\hat \rho := \hat \rho_{\rm S} \otimes | 0 \rangle \langle 0 | \otimes  | E_{\rm i} \rangle \langle E_{\rm i} | $
and
$\hat \rho' := \hat \rho_{\rm S}' \otimes | 1 \rangle \langle 1 | \otimes | E_{\rm f} \rangle  \langle E_{\rm f}  |$.
\label{def:approx_singleshot_thermo}
\end{definition}

We remark that in Definition~\ref{def:approx_singleshot_thermo} the Hamiltonian of W may include additional levels;  the Hamiltonian (\ref{Hamiltonian_W}) can be replaced by $\hat H_{\rm W} = E_{\rm i} | E_{\rm i} \rangle \langle E_{\rm i} |+  E_{\rm f}  |E_{\rm f} \rangle \langle E_{\rm f} | + \sum_k E_k | E_k \rangle \langle E_k|$ with $E_{\rm i} - E_{\rm f} =w$.
For example, if the work extraction fails, the final state of W might end up with some state other than $ | E_{\rm i} \rangle$,  $| E_{\rm f} \rangle$.

From Definition~\ref{def:approx_singleshot_thermo}, we see that the deviation of the final state of S and the failure probability of CW are both suppressed by $\varepsilon$. 
In fact, from the monotonicity of the trace distance under the partial trace,
the condition~(\ref{approximate_process_trace_distance}) implies that
$D (  \hat \rho_{\rm S}''  , \hat \rho_{\rm S}' ) \leq \varepsilon$
with $\hat \rho_{\rm S}'' := {\rm tr}_{\rm CW}[ \mathcal E_{\rm SCW} ( \hat \rho ) ]$.
The condition~(\ref{approximate_process_trace_distance})  also implies that the success probability of the clock and the work storage (i.e., the probability that C ends up with $| 1\rangle$ and W ends up with $| E_{\rm f} \rangle$) equals or is greater than $1 - \varepsilon$, which is shown as follows.
Let $\hat \tau := \mathcal E_{\rm SCW} ( \hat \rho )$ and $| {\rm f} \rangle :=| 1\rangle | E_{\rm f} \rangle$.
Then, $D ( \hat \tau , \hat \rho')  \leq  \varepsilon$ implies that there exists $\hat \Delta$ such that $\hat \tau \geq \hat \rho' - \hat \Delta$, $\hat \Delta \geq 0$, and ${\rm tr}[\hat \Delta] \leq \varepsilon$.
We then have  $\langle {\rm f }| \hat \tau |  {\rm f} \rangle \geq \hat \rho_{\rm S}' - \langle {\rm f }|  \hat \Delta   |  {\rm f} \rangle$,
and  therefore
\begin{equation}
{\rm tr}[\langle {\rm f }| \hat \tau |  {\rm f} \rangle] \geq  1 -   \varepsilon.
\end{equation}

Another advantage of Definition~\ref{def:approx_singleshot_thermo} is that it is directly related to  the asymptotic theory discussed in Section~\ref{sec:work_asymptotic}.
Under this definition, however, the necessary condition for approximate state conversion is given by an apparently involved inequality with the  hypothesis testing divergence $S_{\rm H}$ introduced in Appendix~\ref{appx:hypothesis_testing}, as shown in the following theorem (Proposition 4 of Ref.~\cite{Sagawa2019}).
(We will discuss  alternative definitions of approximate processes later in this section, where the smooth R\'enyi $0$- and $\infty$-divergences gives simpler thermodynamic bounds, as obtained in Refs.~\cite
{Horodecki2013,Aberg2013}.)

\begin{theorem}[Necessary condition for approximate state conversion]
Let $\varepsilon \geq 0$.
If $\hat \rho_{\rm S}$ is $\varepsilon$-approximate $w$-assisted transformable to  $\hat \rho_{\rm S}'$ with the Hamiltonians $\hat H_{\rm S}$ and $\hat H_{\rm S}'$, then for $0 < \eta < 1 - \varepsilon$,
\begin{equation}
\beta (w - \Delta F_{\rm S}) \geq S_{\rm H}^{\eta + \varepsilon } (\hat \rho_{\rm S}'  \| \hat \rho_{\rm S}^{\rm G}{}' )  - S_{\rm H}^{\eta} (\hat \rho_{\rm S}  \| \hat \rho_{\rm S}^{\rm G}{} ) -\ln \left( \frac{\eta + \varepsilon }{\eta} \right).
\label{hypothesis_testing_work_ap}
\end{equation}
\label{thm:hypothesis_testing_work_ap}
\end{theorem}

\begin{proof}
We use the notations $\hat \rho$ and $\hat \rho'$ in Definition~\ref{def:approx_singleshot_thermo}.
First, we have
\begin{eqnarray}
&{}&S_{\rm H}^{\eta + \varepsilon } (\hat \rho_{\rm S}'  \| \hat \rho_{\rm S}^{\rm G}{}' ) 
= S_{\rm H}^{\eta + \varepsilon } (\hat \rho'  \| \hat \rho_{\rm SCW}^{\rm G}  ) + \ln  \frac{Z_{\rm S}'}{Z_{\rm S} + Z_{\rm S}'}  - \beta E_{\rm f} \\
&\leq& S_{\rm H}^{\eta } ( \mathcal E_{\rm SCW} (\hat \rho) \| \hat \rho_{\rm SCW}^{\rm G}  ) + \ln \left( \frac{\eta + \varepsilon }{\eta} \right) + \ln  \frac{Z_{\rm S}'}{Z_{\rm S} + Z_{\rm S}'}  - \beta E_{\rm f},
\end{eqnarray}
where the equality in the first line is from Lemma~\ref{lemma:Sagawa_Faist_eq} and the inequality in the second line is from Lemma~\ref{lemma:Sagawa_Faist_ineq}.
From the monotonicity under $\mathcal{E}_{\rm SCW} $,
\begin{equation}
S_{\rm H}^{\eta } ( \mathcal E_{\rm SCW} (\hat \rho)   \| \hat \rho_{\rm SCW}^{\rm G}  ) \leq S_{\rm H}^{\eta } (  \hat \rho     \| \hat \rho_{\rm SCW}^{\rm G}  ).
\end{equation}
Again from Lemma~\ref{lemma:Sagawa_Faist_eq}, we have
\begin{equation}
S_{\rm H}^{\eta } (  \hat \rho     \| \hat \rho_{\rm SCW}^{\rm G}  ) = S_{\rm H}^{\eta} (\hat \rho_{\rm S}  \| \hat \rho_{\rm S}^{\rm G}{} ) - \ln  \frac{Z_{\rm S}}{Z_{\rm S} + Z_{\rm S}'}  + \beta E_{\rm i}.
\end{equation}
By combining the foregoing relations, we obtain
\begin{equation}
S_{\rm H}^{\eta + \varepsilon } (\hat \rho_{\rm S}'  \| \hat \rho_{\rm S}^{\rm G}{}' )  \leq S_{\rm H}^{\eta} (\hat \rho_{\rm S}  \| \hat \rho_{\rm S}^{\rm G}{} )  + \ln \frac{Z_{\rm S}'}{Z_{\rm S}} + \beta (E_{\rm i} - E_{\rm f})+ \ln \left( \frac{\eta + \varepsilon }{\eta} \right),
\end{equation}
which implies inequality~(\ref{hypothesis_testing_work_ap}).
$\Box$
\end{proof}

By using Proposition~\ref{Faist_prop}, inequality~(\ref{hypothesis_testing_work_ap}) can be rewritten in terms of the smooth R\'enyi $0$- and $\infty$-divergences, while we will not explicitly write it down here, as it is still involved.

We next consider a sufficient condition for state conversion, which is an approximate version of  Theorem~\ref{general_single_shot_work} (b) and now is given by a simple form (Proposition 16 of Ref.~\cite{Sagawa2019}).

\begin{theorem}[Sufficient condition for approximate state conversion]
$\hat \rho_{\rm S}$ is $\varepsilon$-approximate $w$-assisted  transformable to $\hat \rho_{\rm S}'$ with the Hamiltonians $\hat H_{\rm S}$ and $\hat H_{\rm S}'$, if (but not only if)
\begin{equation}
\beta ( w  - \Delta F_{\rm S}  ) \geq S_\infty^{\varepsilon / 2} (\hat \rho_{\rm S}' \| \hat \rho_{\rm S}^{\rm G}{}') - S_0^{\varepsilon / 2} (\hat \rho_{\rm S} \| \hat \rho_{\rm S}^{\rm G}).
\label{approximate_sufficient}
\end{equation}
\label{thm:approximate_sufficient}
\end{theorem}

\begin{proof}
Let $\hat \tau_{\rm S}$ be optimal for $S_0^{\varepsilon / 2} (\hat \rho_{\rm S} \| \hat \rho_{\rm S}^{\rm G})$ and  $\hat \tau_{\rm S}'$ be optimal for $S_\infty^{\varepsilon / 2} (\hat \rho_{\rm S}' \| \hat \rho_{\rm S}^{\rm G}{}')$.
We then have 
\begin{equation}
\beta ( w  - \Delta F_{\rm S}  ) \geq S_\infty (\hat \tau_{\rm S}' \| \hat \rho_{\rm S}^{\rm G}{}') - S_0 (\hat \tau_{\rm S} \| \hat \rho_{\rm S}^{\rm G}).
\label{approximate_sufficient}
\end{equation}
From this inequality and Theorem~\ref{general_single_shot_work} (b), $\hat \tau_{\rm S}$ is $w$-assisted transformable to $\hat \tau_{\rm S}'$.
Define $\hat \tau:= \hat \tau_{\rm S} \otimes | 0 \rangle \langle 0 | \otimes | E_{\rm i} \rangle \langle E_{\rm i} |$ and $\hat \tau' := \hat \tau_{\rm S}' \otimes | 1 \rangle \langle 1 | \otimes | E_{\rm f} \rangle \langle E_{\rm f} |$,  and note that $D(\hat \rho, \hat \tau )\leq \varepsilon / 2$ and $ D(\hat \rho', \hat \tau') \leq \varepsilon / 2$. 
Then, there exists a Gibbs-preserving map $\mathcal E_{\rm SCW}$ such that $\mathcal E_{\rm SCW} (\hat \tau ) = \hat \tau'$.
Then, $D(\mathcal E_{\rm SCW} (\hat \rho), \hat \rho' ) \leq D(\mathcal E_{\rm SCW} (\hat \rho), \hat \tau' ) + D(\hat \tau', \hat \rho') \leq D(\hat \rho, \hat \tau ) + D(\hat \tau', \hat \rho') \leq \varepsilon$.
$\Box$
\end{proof}

As mentioned before, we can also introduce alternative definitions of approximate processes, which are given by slightly stronger conditions  than Definition~\ref{def:approx_singleshot_thermo}, leading to simpler work bounds.
We here define two kinds of approximate processes, which we label by $\alpha = 0, \infty$, having in mind that these two definitions correspond to the R\'enyi $0$- and $\infty$-divergences.

\begin{definition}[Strongly $\varepsilon$-approximate single-shot thermodynamic process]
Let $w \in \mathbb R$ and $\varepsilon \geq 0$.  
A state $\hat \rho_{\rm S}$ is $\alpha$-strongly $\varepsilon$-approximate $w$-assisted  transformable to another state $\hat \rho_{\rm S}'$ with Hamiltonians $\hat H_{\rm S}$ and $\hat H_{\rm S}'$, if:
\begin{description}
\item[$\alpha = 0$;] there exists a state $\hat \tau_{\rm S}$ satisfying $D(\hat \tau_{\rm S}, \hat \rho_{\rm S}) \leq \varepsilon$ such that $\hat \tau_{\rm S}$ is $w$-assisted transformable to $\hat \rho_{\rm S}'$.
\item[$\alpha=\infty$;]  there exists a state $\hat \tau_{\rm S}'$ satisfying $D(\hat \tau_{\rm S}', \hat \rho_{\rm S}') \leq \varepsilon$ such that $\hat \rho_{\rm S}$ is $w$-assisted transformable to $\hat \tau_{\rm S}'$.
\end{description}
\label{def:approx_singleshot_thermo_strong}
\end{definition}

Here, the roles of the initial and final states are exchanged for $\alpha = 0, \infty$; in this sense, these two definitions are dual to each other.  More importantly, for both of  $\alpha = 0,1$, ``$\alpha$-strongly $\varepsilon$-approximate transformable'' automatically implies  ``$\varepsilon$-approximate transformable''  (Definition~\ref{def:approx_singleshot_thermo}).

To see this, we set the notations $\hat \rho := \hat \rho_{\rm S} \otimes | 1 \rangle \langle 1 | \otimes | E_{\rm f} \rangle  \langle E_{\rm f}  |$, 
$\hat \tau := \hat \tau_{\rm S} \otimes | 0 \rangle \langle 0 | \otimes  | E_{\rm i} \rangle \langle E_{\rm i} |$, $\hat \rho' := \hat \rho_{\rm S}' \otimes | 1 \rangle \langle 1 | \otimes | E_{\rm f} \rangle  \langle E_{\rm f}  |$, $\hat \tau' := \hat \tau_{\rm S}' \otimes | 1 \rangle \langle 1 | \otimes | E_{\rm f} \rangle  \langle E_{\rm f}  |$,
 and let $\mathcal E_{\rm SCW}$ be the relevant Gibbs preserving map  with  respect to  the Hamiltonian $\hat H_{\rm SCW}$  of Eq.~(\ref{Hamiltonian_SCW}), satisfying
$ \mathcal E_{\rm SCW} ( \hat \tau )  = \hat \rho'$ for $\alpha = 0$ or $ \mathcal E_{\rm SCW} ( \hat \rho )  = \hat \tau'$ for $\alpha = \infty$ .
For $\alpha = 0$, if the condition of Definition~\ref{def:approx_singleshot_thermo_strong} is satisfied, then $D(\mathcal E_{\rm SCW} (\hat \rho), \hat \rho') \leq D(\hat \rho, \hat \tau ) = D(\hat \rho_{\rm S}, \hat \tau_{\rm S} ) \leq \varepsilon$ holds from the monotonicity of the trace distance.
For $\alpha = \infty$, if the condition of Definition~\ref{def:approx_singleshot_thermo_strong} is satisfied, then
obviously $D(\mathcal E_{\rm SCW} (\hat \rho), \hat \rho') = D(\hat \tau_{\rm S}', \hat \rho_{\rm S}') \leq \varepsilon$.

A physical motivation behind Definition~\ref{def:approx_singleshot_thermo_strong} can be illustrated as follows. 
For $\alpha = 0$, suppose that $\hat \rho_{\rm S}$ can be written as $\hat \rho_{\rm S} = (1 - \varepsilon ) \hat \tau_{\rm S} + \varepsilon \hat \gamma_{\rm S}$, where $0< \varepsilon < 1$ is some constant and $\hat \tau_{\rm S}$ and $\hat \gamma_{\rm S}$ are normalized states orthogonal to each other.
Note that $D(\hat \tau_{\rm S}, \hat \rho_{\rm S}) \leq \varepsilon$.
Also suppose that the condition of Definition~\ref{def:approx_singleshot_thermo_strong} is satisfied for this $\hat \tau_{\rm S}$ (i.e., $\mathcal E_{\rm SCW} (\hat \tau ) = \hat \rho'$ with the above notation).
In such a case,  $\hat \tau_{\rm S}$ represents the ``success event'' in the initial state of the transformation with success probability $1-\varepsilon$.
On the other hand, for $\alpha = \infty$, Definition~\ref{def:approx_singleshot_thermo_strong} simply means that state transformation of S is imperfect but CW goes to the desired final state with unit probability.

With the foregoing ``strong'' definitions of approximate transformation, we have the following simpler work bounds.
They are approximate versions of Corollary~\ref{main_single_shot_work} (a), (b), and
can be regarded as essentially the same bounds as obtained in Refs.~\cite
{Horodecki2013,Aberg2013}.
(See also Refs.~\cite{Dahlsten2011,Rio2011,Faist2015c,Loomis2020} for the related work bounds for information processing.)

\begin{theorem}[Simple work bounds for strongly approximate processes]
Let $w \in \mathbb R$ and $\varepsilon \geq 0$. Consider states $\hat \rho_{\rm S}$, $\hat \rho_{\rm S}'$ and a fixed Hamiltonian  $\hat H_{\rm S}$ with the Gibbs state $\hat \rho_{\rm S}^{\rm G}$.
\begin{description}
\item[(a)]
$\hat \rho_{\rm S}$ is $0$-strongly $\varepsilon$-approximate $w$-assisted  transformable to  $\hat \rho_{\rm S}^{\rm G}$, if and only if
\begin{equation}
-\beta w \leq S_0^\varepsilon ( \hat \rho_{\rm S} \| \hat \rho_{\rm S}^{\rm G} ).
\label{Aberg_inequality0}
\end{equation}
\item[(b)]
$\hat \rho_{\rm S}^{\rm G}$ is $\infty$-strongly $\varepsilon$-approximate $w$-assisted  transformable to  $\hat \rho_{\rm S}'$, if and only if
\begin{equation}
\beta w \geq S_\infty^\varepsilon ( \hat \rho_{\rm S}' \| \hat \rho_{\rm S}^{\rm G} ).
\label{Aberg_inequality_inf}
\end{equation}
\end{description}
\label{thm:Aberg_inequality}
\end{theorem}

\begin{proof}
We first prove (a).
For the sake of generality at this stage, suppose that $\hat \rho_{\rm S}$ is $0$-strongly $\varepsilon$-approximate $w$-assisted  transformable to  $\hat \rho_{\rm S}'$, by allowing different initial and final Hamiltonians $\hat H_{\rm S}$ and $\hat H_{\rm S}'$.
From inequality~(\ref{theorem_single_work1}) of Theorem~\ref{general_single_shot_work} (a), we have $\beta (w -  \Delta F_{\rm S} ) \geq S_0 (\hat \rho_{\rm S}' \| \hat \rho_{\rm S}^{\rm G}{}') - S_0 (\hat \tau_{\rm S} \| \hat \rho_{\rm S}^{\rm G})$.
Meanwhile, because $\hat \tau_{\rm S}$ is a candidate for optimization in $S_0^\varepsilon (\hat \rho_{\rm S} \| \hat \rho_{\rm S}^{\rm G})$, we have $S_0 (\hat \tau_{\rm S} \| \hat \rho_{\rm S}^{\rm G} ) \leq S_0^\varepsilon (\hat \rho_{\rm S} \| \hat \rho_{\rm S}^{\rm G})$.
We thus have
\begin{equation}
\beta (w -  \Delta F_{\rm S} ) \geq S_0 (\hat \rho_{\rm S}' \| \hat \rho_{\rm S}^{\rm G}{}') - S_0^\varepsilon (\hat \rho_{\rm S} \| \hat \rho_{\rm S}^{\rm G}).
\end{equation}
By letting $\hat H_{\rm S} = \hat H_{\rm S}'$ and $\hat \rho_{\rm S}' = \hat \rho_{\rm S}^{\rm G}$, we obtain inequality~(\ref{Aberg_inequality0}).

To prove the converse, suppose that 
\begin{equation}
\beta (w -  \Delta F_{\rm S} )  \geq S_\infty (\hat \rho_{\rm S}' \| \hat \rho_{\rm S}^{\rm G}{}') - S_0^\varepsilon (\hat \rho_{\rm S} \| \hat \rho_{\rm S}^{\rm G} )
\end{equation}
for general  $\hat \rho_{\rm S}$, $\hat \rho_{\rm S}'$, $\hat H_{\rm S}$, $\hat H_{\rm S}'$.
Let $\hat \tau_{\rm S}$ be optimal for $S_0^\varepsilon (\hat \rho_{\rm S} \| \hat \rho_{\rm S}^{\rm G} )$, which satisfies $S_0^\varepsilon (\hat \rho_{\rm S} \| \hat \rho_{\rm S}^{\rm G} ) = S_0 (\hat \tau_{\rm S} \| \hat \rho_{\rm S}^{\rm G} )$ and  $D( \hat \tau_{\rm S}, \hat \rho_{\rm S} ) \leq \varepsilon$.
We  have
$\beta (w -  \Delta F_{\rm S} )  \geq S_\infty (\hat \rho_{\rm S}' \| \hat \rho_{\rm S}^{\rm G}{}') - S_0 (\hat \tau_{\rm S} \| \hat \rho_{\rm S}^{\rm G} )$.
Then, from Theorem~\ref{general_single_shot_work} (b),
$\hat \tau_{\rm S}$ is  $w$-assisted  transformable to $\hat \rho_{\rm S}'$,
which implies that $\hat \rho_{\rm S}$ is $0$-strongly $\varepsilon$-approximate $w$-assisted  transformable to  $\hat \rho_{\rm S}'$.
By letting $\hat H_{\rm S} = \hat H_{\rm S}'$ and $\hat \rho_{\rm S}' = \hat \rho_{\rm S}^{\rm G}$, we obtain the claim of (a).

We next prove (b).  Again for the sake of generality at this stage, suppose that $\hat \rho_{\rm S}$ is $\infty$-strongly $\varepsilon$-approximate $w$-assisted  transformable to  $\hat \rho_{\rm S}'$, by allowing different initial and final Hamiltonians $\hat H_{\rm S}$ and $\hat H_{\rm S}'$.
From inequality~(\ref{theorem_single_work1}) of Theorem~\ref{general_single_shot_work} (a), we have $\beta (w -  \Delta F_{\rm S} ) \geq S_\infty (\hat \tau_{\rm S}' \| \hat \rho_{\rm S}^{\rm G}{}') - S_\infty (\hat \rho_{\rm S} \| \hat \rho_{\rm S}^{\rm G})$.
Meanwhile, because $\hat \tau_{\rm S}'$ is a candidate for optimization in $S_\infty^\varepsilon (\hat \rho_{\rm S}' \| \hat \rho_{\rm S}^{\rm G}{}')$, we have $S_\infty (\hat \tau_{\rm S}' \| \hat \rho_{\rm S}^{\rm G}{}' ) \geq S_\infty^\varepsilon (\hat \rho_{\rm S}' \| \hat \rho_{\rm S}^{\rm G}{}')$.
We thus have
\begin{equation}
\beta (w -  \Delta F_{\rm S} ) \geq S_\infty^\varepsilon (\hat \rho_{\rm S}' \| \hat \rho_{\rm S}^{\rm G}{}') - S_\infty (\hat \rho_{\rm S} \| \hat \rho_{\rm S}^{\rm G}).
\end{equation}
By letting $\hat H_{\rm S} = \hat H_{\rm S}'$ and $\hat \rho_{\rm S} = \hat \rho_{\rm S}^{\rm G}$, we obtain inequality~(\ref{Aberg_inequality_inf}).

To prove the converse, suppose that 
\begin{equation}
\beta (w -  \Delta F_{\rm S} )  \geq S_\infty^\varepsilon (\hat \rho_{\rm S}' \| \hat \rho_{\rm S}^{\rm G}{}') - S_0 (\hat \rho_{\rm S} \| \hat \rho_{\rm S}^{\rm G} )
\end{equation}
for general  $\hat \rho_{\rm S}$, $\hat \rho_{\rm S}'$, $\hat H_{\rm S}$, $\hat H_{\rm S}'$.
Let $\hat \tau_{\rm S}'$ be  optimal for $S_0^\varepsilon (\hat \rho_{\rm S}' \| \hat \rho_{\rm S}^{\rm G}{}' )$, which satisfies $S_0^\varepsilon (\hat \rho_{\rm S}' \| \hat \rho_{\rm S}^{\rm G}{}' ) = S_0 (\hat \tau_{\rm S}' \| \hat \rho_{\rm S}^{\rm G}{}' )$ and  $D( \hat \tau_{\rm S}', \hat \rho_{\rm S}' ) \leq \varepsilon$.
We  have
$\beta (w -  \Delta F_{\rm S} )  \geq S_\infty (\hat \tau_{\rm S}' \| \hat \rho_{\rm S}^{\rm G}{}') - S_0 (\hat \rho_{\rm S} \| \hat \rho_{\rm S}^{\rm G} )$.
Then, from Theorem~\ref{general_single_shot_work} (b),
$\hat \rho_{\rm S}$ is  $w$-assisted  transformable to $\hat \tau_{\rm S}'$,
which implies that $\hat \rho_{\rm S}$ is $\infty$-strongly $\varepsilon$-approximate $w$-assisted  transformable to  $\hat \rho_{\rm S}'$.
By letting $\hat H_{\rm S} = \hat H_{\rm S}'$ and $\hat \rho_{\rm S} = \hat \rho_{\rm S}^{\rm G}$, we obtain the claim of (a).
$\Box$
\end{proof}

A protocol that saturates the equality in inequality~(\ref{Aberg_inequality0}) is presented in Ref.~\cite{Aberg2013}, which is a slight modification of the protocol of Fig.~\ref{fig:equality_protocol_Renyi0}  in Section~\ref{sec:classical_work_bound}, by allowing failure of work extraction with small probability ($\leq \varepsilon$).
We also note that, by substituting $\hat \rho_{\rm S} = \hat \rho_{\rm S}^{\rm G}$, inequality~(\ref{Aberg_inequality0}) suggests  that a positive amount of work can be extracted even from the Gibbs state if a small probability of failure is allowed.


\section{Single-shot work bound: Asymptotic case}
\label{sec:work_asymptotic}

We consider  the single-shot work bound   for macroscopic systems, by taking the asymptotic limit of approximate thermodynamic processes in the sense of Definition~\ref{def:approx_singleshot_thermo}.
This is  an application of the general asymptotic theory developed in Section~\ref{sec:asymptotic} and Section~\ref{sec_condition_information}.
In terms of thermodynamics, the asymptotic limit represents the thermodynamic limit of many-body systems, where we do not necessarily assume the i.i.d. setup.
Physically, i.i.d. systems represent non-interacting systems, while we can treat interacting systems  in the following argument.
In particular, we show that if the state is ergodic and the Hamiltonian is local and translation invariant, then the KL divergence provides the work bound in a necessary and sufficient manner, which is a thermodynamic consequence of the quantum AEP (Theorem~\ref{thm:main_ergodic}).

Let $\widehat{P}_{\rm S} := \{ \hat \rho_{{\rm S}, n} \}_{n \in \mathbb N}$ and $\widehat{P}_{\rm S}' := \{ \hat \rho_{{\rm S}, n}' \}_{n \in \mathbb N}$ be sequences of states of S and  $\widehat{H}_{\rm S} := \{ \hat H_{{\rm S}, n} \}_{n \in \mathbb N}$ and $\widehat{H}_{\rm S}' := \{ \hat H_{{\rm S}, n}' \}_{n \in \mathbb N}$ be sequences of Hamiltonians of S.
Let $\widehat{\Sigma}_{\rm S} := \{ \hat \rho_{{\rm S},n}^{\rm G} \}_{n \in \mathbb N}$ and 
$\widehat{\Sigma}_{\rm S}' := \{ \hat \rho_{{\rm S},n}^{\rm G}{}' \}_{n \in \mathbb N}$ be the corresponding sequences of the Gibbs states.
Suppose that  the equilibrium free-energy rates exist,  defined as
\begin{equation}
F_{\rm S} := \lim_{n \to \infty} \frac{1}{n} F_{{\rm S},n}, \ \ \ F_{\rm S}' := \lim_{n \to \infty} \frac{1}{n} F_{{\rm S},n}',
\end{equation} 
where $F_{{\rm S},n}$ and $F_{{\rm S},n}'$ are the free energies corresponding to $\hat H_{{\rm S}, n}$ and  $\hat H_{{\rm S}, n}'$, respectively.
Let $\Delta F_{{\rm S},n} :=  F_{{\rm S},n}' -  F_{{\rm S},n}$ and $\Delta F_{\rm S} := F_{\rm S}' - F_{\rm S}$.
We then define the asymptotic limit of single-shot thermodynamic processes.

\begin{definition}[Asymptotic thermodynamic process]
Let $w \in \mathbb R$.
A sequence $\widehat{P}_{\rm S}$ is asymptotically  $w$-assisted  transformable to another sequence $\widehat{P}_{\rm S}'$ with respect to $\widehat{H}_{\rm S}$ and $\widehat{H}_{\rm S}'$, 
if there exist sequences $\{ w_n \}_{n \in \mathbb N}$ and $\{ \varepsilon_n \}_{n \in \mathbb N}$ with $w_n \in \mathbb R$ and $\varepsilon_n > 0$,
such that
$\hat \rho_{{\rm S}, n}$ is  $\varepsilon_n$-approximate $w_n$-assisted transformable to $\hat \rho_{{\rm S}, n}'$, and $\lim_{n \to \infty}w_n / n = w$, $\lim_{n \to \infty} \varepsilon_n = 0$.
\label{def:asymptotic_work}
\end{definition}

Based on this definition, we consider the work bounds.  We first state the necessary conditions corresponding to Theorem~\ref{general_single_shot_work} (a) [Proposition 5 of Ref.~\cite{Sagawa2019}].

\begin{theorem}[Necessary condition for asymptotic state conversion]
If $\widehat{P}_{\rm S}$ is asymptotically  $w$-assisted transformable to $\widehat{P}_{\rm S}'$ with respect to $\widehat{H}_{\rm S}$ and $\widehat{H}_{\rm S}'$, then
\begin{equation}
\beta ( w - \Delta F_{\rm S}) \geq \underline{S} (\widehat{P}_{\rm S}' \| \widehat{\Sigma}_{\rm S}') -  \underline{S} (\widehat{P}_{\rm S} \| \widehat{\Sigma}_{\rm S}),
\end{equation}
\begin{equation}
\beta ( w - \Delta F_{\rm S}) \geq \overline{S} (\widehat{P}_{\rm S}' \| \widehat{\Sigma}_{\rm S}') -  \overline{S} (\widehat{P}_{\rm S} \| \widehat{\Sigma}_{\rm S}).
\end{equation}
\end{theorem}

\begin{proof}
Take the limit of Theorem~\ref{thm:hypothesis_testing_work_ap}, by applying Proposition~\ref{Faist_prop}.
$\Box$
\end{proof}

We next state the sufficient condition of state conversion corresponding to Theorem~\ref{general_single_shot_work} (b) [Proposition 16 of Ref.~\cite{Sagawa2019}].

\begin{theorem}[Sufficient condition for asymptotic state conversion]
$\widehat{P}_{\rm S}$ is asymptotically  $w$-assisted transformable to $\widehat{P}_{\rm S}'$ with respect to $\widehat{H}_{\rm S}$ and $\widehat{H}_{\rm S}'$, if (but not only if)
\begin{equation}
\beta ( w - \Delta F_{\rm S}) \geq \overline{S} (\widehat{P}_{\rm S}' \| \widehat{\Sigma}_{\rm S}') -  \underline{S} (\widehat{P}_{\rm S} \| \widehat{\Sigma}_{\rm S}).
\label{spectral_sufficient_w}
\end{equation}
\label{thm:spectral_sufficient_w}
\end{theorem}

\begin{proof}
Suppose that inequality~(\ref{spectral_sufficient_w}) holds.
Then, for any $\varepsilon > 0$ and $n \in \mathbb N$, there exists $\Delta_{n, \varepsilon} \geq  0$ such that 
\begin{equation}
\frac{\beta}{n} (n w - \Delta F_{{\rm S}, n }) + \Delta_{n, \varepsilon} \geq \frac{1}{n} S_\infty^{\varepsilon / 2} ( \hat \rho_{{\rm S}, n}'  \| \hat \rho_{{\rm S},n}^{\rm G}{}'  ) -  \frac{1}{n} S_0^{\varepsilon / 2} ( \hat \rho_{{\rm S}, n} \| \hat \rho_{{\rm S},n}^{\rm G} )
\end{equation}
and 
\begin{equation}
\lim_{\varepsilon \to + 0} \limsup_{n \to \infty} \Delta_{n,\varepsilon} = 0.
\label{proof_Delta_lim}
\end{equation}
Let $w_{n, \varepsilon} := n w + n\Delta_{n, \varepsilon}$ that satisfies $\lim_{\varepsilon \to + 0} \limsup_{n \to \infty} w_{n, \varepsilon} / n = w$.
From Theorem~\ref{thm:approximate_sufficient},  $\hat \rho_{{\rm S}, n}$ is $\varepsilon$-approximate $w_{n,\varepsilon}$-assisted  transformable to $\hat \rho_{{\rm S}, n}'$.
From Lemma 13 of Ref.~\cite{Sagawa2019},  this implies that $\widehat{P}_{\rm S}$ is asymptotically  $w$-assisted  transformable to $\widehat{P}_{\rm S}'$.

For the sake of self-containedness, we here reproduce the proof of the above last step.
Let $w_\varepsilon := \limsup_{n \to \infty} w_{n, \varepsilon}/n$,
 $N(\varepsilon ) := \min \{  N  \  :  \  \forall n \geq N, w_{n,\varepsilon}  / n \leq w_\varepsilon + \varepsilon \}$,
and  $\varepsilon (n) := \inf \{ \varepsilon  :  N(\varepsilon ) \leq n \}$.
Because of the existence of the limit superior, $N(\varepsilon)$ is finite for any $\varepsilon > 0$, and thus we see that $\lim_{n \to \infty} \varepsilon (n) = 0$.
Then, by defining $w_n' := w_{n, \varepsilon (n)}$, we obtain
$\limsup_{n \to \infty} w_n'/n \leq \limsup_{n \to \infty}  ( w_{\varepsilon (n)} + \varepsilon (n) ) = w$.
Because the state convertibility is not affected by adding any positive amount of work, we can construct $w_n$ such that $\lim_{n \to \infty} w_n = w$, by adding some positive constant to $w_n'$ if necessary. 
$\Box$
\end{proof}

In contrast to Theorem~\ref{thm:spectral_sufficient}, inequality~(\ref{spectral_sufficient_w})  in the above theorem includes the equality case, thanks to the role of $w$ in Definition~\ref{def:asymptotic_work}.
The following are the two special cases corresponding to Corollary~\ref{main_single_shot_work} (a) (b).

\begin{corollary}[Asymptotic work extraction / state formation]
Let $w \in \mathbb R$.\\
(a) $\widehat{P}_{\rm S}$ is  asymptotically $w$-assisted transformable to $\widehat{\Sigma}_{\rm S}$ with  $\widehat{H}_{\rm S} = \widehat{H}_{\rm S}'$, if and only if 
\begin{equation}
- \beta w \leq \underline{S} (\widehat{P}_{\rm S} \| \widehat{\Sigma}_{\rm S}).
\end{equation}
(b)
$\widehat{\Sigma}_{\rm S}$ is asymptotically   $w$-assisted  transformable to $\widehat{P}_{\rm S}'$  with  $\widehat{H}_{\rm S} = \widehat{H}_{\rm S}'$, if and only if 
\begin{equation}
\beta w \geq\overline{S} (\widehat{P}_{\rm S}' \| \widehat{\Sigma}_{\rm S}) .
\end{equation}
\end{corollary}

Finally, we state a necessary and sufficient characterization of state conversion in terms of the work bound for the case where  the upper and lower spectral divergence rates collapse to a single value for the initial and final states.

\begin{corollary}
Suppose that $\overline{S} (\widehat{P}_{\rm S} \| \widehat{\Sigma}_{\rm S}) = \underline{S} (\widehat{P}_{\rm S} \| \widehat{\Sigma}_{\rm S}) =: S (\widehat{P}_{\rm S} \| \widehat{\Sigma}_{\rm S})$ and $\overline{S} (\widehat{P}_{\rm S}' \| \widehat{\Sigma}_{\rm S}') = \underline{S} (\widehat{P}_{\rm S}' \| \widehat{\Sigma}_{\rm S}') =: S (\widehat{P}_{\rm S}' \| \widehat{\Sigma}_{\rm S}') $. 
Then, $\widehat{P}_{\rm S}$ is asymptotically  $w$-assisted transformable to $\widehat{P}_{\rm S}'$ with respect to $\widehat{H}_{\rm S}$ and $\widehat{H}_{\rm S}'$, if and only if
\begin{equation}
\beta ( w - \Delta F_{\rm S}) \geq S (\widehat{P}_{\rm S}' \| \widehat{\Sigma}_{\rm S}') -  S (\widehat{P}_{\rm S} \| \widehat{\Sigma}_{\rm S}).
\end{equation}
\label{spectral_sufficient_w_cor}
\end{corollary}

The above corollary implies that the upper and lower spectral divergence rates give a complete thermodynamic potential that can be defined for out-of-equilibrium situations. 
Correspondingly, we can introduce the nonequilibrium free-energy rate by
\begin{equation}
F(\widehat{P}_{\rm S}; \widehat{H}_{\rm S}) := \beta^{-1}S (\widehat{P}_{\rm S} \| \widehat{\Sigma}_{\rm S}) + F_{\rm S}.
\end{equation}
This illustrates the  significance of information spectrum in thermodynamics.

We now consider a many-body quantum spin system on a lattice $\mathbb Z^d$ as in Section \ref{sec_condition_information}, where
 an explicit condition of the collapse of the upper and lower spectral divergence rates has been shown in  Theorem~\ref{thm:main_ergodic}.
Combining it with Corollary~\ref{spectral_sufficient_w_cor} above, we have the following asymptotic work bound, which states that the complete thermodynamic potential emerges and is given by the KL divergence rate even for out-of-equilibrium and quantum situations, if the state is ergodic and the Hamiltonian is local and translation invariant.

\begin{corollary}
Suppose that $\widehat{P}_{\rm S}$ and $\widehat{P}_{\rm S}'$ are translation invariant and ergodic, and that the Hamiltonians $\widehat{H}_{\rm S}$ and $\widehat{H}_{\rm S}'$ are local and translation invariant with the corresponding Gibbs states $\widehat{\Sigma}_{\rm S}$ and $\widehat{\Sigma}_{\rm S}'$.
Then,  $\widehat{P}_{\rm S}$ is asymptotically  $w$-assisted  transformable to $\widehat{P}_{\rm S}'$ with respect to $\widehat{H}_{\rm S}$ and $\widehat{H}_{\rm S}'$, if and only if
\begin{equation}
\beta ( w - \Delta F_{\rm S}) \geq S_1 (\widehat{P}_{\rm S}' \| \widehat{\Sigma}_{\rm S}') -  S_1 (\widehat{P}_{\rm S} \| \widehat{\Sigma}_{\rm S}).
\end{equation}
\label{co:q_themro_macro}
\end{corollary}

This may provide  an information-theoretic and statistical-mechanical foundation of, and furthermore a nonequilibrium generalization of, the phenomenological  thermodynamics of Lieb and Yngvason~\cite{Lieb1999} where the entropy is a  complete monotone of equilibrium transitions.

The emergence of the complete thermodynamic potential ensures the reversibility of thermodynamic transformations in the macroscopic limit.
As an illustrative situation, let us consider a simple cycle in the single-shot scenario: one first transforms  an equilibrium state $\hat \rho_{\rm S}^{\rm G}$ to a nonequilibrium state $\hat \rho_{\rm S}$, and then restores $\hat \rho_{\rm S}$ to $\hat \rho_{\rm S}^{\rm G}$.  As shown in Corollary \ref{main_single_shot_work}, this cycle requires work of $S_\infty (\hat \rho_{\rm S} \| \hat \rho_{\rm S}^{\rm G} ) - S_0 (\hat \rho_{\rm S} \| \hat \rho_{\rm S}^{\rm G} )$, which is positive if $\hat \rho_{\rm S}$ is out of equilibrium.
This implies that it is in general impossible to obtain a single thermodynamic potential that can completely characterize state convertibility in the single-shot scenario.
This is also contrastive to conventional thermodynamics for equilibrium transitions, in which a crucial postulate is that one can perform a reversible cyclic operation without remaining any effect on the outside world.
On the other hand, Corollary \ref{co:q_themro_macro} reveals that a thermodynamic potential emerges in the asymptotic limit, which is related to the concept of reversibility  in resource theory~\cite{Horodecki2002,Brandao2015b}.

We also note that  one can see a characteristic of the single-shot scenario in light of the asymptotic theory;  Thanks to the single-shot formulation, the work automatically becomes a deterministic quantity in the asymptotic limit, which is a desirable property in macroscopic thermodynamics.

Furthermore, as shown in Refs.~\cite{Sagawa2019,Faist2019}, the operation in Corollary \ref{co:q_themro_macro} can be replaced by an asymptotic thermal operation even in the fully quantum case, if the aid of a small amount of quantum coherence is available (Theorem 2 of ~Ref.~\cite{Sagawa2019}).
This implies that thermal operation can work even in the fully quantum regime, if we take the asymptotic limit.
The key of the proof is the fact that any ergodic state (more generally, any state with which the upper and lower spectral divergence rates are close) has a small coherence in the energy basis (Lemma 4 of Ref.~\cite{Sagawa2019}).

\section{Trace-nonincreasing formulation}
\label{sec:trace_nonincreasing}

As a side remark, we consider a way to ``trace out'' the clock degrees of freedom from our formulation introduced in Section~\ref{sec:clock_work}, where the notion of trace-nonincreasing naturally appears as a consequence of the assumption that the clock does not necessarily work perfectly.

For simplicity, we ignore  the work storage W and only focus on the role of the clock C.
Suppose that the Hamiltonian is given by $\hat H_{\rm SC}$ of the form Eq.~(\ref{Hamiltonian_SCW}), 
where the corresponding Gibbs state $\hat \rho_{\rm SC}$ is given by Eq.~(\ref{Gibbs_SCW}).

\begin{lemma}
Let $\mathcal E_{\rm SC}$ be a Gibbs-preserving map with the Hamiltonian $\hat H_{\rm SC}$ of Eq.~(\ref{Hamiltonian_SCW}).
Then,
\begin{equation}
\mathcal E_{\rm S} (\hat \rho_{\rm S}) := \langle 1 | \mathcal E_{\rm SC} (\hat \rho_{\rm S} \otimes | 0 \rangle \langle 0 | ) | 1 \rangle
\label{trace_nonincreasing_Gibbs}
\end{equation}
is CP and trace-nonincreasing, and satisfies
\begin{equation}
\mathcal E_{\rm S} ( e^{-\beta \hat H_{\rm S}} ) \leq  e^{-\beta \hat H_{\rm S}'}.
\label{GMP_H_change}
\end{equation}
\label{lemma:GMP_H_change}
\end{lemma}

\begin{proof}
We  prove inequality~(\ref{GMP_H_change}).
Let 
$\mathcal E_{\rm SC} (\hat \rho_{\rm S}^{\rm G} \otimes  | 0 \rangle \langle 0 | )= \sum_{i,j = 0,1} \hat \sigma_{ij} \otimes | i \rangle \langle j |$
and
$\mathcal E_{\rm SC} (\hat \rho_{\rm S}^{\rm G}{}' \otimes  | 1 \rangle \langle 1 | )= \sum_{i,j = 0,1} \hat \sigma_{ij}' \otimes | i \rangle \langle j |$.
Note that $\mathcal E_{\rm S} (\hat \rho_{\rm S}^{\rm G} ) = \hat \sigma_{11}$.
From Eq.~(\ref{Gibbs_SCW}), we have
\begin{equation}
\mathcal E_{\rm SC} ( \hat \rho_{\rm SC}^{\rm G} ) = \frac{Z_{\rm S}}{Z_{\rm S} + Z_{\rm S}'}\sum_{i,j = 0,1} \hat \sigma_{ij} \otimes | i \rangle \langle j | +  \frac{Z_{\rm S}'}{Z_{\rm S} + Z_{\rm S}'} \sum_{i,j = 0,1} \hat \sigma_{ij}' \otimes | i \rangle \langle j |,
\end{equation}
which equals $\hat \rho_{\rm SC}^{\rm G}$ because $\mathcal E_{\rm SC}$ is Gibbs-preserving.
Thus, we find that 
$Z_{\rm S} \hat \sigma_{00} + Z_{\rm S}' \hat \sigma_{00}' = Z_{\rm S} \hat \rho_{\rm S}^{\rm G}$, $Z_{\rm S} \hat \sigma_{11} + Z_{\rm S}' \hat \sigma_{11}' = Z_{\rm S}' \hat \rho_{\rm S}^{\rm G}{}'$,
$Z_{\rm S} \hat \sigma_{01} + Z_{\rm S}' \hat \sigma_{01}' =0$, 
$Z_{\rm S} \hat \sigma_{10} + Z_{\rm S}' \hat \sigma_{10}' = 0$.
Because $\mathcal E_{\rm SC}$ is CPTP, $\hat \sigma_{11}' \geq 0$.
Thus, 
we obtain $Z_{\rm S} \hat \sigma_{11} \leq Z_{\rm S}' \hat \rho_{\rm S}^{\rm G}{}'$, which implies inequality~(\ref{GMP_H_change}).
$\Box$
\end{proof}

If the clock works perfectly for all the initial states, i.e., if  for any $\hat \rho_{\rm S}$ there exists $\hat \rho_{\rm S}'$  such that 
\begin{equation}
\mathcal E_{\rm SC} (\hat \rho_{\rm S} \otimes | 0 \rangle \langle 0 |) = \hat \rho_{\rm S}' \otimes | 1 \rangle \langle 1 |
\label{GMP_H_change_eq}
\end{equation}
holds, then $\mathcal E_{\rm S} : \hat \rho_{\rm S} \mapsto \hat \rho_{\rm S}' $ is TP.
We have inequality~(\ref{GMP_H_change}) also in this case.
However, the condition (\ref{GMP_H_change_eq}) for \textit{all} $\hat \rho_{\rm S}$ is very strong, which is not necessarily satisfied.
If fact,  if  $\mathcal E_{\rm S}$ is TP, we have
${\rm tr}[e^{-\beta \hat H_{\rm S}}] \leq {\rm tr}[e^{-\beta \hat H_{\rm S}'}]$ from~(\ref{GMP_H_change}), or equivalently $F_{\rm S} \geq F_{\rm S}'$, which is not necessarily satisfied in thermodynamic processes.

The converse of Lemma~\ref{lemma:GMP_H_change} is also true in the following sense.

\begin{lemma}
For any CP and trace-nonincreasing map $\mathcal E_{\rm S}$ satisfying inequality~(\ref{GMP_H_change}), 
there exists a CPTP Gibbs-preserving map of SC with the Hamiltonian $\hat H_{\rm SC}$, written as $\mathcal E_{\rm SC}$, such that Eq.~(\ref{trace_nonincreasing_Gibbs}) holds.
\label{lemma:GMP_H_change_converse}
\end{lemma}

\begin{proof}
We can  construct $\mathcal E_{\rm SC}$ as follows.
First, there exists a CP and trace-nonincreasing map $\mathcal E_{\rm S}'$ such that $\mathcal E_{\rm S}'( \hat \rho_{\rm S}) \leq \hat \rho_{\rm S}$ holds for any $\hat \rho_{\rm S}$ and  $\mathcal E_{\rm S} + \mathcal E_{\rm S}'$  is CPTP.
In fact, if the Kraus representation of $\mathcal E_{\rm S}$ is given by $\mathcal E_{\rm S} (\hat \rho_{\rm S} ) = \sum_k \hat M_k \hat \rho_{\rm S} \hat M_k^\dagger$ with $\sum_k \hat M_k^\dagger \hat M_k \leq \hat I$,
we can define $\mathcal E_{\rm S}' (\hat \rho_{\rm S} ) := \sqrt{\hat I - \hat E} \hat \rho_{\rm S}  \sqrt{\hat I - \hat E}$ with $\hat E := \sum_k \hat M_k^\dagger \hat M_k$.
Then define
\begin{equation}
\mathcal E_{\rm SC} ( \hat \rho_{\rm S} \otimes |  0 \rangle \langle 0 | ) := \mathcal E_{\rm S}' (\hat \rho_{\rm S} ) \otimes | 0 \rangle \langle 0 | + \mathcal E_{\rm S} (\hat \rho_{\rm S} ) \otimes | 1 \rangle \langle 1 |.
\end{equation}
We also define
\begin{equation}
\mathcal E_{\rm SC} ( \hat \rho_{\rm S} \otimes |  1 \rangle \langle 1 | ) :=  \frac{e^{-\beta \hat H_{\rm S}} - \mathcal E_{\rm S}' ( e^{-\beta \hat H_{\rm S}} ) }{Z_{\rm S}'} \otimes | 0 \rangle \langle 0 | 
  \frac{e^{-\beta \hat H_{\rm S}'} - \mathcal E_{\rm S} ( e^{-\beta \hat H_{\rm S}} ) }{Z_{\rm S}'}   \otimes | 1 \rangle \langle 1 |,
\end{equation}
for any $\hat \rho_{\rm S}$.
Finally, we let $\mathcal E_{\rm SC} ( \hat \rho_{\rm S} \otimes |  0 \rangle \langle 1 | ) := 0$ and $\mathcal E_{\rm SC} ( \hat \rho_{\rm S} \otimes |  1 \rangle \langle 0 | ) := 0$.
By construction, along with the assumption~(\ref{GMP_H_change}), $\mathcal E_{\rm SC}$ is CPTP.
It is also straightforward to check that $\mathcal E_{\rm SC}$  is Gibbs-preserving.
$\Box$
\end{proof}

Now, we may adopt inequality (\ref{GMP_H_change}) as an alternative generalized definition  of Gibbs-preserving maps for the situation that the input and the output Hamiltonians are not the same and the clock does not necessarily work perfectly~\cite{Faist2018,Sagawa2019}, 
which we refer to as Gibbs-sub-preserving maps.
We formally state the definition as follows:

\begin{definition}[Gibbs-sub-preserving maps]
A CP and trace-nonincreasing map $\mathcal E_{\rm S}$ is Gibbs-sub-preserving with respect to the initial and final Hamiltonians $\hat H_{\rm S}$ and $\hat H_{\rm S}'$, if 
\begin{equation}
\mathcal E (e^{-\beta \hat H_{\rm S}}) \leq e^{-\beta H_{\rm S}'}.
\end{equation}
\label{def:GMP_sub}
\end{definition}

We can rephrase the second law in the form of Theorem~\ref{general_single_shot_work} (a) as follows.

\begin{corollary}[Proposition 3 of \cite{Sagawa2019}]
Let $\hat \rho_{\rm S}$, $\hat \rho_{\rm S}'$ be (normalized) states satisfying $\hat \rho_{\rm S}' = \mathcal E_{\rm S} (\hat \rho_{\rm S})$, where   $\mathcal E_{\rm S}$ is a Gibbs-sub-preserving (CP and trace-nonincreasing) map with respect to the initial and final Hamiltonians $\hat H_{\rm S}$ and $\hat H_{\rm S}'$.  Then, for $\alpha = 0,1,\infty$,
\begin{equation}
S_\alpha (\hat \rho_{\rm S} \| e^{- \hat H_{\rm S}} ) \geq S_\alpha (\hat \rho_{\rm S}' \| e^{- \hat H_{\rm S}'} ).
\label{Gibbs_sub_monotonicity}
\end{equation}
\end{corollary}

\begin{proof}
We embed $\mathcal E_{\rm S}$ into a CPTP Gibbs-preserving map $\mathcal E_{\rm SC}$ with Hamiltonian $\hat H_{\rm SC}$ in the same manner as Lemma~\ref{lemma:GMP_H_change_converse}.
As in the proof of Theorem~\ref{general_single_shot_work}, we have $S_\alpha (\hat \rho \| \hat \rho^{\rm G}_{\rm SC}) = S_\alpha (\hat \rho_{\rm S} \| e^{- \hat H_{\rm S}}) + C$ and  $S_\alpha (\hat \rho' \| \hat \rho^{\rm G}_{\rm SC}) = S_\alpha (\hat \rho_{\rm S}' \| e^{- \hat H_{\rm S}'}) + C$, where $\hat \rho := \hat \rho_{\rm S} \otimes | 0 \rangle \langle 0|$, $\hat \rho' := \hat \rho_{\rm S}' \otimes | 1 \rangle \langle 1|$, and $C := \ln (Z_{\rm S} + Z_{\rm S}')$.  
Because $\hat \rho_{\rm S}$ and $\hat \rho_{\rm S}'$ are both normalized, $\mathcal E_{\rm S}' (\hat \rho_{\rm S} ) = 0$ must be satisfied and thus  $\hat \rho' = \mathcal E_{\rm SC} (\hat \rho)$.
Therefore, the monotonicity of $\mathcal E_{\rm SC}$ implies inequality~(\ref{Gibbs_sub_monotonicity}). $\Box$
\end{proof}

We now explicitly recover the work storage W and consider the single-shot situations.
From the above argument, we have the following. 

\begin{corollary}
A (normalized) state $\hat \rho_{\rm S}$ is $w$-assisted single-shot Gibbs-preserving transformable to another (normalized) state $\hat \rho_{\rm S}'$ with the initial and final Hamiltonians $\hat H_{\rm S}$ and $\hat H_{\rm S}'$ in the sense of Definition~\ref{def:single_work_trans},
 if and only if there exists a Gibbs-sub-preserving map $\mathcal E_{\rm SW}$ on  SW such that 
\begin{equation}
\mathcal{E}_{\rm SW} \left(\hat \rho_{\rm S} \otimes   | E_{\rm i} \rangle \langle E_{\rm i} |  \right)  =  \hat \rho_{\rm S}' \otimes  | E_{\rm f} \rangle \langle E_{\rm f}  |.
\end{equation}
\end{corollary}

We note that the initial and final Hamiltonians of W can be different by applying the clock to W as well.
 We can thus restrict the Hilbert spaces of W for the initial and final states to one-dimensional;
the initial (resp. final) Hilbert space of W is spanned only by $| E_{\rm i} \rangle $ (resp. $| E_{\rm f} \rangle $) with the initial and final Hamiltonians $\hat H_{\rm W,i} := E_{\rm i} | E_{\rm i} \rangle \langle E_{\rm i} |$ (resp. $\hat H_{\rm W,f} := E_{\rm f} | E_{\rm f} \rangle \langle E_{\rm f} |$).
In this setup, the condition (\ref{def:GMP_sub})  of Gibbs-sub-preserving maps on SW is written as~\cite{Faist2018}
\begin{equation}
\mathcal E_{\rm S} ( e^{-\beta H_{\rm S}} ) \leq e^{\beta w} e^{-\beta \hat H_{\rm S}'}.
\end{equation}

Meanwhile, the corresponding, trace-nonincreasing, notion of  thermal operations can be defined as follows~\cite{Sagawa2019}.
We again ignore W.

\begin{definition}[Generalized thermal operations]
A CP and trace-nonincreasing map $\mathcal E_{\rm S}$  is a (generalized) thermal operation with the initial and final Hamiltonians $\hat H_{\rm S}$ and $\hat H_{\rm S}'$, if there exists a heat bath B with Hamiltonian $\hat H_{\rm B}$ and the corresponding Gibbs state $\hat \rho_{\rm B}^{\rm G}$, and exists a partial isometry $\hat V$ such that
\begin{equation}
\mathcal E_{\rm S} (\hat \rho_{\rm S}) = {\rm tr}_{\rm B} \left[ \hat V \hat \rho_{\rm S} \otimes \hat \rho_{\rm B}^{\rm G} \hat V^\dagger  \right]
\end{equation}
and
\begin{equation}
\hat V (\hat H_{\rm S} + \hat H_{\rm B}) = (\hat H_{\rm S}' + \hat H_{\rm B} ) \hat V.
\label{condition_TO2}
\end{equation}
The corresponding non-exact (generalized) thermal operation is defined in the same manner as Definition~\ref{def:thermal_operation}.
\label{def:thermal_operation2}
\end{definition}

Here, an operator $\hat V$ is a partial isometry if $\hat V \hat V^\dagger$ and $\hat V^\dagger \hat V$ are projectors.
The condition~(\ref{condition_TO2}) above implies that the sum of the energies of the system and the bath is conserved even when the Hamiltonian of the system is changed.
In parallel to Lemma~\ref{lemma:GMP_H_change_converse}, we can also construct a TP thermal operation (with a unitary operator) of an extended system, starting from any  trace-nonincreasing thermal operation (Proposition 13 of Ref.~\cite{Sagawa2019}).

It is also known that any trace-nonincreasing thermal operation $\mathcal E_{\rm S}$ in the above sense is a trace-nonincreasing Gibbs-sub-preserving map that satisfies $\mathcal E_{\rm S} ( e^{-\beta \hat H_{\rm S}} ) \leq  e^{-\beta \hat H_{\rm S}'}$  (Lemma 1 of \cite{Sagawa2019}).
We note that if  $\hat V$ is unitary and the thermal operation is TP, we have $\mathcal E (e^{-\beta \hat H_{\rm S}}) = e^{-\beta H_{\rm S}'}$.


\appendix


\chapter{General quantum divergences and their monotonicity}
\label{apx:general_monotonicity}

We prove the monotonicity of quantum divergences discussed in Chapter~\ref{chap:quantum_entropy} from a general point of view, by introducing a class of general quantum divergence-like quantities called the Petz's quasi-entropies~\cite{Petz1985,Petz1986,Hiai2011f}. 
The classical counterpart is presented mainly in Section~\ref{sec:classical_general_divergence}
 and Section~\ref{sec:classical_Fisher}.
 In this Appendix, superoperator $\mathcal E : \mathcal L (\mathcal H) \to \mathcal L (\mathcal H')$ allows different input and output  spaces.
 
We start with proving some operator inequalities in Section~\ref{sec:operator_inequalities}, and discuss operator convex and operator monotone in Section~\ref{sec:operator_convex}.
In Section~\ref{sec:general_monotonicity}, we prove the monotonicity of general divergence-like quantities.
In addition, we prove the monotonicity of the quantum Fisher information in Section~\ref{sec:quantum_Fisher}.


\section{Some operator inequalities}
\label{sec:operator_inequalities}

As a preliminary, we here note  some useful operator inequalities.

\begin{lemma}
 $\hat A \leq \hat B$ implies $\hat A^{-1} \geq \hat B^{-1}$ for $\hat A > 0$ and $\hat B > 0$.
\end{lemma}

\begin{proof}
$\hat B^{-1/2} \hat A \hat B^{-1/2} \leq \hat I$ implies $\hat B^{1/2} \hat A^{-1} \hat B^{1/2} = (\hat B^{-1/2} \hat A \hat B^{-1/2})^{-1} \geq \hat I$.
 $\Box$
\end{proof}

\begin{lemma}
Suppose that $\hat Z$ is positive definite.
An operator-valued matrix
\begin{equation}
\left[
\begin{array}{cc}
\hat X & \hat Y \\
\hat Y^\dagger & \hat Z \\
\end{array}
\right]
\end{equation}
is positive, if and only if
\begin{equation}
\hat X \geq \hat Y \hat Z^{-1} \hat Y^\dagger.
\end{equation}
\label{matrix_positive_lemma}
\end{lemma}

\begin{proof}
This can be seen from 
\begin{equation}
\left[
\begin{array}{cc}
\hat I &  - \hat Y \hat Z^{-1} \\
0 & \hat I 
\end{array}
\right]
\left[
\begin{array}{cc}
\hat X &  \hat Y \\
\hat Y^\dagger & \hat Z 
\end{array}
\right]
\left[
\begin{array}{cc}
\hat I &  - \hat Y \hat Z^{-1} \\
0 & \hat I 
\end{array}
\right]^\dagger
=
\left[
\begin{array}{cc}
\hat X - \hat Y \hat Z^{-1} \hat Y^\dagger & 0 \\
0 & \hat Z
\end{array}
\right].
\end{equation}
$\Box$
\end{proof}

\begin{proposition}[Kadison's inequality, Lemma  3.5 of \cite{Hiai2011f}]
Let $\mathcal E$  be a positive superoperator.
Let $\hat X$ be Hermitian and  $\hat Y$, $\mathcal E (\hat Y)$ be positive definite.
Then,
\begin{equation}
\mathcal E (\hat X \hat Y^{-1} \hat X) \geq \mathcal E (\hat X) \mathcal E (\hat Y )^{-1} \mathcal E (\hat X).
\label{genera_Kadison_inequality}
\end{equation}
In particular, if $\mathcal E$ is positive and unital,
\begin{equation}
\mathcal E (\hat X^2) \geq \mathcal E (\hat X)^2.
\label{Kadison_inequality}
\end{equation}
\label{prop:Kadison}
\end{proposition}

\begin{proof}
Let $\hat X = \sum_k x_k | \varphi_k \rangle \langle \varphi_k |$ be the spectral decomposition of $\hat X$.
We first consider the case that $\hat Y = \hat I$ and define
\begin{equation}
\hat X' :=
\left[
\begin{array}{cc}
\mathcal E (\hat X^2 ) & \mathcal E (\hat X ) \\
 \mathcal E (\hat X ) & \mathcal E (\hat I )
\end{array}
\right]
= \sum_k
\left[
\begin{array}{cc}
x_k^2 & x_k \\
x_k & 1
\end{array}
\right]
\otimes \mathcal E (| \varphi_k \rangle \langle \varphi_k |).
\end{equation}
The right-hand side  is positive, because the $2 \times 2$ matrix and $\mathcal E (| \varphi_k \rangle \langle \varphi_k |)$ are both positive.
Thus, from Lemma~\ref{matrix_positive_lemma}, we obtain
\begin{equation}
\mathcal E (\hat X^2) \geq  \mathcal E (\hat X) \mathcal E (\hat I )^{-1} \mathcal E (\hat X).
\end{equation}
If $\hat Y \neq \hat I$, we replace  $\hat X$ by $\hat Y^{-1/2} \hat X \hat Y^{-1/2}$ and $\mathcal E (\ast )$ by $\mathcal E (\hat Y^{1/2} \ast \hat Y^{1/2} )$ (that is also positive),
 and then obtain inequality (\ref{genera_Kadison_inequality}).
$\Box$
\end{proof}

\begin{corollary}
Let $\mathcal E$ be positive and TP.
Let $\hat X$ be Hermitian and $\hat Y$, $\mathcal E (\hat Y)$ be positive definite.
Then,
\begin{equation}
{\rm tr}[\hat X^2 \hat Y^{-1}] \geq {\rm tr}[\mathcal E (\hat X)^2  \mathcal E (\hat Y)^{-1}]. 
\label{Renyi_2_monotonicity}
\end{equation}
\label{cor:Renyi_2_monotonicity}
\end{corollary}

\begin{proof}
Take the trace of inequality (\ref{genera_Kadison_inequality}).
$\Box$
\end{proof}

The quantum R\'enyi 2-divergence introduced in Eq.~(\ref{simple_Renyi_divergence}) is given by $\tilde S_2 (\hat \rho \| \hat \sigma):=  \ln \left( {\rm tr}[\hat \rho^2  \hat \sigma^{-1}] \right)$.
Thus, inequality~(\ref{Renyi_2_monotonicity}) implies the monotonicity of $\tilde S_2 (\hat \rho \| \hat \sigma)
$.

The above proof of Proposition~\ref{prop:Kadison} does not work in general if $\hat X$ is not Hermitian.  If $\hat X$ is not necessarily Hermitian, we have the following proposition by additionally assuming that $\mathcal E$ is 2-positive.

\begin{proposition}[Schwarz's operator inequality~\cite{Choi1974}]
Let $\mathcal E$ be 2-positive and let $\hat Y$, $\mathcal E (\hat Y)$ be positive definite.
Then,
\begin{equation}
\mathcal E (\hat X^\dagger \hat Y^{-1} \hat X) \geq \mathcal E (\hat X)^\dagger \mathcal E(\hat Y)^{-1} \mathcal E (\hat X). 
\label{Schwarz_operator_inequality}
\end{equation}
In particular, if $\mathcal E$ is 2-positive and unital,
\begin{equation}
\mathcal E (\hat X^\dagger \ \hat X) \geq \mathcal E (\hat X)^\dagger \mathcal E (\hat X). 
\label{Schwarz_operator_inequality_unital}
\end{equation}
\label{prop:Schwarz_operator_inequality}
\end{proposition}

\begin{proof}
We define
\begin{equation}
\hat X' := \left[
\begin{array}{cc}
\hat Y & \hat X \\
\hat X^\dagger & \hat X^\dagger \hat Y^{-1} \hat X \\
\end{array}
\right].
\end{equation}
Because $\mathcal E$ is 2-positive, $\mathcal E \otimes \mathcal I_2$ is positive, and thus
\begin{equation}
( \mathcal E \otimes \mathcal I_2 ) (\hat X' ) = \left[
\begin{array}{cc}
 \mathcal E ( \hat Y ) &   \mathcal E (X) \\
 \mathcal E ( \hat X )^\dagger &  \mathcal E ( \hat X^\dagger \hat Y^{-1} \hat X ) \\
\end{array}
\right]
\end{equation}
is positive.
From Lemma~\ref{matrix_positive_lemma}, we obtain inequality (\ref{Schwarz_operator_inequality}). $\Box$
\end{proof}

\section{Operator convex and operator monotone}
\label{sec:operator_convex}

We consider operator convex/concave functions and operator monotone functions.
The operator convexity and the operator monotonicity are much stronger properties than the ordinary convexity and the ordinary monotonicity of functions. 
In this section, we will omit proofs of several important theorems; see Refs.~\cite{Bhatia,Hiai2010,Hiai2011f} for details.

\begin{definition}[Operator convex/concave functions]
A function $f : (0, \infty) \to \mathbb R$ is operator convex, if for any positive-definite operators $\hat X, \hat Y$ and any  $0 \leq \lambda \leq 1$, 
 \begin{equation}
f(\lambda \hat X + (1-\lambda)\hat Y) \leq \lambda f(\hat X) + (1-\lambda)f(\hat Y)
\label{def:operator_convex}
\end{equation}
holds. If $\leq$ above is replaced by $\geq$, $f$ is operator concave.  
\end{definition}

Obviously, $f$ is operator convex if and only if $-f$ is operator concave.
Since every convex/concave function on any open interval is continuous, every operator convex/concave function is automatically continuous on $(0,\infty)$.

\begin{definition}[Operator monotone functions]
A function $f: (0, \infty) \to \mathbb R$ is operator monotone,  if for any positive-definite operators $\hat X, \hat Y$  with $\hat X \leq \hat Y$, $f(\hat X)\leq f(\hat Y)$ holds.
If $f(\hat X)\geq f(\hat Y)$ holds, $f$ is operator decreasing-monotone. 
\end{definition}

Obviously, $f$ is operator monotone if and only if $-f$ is operator decreasing-monotone.
We  note the following propositions.

\begin{proposition}[Theorem 2.4 of~\cite{Hansen1982}]
Let $f : (0 , \infty ) \to \mathbb R$ be a continuous function and suppose that $f(0) := \lim_{x \to +0} f(x) \in (-\infty, 0]$ exists.
Then, $f$ is operator convex if and only if $g(x):= x^{-1}f(x)$ is operator monotone on $(0,\infty)$.
\end{proposition}

\begin{proposition}[Theorem 2.5 of~\cite{Hansen1982}]
Let $f : (0 , \infty ) \to (-\infty, 0]$ be a continuous function and suppose that $f(0) := \lim_{x \to +0} f(x)  \in (-\infty, 0]$ exists.
Then, $f$ is operator convex if and only if it is operator decreasing-monotone.
\label{prop:concave_monotone}
\end{proposition}



It is known that the full characterizations of operator convex functions and operator monotone functions are given by some integral representations, which is often referred to as the L\"owner's theorem~\cite{Bhatia,Hiai2010,Hiai2011f}.
While there are several variants of such integral representations, we here state one of them for operator convex functions.

\begin{theorem}[Theorem 8.1 of \cite{Hiai2011f}]
Suppose that $f : (0,\infty) \to \mathbb R$ is a continuous function and  $f(0) := \lim_{x \to +0} f(x) \in \mathbb R$ exists.
$f$ is operator convex, if and only if there exist $a \in \mathbb R$, $b \geq 0$, and a non-negative measure $\mu$ on $(0,\infty)$ satisfying $\int_{(0,\infty)} (1+t)^{-2} d\mu (t) < \infty$,  such that
\begin{equation}
f(x) = f(0) + ax + bx^2 + \int_{(0,\infty)} \left( \frac{x}{1+t} - \frac{x}{x+t} \right) d\mu (t).
\label{Lowner_convex}
\end{equation}
Moreover, $a$, $b$, and $\mu$ are uniquely determined by $f$, and 
\begin{equation}
b = \lim_{x \to \infty} \frac{f(x)}{x^2}, \ \ a = f(1) - f(0) - b.
\end{equation}
\label{thm:Lowner_convex}
\end{theorem}

A simple example of  the integral representation~(\ref{Lowner_convex})  is given by
\begin{equation}
x \ln x = \int_0^\infty \left( \frac{x}{1+t} - \frac{x}{x+t} \right) dt.
\label{ln_operator_convex}
\end{equation}
We here list  examples of operator convex/concave functions and operator monotone functions~\cite{Bhatia,Hiai2010,Hiai2011f}.

\begin{proposition}[L\"owner-Heinz Theorem]
On $(0,\infty)$,  
\begin{description}
\item[(a)] For $0 \leq t < \infty$, $f(x) =  (x+t)^{-1}$ is operator convex and operator decreasing-monotone. 
\item[(b)] For  $0 < t < \infty$, $f(x) = x / (x+t)$ is operator concave and operator monotone.
\item[(c)] $f(x) = - \ln x$ is operator convex and operator decreasing-monotone.
\item[(d)] $f(x) = x \ln x$ is operator convex (but not (decreasing-)monotone).
\item[(e)] For $0 \leq \alpha \leq 1$, $f(x) = x^\alpha$ is operator concave and operator monotone.
\item[(f)] For $1 < \alpha \leq 2$, $f(x) = x^\alpha$ is operator convex (but not operator monotone).
\end{description}
\end{proposition}

\begin{proof}
Here we only prove (a)-(d).

We first show (a).
From the convexity of $x^{-1}$, we have 
$(\lambda \hat Y^{-1/2}\hat X \hat Y^{-1/2} + ( 1 - \lambda ) \hat I )^{-1} \leq \lambda \hat Y^{1/2} \hat X^{-1} \hat Y^{1/2}+ ( 1-\lambda) \hat I$, which implies 
$(\lambda \hat X + (1-\lambda ) \hat Y )^{-1} \leq \lambda \hat X^{-1} + (1-\lambda ) \hat Y^{-1}$.  By shifting $\hat X$ and $\hat Y$ by $t \hat I$, we prove the operator convexity.
The operator decreasing-monotonicity is obvious, because $\hat X + t\hat I \leq \hat Y + t \hat I$ implies  $(\hat X + t\hat I)^{-1} \geq ( \hat Y + t \hat I)^{-1}$.

Then, (b) follows from (a) and $x/(x+t) = 1 - t/(x+t)$;
(c) follows from (a) and $- \ln x = \int_0^\infty \left( (x+t)^{-1} - (1+t)^{-1} \right) dt$;
(d) follows from (b) and Eq.~(\ref{ln_operator_convex}).

(e) is Theorem V.1.9 of Ref.~\cite{Bhatia} or Example 8.3 of Ref.~\cite{Hiai2011f}  or Example 2.5.9 of Ref.~\cite{Hiai2010}; (f) is  Example 8.3 of Ref.~\cite{Hiai2011f} or Example 2.5.9 of Ref.~\cite{Hiai2010}.
Here we only note that the operator convexity of $f(x) = x^2$ is obvious from $\lambda \hat X^2 + (1-\lambda) \hat Y^2  - (\lambda  \hat X + (1-\lambda ) \hat Y )^2 = \lambda (1-\lambda ) (\hat X - \hat Y)^2 \geq 0$
(Example V.1.3 of~\cite{Bhatia}).
See also Ref.~\cite{Bhatia} for that $f(x)=x^2$ is not operator monotone (Example V.1.2).

We finally note that $f(x) = x^\alpha$ for $2 < \alpha <\infty$ is not operator convex (Example 2.5.9 of Ref.~\cite{Hiai2010}; see also Example V.1.4 of Ref.~\cite{Bhatia}).
$\Box$
\end{proof}

We next show a useful property of operator convex functions (see also Theorem V.2.3 of \cite{Bhatia}).

\begin{proposition}[Jensen's operator inequality; Theorem 2.1 of~\cite{Hansen1982}]
Consider  $f: (0,\infty) \to \mathbb R$ and suppose that $f(0) := \lim_{x \to +0}f(x) \in \mathbb R$ exists.  Then, the following are equivalent.\\
(i) $f$ is operator convex on $(0,\infty)$ and $f(0) \leq 0$. \\
(ii)  For any $\hat X \geq 0$ and any contraction $\hat V$,
\begin{equation}
f (\hat V^\dagger \hat X \hat V ) \leq \hat V^\dagger f (\hat X ) \hat V.
\label{Jensen_operator_inequality}
\end{equation}
Here, a linear operator $\hat V$ is called a \textit{contraction}, if it satisfies $\langle \varphi | \hat V^\dagger \hat V | \varphi \rangle \leq \langle \varphi | \varphi \rangle$ for all $| \varphi \rangle$.
\label{prop:Jensen_operator_inequality}
\end{proposition}

\begin{proof}
We only prove (i) $\Rightarrow$ (ii) here (see Ref.~\cite{Bhatia} for the converse).
Note that $\hat V^\dagger \hat V \leq \hat I$ and $\hat V \hat V^\dagger \leq \hat I$ because $\hat V$ is a contraction.
We can define  $\hat U := (I - \hat  V \hat V^\dagger)^{1/2}$ and $\hat U' := (I - \hat V^\dagger \hat V)^{1/2}$.
Then we consider
\begin{equation}
\hat X' := \left[
\begin{array}{cc}
\hat X & 0 \\
0 & 0
\end{array}
\right],
\ 
\hat V_1 := \left[
\begin{array}{cc}
\hat V &  \hat U  \\
\hat U' & - \hat V^\dagger
\end{array}
\right],
\
\hat V_2 := \left[
\begin{array}{cc}
\hat V & - \hat U \\
\hat U' & \hat V^\dagger
\end{array}
\right].
\end{equation}
By using the singular-value decomposition of $\hat V$, we can see that $\hat V_1$ and $\hat V_2$ are unitary. We compute
\begin{equation}
\hat V_1^\dagger \hat X'  \hat V_1 = \left[
\begin{array}{cc}
\hat V^\dagger \hat X \hat V & \hat V^\dagger \hat X \hat U \\
\hat U \hat X \hat V & \hat U \hat X \hat U
\end{array}
\right], \ 
\hat V_2^\dagger \hat X' \hat V_2 = \left[
\begin{array}{cc}
\hat V^\dagger \hat X \hat V & - \hat V^\dagger \hat X \hat U \\
-\hat U \hat X \hat V & \hat U \hat X \hat U
\end{array}
\right], 
\end{equation}
and  thus obtain
\begin{equation}
\frac{\hat V_1^\dagger \hat X'  \hat V_1 + \hat V_2^\dagger \hat  X' \hat V_2}{2} = \left[
\begin{array}{cc}
\hat V^\dagger \hat  X \hat V &0 \\
0 & \hat U \hat X \hat U
\end{array}
\right].
\end{equation}
From the operator convexity of $f(x)$, we have
\begin{equation}
f \left( \frac{\hat V_1^\dagger \hat X'  \hat V_1 + \hat V_2^\dagger \hat X' \hat V_2}{2} \right) \leq  \frac{f(\hat V_1^\dagger \hat X'  \hat V_1) + f(\hat V_2^\dagger \hat X' \hat V_2)}{2} 
= \frac{\hat V_1^\dagger f(\hat X')  \hat V_1+ \hat V_2^\dagger f( \hat X')\hat V_2}{2},
\end{equation}
which leads to
\begin{equation}
\left[
\begin{array}{cc}
f(\hat V^\dagger \hat  X \hat V) &0 \\
0 & f(\hat U \hat X \hat U) 
\end{array}
\right] 
\leq 
 \left[
\begin{array}{cc}
\hat V^\dagger f(\hat X) \hat V + \hat U' f(0) \hat U' &0 \\
0 & \hat U f(\hat X) \hat U + \hat V f(0) \hat V^\dagger
\end{array}
\right].
\end{equation}
 From $f(0) \leq 0$, we finally obtain
\begin{equation}
\left[
\begin{array}{cc}
f(\hat V^\dagger \hat X \hat V) &0 \\
0 & f(\hat U \hat X \hat U) 
\end{array}
\right] \leq  \left[
\begin{array}{cc}
\hat V^\dagger f(\hat X) \hat V &0 \\
0 &  \hat U f ( \hat X) \hat U
\end{array}
\right], 
\end{equation}
which implies (ii).
$\Box$
\end{proof}



\section{General monotonicity properties}
\label{sec:general_monotonicity}

We now go to the main part of this Appendix: the proof of the monotonicity.
In this section, we basically assume that  $\hat \rho$, $\hat \sigma$, $\mathcal E (\hat \rho)$, $\mathcal E (\hat \sigma)$ are all positive definite, while $\hat \rho$ and $\mathcal E (\hat \rho)$ can be just positive if  $f(0) := \lim_{x \to +0 }f(x) \in \mathbb R$ exists for  function $f : (0, \infty ) \to \mathbb R$ introduced below.
On the other hand,  we do not necessarily assume that these  operators are normalized, unless stated otherwise.

To introduce general divergence-like quantities, we first define the left and right multiplications of $\hat \rho$:
\begin{equation}
\mathcal L_{\hat \rho} (\hat X ) := \hat \rho \hat X, \ \ \mathcal R_{\hat \rho} (\hat X ) := \hat X \hat \rho. 
\end{equation}
Here, $\mathcal L_{\hat \rho} $ and $\mathcal R_{\hat \rho}$ are commutable and are both Hermitian  with respect to the Hilbert-Schmidt inner product.
Then,  we introduce the modular operator $\mathcal D_{\hat \rho, \hat \sigma}$  for  $\hat \rho$, $\hat \sigma$ by
\begin{equation}
\mathcal  D_{\hat \rho, \hat \sigma} (\hat X ) := \mathcal L_{\hat \rho} \mathcal R_{\hat \sigma^{-1}} (\hat X ) = \hat \rho \hat X \hat \sigma^{-1},
\end{equation}
which is also Hermitian.

Let $\hat \rho = \sum_p p_i \hat P_i$ and $\hat \sigma = \sum_i q_i \hat Q_i$ be the spectrum decompositions, where  $\hat P_i$, $\hat Q_i$ are the projectors onto the eigenspaces and we set $p_i \neq p_j$ and $q_i \neq q_j$ for $i \neq j$.
The spectrum decomposition of $\mathcal  D_{\hat \rho, \hat \sigma}$ is then given by
\begin{equation}
\mathcal  D_{\hat \rho, \hat \sigma} = \sum_{ij} \frac{p_i}{q_j} \mathcal P_{ij},
\end{equation}
where $\mathcal P_{ij}$ is a projection superoperator defined as $\mathcal P_{ij} (\hat X ) := \hat P_i \hat X  \hat Q_j$.

Let  $f :  ( 0,\infty ) \to \mathbb R$ be a function.
We can define $f (\mathcal  D_{\hat \rho, \hat \sigma})$ by
\begin{equation}
f (\mathcal  D_{\hat \rho, \hat \sigma}) = \sum_{ij} f \left( \frac{p_i}{q_j} \right) \mathcal P_{ij},
\end{equation}
or equivalently 
\begin{equation}
f (\mathcal  D_{\hat \rho, \hat \sigma}) (\hat X ) = \sum_{ij} f \left( \frac{p_i}{q_j} \right) \hat P_i \hat X \hat Q_j.
\label{D_rho_sigma}
\end{equation}
Now, we consider a divergence-like quantity
\begin{equation}
D_f (\hat \rho \| \hat \sigma ) := \langle \hat \sigma^{1/2}, f (\mathcal D_{\hat \rho, \hat \sigma}) ( \hat \sigma^{1/2}) \rangle_{\rm HS} = \sum_{ij} q_j  f \left( \frac{p_i}{q_j} \right)  {\rm tr}[\hat P_i \hat Q_j],
\label{quantum_f_general}
\end{equation}
which is a special case of the Petz's quasi-entropies~\cite{Petz1985,Petz1986} (see also Refs.~\cite{Ohya,Hiai2011f}).
We discuss two special examples of this quantity:

\noindent\textit{Quantum KL divergence.}
The first example is $f(x) = x \ln x$, which is operator convex. 
By noting that $f ( \mathcal  D_{\hat \rho, \hat \sigma})  = \mathcal L_{\hat \rho} \mathcal R_{\hat \sigma^{-1}} ( \ln \mathcal L_{\hat \rho} + \ln \mathcal R_{\hat \sigma^{-1}} )$ and $\mathcal L_{\hat \rho}$ and $\mathcal R_{\hat \sigma^{-1}}$ are commutable, we obtain
\begin{equation}
f ( \mathcal  D_{\hat \rho, \hat \sigma}) (\hat X ) = (\hat \rho \ln \hat \rho ) \hat X \hat \sigma^{-1} - \hat \rho \hat X ( \hat \sigma^{-1} \ln \hat \sigma),
\end{equation}
which leads to
\begin{equation}
D_f (\hat \rho \| \hat \sigma )  =  {\rm tr} [ \hat \sigma^{1/2}(\hat \rho \ln \hat \rho ) \hat \sigma^{1/2} \hat \sigma^{-1}  ] - {\rm tr} [\hat \sigma^{1/2} \hat \rho \hat \sigma^{1/2} ( \hat \sigma^{-1} \ln \hat \sigma )  ] = S_1 (\hat \rho \| \hat \sigma).
\label{f_divergence_KL}
\end{equation}
Note that Eq.~(\ref{quantum_f_general}) implies
\begin{equation}
D_{f} (\hat \rho \| \hat \sigma) =  \sum_{ij} p_i \ln \frac{p_i}{q_j}{\rm tr}[\hat P_i \hat Q_j].
\label{KL_f_quantum_classical}
\end{equation}

\noindent\textit{Quantum R\'enyi divergence.}
The second example is $f_\alpha(x) := x^\alpha$ with  $0 < \alpha < 1$ and $1 < \alpha < \infty$.
By noting that $f_\alpha ( \mathcal  D_{\hat \rho, \hat \sigma}) (\hat X ) = \hat \rho^\alpha \hat X \hat \sigma^{-\alpha}$, we have
\begin{equation}
D_{f_\alpha} (\hat \rho \| \hat \sigma ) = {\rm tr}[\hat \sigma^{1/2} \hat \rho^\alpha \hat \sigma^{1/2} \hat \sigma^{-\alpha}] =  {\rm tr}[\hat \rho^\alpha \hat \sigma^{1-\alpha} ].
\end{equation}
We note that, from Eq.~(\ref{quantum_f_general}),
\begin{equation}
D_{f_\alpha} (\hat \rho \| \hat \sigma) =  \sum_{ij} \frac{p_i^\alpha}{q_j^{\alpha-1}}  {\rm tr}[\hat P_i \hat Q_j].
\label{Renyi_f_quantum_classical}
\end{equation}
Then, we define a simple version of the quantum R\'enyi $\alpha$-divergence for  $0< \alpha < 1$ and $1 < \alpha < \infty$ as~\cite{Hiai2011f,Tomamichel}
\begin{equation}
\tilde S_\alpha ( \hat \rho \| \hat \sigma ) := \frac{1}{\alpha - 1} \ln D_{f_\alpha} (\hat \rho \| \hat \sigma) = \frac{1}{\alpha - 1} \ln \left( {\rm tr} [\hat \rho^\alpha \hat \sigma^{1-\alpha }]  \right).
\label{simple_Renyi_divergence}
\end{equation}
It is straightforward to see that
\begin{equation}
\tilde S_0 (\hat \rho \| \hat \sigma) := \lim_{\alpha \to +0} \tilde S_\alpha (\hat \rho \| \hat \sigma)  = S_0  (\hat \rho \| \hat \sigma),
\end{equation}
and we can also show that
\begin{equation}
\tilde S_1 (\hat \rho \| \hat \sigma)  := \lim_{\alpha \to 1} \tilde S_\alpha (\hat \rho \| \hat \sigma)  = S_1  (\hat \rho \| \hat \sigma).
\label{simple_Renyi_1}
\end{equation}
Thus  $\tilde S_\alpha (\hat \rho \| \hat \sigma)$ is well-defined for $0 \leq \alpha < \infty$.
On the other hand,  $\lim_{\alpha \to \infty} \tilde S_\alpha (\hat \rho \| \hat \sigma)$ does not equal $S_\infty (\hat \rho \| \hat \sigma)$.
We finally note that, if $\hat \rho$ is normalized,
\begin{equation}
\tilde S_\alpha  ( \hat \rho \| \hat \sigma ) \leq \tilde S_{\alpha'} ( \hat \rho  \| \hat \sigma ) \ \ \rm{for} \ \ \alpha \leq \alpha',
\end{equation}
which can be proved in the same manner as the proof of Proposition~\ref{prop:classical_alpha_monotone}, by noting that $\sum_{ij} p_i{\rm tr}[\hat P_i \hat Q_j] = 1$.

\

We now investigate the fundamental properties of $D_{f} (\hat \rho \| \hat \sigma)$.  First, we consider the following theorem, from which the non-negativity of the quantum KL divergence follows.

\begin{theorem}
Suppose that $\hat \rho (>0)$, $\hat \sigma (>0)$ are both normalized.  If $f$ is convex on $(0, \infty )$ and strictly-convex at $x=1$,
\begin{equation}
D_f (\hat \rho \| \hat \sigma) \geq f(1),
\end{equation}
 where the equality holds if and only if  $\hat \rho = \hat \sigma$.
 If $f$ is concave on $(0,\infty )$ and strictly-concave at $x=1$, the opposite inequality holds.
\label{quantum_f_positive}
\end{theorem}

\begin{proof}
Consider the expression of $D_f (\hat \rho \| \hat \sigma) $ of Eq.~(\ref{quantum_f_general}).
Because $T_{ij} := {\rm tr}[\hat P_i \hat Q_j ]$ is a doubly stochastic matrix and $f$ is convex, we have $\sum_i f(p_i/q_j)T_{ij} \geq f(p_j' / q_j)$ with $p_j' := \sum_i p_i T_{ij}$.  We thus have
\begin{equation}
D_f (\hat \rho \| \hat \sigma) \geq D_f (p' \| q),
\label{quantum_classica_f_divergence}
\end{equation}
where $ D_f (p' \| q) := \sum_i q_i f( p_i' / q_i)$.
Also, as is the case for  inequality (\ref{classical_f_positive}), we have
\begin{equation}
 D_f (p' \| q) \geq f(1),
\end{equation}
where the equality hods if and only if $p' = q$, because of the strict convexity of $f$ at $x=1$.
In the case that $p' = q$ holds, the equality of (\ref{quantum_classica_f_divergence}) holds if and only if $T_{ij} = \delta_{ij}$, again from the strict convexity of $f$ at $x=1$ (note that $p_i \neq p_j$ for $i \neq j$).  Thus, $D_f (\hat \rho \| \hat \sigma) = f(1)$ holds if and only if $\hat \rho = \hat \sigma$.
We can prove the concave case in the same manner.
$\Box$
\end{proof}

We emphasize that we only used the ordinary  convexity, not the operator convexity, of $f$ (and the strict convexity at $x=1$) in the above proof.
The following are two special cases.

\begin{corollary}
Suppose that $\hat \rho (\geq 0)$, $\hat \sigma (>0)$ are both normalized.
The quantum KL divergence satisfies  $S_1 (\hat \rho \| \hat \sigma) \geq 0$ [inequality (\ref{quantum_KL_positive})], where the equality holds if and only if  $\hat \rho = \hat \sigma$.
\label{cor:quantum_KL_positive}
\end{corollary}

\begin{proof}
Take $f(x) = x \ln x$ with  $f(1) = 0$, which is convex on $(0,\infty)$ and strictly-convex at $x=1$.
$\Box$
\end{proof}

\begin{corollary}
Suppose that $\hat \rho (\geq 0 )$, $\hat \sigma (>0)$ are both normalized.
For $0< \alpha < 1$ and $1 < \alpha < \infty$, the quantum R\'enyi divergence satisfies  $\tilde S_\alpha (\hat \rho \| \hat \sigma) \geq 0$, where the equality holds if and only if  $\hat \rho = \hat \sigma$.
\label{cor:quantum_Renyi_positive}
\end{corollary}

\begin{proof}
We take $f(x) = x^\alpha$ with $f(1) = 1$, which is concave (convex)  on $(0,\infty)$ and strictly-concave (convex) at $x=1$ for $0 < \alpha < 1$ ($1< \alpha < \infty$).  By noting the sign of $\alpha - 1$, we prove the claim.
$\Box$
\end{proof}

We now consider the monotonicity, where $\hat \rho$, $\hat \sigma$ can be unnormalized.

\begin{theorem}[Monotonicity, Theorem 4.3 of \cite{Hiai2011f}]
Let $\hat \rho$,  $\mathcal E (\hat \rho)$ be positive and $\hat \sigma$, $\mathcal E (\hat \sigma)$ be positive definite.
Suppose that $f : (0,\infty) \to \mathbb R$ is operator convex and  $f(0) := \lim_{x \to +0} f(x) \in \mathbb R$ exists.
Let $\mathcal E$ be 2-positive and TP.  
Then,
\begin{equation}
D_f (\hat \rho \| \hat \sigma) \geq D_f (\mathcal E ( \hat \rho ) \| \mathcal E (\hat \sigma )).
\end{equation}
If $f$ is operator concave, the opposite inequality holds.
\label{quantum_f_monotonicity}
\end{theorem}

To prove this theorem, we prepare the following lemma.

\begin{lemma}[Lemma 4.2 of  \cite{Hiai2011f}]
Let $\mathcal E: \mathcal L (\mathcal H) \to \mathcal L (\mathcal H')$ be 2-positive and TP. 
Define $\mathcal V : \mathcal L (\mathcal H') \to \mathcal L (\mathcal H)$ by
\begin{equation}
\mathcal V:= \mathcal R_{\hat \sigma}^{1/2} \mathcal E^\dagger \mathcal R_{\mathcal E (\hat \sigma)}^{-1/2},
\end{equation}
or equivalently,
\begin{equation}
\mathcal V (\hat X ) := \mathcal E^\dagger (\hat X \mathcal E (\hat \sigma)^{-1/2} )\hat \sigma^{1/2} \ \  \Leftrightarrow \ \ \mathcal V (\hat X \mathcal E (\hat \sigma)^{1/2} ) = \mathcal E^\dagger (\hat X ) \hat \sigma^{1/2}. 
\end{equation}
Then,  $\mathcal V$ is a contraction with respect to the Hilbert-Schmidt inner product.  Moreover,
 \begin{equation}
\mathcal V^\dagger \mathcal D_{\hat \rho, \hat \sigma} \mathcal V 
\leq
\mathcal D_{\mathcal E ( \hat \rho) , \mathcal E (\hat \sigma)}.
\label{proof_monotonicity2}
\end{equation}
\label{lemma:V_contraction}
\end{lemma}

\begin{proof}
Because  $\mathcal E^\dagger$ that is 2-positive and unital, we apply the Schwarz's operator inequality (\ref{Schwarz_operator_inequality_unital}) of Proposition \ref{prop:Schwarz_operator_inequality}  and obtain
\begin{eqnarray}
&{}& \langle \hat X \mathcal E (\hat \sigma )^{1/2}, \hat X \mathcal E (\hat \sigma )^{1/2} \rangle_{\rm HS}
= {\rm tr} [\hat \sigma \mathcal E^\dagger (\hat X^\dagger \hat X ) ] \\
&\geq&
{\rm tr} [\hat \sigma \mathcal E^\dagger (\hat X)^\dagger \mathcal E^\dagger (\hat X ) ] 
=
\langle \mathcal E^\dagger (\hat X) \hat \sigma^{1/2},   \mathcal E^\dagger (\hat X) \hat \sigma^{1/2} \rangle_{\rm HS},
\end{eqnarray}
which implies that $\mathcal V$ is a contraction.

Inequality~(\ref{proof_monotonicity2}) is shown as
\begin{eqnarray}
&{}& \langle \hat X \mathcal E(\hat \sigma)^{1/2}, \mathcal V^\dagger \mathcal D_{\hat \rho, \hat \sigma} \mathcal V ( \hat X \mathcal E(\hat \sigma)^{1/2}) \rangle_{\rm HS} \\
&=& \langle \mathcal V (\hat X \mathcal E(\hat \sigma)^{1/2} ), \mathcal D_{\hat \rho, \hat \sigma} \mathcal V ( \hat X \mathcal E(\hat \sigma)^{1/2}) \rangle_{\rm HS} \\
&=& \langle \mathcal E^\dagger (\hat X) \hat \sigma^{1/2}, \mathcal D_{\hat \rho, \hat \sigma} (\mathcal E^\dagger (\hat X) \hat \sigma^{1/2} ) \rangle_{\rm HS} \\
&=& {\rm tr} \left[ \hat \sigma^{1/2} \mathcal E^\dagger (\hat X^\dagger) \mathcal D_{\hat \rho, \hat \sigma} ( \mathcal E^\dagger (\hat X) \hat \sigma^{1/2}) \right]\\
&=& {\rm tr} \left[ \hat \rho \mathcal E^\dagger (\hat X) \mathcal E^\dagger (\hat X^\dagger )  \right] \\
&\leq&  {\rm tr} \left[ \hat \rho \mathcal E^\dagger (\hat X \hat X^\dagger )  \right]\\
&=& {\rm tr} \left[ \mathcal E (\hat \rho)\hat X \hat X^\dagger   \right]  \\
&=& {\rm tr} [\mathcal E (\hat \sigma)^{1/2} \hat X^\dagger \mathcal E(\hat \rho) \hat X \mathcal E(\hat \sigma)^{1/2} \mathcal E(\hat \sigma)^{-1} ] \\
&=& \langle \hat X \mathcal E(\hat \sigma)^{1/2}, \mathcal D_{\mathcal E ( \hat \rho) , \mathcal E (\hat \sigma)} ( \hat X \mathcal E(\hat \sigma)^{1/2}) \rangle_{\rm HS},
\end{eqnarray}
where we again used the Schwarz's operator inequality (\ref{Schwarz_operator_inequality_unital})  for $\mathcal E^\dagger$.
$\Box$
\end{proof}

\noindent\textbf{Proof of Theorem~\ref{quantum_f_monotonicity}.}
First, we assume that $f$ is operator convex and operator decreasing-monotone, and satisfies $f(0) \leq 0$.
Consider $\mathcal V$ defined in Lemma~\ref{lemma:V_contraction}.

By applying the Jensen's operator inequality~(\ref{Jensen_operator_inequality}) in Proposition~\ref{prop:Jensen_operator_inequality} to $f$  and contraction $\mathcal V$, we have
\begin{equation}
f (\mathcal  V^\dagger \mathcal D_{\hat \rho, \hat \sigma} \mathcal V) 
\leq
\mathcal V^\dagger f ( \mathcal D_{\hat \rho, \hat \sigma} ) \mathcal V.
\label{proof_monotonicity1}
\end{equation}
Also, by substituting inequality~(\ref{proof_monotonicity2})  in Lemma~\ref{lemma:V_contraction} to operator-decreasing $f$, we have
\begin{equation}
f ( \mathcal D_{\mathcal E ( \hat \rho) , \mathcal E (\hat \sigma)})
\leq 
f ( \mathcal V^\dagger \mathcal D_{\hat \rho, \hat \sigma} \mathcal V ).
\label{proof_monotonicity3}
\end{equation}
By combining inequalities (\ref{proof_monotonicity1}) and  (\ref{proof_monotonicity3}),  we obtain
\begin{equation}
f ( \mathcal D_{\mathcal E ( \hat \rho) , \mathcal E (\hat \sigma)})
\leq 
\mathcal V^\dagger f ( \mathcal D_{\hat \rho, \hat \sigma} ) \mathcal V.
\label{proof_monotonicity4}
\end{equation}
Thus, 
\begin{equation}
 \langle \mathcal E (\hat \sigma)^{1/2}, f (\mathcal D_{\mathcal E (\hat \rho), \mathcal E (\hat \sigma)}) (\mathcal E (\hat \sigma)^{1/2} ) \rangle_{\rm HS} 
 \leq 
 \langle \mathcal E (\hat \sigma)^{1/2}, \mathcal V^\dagger f (\mathcal D_{\hat \rho, \hat \sigma} )\mathcal V ( \mathcal E (\hat \sigma )^{1/2} ) \rangle_{\rm HS},
\end{equation}
where the left-hand side equals $D_f (\mathcal E ( \hat \rho ) \| \mathcal E (\hat \sigma ))$ and the right-hand side equals
$D_f ( \hat \rho \| \hat \sigma)$ because of $\mathcal V (\mathcal E (\hat \sigma )^{1/2} ) = \hat \sigma^{1/2}$.
Thus, we obtain
\begin{equation}
D_f (\mathcal E ( \hat \rho ) \| \mathcal E (\hat \sigma )) \leq D_f ( \hat \rho \| \hat \sigma).
\end{equation}

For general $f$ that is just operator convex, we invoke the  integral representation (\ref{Lowner_convex}), which leads to
\begin{equation}
D_f ( \hat \rho \| \hat \sigma) ={\rm tr}[\hat \sigma] f(0) +  {\rm tr}[\hat \rho] a +{\rm tr}[\hat \rho^2 \hat \sigma^{-1}] b  + \int_{(0,\infty)} \left( \frac{{\rm tr}[\hat \rho]}{1+t} + D_{\varphi_t} (\hat \rho \| \hat \sigma ) \right) d\mu (t),
\end{equation}
where we defined $\varphi_t (x) :=  - x / (x+t)$ for $0<t<\infty$.
The first and the second terms on the right-hand side above  are just constants, because $\mathcal E$ is TP. 
The monotonicity of the third term, ${\rm tr}[\hat \rho^2 \hat \sigma^{-1}]b $ with $b \geq 0$, was already proved in Corollary \ref{cor:Renyi_2_monotonicity}.
The first term inside the integral again gives a constant.
Finally, $D_{\varphi_t} (\hat \rho \| \hat \sigma )$ satisfies the monotonicity, because $\varphi_t$ is operator convex and operator decreasing-monotone, and $\varphi_t (0) \leq 0$.
$\Box$

\

We discuss the following two special cases.

\begin{corollary}
Let $\hat \rho$,  $\mathcal E (\hat \rho)$ be positive and $\hat \sigma$, $\mathcal E (\hat \sigma)$ be positive definite.
Let $\mathcal E$ be 2-positive and TP.  The quantum KL divergence $S_1 (\hat \rho \| \hat \sigma) $ satisfies the monotonicity
\begin{equation}
S_1(\hat \rho \| \hat \sigma ) \geq S_1(\mathcal{E}( \hat \rho ) \| \mathcal{E}( \hat \sigma )). 
\end{equation}
From this, Theorem~\ref{thm:monotone} immediately follows.
\label{cor:KL_monotonicity}
\end{corollary}

\begin{proof}
Take $f(x) = x \ln x$, which is operator convex.
$\Box$
\end{proof}

\begin{corollary}
[\cite{Hiai2011f,Tomamichel}]
Let $\hat \rho$,  $\mathcal E (\hat \rho)$ be positive and $\hat \sigma$, $\mathcal E (\hat \sigma)$ be positive definite.
Let $\mathcal E$ be 2-positive and TP.
For $0 \leq \alpha  \leq 2$, the quantum R\'enyi $\alpha$-divergence (\ref{simple_Renyi_divergence}) satisfies the monotonicity:
\begin{equation}
\tilde S_\alpha (\hat \rho \| \hat \sigma) \geq \tilde S_\alpha (\mathcal E (\hat \rho ) \| \mathcal E ( \hat \sigma )).
\label{eq:Renyi_monotonicity}
\end{equation}
\label{cor:Renyi_monotonicity}
\end{corollary}

\begin{proof}
For  $\alpha \neq 0, 1$, take $f(x) = x^\alpha$, which is operator concave  for $0 < \alpha < 1$ and operator convex for $1 < \alpha \leq 2$.
By noting the sign  of $\alpha - 1$, we obtain inequality~(\ref{eq:Renyi_monotonicity}).
For  the case of $\alpha = 0, 1$, we can take the limit $\alpha \to +0$ and $\alpha \to 1$, respectively.
$\Box$
\end{proof}


\

From the monotonicity, we have the joint convexity of $D_f (\hat \rho \| \hat \sigma )$; we will omit the proof, as it is completely parallel to that of Theorem~\ref{thm:joint_convexity_KL} for the quantum KL divergence.

\begin{corollary}[Joint convexity, Corollary 4.7 of \cite{Hiai2011f}]
Let $\hat \rho$,  $\mathcal E (\hat \rho)$ be positive and $\hat \sigma$, $\mathcal E (\hat \sigma)$ be positive definite.
Suppose that $f : (0,\infty) \to \mathbb R$ is operator convex and  $f(0) := \lim_{x \to +0} f(x) \in \mathbb R$ exists.
Let $\hat \rho = \sum_k p_k \hat \rho_k$ and $\hat \sigma = \sum_k p_k \hat \sigma_k$, where $\hat \rho_k$ and $\hat \sigma_k$ are quantum states and $p_k$'s represent a classical distribution  with $p_k > 0$.
Then,
\begin{equation}
D_f (\hat \rho \| \hat \sigma ) \leq \sum_k p_k D_f (\hat \rho_k \| \hat \sigma_k) .
\label{joint_convexity_f}
\end{equation}
If $f$ is operator concave, the opposite inequality holds.
The equality holds if $\mathcal P_k$'s are orthogonal with each other, where $\mathcal P_k$ is the subspace spanned by the supports of $\hat \rho_k$ and $\hat \sigma_k$. 
\label{thm:joint_convexity_f}
\end{corollary}


\begin{corollary}[Joint convexity of the quantum R\'enyi divergence]
Let $\hat \rho$,  $\mathcal E (\hat \rho)$ be positive and $\hat \sigma$, $\mathcal E (\hat \sigma)$ be positive definite.
For $0 < \alpha < 1$, the quantum R\'enyi $\alpha$-divergence (\ref{simple_Renyi_divergence}) satisfies the joint convexity (with the same notations as Corollary~\ref{thm:joint_convexity_f} and with $p_k > 0$):
\begin{equation}
\tilde S_\alpha (\hat \rho \| \hat \sigma ) \leq \sum_k p_k \tilde S_\alpha (\hat \rho_k \| \hat \sigma_k) .
\label{joint_convexity_Renyi}
\end{equation}
The equality holds if the same equality condition as in Corollary \ref{thm:joint_convexity_f} is satisfied and $\tilde S_\alpha (\hat \rho_k \| \hat \sigma_k)$'s are the same for all $k$. 
\label{cor:joint_convexity_Renyi}
\end{corollary}

\begin{proof}
 From Corollary~\ref{thm:joint_convexity_f} with $f_\alpha$ being concave  and from the convexity of $\frac{1}{\alpha - 1}\ln t$,
\begin{equation}
\frac{1}{\alpha - 1} \ln D_{f_\alpha} (\hat \rho \| \hat \sigma ) 
\leq \frac{1}{\alpha - 1}  \ln \sum_k p_k D_{f_\alpha} (\hat \rho_k \| \hat \sigma_k)  
\leq \frac{1}{\alpha - 1} \sum_k  p_k \ln   D_{f_\alpha} (\hat \rho_k \| \hat \sigma_k).  
\end{equation}
We note that for $\alpha > 1$, the right inequality fails.
The equality  in (\ref{joint_convexity_Renyi}) holds if the equality conditions for the above two inequalities are satisfied. 
$\Box$
\end{proof}

We can take the limit of inequality (\ref{joint_convexity_Renyi}) for $\alpha = 0,1$, while the joint convexity for these cases is already proved in Theorem~\ref{thm:joint_convexity_0} and Theorem~\ref{thm:joint_convexity_KL}, respectively.

\

\noindent\textit{Quantum f-divergence.} 
We remark on the quantum $f$-divergence~\cite{Hiai2011f}, which is a special case of the Petz's quasi-entropies~\cite{Petz1985,Petz1986}.
Let $\hat \rho$ and $\hat \sigma$ be normalized states.
 $D_f (\hat \rho \| \hat \sigma)$ is called the quantum $f$-divergence, if $f$ is operator convex on $(0,\infty)$, strictly convex at $x=1$, and $f(1) = 0$.
For example, $S_1 (\hat \rho \| \hat \sigma)$ is the quantum $f$-divergence in this sense. 
From Theorem~\ref{quantum_f_positive}, the quantum $f$-divergence is non-negative:  $D_f (\hat \rho \| \hat \sigma) \geq 0$, where the equality holds if and only if  $\hat \rho = \hat \sigma$.
From Theorem~\ref{quantum_f_monotonicity}, $D_f (\hat \rho \| \hat \sigma)$ satisfies the monotonicity under CPTP maps.

The trace distance $D( \hat \rho, \hat \sigma )$ is not an $f$-divergence in the quantum case.  We  note that $f(x) = |x-1|$ is not operator convex in any interval that contains $x=1$~\cite{Bhatia}.  The quantum fidelity~(\ref{quantum_fidelity}) is not related to a quantum $f$-divergence too.  In fact, $f(x) = 1-  \sqrt{x}$ is operator convex and gives $D_f (\hat \rho \| \hat \sigma) = 1 - {\rm tr}[\hat \rho^{1/2} \hat \sigma^{1/2}]$, which is not equivalent to Eq.~(\ref{quantum_fidelity}).

\

\noindent\textit{Sandwiched R\'enyi divergence.}
We can also introduce another version of the quantum R\'enyi $\alpha$-divergence, called the sandwiched R\'enyi $\alpha$-divergence~\cite{Wilde2014,Lennert2013,Frank2013,Beigi2013}.
Let $\hat \rho$, $\hat \sigma$ be normalized states.
The sandwiched R\'enyi $\alpha$-divergence with $0< \alpha < 1$ and $1 < \alpha < \infty$ is defined as
\begin{equation}
S_\alpha (\hat \rho \| \hat \sigma ) := \frac{1}{\alpha -1} \ln \left(  {\rm tr} \left[   (\hat \sigma^{\frac{1 - \alpha}{2\alpha}} \hat \rho \hat \sigma^{\frac{1 - \alpha}{2\alpha}} ) ^{\alpha} \right] \right).
\label{def_sandwiched}
\end{equation}
It is known that $S_\alpha (\hat \rho \| \hat \sigma ) \geq 0$ holds, where the equality is achieved if and only if $\hat \rho = \hat \sigma$ (e.g., Theorem 5 of \cite{Beigi2013}).
The limit $\alpha \to \infty$ of  $S_\alpha (\hat \rho \| \hat \sigma ) $ equals $S_\infty (\hat \rho \| \hat \sigma)$ defined by Eq.~(\ref{def_Renyi_infty}) (Theorem 5 of \cite{Lennert2013}).
Also, the limit $\alpha \to 1$ gives the quantum KL divergence.
Thus, $S_\alpha (\hat \rho \| \hat \sigma )$ is well-defined for $0 < \alpha \leq \infty$. 
It is known~\cite{Lennert2013,Beigi2013} that $S_\alpha (\hat \rho \| \hat \sigma ) \leq S_{\alpha'} (\hat \rho \| \hat \sigma )$ for $\alpha \leq \alpha'$.
We note that  the fidelity $F(\hat \rho, \hat \sigma)$ corresponds to $S_{1/2}  (\hat \rho \| \hat \sigma )$ as shown in Eq.~(\ref{quantum_fidelity_Renyi}). 
It is also known that the sandwiched R\'enyi divergence satisfies the monotonicity under CPTP map $\mathcal E$ (Theorem 1 of \cite{Frank2013}): For  $1/2 \leq \alpha \leq \infty$,
\begin{equation}
S_\alpha (\hat \rho \| \hat \sigma ) \geq S_\alpha (\mathcal E (\hat \rho ) \| \mathcal E (\hat \sigma )).
\end{equation}
We remark that it has been further proved in Ref.~\cite{Hermes2017} that $S_\alpha (\hat \rho \| \hat \sigma ) $  satisfies the monotonicity under \textit{positive} and TP maps for $\alpha \geq 1$; in particular, the quantum KL divergence satisfies the monotonicity under positive and TP maps.
$S_{1/2}  (\hat \rho \| \hat \sigma )$  also satisfies the monotonicity under positive and TP maps, as the fidelity satisfies it as mentioned in Section~\ref{sec:misc}.

\

\noindent
\textit{Characterization of the KL divergence.}
So far, we have shown the monotonicity  of various divergence-like quantities.  
To put it in another way, the monotonicity is not sufficient to specify a single divergence such as the KL divergence.
Then, an interesting question is: Under which additional conditions, the KL divergence can be uniquely characterized?  This question has been answered in Refs.~\cite{Wilming2017,Matsumoto2010} as follows.

\begin{theorem}[Theorem 1 of~\cite{Wilming2017}]
Let $S(\hat \rho \| \hat \sigma)$ be a real-valued function of two normalized quantum states with $\hat \sigma$ being positive definite.  Suppose that the following four properties are satisfied:
\begin{description}
\item[Continuity:]  $S(\hat \rho \| \hat \sigma)$ is a continuous function of $\hat \rho$.
\item[Monotonicity:] For any CPTP map $\mathcal E$,  $S(\hat \rho \| \hat \sigma) \geq S(\mathcal E (\hat \rho ) \| \mathcal E ( \hat \sigma ))$.
\item[Additivity:] $S(\hat \rho_{\rm A} \otimes \hat \rho_{\rm B} \| \hat \sigma_{\rm A} \otimes \hat \sigma_{\rm B} ) = S(\hat \rho_{\rm A} \| \hat \sigma_{\rm A} ) + S( \hat \rho_{\rm B} \|  \hat \sigma_{\rm B} )$.
\item[Super-additivity:] Let $\hat \rho_{\rm AB}$ be a state of a composite system AB with partial states $\hat \rho_{\rm A}$ and $\hat \rho_{\rm B}$.  Then, $S(\hat \rho_{\rm AB} \| \hat \sigma_{\rm A} \otimes \hat \sigma_{\rm B}) \geq S(\hat \rho_{\rm A} \| \hat \sigma_{\rm A} )  + S(\hat \rho_{\rm B} \|  \hat \sigma_{\rm B}) $.
\end{description}
Then,  $S(\hat \rho \| \hat \sigma)$ equals the KL divergence up to normalization, i.e.,  $S(\hat \rho \| \hat \sigma) = C S_1 (\hat \rho \| \hat \sigma)$ for some constant $C>0$.
\end{theorem}

\section{Quantum Fisher information}
\label{sec:quantum_Fisher}

We briefly discuss a quantum generalization of the Fisher information in line with Ref.~\cite{Petz1996,Amari2000}, where again the concepts of operator monotone and operator concave play crucial roles.
The quantum Fisher information is related to quantum estimation theory, which we will not go into details in this book.
In contrast to the classical case discussed in Section~\ref{sec:classical_Fisher}, the quantum Fisher information is not unique, but can be characterized by operator monotone (and operator concave) functions.

We consider smooth parametrization of positive-definite states, written as $\hat \rho(\theta )$ with $\theta := (\theta^1, \theta^2, \cdots, \theta^m) \in \mathbb R^m$.
We denote $\partial_k := \partial / \partial \theta^k$.

\begin{definition}[Quantum Fisher information]
Let $f : (0 , \infty ) \to (0, \infty )$ be operator monotone and suppose that $f(0) := \lim_{x \to +0} f(x) \in [0,\infty )$ exists.
($f$ is thus operator concave from Proposition~\ref{prop:concave_monotone}.)  
Let  $\hat \rho$ be a positive-definite (normalized) state.
We define 
\begin{equation}
\mathcal K_{\hat \rho} := f(\mathcal D_{\hat \rho, \hat \rho}) \mathcal R_{\hat \rho}.
\end{equation}
Then, the quantum Fisher information matrix $K_{\hat \rho(\theta )}$, whose matrix component is written as $K_{\hat \rho(\theta ), kl}$, is defined as
\begin{equation}
K_{\hat \rho(\theta ), kl}  :=  \langle \partial_k \hat \rho(\theta ), \mathcal K_{\hat \rho(\theta )}^{-1} (\partial_l \hat \rho(\theta ) ) \rangle_{\rm HS} ={\rm tr} [ \partial_k \hat \rho(\theta ) \mathcal K_{\hat \rho(\theta )}^{-1} (\partial_l \hat \rho(\theta ) ) ] .
\end{equation}
\label{def:quantum_Fisher_information}
\end{definition}

We discuss two important examples.
A simplest case is $f(x) = 1$, for which the quantum Fisher information  is called the RLD (right logarithmic derivative) Fisher information.  
In this case,
\begin{equation}
K_{\hat \rho(\theta ), kl}   ={\rm tr} [ \partial_k \hat \rho(\theta ) \partial_l \hat \rho(\theta )  \hat \rho ( \theta)^{-1} ] .
\end{equation}
We note that this can be rewritten as
\begin{equation}
K_{\hat \rho(\theta ), kl}  ={\rm tr} [ \hat \rho ( \theta) \hat L_k (\theta ) \hat L_l (\theta )  ] ,
\end{equation}
where  $ \hat L_k (\theta )$ is  defined as the solution of $ \partial_k  \hat \rho (\theta ) = \hat \rho (\theta ) \hat L_k (\theta)$ and is given by $ \hat L_k (\theta ) = \hat \rho (\theta )^{-1} \partial_k \hat \rho (\theta )$.

Another example is $f(x) = (1+x)/2$, for which the quantum Fisher information  is called the SLD (symmetric logarithmic derivative) Fisher information.
In this case,  we have
$\mathcal K_{\hat \rho(\theta )} = ( \mathcal L_{\hat \rho (\theta )} + \mathcal R_{\hat \rho(\theta )} ) /2$, and  $\hat L_k (\theta ) :=  \mathcal K_{\hat \rho(\theta )}^{-1} (\partial_k \hat \rho (\theta))$
is given by the solution of $ \partial_k  \hat \rho (\theta ) = ( \hat \rho (\theta ) \hat L_k (\theta) + \hat L_k (\theta) \hat \rho (\theta ) ) /2$.
Thus,
\begin{equation}
K_{\hat \rho(\theta ), kl}  ={\rm tr} [  \partial_l  \hat \rho (\theta )  \hat L_k (\theta )  ] = \frac{1}{2} {\rm tr} [ \hat \rho ( \theta) ( \hat L_k (\theta ) \hat L_l (\theta ) +  \hat L_l (\theta ) \hat L_k (\theta ) )  ] .
\end{equation}

 These two versions of the quantum Fisher information are more or less straightforward generalizations of the classical Fisher information, but here we need to take  into account the fact that the logarithmic derivative $\partial_k \ln p(\theta )$ is not uniquely extended to the quantum case because of the non-commutability of operators.
We note that there is also another useful quantity called the ALD (anti-symmetric logarithmic derivative) Fisher information~\cite{Hayashi2005c}, while it is not one of the quantum Fisher information in the sense of of Definition~\ref{def:quantum_Fisher_information}.

The classical Fisher information is a special case of the quantum Fisher information:

\begin{lemma}
Suppose that $\hat \rho (\theta )$'s are diagonalizable in the same basis for all $\theta$.
Let $p (\theta)$ be the diagonal distribution of $\hat \rho (\theta )$.
Then, the quantum Fisher information matrix of $\hat \rho (\theta )$ for any $f$ reduces to the classical Fisher information matrix of $p (\theta)$ up to normalization.
\label{prop:Fisher_classical_quantum}
\end{lemma}

\begin{proof}
Let $\hat \rho (\theta ) := \sum_i p_i (\theta ) \hat P_i$ be the spectral decomposition.
From the expression like Eq.~(\ref{D_rho_sigma}), if $\hat X =  \sum_i x_i  \hat P_i$ is also diagonalizable in the same basis, we have
$\mathcal K_{\hat \rho(\theta )} (\hat X) = f(1) \sum_i p_i (\theta )  x_i \hat P_i$, 
and thus
 $\mathcal K_{\hat \rho(\theta )}^{-1} (\hat X) = \frac{1}{ f(1) } \sum_i \frac{x_i}{p_i (\theta )}  \hat P_i$. 
 By substituting $\hat X = \partial_l \hat \rho (\theta ) = \sum_i \partial_l p_i (\theta ) \hat P_i$, we obtain the claim of the lemma.
 $\Box$
\end{proof}

As in the classical case, we can adopt the information geometry perspective.
For $\hat \rho > 0$, the quantum Fisher information metric on the operator space is given by a map $G_{\hat \rho} : \mathcal L (\mathcal H) \times \mathcal L (\mathcal H ) \to \mathbb C$  defined as
\begin{equation}
G_{\hat \rho} (\hat X , \hat Y ) := \langle \hat X, \mathcal K_{\hat \rho}^{-1} (\hat Y ) \rangle_{\rm HS} = {\rm tr} [ \hat X^\dagger \mathcal K_{\hat \rho}^{-1} (\hat Y ) ].
\label{quantum_Fisher_metric}
\end{equation}
The quantum Fisher information matrix is  then represented as
\begin{equation}
K_{\hat \rho(\theta ), kl}  = G_{\hat \rho (\theta ) } ( \partial_k \hat \rho(\theta ) , \partial_l \hat \rho(\theta ) ).
\end{equation}

As in the classical case, the quantum Fisher information also satisfies the monotonicity, whose proof invoke a similar technique to the proof of the monotonicity of divergences (Theorem~\ref{quantum_f_monotonicity}).

\begin{theorem}[Theorem 3 of \cite{Petz1996}]
Let $\mathcal E$ be 2-positive and TP.  
Suppose that $\hat \rho$, $\mathcal E (\hat \rho)$ are (normalized) positive-definite states.
For any $\hat X \in \mathcal L (\mathcal H)$,
\begin{equation}
G_{\hat \rho} (\hat X , \hat X ) \geq G_{\mathcal E (\hat \rho ) } (\mathcal E (\hat X )  ,\mathcal E ( \hat X ) ). 
\label{quantum_Fisher_monotonicity_metric}
\end{equation}
In particular,  the quantum Fisher information matrix satisfies 
\begin{equation}
K_{\hat \rho(\theta )} \geq K_{\mathcal E (\hat \rho(\theta ))}.
\label{quantum_Fisher_monotonicity}
\end{equation}
\end{theorem}

\begin{proof}
It is sufficient to prove that
\begin{equation}
\mathcal K_{\hat \rho}^{-1} \geq \mathcal E^\dagger \mathcal K_{\mathcal E(\hat \rho)}^{-1} \mathcal E. 
\end{equation}
This is equivalent to
\begin{equation}
\mathcal I \geq \mathcal K_{\hat \rho}^{1/2}  \mathcal E^\dagger \mathcal K_{\mathcal E(\hat \rho)}^{-1} \mathcal E \mathcal K_{\hat \rho}^{1/2}  \ \ 
\Leftrightarrow \ \ 
 \mathcal I \geq   \mathcal K_{\mathcal E(\hat \rho)}^{-1/2} \mathcal E \mathcal K_{\hat \rho}  \mathcal E^\dagger  \mathcal K_{\mathcal E(\hat \rho)}^{-1/2} ,
\end{equation}
where $\mathcal I$ is the identity.  Thus, what we  will prove is that
\begin{equation}
\mathcal K_{\mathcal E(\hat \rho)} \geq \mathcal E  \mathcal K_{\hat \rho} \mathcal E^\dagger. 
\label{proof_F_monotonicity0}
\end{equation}
By using a contraction $\mathcal V$  in Lemma~\ref{lemma:V_contraction} by replacing $\hat \sigma$ with $\hat \rho$, we have
\begin{equation}
\mathcal E  \mathcal K_{\hat \rho} \mathcal E^\dagger 
= \mathcal R_{\mathcal E (\hat \rho)}^{1/2} \mathcal V^\dagger f (\mathcal D_{\hat \rho, \hat \rho} ) \mathcal V \mathcal R_{\mathcal E (\hat \rho)}^{1/2}.
\label{proof_F_monotonicity1}
\end{equation}
We apply  inequality~(\ref{proof_monotonicity4}), but now $f$ is operator concave and operator monotone with $f(0) \geq 0$:
\begin{equation}
\mathcal V^\dagger f (\mathcal D_{\hat \rho, \hat \rho} ) \mathcal V  \leq f  (  \mathcal D_{\mathcal E ( \hat \rho) , \mathcal E ( \hat \rho )} ).
\label{proof_F_monotonicity2}
\end{equation}
From inequalities (\ref{proof_F_monotonicity1}) and (\ref{proof_F_monotonicity2}), we finally obtain inequality~(\ref{proof_F_monotonicity0}).
$\Box$
\end{proof}

To relate the quantum Fisher information to parameter estimation, we consider quantum measurement on $\hat \rho (\theta )$, where the measurement itself is assumed to be independent of $\theta$.
Let $p(\theta )$ be the classical probability distribution obtained  by the measurement, and let $\hat p (\theta ) := \sum_i p(\theta ) \hat P_i$ be the corresponding diagonal density operator with a fixed basis represented by projectors $\{ \hat P_i \}$.
Let $J_{p ( \theta)}$ be the classical Fisher information matrix of  $p(\theta)$, which is equivalent to the quantum Fisher information matrix of $\hat p (\theta )$ from Lemma~\ref{prop:Fisher_classical_quantum} (here we take normalization $f(1) = 1$).
Since $\hat \rho (\theta ) \mapsto \hat p (\theta )$ is a CPTP map, the monotonicity (\ref{quantum_Fisher_monotonicity}) implies
\begin{equation}
K_{\hat \rho ( \theta ) } \geq J_{p ( \theta) },
\end{equation}
which is sometimes called the quantum Cramer-Rao inequality.  
By combining this with the classical Cramer-Rao bound~(\ref{Cramer_Rao}), we obtain a bound of the accuracy of unbiased estimation of $\theta$ from quantum measurement:
\begin{equation}
{\rm Cov}_\theta (\theta_{\rm est} )  \geq K_{\hat \rho ( \theta ) }^{-1},
\end{equation}
where the left-hand side depends on the choice of quantum measurement, but the right-hand side does not.

\

We finally note the quantum analogue of the Chentsov theorem (Theorem~\ref{thm:Chentsov}), which is referred to as the Petz's theorem.
The quantum monotone metric is defined as follows.

\begin{definition}[Quantum monotone metric]
Suppose that $G_{\hat \rho} : \mathcal L (\mathcal H ) \times \mathcal L (\mathcal H ) \to \mathbb C$ is defined for positive-definite states $\hat \rho$.
We call $G_{\hat \rho}$ a quantum monotone metric on the operator space, if it satisfies the following.
\begin{itemize}
\item $G_{\hat \rho}$ is sesquilinear (as is the case for the ordinary complex inner product).
\item  $G_{\hat \rho} (\hat X, \hat X) \geq 0$ holds for any $\hat \rho$, where the equality is achieved if and only if $\hat X = 0$.
\item $\hat \rho \mapsto G_{\hat \rho} (\hat X, \hat X) $ is continuous for any $\hat X$.
\item The monotonicity (\ref{quantum_Fisher_monotonicity_metric}) holds for any CPTP $\mathcal E$ and for any $\hat \rho$, $\hat X$.
\end{itemize}
\end{definition}

The quantum Fisher information metric (\ref{quantum_Fisher_metric}) is a quantum monotone metric.  The Petz's theorem states the converse:

\begin{theorem}[Petz's theorem, Theorem 5 of \cite{Petz1996}]
Any quantum monotone metric is a quantum Fisher information metric with some $f$.
\end{theorem}

We finally note that a metric is called symmetric, if 
$G_{\hat \rho} (\hat X, \hat Y) = G_{\hat \rho} (\hat Y^\dagger, \hat X^\dagger)$.
It is known that the quantum Fisher information metric is symmetric if and only if  $f$ satisfies $f(x) = x f(x^{-1})$ (Theorem 7 of Ref.~\cite{Petz1996}).
In this case, the corresponding quantum Fisher information matrix is a symmetric matrix.
For example, the SLD Fisher information is symmetric.
Moreover, it is known that there are the smallest and the largest symmetric Fisher information metrics (Corollary 9 of  Ref.~\cite{Petz1996});
In particular, the SLD metric is the smallest one.


\chapter{Hypothesis testing}
\label{appx:hypothesis_testing}

As discussed in Chapter~\ref{chap:approximate_asymptotic},  hypothesis testing is  an important tool to analyze the smooth divergences and their asymptotic limit.
In this Appendix, we briefly overview the basic concepts of quantum hypothesis testing.
In Section~\ref{sec:hypothesis_testing_divergence}, we introduce the hypothesis testing divergence and discuss its relation to semidefinite programming.
In Section~\ref{sec:quantum_stein_lemma}, we discuss a most fundamental theorem regarding quantum hypothesis testing: the quantum Stein's lemma, which is another representation of the quantum relative  AEP discussed in Section~\ref{sec_condition_information}.

\

The basic task of quantum hypothesis testing can be stated as follows:
 we want to distinguish two quantum states $\hat \rho$ and $\hat \sigma$ with $\hat \sigma$ being the false null hypothesis, and  minimize the error probability of the second kind given by ${\rm tr} [\hat \sigma \hat Q]$ with $0 \leq \hat Q \leq \hat I$, while keeping ${\rm tr}[\hat \rho \hat Q] \geq \eta$ for a given $0 < \eta < 1$.

\section{Hypothesis testing divergence}
\label{sec:hypothesis_testing_divergence}

The divergence-like quantity related to the above-mentioned error probability is the hypothesis testing divergence~\cite{Faist2018,Dupuis2012,Wang2012}.
For quantum states $\hat \rho$, $\hat \sigma$ and $0 < \eta < 1$, it is defined as 
\begin{equation}
S_{\rm H}^\eta (\hat \rho \| \hat \sigma) := - \ln \left( \frac{1}{\eta} \min_{0 \leq \hat Q \leq \hat I,  {\rm tr}[\hat \rho \hat Q] \geq \eta } {\rm tr} [\hat \sigma \hat Q] \right).
\label{def_SH}
\end{equation}
Because the argument of the logarithm above can be rewritten as $\min_{0 \leq \hat Q' \leq \hat I / \eta,  {\rm tr}[\hat \rho \hat Q'] \geq 1 } {\rm tr} [\hat \sigma \hat Q']$ with $\hat Q' := \hat Q / \eta$,
we have
\begin{equation}
S_{\rm H}^\eta (\hat \rho \| \hat \sigma) \geq S_{\rm H}^{\eta'} (\hat \rho \| \hat \sigma)  \ \ \ \rm{for} \ \ \   \eta \leq \eta'.
\label{eta_SH}
\end{equation}

We can also  define $S_{\rm H}^\eta (\hat \rho \| \hat \sigma) $ for subnormalized states $\hat \rho$, $\hat \sigma$.
We note the scaling property:
\begin{equation}
S_{\rm H}^\eta (\hat \rho \| \hat \sigma / Z ) = S_{\rm H}^\eta (\hat \rho \| \hat \sigma ) + \ln Z.
\label{hypothesis_scaling}
\end{equation}
Moreover, if subnormalized states satisfy $\hat \sigma \leq \hat \sigma'$,
\begin{equation}
S_{\rm H}^\eta (\hat \rho \| \hat \sigma' ) \leq S_{\rm H}^\eta (\hat \rho \| \hat \sigma ).
\end{equation}

A particularly important property of the hypothesis testing divergences in our context is that $S_{\rm H}^\eta (\hat \rho \| \hat \sigma)$ with $\eta \simeq 0$ and $\eta \simeq 1$ are essentially equivalent to the smooth $\infty$- and $0$-divergences, respectively:

\begin{proposition}[Lemma 40 of  Ref.~\cite{Faist2018}; Proposition 1 of Ref.~\cite{Sagawa2019}]
For  $0 < \varepsilon < 1/2$,
\begin{equation}
 S_{\rm H}^{1-\varepsilon^2 / 6}(\hat \rho \| \hat \sigma ) - \ln \frac{1-\varepsilon^2 /6}{\varepsilon^2 / 6} \leq S_0^\varepsilon (\hat \rho \| \hat \sigma ) \leq S_{\rm H}^{1-\varepsilon} (\hat \rho \| \hat \sigma) - \ln (1-\varepsilon),
\end{equation}
\begin{equation}
S_{\rm H}^{2\varepsilon} (\hat \rho \| \hat \sigma ) -\ln 2 \leq S_\infty^\varepsilon (\hat \rho \| \hat \sigma ) \leq S_{\rm H}^{\varepsilon^2 / 2}(\hat \rho \| \hat \sigma ) - \ln (1 -  \varepsilon).
\label{Faist_prop2}
\end{equation}
\label{Faist_prop}
\end{proposition}

See Ref.~\cite{Faist2018} for the proof.
We note that Ref.~\cite{Faist2018,Sagawa2019} adopted different ways of smoothing, and thus inequality~(\ref{Faist_prop2}) is slightly different (see also Section~\ref{sec:misc}).

We next discuss the intuitive meaning of  Proposition~\ref{Faist_prop} in a non-rigorous manner.
For the case of $S_0^\varepsilon (\hat \rho \| \hat \sigma )$, by letting $\eta \simeq 1$, we have
\begin{equation}
S_{\rm H}^{\eta \simeq 1} (\hat \rho \| \hat \sigma) \simeq - \ln \left(  \min_{0 \leq \hat Q \leq \hat I, \ {\rm tr}[\hat \rho \hat Q] \simeq  1  } {\rm tr} [\hat \sigma \hat Q] \right) \simeq - \ln \left( {\rm tr} [\hat \sigma \hat P_{\hat \rho}]  \right) \simeq S^{\varepsilon \simeq 0}_0  (\hat \rho \| \hat \sigma ),
\label{SH_S0_1}
\end{equation}
where we chose $\hat Q$ as the projection on the support of $\hat \rho$, written as $\hat P_{\hat \rho}$.
See also the proof of Theorem~\ref{thm:asymp0_q}, where we discussed that  $S_0 (\hat \rho \| \hat \sigma )$ can be identified with $S_{\rm H}^{\eta = 1}  (\hat \rho \| \hat \sigma )$.

For the case of $S_\infty^\varepsilon (\hat \rho \| \hat \sigma )$, we invoke the dual expression of Eq.~(\ref{def_SH}), which is given by Eq.~(\ref{strong_duality}) at the end of this section.  By letting $\eta \simeq 0$, $ - {\rm tr}[\hat X] / \eta$ in Eq.~(\ref{strong_duality}) would take a very large negative value, and thus $ {\rm tr}[\hat X]  \simeq 0$ should hold for the maximum. Then,
\begin{equation}
S_{\rm H}^{\eta \simeq 0} (\hat \rho \| \hat \sigma) 
\simeq - \ln \left( \max_{\mu \geq 0, \mu \hat \rho \lessapprox \hat \sigma} \mu \right)  \simeq \ln \| \hat \sigma^{-1/2} \hat \rho \hat \sigma^{-1/2} \|_\infty \simeq S_\infty^{\varepsilon \simeq 0} (\hat \rho \| \hat \sigma ).
\end{equation}

\

We note some general properties of the hypothesis testing divergences.
First,  $S_{\rm H}^\eta (\hat \rho \| \hat \sigma)$ satisfies the monotonicity in the following form:

\begin{proposition}
For $0 < \eta < 1$  and any positive and trace-nonincreasing  map $\mathcal E$, 
\begin{equation}
S_{\rm H}^\eta (\hat \rho \| \hat \sigma) \geq S_{\rm H}^\eta (\mathcal E(\hat \rho ) \| \mathcal E( \hat \sigma ) ).
\end{equation}
\label{prop:H_monotone}
\end{proposition}

\begin{proof}
We note that the trace-nonincreasing map satisfies ${\rm tr}[\mathcal E (\hat \rho )] \leq {\rm tr}[\hat \rho]$ for any $\hat \rho \geq 0$.
Its adjoint is sub-unital, i.e., $\mathcal E^\dagger (\hat I ) \leq \hat I$.
Then  we  have
\begin{eqnarray}
\min_{0 \leq \hat Q \leq \hat I,  {\rm tr}[\mathcal E(\hat \rho ) \hat Q] \geq \eta } {\rm tr} [\mathcal E(\hat \sigma ) \hat Q] 
&=& \min_{0 \leq \hat Q \leq \hat I,  {\rm tr}[\hat \rho \mathcal E^\dagger ( \hat Q )] \geq \eta } {\rm tr} [\hat \sigma \mathcal E^\dagger (\hat Q )] \\
&\geq& \min_{0 \leq \hat Q' \leq \hat I,  {\rm tr}[\hat \rho \hat Q'] \geq \eta } {\rm tr} [\hat \sigma \hat Q'],
\end{eqnarray}
where we used that $0 \leq \hat Q \leq \hat I$ implies $0 \leq \mathcal E^\dagger (\hat Q ) \leq \mathcal E^\dagger (\hat I ) \leq \hat I$ for positive and sub-unital $\mathcal E^\dagger$ to obtain the second line.
$\Box$
\end{proof}


\begin{lemma}[Inequality (22) of~\cite{Sagawa2019}]
Let $\varepsilon \geq 0$ and $D( \hat \rho, \hat \tau ) \leq \varepsilon$.  For any $0 < \eta < 1 - \varepsilon$,
\begin{equation}
S_{\rm H}^{\eta + \varepsilon } (\hat \tau \| \hat \sigma ) \leq S_{\rm H}^\eta (\hat \rho \| \hat \sigma ) + \ln \left( \frac{\eta + \varepsilon }{\eta} \right).
\label{Sagawa_Faist_ineq_H}
\end{equation}
\label{lemma:Sagawa_Faist_ineq}
\end{lemma}

\begin{proof}
There exists $\hat \Delta$ such that $\hat \rho + \hat \Delta \geq \hat \tau$, $\hat \Delta \geq 0$, and ${\rm tr}[\hat \Delta] \leq \varepsilon$.
If $0 \leq \hat Q \leq \hat I$, we have ${\rm tr}[\hat \tau \hat Q] \leq {\rm tr}[(\hat \rho + \hat \Delta ) \hat Q] \leq  {\rm tr}[\hat \rho \hat Q] + \varepsilon$.
Now let $\hat Q$ be the optimal choice for the left-hand side of (\ref{Sagawa_Faist_ineq_H}).
Then, $ {\rm tr}[\hat \rho \hat Q] \geq \eta$ holds and thus
\begin{equation}
\frac{1}{\eta + \varepsilon} {\rm tr} [\hat \sigma \hat Q]
\geq
\frac{\eta}{\eta + \varepsilon}\frac{1}{\eta} \min_{0 \leq \hat Q' \leq \hat I, {\rm tr}[\hat \rho \hat Q'] \geq \eta } {\rm tr} [\hat \sigma \hat Q'].
\end{equation}
$\Box$
\end{proof}

\begin{lemma}
Let $\hat \rho$, $\hat \sigma_k$ be subnormalized states and define $\hat \sigma := \sum_k \hat \sigma_k \otimes  r_k | k \rangle \langle k |$ with $\{ | k \rangle \} $ being orthonormal and $r_k> 0$.
Then,
\begin{equation}
S_{\rm H}^\eta ( \hat \rho \otimes | k \rangle \langle k | \| \hat \sigma ) = S_{\rm H}^\eta ( \hat \rho \| \hat \sigma_k ) - \ln r_k.
\label{Sagawa_Faist_eq}
\end{equation}
\label{lemma:Sagawa_Faist_eq}
\end{lemma}

\begin{proof}
Let $\hat Q$ be the optimal choice for  $S_{\rm H}^\eta ( \hat \rho \otimes | k \rangle \langle k | \| \hat \sigma  )$.
 Then $\langle k | \hat Q | k \rangle$ is a candidate for $S_{\rm H}^\eta ( \hat \rho \| \hat \sigma_k )$ because ${\rm tr}[\hat \rho \otimes  | k \rangle \langle k | \hat Q]  = {\rm tr}[\hat \rho \langle k | \hat Q | k \rangle]$.  
From this and ${\rm tr}[\hat \sigma \hat Q ] \geq {\rm tr}[ \hat \sigma_k \langle k | \hat Q |k \rangle ] r_k$, we have $S_{\rm H}^\eta ( \hat \rho \otimes | k \rangle \langle k | \| \hat \sigma ) \leq S_{\rm H}^\eta ( \hat \rho \| \hat \sigma_k ) - \ln r_k$.

Conversely, let $\hat Q'$ be the optimal choice for $S_{\rm H}^\eta ( \hat \rho \| \hat \sigma_k )$.
Then $\hat Q' \otimes | k \rangle \langle k |$ is a candidate for  $S_{\rm H}^\eta ( \hat \rho \otimes | k \rangle \langle k | \| \hat \sigma )$.
From this and ${\rm tr}[\hat \sigma_k \hat Q'] = {\rm tr} [\hat \sigma  \hat Q' \otimes | k \rangle \langle k  |]r_k^{-1} $, we obtain $S_{\rm H}^\eta ( \hat \rho \| \hat \sigma_k ) \leq S_{\rm H}^\eta ( \hat \rho \otimes | k \rangle \langle k | \| \hat \sigma  )  + \ln r_k $. 
$\Box$
\end{proof}

We note that the above lemma also holds for $S_0$ and $S_\infty$.
However, it does not hold for $S_0^\varepsilon$ and $S_\infty^\varepsilon$ with $\varepsilon > 0$, which is the reason why we used the hypothesis testing divergence in Theorem~\ref{thm:hypothesis_testing_work_ap} instead of directly addressing $S_0^\varepsilon$ and $S_\infty^\varepsilon$.

\

We next briefly overview  semidefinite programing  in line with Refs.~\cite{Faist2018,Dupuis2012,Watrous2009}, from which we obtain the dual expression of the hypothesis testing divergence.
(See Ref.~\cite{Boyd} for a comprehensive textbook of convex optimization.)
We note that our terminologies ``primal'' and ``dual'' are the same as in Ref.~\cite{Watrous2009}.
On the other hand,  ``dual'' and ``primal'' are exchanged in some papers (e.g., Ref.~\cite{Dupuis2012}), where the standard hypothesis testing becomes the ``primal'' program.

Let $\hat C$ and $\hat D$ be Hermitian matrices and $\mathcal E$ be a Hermitian-preserving map.
The primal program is the optimization problem given by
\begin{equation}
\sup_{\hat X \geq 0,  \mathcal E(\hat X) \leq \hat D }{\rm tr}[\hat C \hat X ]  =: \alpha,
\end{equation}
while the dual program is given by
\begin{equation}
\inf_{\hat Y \geq 0, \mathcal E^\dagger (\hat Y) \geq \hat C} {\rm tr} [\hat D \hat Y] =: \beta.
\end{equation}
Without any further assumption, we have
\begin{equation}
\alpha \leq \beta,
\label{weak_duality}
\end{equation}
which is called the weak duality.  
The proof is straightforward: For $\hat X \geq 0$ with $\mathcal E(\hat X) \leq \hat D$ and $\hat Y \geq 0$ with $\mathcal E^\dagger (\hat Y) \geq \hat C$, we have
\begin{equation}
{\rm tr}[\hat C \hat X ] \leq {\rm tr}[\mathcal E^\dagger (\hat Y) \hat X] = {\rm tr} [\hat Y \mathcal E (\hat X)] \leq  {\rm tr} [\hat Y \hat D],
\end{equation}
and then take the supremum over $\hat X$ and the infimum over $\hat Y$.

Moreover, it is known that the \textit{strong duality} holds under certain conditions, such as  the Slater's conditions (e.g., Theorem 2.2 of \cite{Watrous2009}):
\begin{itemize}
\item If $\beta$ is finite and there exists $\hat X>0$ such that $ \mathcal E (\hat X ) < \hat D$, then $\alpha = \beta$ holds and there exists $\hat Y$ such that $\hat Y \geq 0$, $\mathcal E^\dagger (\hat Y ) \geq \hat C$ and  ${\rm tr} [\hat D \hat Y] = \beta$.
\item  If $\alpha$ is finite and there exists $\hat Y>0$ such that $ \mathcal E^\dagger (\hat Y ) > \hat C$, then $\alpha = \beta$ holds and there exists $\hat X$ such that $\hat X \geq 0$, $\mathcal E (\hat X ) \leq \hat D$ and  ${\rm tr} [\hat C \hat X] = \alpha$.
\end{itemize}

Our hypothesis testing is a dual program in the above definition, where $\hat D := \hat \sigma$, $\hat Y := \hat Q$,  and
\begin{eqnarray}
\mathcal E^\dagger(\hat Y) := \left[
\begin{array}{cc}
-\hat Y & 0 \\
0 & {\rm tr}[\hat \rho  \hat Y]
\end{array}
\right], \ \ 
\hat C := \left[
\begin{array}{cc}
-\hat I & 0 \\
0 & \eta
\end{array}
\right].
\end{eqnarray}
We note that $e^{-S_{\rm H}^\eta (\hat \rho \| \hat \sigma)} = \eta^{-1} \beta$.
The corresponding primal program is given by
\begin{equation}
\mathcal E( \hat X) = \mu \hat \rho - \hat X_{11}, 
\ \ {\rm where} \ \
\hat X =
\left[
\begin{array}{cc}
\hat X_{11} & \hat X_{12} \\
\hat X_{21} & \mu
\end{array}
\right] 
\ \ {\rm with} \ \mu \in \mathbb R.
\end{equation}
The condition $\mathcal E(\hat X) \leq \hat D$ then gives $\mu \hat \rho - \hat X_{11} \leq \hat \sigma$, and ${\rm tr}[\hat C \hat X] = \eta \mu - {\rm tr}[X_{11}]$.
From the weak duality (\ref{weak_duality}), we have  $e^{-S_{\rm H}^\eta (\hat \rho \| \hat \sigma)} \geq \eta^{-1} \alpha$.
Moreover, it is easy to see that there exists $\hat Y > 0$ such that $\hat Y < \hat I$ and ${\rm tr}[\hat \rho \hat Y] > \eta$.
Thus, the condition of the strong duality is satisfied, and we have  $e^{-S_{\rm H}^\eta (\hat \rho \| \hat \sigma)} = \eta^{-1} \alpha$  and we can replace sup by max for $\alpha$.
Then, by rewriting $\hat X_{11}$ as just $\hat X$ and noting that $\hat X_{12}$ and $\hat X_{21}$ are irrelevant to the optimization, we obtain ~\cite{Faist2018,Dupuis2012}
\begin{equation}
S_{\rm H}^\eta (\hat \rho \| \hat \sigma ) = - \ln \left( \max_{\mu \geq 0,  \hat X \geq 0,  \mu \hat \rho \leq \hat \sigma + \hat X} \left\{ \mu - \frac{{\rm tr}[\hat  X]}{\eta } \right\}  \right).
\label{strong_duality}
\end{equation}

\section{Quantum Stein's lemma}
\label{sec:quantum_stein_lemma}

We consider the asymptotic behavior of the hypothesis testing divergence, which is characterized by the quantum Stein's lemma.
See also Appendix~\ref{app_classical} for the classical Stein's lemma.

First, we define the hypothesis testing divergence rate for sequences $\widehat{P} = \{ \hat \rho_n \}_{n=1}^\infty$ and $\widehat{\Sigma} = \{ \hat \sigma_n \}_{n=1}^\infty$ of quantum states:
\begin{equation}
S_{\rm H}^\eta (\widehat{P} \| \widehat{\Sigma}) := \lim_{n \to \infty}\frac{1}{n} S_{\rm H}^\eta (\hat \rho_n \| \hat \sigma_n).
\end{equation}
We note that the limit does not necessarily exist, while we can always define the upper and lower limits, which are equivalent to the lower and upper spectral divergence rates introduced in Section~\ref{sec:asymptotic}. In fact, from  Proposition \ref{Faist_prop}:
\begin{equation}
\lim_{\varepsilon \to +0}\limsup_{n \to \infty} \frac{1}{n} S_{\rm H}^\varepsilon (\hat \rho_n \| \hat \sigma_n) = \overline{S} (\widehat{P} \| \widehat{\Sigma}), 
\label{hypothesis_spectral_eq1}
\end{equation}
\begin{equation}
\lim_{\varepsilon \to +0}\liminf_{n \to \infty} \frac{1}{n} S_{\rm H}^{1-\varepsilon} (\hat \rho_n \| \hat \sigma_n)  = \underline{S} (\widehat{P} \| \widehat{\Sigma}).
\label{hypothesis_spectral_eq2}
\end{equation}

The quantum Stein's lemma states that the minimized error probability of  the second kind  is asymptotically characterized by  the KL divergence rate.
More precisely: 
\textit{For any} $0 < \eta < 1$\textit{,}
\begin{equation}
S_{\rm H}^\eta  (\widehat{P} \| \widehat{\Sigma})  = S_1  (\widehat{P} \| \widehat{\Sigma})
\label{Stein_lemma}
\end{equation}
\textit{holds.}

From Eq.~(\ref{def_SH}), we see that Eq.~(\ref{Stein_lemma}) implies that ${\rm tr}[\hat \sigma_n \hat Q_n ] \sim e^{-n S_1 (\widehat{P} \| \widehat{\Sigma})}$.
That is, the quantum Stein's lemma states that the KL divergence rate characterizes the large deviation property of the error probability in hypothesis testing.

The quantum Stein's lemma can be equivalently rephrased in terms of the spectral divergence rates from Eq.~(\ref{hypothesis_spectral_eq1}), Eq.~(\ref{hypothesis_spectral_eq2}),  and inequality~(\ref{eta_SH}):

\begin{proposition}
The quantum Stein's lemma (\ref{Stein_lemma}) holds, if and only if 
\begin{equation}
\underline{S} (\widehat{P} \| \widehat{\Sigma}) = \overline{S} (\widehat{P} \| \widehat{\Sigma}) =  S_1  (\widehat{P} \| \widehat{\Sigma}).
\label{eq_Stein_spectral}
\end{equation}
\label{hypothesis_spectral}
\end{proposition}

\begin{proof}
This is obvious, but we here prove the ``if'' part.
For $0 < \eta \leq 1/2$, we have $S_1 (\widehat{P} \| \widehat{\Sigma}) = \lim_{\varepsilon \to +0} \limsup_{n \to \infty} S_{\rm H}^\varepsilon (\hat \rho_n \| \hat \sigma_n ) / n \geq \limsup_{n \to \infty} S_{\rm H}^\eta (\hat \rho_n \| \hat \sigma_n ) / n \geq \liminf_{n \to \infty} S_{\rm H}^\eta (\hat \rho_n \| \hat \sigma_n ) / n \geq \liminf_{n \to \infty} S_{\rm H}^{1-\eta} (\hat \rho_n \| \hat \sigma_n ) / n \geq \lim_{\varepsilon \to +0} \liminf_{n \to \infty} S_{\rm H}^{1-\varepsilon} (\hat \rho_n \| \hat \sigma_n ) / n  = S_1 (\widehat{P} \| \widehat{\Sigma})$.
The proof is almost the same for $1/2 \leq \eta < 1$.
$\Box$
\end{proof}

Eq.~(\ref{eq_Stein_spectral}) is exactly the same as the relative quantum AEP~(\ref{eq_relative_quantum_AEP}).
Moreover, from Proposition~\ref{Faist_prop}, we have the following~\cite{Faist2018}.

\begin{proposition}
The quantum Stein's lemma (\ref{Stein_lemma}) holds, if and only if for any $0 < \varepsilon < 1/2$,
\begin{equation}
\lim_{n \to \infty}\frac{1}{n} S_0^\varepsilon (\hat \rho_n \| \hat \sigma_n ) 
=\lim_{n \to \infty}\frac{1}{n} S_\infty^\varepsilon (\hat \rho_n \| \hat \sigma_n )  
=   S_1  (\widehat{P} \| \widehat{\Sigma}).
\end{equation}
\end{proposition}

To summarize, in order to prove the collapse of $\overline{S}(\widehat{P}\| \widehat{\Sigma} )$ and  $\underline{S}(\widehat{P} \| \widehat{\Sigma})$ to $S_1(\widehat{P}\| \widehat{\Sigma} )$,
it is necessary and  sufficient to prove the quantum Stein's lemma.
In other words, the quantum Stein's lemma  is regarded as a representation of the quantum relative AEP (see also Appendix~\ref{app_classical} for the classical case).
In fact, if $\widehat{P}$ is translation invariant and ergodic and $\widehat{\Sigma}$ is the Gibbs state of a local and translation invariant Hamiltonian, 
then the quantum Stein's lemma (\ref{Stein_lemma}) holds, which is equivalent to Theorem \ref{thm:main_ergodic}  of  Section~\ref{sec_condition_information} (Theorem 3 of Ref.~\cite{Sagawa2019}).

The simplest case is i.i.d.   with  $\widehat{P}:= \{ \hat \rho^{\otimes n} \}$ and $\widehat{\Sigma} := \{ \hat \sigma^{\otimes n} \}$.
In this case, the quantum Stein's lemma states that   for any $0 < \eta < 1$, 
\begin{equation}
S_{\rm H}^\eta (\widehat{P} \| \widehat{\Sigma}) =  S_1(\hat \rho \| \hat \sigma ),
\end{equation}
which has been proved in Refs.~\cite{Hiai_Petz,Ogawa2000}.
Given Proposition~\ref{hypothesis_spectral},  this is equivalent to Theorem~\ref{thm:iid_divergence} (Theorem 2 of Ref.~\cite{Nagaoka2007}).


\chapter{Classical asymptotic equipartition properties}
\label{app_classical}

In this Appendix, we consider the classical asymptotic equipartition properties (AEP) of stochastic processes, as the classical counterpart  of Section~\ref{sec_condition_information}.
Specifically, we discuss the Shannon-McMillan theorem and the classical Stein's lemma.

\

Let $\{ x_l \}_{l\in \mathbb Z}$ be a two-sided stochastic process, where $x_l \in B$ with $B$ being a finite set of alphabets.
This is equivalent to a one-dimensional classical spin system.
Formally, the stochastic process is given by a measure $\mu$ over $K := B^{\mathbb Z}$.
We introduce a shift operator $T$ such that $(Tx)_l := x_{l+1}$.

Let $X_n := (x_{-l}, x_{-l+1}, \cdots, x_l)$ with $n:= 2l+1$.
We consider sequences of probability distributions  $\widehat{P} := \{ \rho_n (X_n) \}$ and $\widehat{\Sigma} := \{ \sigma_n (X_n ) \}$.
The spectral divergence rates of these sequences can be defined in completely parallel manner to the quantum case discussed in Chapter~\ref{chap:approximate_asymptotic}.
We sometimes loosely identify $\{ x_l \}_{l\in \mathbb Z}$ and $\widehat{P}$.

In the classical case, we can directly show that if $\widehat{P}$ and $\widehat{\Sigma}$ satisfy a relative version of the AEP: the lower and the upper divergence rates coincide and further equal the KL divergence rate.
We first define the AEP  in the form of convergence in probability.

\begin{definition}[Classical AEP]
$\widehat{P}$ satisfies  the AEP, if the Shannon entropy rate $S_1 (\widehat{P} )$ exists, and $- \frac{1}{n} \ln \rho_n (X_n)$ converges to $S_1 (\widehat{P} )$ in probability by sampling $X_n$ according to $\rho_n$.
\label{def_AEP}
\end{definition}

This definition is equivalent to a statement with \textit{typical set} as described in the following proposition, which can be shown in the same manner as in Theorem 3.1.2 of Ref.~\cite{Cover_Thomas}.

\begin{proposition}
$\widehat{P}$ satisfies the AEP, if and only if for any $ 0 < \varepsilon < 1$, there exists a typical set $Q_n^\varepsilon \subset \{ X_n \}$, satisfying the following properties for sufficiently large $n$:
\begin{description}
\item[(a)] For any $X_n \in Q_n^\varepsilon$, $\exp ( - n(S_1 (\widehat{P} ) + \varepsilon ) ) \leq \rho_n (X_n) \leq \exp (- n(S_1 (\widehat{P}) - \varepsilon ))$.
\item[(b)] $\rho_n [Q_n^\varepsilon] > 1-\varepsilon$, where $\rho_n [Q_n^\varepsilon]$ is the probability of $Q_n^\varepsilon$ according to the distribution $\rho_n$.
\item[(c)] $(1-\varepsilon) \exp (n (S_1 (\widehat{P}  ) - \varepsilon)) \leq   | Q_n^\varepsilon | \leq \exp (n (S_1 (\widehat{P} ) + \varepsilon))$, where $| Q_n^\varepsilon |$ describes the number of elements of $Q_n^\varepsilon $.
\end{description} 
\label{prop_c_AEP0}
\end{proposition}

This formulation of the AEP clearly represents the meaning of ``equipartition'';
Almost all events (i.e., events in the typical set) have almost the same probability.

We next define the relative AEP, which is the classical counterpart of our quantum formulation in Section~\ref{sec_condition_information}.

\begin{definition}[Relative AEP]
$\widehat{P}$ and $\widehat{\Sigma}$ satisfy the relative AEP, if  the KL divergence rate $S_1 (\widehat{P} \| \widehat{\Sigma})$ exists, and $\frac{1}{n} \ln \frac{\rho_n (X_n) }{\sigma_n (X_n)}$ converges to $S_1 (\widehat{P} \| \widehat{\Sigma})$ in probability by sampling $X_n$ according to $\rho_n$.
\label{def_rAEP}
\end{definition}

The following proposition can be shown in the same manner as Theorem 11.8.2 of Ref.~\cite{Cover_Thomas}.

\begin{proposition}
$\widehat{P}$ and $\widehat{\Sigma}$ satisfy the relative AEP, if and only if for any $ 0 < \varepsilon < 1$, there exists a relative typical set $Q_n^\varepsilon \subset \{ X_n \}$, satisfying the following properties for sufficiently large $n$:
\begin{description}
\item[(a)] For any $X_n \in Q_n^\varepsilon$,
\begin{equation}
\exp ( n(S_1 (\widehat{P} \| \widehat{\Sigma}) - \varepsilon ) ) \leq \frac{\rho_n (X_n)}{\sigma_n (X_n)} \leq \exp ( n(S_1 (\widehat{P} \| \widehat{\Sigma}) + \varepsilon )).
\end{equation}
\item[(b)] $\rho_n [Q_n^\varepsilon] > 1-\varepsilon$.
\item[(c)] $(1-\varepsilon) \exp (-n (S_1 (\widehat{P} \| \widehat{\Sigma} ) + \varepsilon)) \leq \sigma_n [Q_n^\varepsilon] \leq \exp (-n (S_1 (\widehat{P} \| \widehat{\Sigma} ) - \varepsilon))$.
\end{description} 
Here, $\rho_n [Q_n^\varepsilon]$ and $\sigma_n [Q_n^\varepsilon]$ represent the probability of $Q_n^\varepsilon$ according to distributions $\rho_n$ and $\sigma_n$, respectively.
\label{prop_c_AEP}
\end{proposition}

Next, to formulate the classical Shannon-McMillan theorem, we define classical ergodicity as follows.

\begin{definition}[Classical ergodicity]
 $\{ x_l \}_{l \in \mathbb Z}$ is translation invariant if $\mu$ is invariant under $T$.
Moreover, a  translation-invariant process $\{ x_l \}_{l \in \mathbb Z}$  is ergodic, if any subset of $K$ that is invariant under $T$ has measure $0$ or $1$. 
\end{definition}

The above definition is a special case of quantum ergodicity (Definition~\ref{def:q_ergodic}) for the situation where the algebra of observables is commutative, which is highlighted by the following lemma.

\begin{lemma}
A two-sided stochastic process is translation invariant and ergodic, if and only if it is an extreme point of the set of translation-invariant stochastic processes.
\end{lemma}

We now state the classical Shannon-McMillan(-Breiman) theorem.  The weaker statement with convergence in probability is enough for our purpose in the following argument.

\begin{theorem}[Classical Shannon-McMillan theorem]
If $\widehat{P}$ is translation invariant and ergodic, $- \frac{1}{n}\ln \rho_n (X_n)$ converges to $S_1 (\widehat{P})$ almost surely (and thus in probability, implying the AEP).
\end{theorem}

In the case of i.i.d., the Shannon-McMillan theorem immediately follows from the law of large numbers, because $- \ln \rho (X_n) = - \sum_{l=-n}^n \ln \rho (x_l)$.
In general, the Birkhoff's ergodic theorem plays a crucial role in the proof of the Shannon-McMillan theorem (see, for example, Theorem 16.8.1 of ~\cite{Cover_Thomas}).

\begin{theorem}[Birkhoff's ergodic theorem]
Let $f$ be a measurable function from $K$ to $\mathbb R$.
If $\{ x_l \}_{l \in \mathbb Z}$ is translation invariant and ergodic, then for almost every $\omega \in K$, 
\begin{equation}
\lim_{n \to \infty} \frac{1}{2n+1} \sum_{l=-n}^n f(T^l \omega ) = \int d\mu (\omega ) f(\omega ).
\end{equation}
\end{theorem}

The Shannon-McMillan theorem can be generalized to a relative version (see, e.g., Ref.~\cite{Algoet1988}), if $\widehat{P}$ is ergodic and $\widehat{\Sigma}$ is Markovian (of finite order).
Here,  $\widehat{\Sigma}$ is Markovian of order $m$, if  for any $n$ the conditional probability satisfies $\sigma (x_n | x_{n-1}, x_{n-2}, \cdots ) = \sigma (x_n | x_{n-1}, \cdots, x_{n-m})$ almost surely (e.g., Ref.~\cite{Doob1990}).

\begin{proposition}[Classical relative Shannon-McMillan theorem]
If $\widehat{P}$ is translation invariant  and ergodic and $\widehat{\Sigma}$ is translation invariant and Markovian (of finite order), then by sampling $X_n$ according to $\rho_n$, $\frac{1}{n} \ln \frac{\rho_n (X_n) }{\sigma_n (X_n)}$ converges to $S_1 (\widehat{P} \| \widehat{\Sigma})$ almost surely (and thus in probability, implying that $\widehat{P}$ and $\widehat{\Sigma}$ satisfy the relative AEP).
\end{proposition}

\begin{proof}
We only need to consider the term of $\ln \sigma_n (X_n)$. 
From the assumption of Markovian, this term is decomposed into the sum of the logarithm of conditional probabilities almost surely, except for the surface term that does not contribute in the limit.
We can then apply the Birkhoff's ergodic theorem.
$\Box$
\end{proof}



Now, the following is the classical counterpart of Theorem~\ref{thm:main_ergodic}.

\begin{proposition}
If $\widehat{P}$ and $\widehat{\Sigma}$ satisfy the relative AEP, we have
\begin{equation}
\underline{S}(\widehat{P} \| \widehat{\Sigma}) = \overline{S} (\widehat{P} \| \widehat{\Sigma})  = S_1 (\widehat{P} \| \widehat{\Sigma}).
\label{c_spectral_KL}
\end{equation}
\label{AEP_spectral}
\end{proposition}

We remark that the Gibbs state of a local and translation-invariant Hamiltonian in one dimension is translation invariant and Markovian (of finite order).
This is essentially the  Hammersley-Clifford  theorem~\cite{Hammersley1971} (see also Ref.~\cite{Kato2016}), and the case of infinite systems (as considered here) is proved in Lemma 19 of Ref.~\cite{Sagawa2019}.

We can directly prove Proposition~\ref{AEP_spectral} from the definition of the information spectrum (see Proposition 21 of Ref.~\cite{Sagawa2019}).
Here, we instead present a proof of the classical Stein's lemma, from which Proposition~\ref{AEP_spectral} follows by using Proposition \ref{hypothesis_spectral}.
The following proof is a slight modification of Theorem 11.8.3 of Ref.~\cite{Cover_Thomas}.

\begin{proposition}[Classical Stein's lemma]
If $\widehat{P}$ and $\widehat{\Sigma}$ satisfy the relative AEP, then for any $0 < \eta < 1$,
\begin{equation}
S_{\rm H}^\eta  (\widehat{P} \| \widehat{\Sigma})  = S_1  (\widehat{P} \| \widehat{\Sigma})
\label{c_Stein_lemma}.
\end{equation}
\end{proposition}

\begin{proof}
We first note that the classical hypothesis testing divergence is given by
\begin{equation}
S_{\rm H}^\eta (\rho_n \| \sigma_n ) := - \ln \left( \frac{1}{\eta} \min_{Q_n : \rho_n [Q_n] \geq \eta} \sigma_n [ Q_n ] \right).
\label{classical_hypothesis_divergence}
\end{equation}
Take $0< \varepsilon < 1- \eta$.
Let $Q_n^\varepsilon$ be defined in Proposition \ref{prop_c_AEP}.
From Proposition \ref{prop_c_AEP} (b), we have $\rho_n [ Q_n^\varepsilon ] > 1- \varepsilon > \eta$, and thus $Q_n^\varepsilon$ is a candidate for the minimization in Eq.~(\ref{classical_hypothesis_divergence}).
From this and Proposition \ref{prop_c_AEP} (c), we obtain
\begin{equation}
S_{\rm H}^\eta (\rho_n \| \sigma_n ) \geq - \ln \left( \frac{1}{\eta} \sigma_n [Q_n^\varepsilon ] \right) \geq n\left( S_1 (\widehat{P} \| \widehat{\Sigma} )  - \varepsilon \right) + \ln \eta.
\end{equation}
By taking the limit,
$S_{\rm H}^\eta (\widehat{P} \| \widehat{\Sigma} ) := \lim_{n \to \infty} \frac{1}{n} S_{\rm H}^\eta (\rho_n \| \sigma_n ) \geq S_1 (\widehat{P} \| \widehat{\Sigma} ) - \varepsilon$.
Since $\varepsilon > 0$ can be taken arbitrarily small, we have 
\begin{equation}
S_{\rm H}^\eta (\widehat{P} \| \widehat{\Sigma} ) \geq S_1 (\widehat{P} \| \widehat{\Sigma} ).
\label{c_Stein1}
\end{equation}
On the other hand, let $Q_n$ be the optimal candidate for the minimization in Eq.~(\ref{classical_hypothesis_divergence}) satisfying $\rho_n [Q_n] \geq \eta$. For any $0 < \varepsilon < \eta$, we can show in the same manner as Lemma 11.8.1 of Ref.~\cite{Cover_Thomas} that
\begin{eqnarray}
\sigma_n [Q_n ]  &\geq& \sigma_n [ Q_n \cap Q_n^\varepsilon ]  \geq \rho_n [ Q_n \cap Q_n^\varepsilon  ] \exp \left( - n (  S_1 (\widehat{P} \| \widehat{\Sigma}  ) + \varepsilon ) \right) \\
&\geq&  \left(  \rho_n [ Q_n] - (1 -  \rho_n [Q_n^\varepsilon  ]  ) \right) \exp \left( - n (  S_1 (\widehat{P} \| \widehat{\Sigma}  ) + \varepsilon ) \right) \\
&>& (\eta - \varepsilon ) \exp \left( - n (  S_1 (\widehat{P} \| \widehat{\Sigma}  ) + \varepsilon ) \right),
\end{eqnarray}
where $Q_n^\varepsilon$ is defined in  Proposition \ref{prop_c_AEP} and we used (a) and (b) there.
Therefore,
\begin{equation}
S_{\rm H}^\eta (\rho_n \| \sigma_n ) = - \ln \left( \frac{1}{\eta} \sigma_n [Q_n ] \right) <   n \left(  S_1 (\widehat{P} \| \widehat{\Sigma}  ) + \varepsilon \right) + \ln \frac{\eta}{\eta - \varepsilon}.
\end{equation}
By taking the limit,
$S_{\rm H}^\eta (\widehat{P} \| \widehat{\Sigma} ) \leq S_1 (\widehat{P} \| \widehat{\Sigma} ) + \varepsilon$.
Since $\varepsilon > 0$ can be taken  arbitrarily small, we have
\begin{equation}
S_{\rm H}^\eta (\widehat{P} \| \widehat{\Sigma} ) \leq S_1 (\widehat{P} \| \widehat{\Sigma} ).
\label{c_Stein2}
\end{equation}
From inequalities (\ref{c_Stein1}) and (\ref{c_Stein2}), we obtain Eq.~(\ref{c_Stein_lemma}).
$\Box$
\end{proof}


\end{document}